%% file: main.tex
\documentclass[
11pt, %
english, %
singlespacing, %
headsepline, %
]{MastersDoctoralThesis} %

\newif\ifprintall
\printallfalse
\printalltrue

\usepackage[utf8]{inputenc} %
\usepackage[T1]{fontenc} %

\usepackage{mathpazo} %

\usepackage[backend=bibtex,style=authoryear,natbib=true]{biblatex}%

\usepackage{xcolor}
\usepackage{amssymb}
\usepackage{amsmath}
\usepackage{amsthm}
\usepackage{multirow}
\usepackage{arydshln}
\usepackage[switch]{lineno} 
\usepackage[misc]{ifsym}

\newcommand{\comment}[1]{}
\newtheorem{prop}{Proposition}
\newlength{\lgCase}

\usepackage{collcell}
\usepackage{booktabs}
\usepackage{siunitx}
\usepackage{etoolbox}
\usepackage{ifthen}
\usepackage[caption=false]{subfig}
\usepackage{rotating}
\usepackage{changepage}
\usepackage{cancel}
\usepackage{caption}
\usepackage{wasysym}
\usepackage{arydshln}
\usepackage{enumitem}
\usepackage[acronym,toc]{glossaries}
\usepackage[normalem]{ulem}
\makeglossaries
\input{Misc/Glossaire}

\usepackage{shorttoc}
\usepackage[subfigure]{tocloft}
\usepackage{hyperref}

\robustify\bfseries
\newcommand{\maxf}[1]{{\bfseries #1}}
\newcommand{\rd}[1]{{\color{red} #1}}
\newcommand{\bl}[1]{{\color{blue} #1}}
\newcommand{\gr}[1]{{\color{green} #1}}
\newcommand{\ora}[1]{{\color{orange} #1}}

\newcommand{\listequationsname}{List of Equations}
\newlistof{myequations}{equ}{\listequationsname}
\newcommand{\myequations}[1]{%
\addcontentsline{equ}{myequations}{\protect\numberline{\theequation}#1}\par}

\usepackage{titlesec}
\titleformat{\paragraph}
{\normalfont\normalsize\bfseries}{}{1em}{}
\titlespacing*{\paragraph}
{0pt}{3.25ex plus 1ex minus .2ex}{1.5ex plus .2ex}

\renewcommand{\thechapter}{\Roman{chapter}}

\usepackage[autostyle=true]{csquotes} %

\thesistitle{Interactions in Information Spread} %
\supervisor{\acrshort{Prof.} Julien \textsc{Velcin}\\\acrshort{Prof.} Sabine \textsc{Loudcher}} %
\examiner{} %
\degree{Doctor of Philosophy} %
\author{Gaël \textsc{Poux-Médard}} %
\addresses{} %

\subject{Applied mathematics and Computer Science} %
\university{\href{https://www.universite-lyon.fr}{Université de Lyon}} %
\department{\href{http://edinfomaths.universite-lyon.fr/l-ed-infomaths/}{Doctoral school Computer Science and Mathematics}} %
\group{\href{https://eric.msh-lse.fr}{ERIC lab}} %
\faculty{\href{}{}} %

\AtBeginDocument{
\hypersetup{pdftitle=\ttitle} %
\hypersetup{pdfauthor=\authorname} %
\hypersetup{pdfkeywords=\keywordnames} %
}

\begin{document}

\frontmatter %

\pagestyle{plain} %

\begin{titlepage}
\begin{center}

\vspace*{.06\textheight}
{\scshape\LARGE \univname\par}\vspace{1.5cm} %
\textsc{\Large Doctoral Thesis}\\[0.5cm] %

\vspace{1cm}

\HRule \\[0.4cm] %
{\Huge \bfseries \ttitle\par}\vspace{0.4cm} %
\HRule \\[1.5cm] %

{\centering \huge \href{https://gaelpouxmedard.github.io}{\authorname}}

\vspace{2.5cm}

{\centering \large \emph{Supervisors:}\\{\supname}}

\vspace{0.2cm}

{\centering \large \emph{Jury:}\\
\acrshort{Prof.} Christine \textsc{Largeron} (Rapporteur)\\
\acrshort{A.R.Prof.} Camille \textsc{Roth} (Rapporteur)\\
\acrshort{D.R.} Pierre \textsc{Borgnat}\\
\acrshort{Prof.} Fabrice \textsc{Rossi}}

\vspace{1cm}

\large \textit{A thesis submitted in fulfillment of the requirements\\ for the degree of \degreename}\\[0.3cm] %
\textit{in the}\\[0.4cm]
\groupname\\\deptname\\[2cm] %
 
\vfill

{\large \today}\\[4cm] %
 
\vfill
\end{center}
\end{titlepage}

\cleardoublepage

{
\ifprintall

\vspace*{0.2\textheight}

\vspace{1cm}

\noindent\enquote{
\itshape Psychohistory was the quintessence of sociology; it was the science of human behavior reduced to mathematical equations. The individual human being is unpredictable, but the reactions of human mobs could be treated statistically. [...]

\vspace{0.2cm}

\noindent
The Three Theorems of Psychohistorical Quantitivity: 
\begin{itemize}
    \item[-] The population under scrutiny is oblivious to the existence of the science of Psychohistory. 
    \item[-] The time periods dealt with are in the region of 3 generations. 
    \item[-] The population must be in the billions for a statistical probability to have a psychohistorical validity.
\end{itemize}   
}\bigbreak
\hfill Isaac Asimov, \textit{Foundation}, 1942

\begin{abstract}
\addchaptertocentry{\abstractname} %

Since the development of writing 5~000 years ago, human-generated data gets produced at an ever-increasing pace. This rate has been greatly influenced by technical innovations, such as clay tablets, papyrus, paper, press, and more recently the Internet. At the same time, new methods designed to handle and archive these growing information flows emerged: clay archives (Nippur, Mari), early libraries (Alexandria, Rome's Tabularia, Athens' Metroon), religious scriptoriums (abbeys, monasteries), modern libraries and, more recently, machine learning. Each of these archival methods aims at easing information retrieval.

Nowadays, archiving is not enough anymore. The amount of data that gets generated daily is beyond human comprehension, and appeals for new information retrieval strategies.
Instead of referencing every single data piece as in traditional archival techniques, a more relevant approach consists in understanding the overall ideas conveyed in data flows. To spot such general tendencies, a precise comprehension of the underlying data generation mechanisms is required.

In the rich literature tackling this problem, the question of information interaction remains nearly unexplored. Explicitly, few works explored the influence of anterior human-generated data on ulterior data creation mechanisms. In this manuscript, we develop a panel of new machine learning methods that explore this specific aspect of online data generation.

First, we investigate the frequency of such interactions. Building on recent advances made in Stochastic Block Modelling, we explore the role of interactions in several social networks. We find that \textbf{interactions are rare} in these datasets. 

Then, we wonder how interactions evolve over time. Earlier data pieces should not have an everlasting influence on ulterior data generation mechanisms
; an ad may exert a short-term influence on buying behaviours, but would have no influence on them a decade later for instance. 
We model this using dynamic network inference advances on social media datasets. We conclude that \textbf{interactions are brief} and that their intensity typically decays in an exponential fashion.

Finally, as an answer to the previous points, we design a framework that jointly \textbf{models rare and brief interactions}. Doing so, we exploit a recent bridge between Dirichlet processes and Point processes. We improve on this advance and discuss the more general Dirichlet-Point processes. We argue that this new class of models readily fits brief and sparse interaction modelling. We conduct a large-scale application on Reddit and find that \textbf{interactions play a minor role} in this dataset. 

From a broader perspective, our work results in a collection of highly flexible models and in a rethinking of core concepts of machine learning. Consequently, we open a range of novel perspectives both in terms of real-world applications and in terms of technical contributions to machine learning.

\end{abstract}

\begin{acknowledgements}
\addchaptertocentry{\acknowledgementname} %

Au terme de cette thèse, tant de gens méritent d'être reconnus pour avoir su me supporter, dans tous les sens du terme. Je vais tenter ici de dépeindre ce cadre, si agréable, à l'intérieur duquel j'ai pu mener les travaux qui composent cet ouvrage.

En tout premier lieu, il convient évidemment de remercier les personnes sans qui cette thèse n'aurait pu avoir eu lieu, mes directeurs de recherche Julien et Sabine. En outre, et en dehors des convenances cette fois, j'aimerais sincèrement les remercier pour la confiance qu'ils m'ont accordée lors de ces trois années de thèse. Cette confiance, qui s'est notamment exprimée au travers de l'autonomie dont j'ai bénéficié, et d'un réel intérêt pour mes productions. Cette confiance qui m'aura permis de mener à bien mon projet de recherche sous une égide qui m'est chère, l'indépendance.

En étendant ce cadre, on retrouve les personnes qui ont réussi à rendre une vie de laboratoire pourtant ponctuée de confinements et d'isolements, stimulante et enrichissante malgré tout. Je tiens donc à sincèrement remercier Antoine, Arwa, Clément, Enzo, Habiba, Jean, Loïc, Margot, Martial, Robin, Adrien, (un autre) Antoine, Camille, Jairo, (un autre) Julien, et Stéphane.

En élargissant encore ce cadre, on trouve les personnes qui, sans avoir directement pris à part aux travaux présentés ici, leur ont tout de même permis de voir le jour, par le soutien moral qu'ils ont su m'apporter --en affection, en présence et en houblon. Merci Téo, Tony, Nora, Noémie, Nadir, Léa, Lucie, Lucas, Kenza, Delphine, Cyril, Claim, Chloé, Benoît, Audrey, (une autre) Audrey, et Amine !

Enfin, pour des raisons évidentes provenant du fond de mon cœur, qui s'ajoutent à la plupart des remerciements précédents, un grand merci à ma famille, merci Elwenn, merci Lucille, merci Maman, merci Papa, grazie Manuela, merci (une autre !) Audrey, merci mes grands-parents.

Afin que le cadre soit complet, je me dois pour finir de remercier les personnes qui, sans le savoir, m'ont mis sur la voie que j'emprunte aujourd'hui. Je me contenterai d'une citation que je juge pertinente au regard de ma situation : \emph{Si je devais résumer ma vie aujourd'hui, je dirais que c'est d'abord des rencontres. Des gens qui m'ont tendu la main [...]. Et c'est assez curieux de se dire que les hasards, les rencontres forgent une destinée.} Ainsi, j'aimerais remercier (la même) Chloé, qui entre autres choses aura eu ce bon goût de me conseiller de lire Asimov, et qui en une simple phrase m'a mené cinq ans plus tard à rédiger ce manuscrit. Gràcies Marta, qui m'a permis à la fois de concrétiser cette petite idée née de la lecture précédente, mais également de m'avoir mis le pied à l'étrier de la recherche il y a quatre ans, et de m'y avoir donné goût par la même occasion. Enfin, merci à ces personnes dont je connais l'existence mais ignore le nom, qui en m'ouvrant une porte il y a cinq et trois ans, ou en m'en claquant une au nez il y a huit, six, cinq \textit{et} trois ans, m'ont permis de guider mes pas là où j'en suis aujourd'hui.

Merci.

\end{acknowledgements}

\else
\fi 
}

\shorttoc{Contents}{2}

\mainmatter %

\pagestyle{thesis} %

\include{Chapters/1_Introduction}

\include{Chapters/2_SBMs}

\include{Chapters/3_InterRate}

\include{Chapters/4_DHPs}

\include{Chapters/5_Conclusion}

\printbibliography[heading=bibintoc]

\appendix %

\include{Appendices/Appendix-2}
\include{Appendices/Appendix-3}
\include{Appendices/Appendix-5}

\listoffigures %
\addchaptertocentry{List of figures}

\listoftables %
\addchaptertocentry{List of tables}

\listofmyequations
\addchaptertocentry{List of equations}

\printglossary[type=\acronymtype, toctitle=Acronyms]

\printglossary[toctitle=Glossary]

\clearpage
\setcounter{tocdepth}{4}
\tableofcontents %
\addchaptertocentry{Full table of contents}

\clearpage
\thispagestyle{empty}
\vspace*{\fill}
\begin{center}
\qedsymbol
\end{center}
\vfill %
\clearpage
\thispagestyle{empty}
\ 
\clearpage

\end{document}

%% file: Misc/Glossaire.tex
\newacronym{SBM}{SBM}{Stochastic Block Model}
\newacronym{MMSBM}{MMSBM}{Mixed Membership Stochastic Block Model}
\newacronym{SIMSBM}{SIMSBM}{Serialized Interacting Mixed Membership Stochastic Block Model}
\newacronym{IMMSBM}{IMMSBM}{Interacting Mixed Membership Stochastic Block Model}
\newacronym{Bi-MMSBM}{Bi-MMSBM}{Bipartite Mixed Membership Stochastic Block Model}
\newacronym{T-MBM}{T-MBM}{Tensorial Mixed Membership Stochastic Block Model}

\newacronym{CRP}{CRP}{Chinese Restaurant Process}
\newacronym{UP}{UP}{Uniform Process}
\newacronym{PY}{PY}{Pitman-Yor Process}
\newacronym{DP}{DP}{Dirichlet Process}
\newacronym{PDP}{PDP}{Powered Dirichlet Process}
\newacronym{DHP}{DHP}{Dirichlet-Hawkes Process}
\newacronym{PDHP}{PDHP}{Powered Dirichlet-Hawkes Process}
\newacronym{MPDHP}{MPDHP}{Multivariate Powered Dirichlet-Hawkes Process}

\newacronym{NMF}{NMF}{Non-negative Matrix Factorization}
\newacronym{SVD}{SVD}{Singular Value Decomposition}
\newacronym{CP}{CP}{Canonical Polyadic Decomposition}
\newacronym{EM}{EM}{Expectation-Maximization Algorithm}
\newacronym{TF}{TF}{Tensor Factorization}
\newacronym{KNN}{KNN}{K-nearest-neighbours}
\newacronym{NB}{NB}{Naive Bayes}
\newacronym{BL}{BL}{Baseline}

\newacronym{P@}{P@}{Precision at}
\newacronym{AUCROC}{AUCROC}{Area Under the Receiving Operator Curve}
\newacronym{AUCPR}{AUCPR}{Area Under the Precision-Recall Curve}
\newacronym{RAP}{RAP}{Rank Average Precision}
\newacronym{AP}{AP}{Average Precision}
\newacronym{NCE}{NCE}{Normalized Coverage Error}
\newacronym{RSS}{RSS}{Residual Sum of Squares}
\newacronym{JS}{JS}{Jensen-Shannon divergence}
\newacronym{MSE}{MSE}{Mean Squared Error}
\newacronym{RMSE}{RMSE}{Root Mean Squared Error}
\newacronym{MAE}{MAE}{Mean Absolute Error}
\newacronym{NMI}{NMI}{Normalized Mutual Information}

\newacronym{RBF}{RBF}{Radial basis function kernel}
\newacronym{SMC}{SMC}{Sequential Monte-Carlo}

\newacronym{Prof.}{Prof.}{Professor}
\newacronym{A.R.Prof.}{A.R.Prof.}{Associate Research Professor}
\newacronym{D.R.}{D.R.}{Director of Research}

\newglossaryentry{independent}
{
    name=independent,
    description={As opposed to interacting, when there is no interaction; the spreading entities do not affect each other. The whole equals the sum of its parts.}
}

\newglossaryentry{interacting}
{
    name=interacting,
    description={As opposed to independent, when there is an interaction; the spreading entities can affect each other. The whole does not equal the sum of its parts.}
}

\newglossaryentry{interaction}
{
    name=interaction,
    description={When several entities are interacting, name given when the whole does not equal the sum of its parts.}
}

\newglossaryentry{infected}
{
    name=infected,
    description={When spreaders take a decision on an entity, they are said to be infected by this entity. Infection denotes the transition between a ``Susceptible'' state and an ``Infected'' state.}
}

\newglossaryentry{entity}
{
    name=entity,
    plural=entities,
    description={Same as ``piece of information''. Any object that is susceptible to have an influence on spreaders. It carries a semantic meaning. For instance: a tweet, a meme, a song, a virus, a news article, etc.}
}

\newglossaryentry{cascade}
{
    name={cascade},
    description={A series of identical decisions taken by different spreaders about an entity. For instance, a tweet that gets retweeted n times makes a cascade of size n+1.}
}

\newglossaryentry{decision}
{
    name={decision},
    description={The action of a spreader when facing an entity. The decision can be endogenous (the spreader acts on the entity alone, e.g. by sharing it, liking it, clicking on it, reacting to it, etc.) or exogenous (the spreader creates a new entity as a consequence of the first one, e.g. a denial, an answer to a mail, to a tweet, etc.).}
}

\newglossaryentry{spread}
{
    name={spread},
    description={An endogenous decision that may infect other spreaders.}
}

\newglossaryentry{outcome}
{
    name={outcome},
    description={A possibly exogenous decision that has been taken given a specific context.}
}

\newglossaryentry{diffusion}
{
    name={diffusion},
    description={A set of individual decisions. For instance: a retweet cascade, an internet buzz, a Youtube trend, etc.}
}

\newglossaryentry{cluster}
{
    name={cluster},
    description={Group formed from a set of objects in such a way that objects in the same group are more similar (in some sense) to each other than to those in other groups.}
}

\newglossaryentry{membership}
{
    name={membership},
    description={The extent to which an entity belongs to a cluster.}
}

\newglossaryentry{type}
{
    name={type},
    description={Entities of the same type have an identical way of expressing their semantic meaning. For instance, two tweets are of the same type ("Tweet"), two news articles are of the same type ("News"), and a tweet and a news article \textit{can} be of the same type ("Textual documents") or not depending on modeling choices.}
}

\newglossaryentry{content}
{
    name={content},
    description={Mean through which entities carry a semantic meaning: text, pixels, sounds, etc.}
}

\newglossaryentry{virality}
{
    name={virality},
    description={Probability of a decision on a piece of information in the absence of interactions.}
}

\newglossaryentry{rich-get-richer}
{
    name={rich-get-richer},
    description={In clustering, a property that gives an entity a higher probability of belonging to the most populated clusters and a lesser probability to belong to the less populated ones.}
}

\newglossaryentry{interaction profile}
{
    name={interaction profile},
    plural={interaction profiles},
    description={A tool for visualizing the strength of interactions between entities as the time separating them grows.}
}

\newglossaryentry{collaborative filtering}
{
    name={collaborative filtering},
    description={A method of making automatic predictions (filtering) about the interests of a user by collecting preferences or taste information from many users (collaborating). Its underlying assumption is that if a person A has the same opinion as a person B on an issue, A is more likely to have B's opinion on a different issue than that of a randomly chosen person.}
}

\newglossaryentry{piece of information}
{
    name={piece of information},
    plural={pieces of information},
    description={Same as ``entity''. Any object that is susceptible to have an influence on spreaders. It carries a semantic meaning. For instance: a tweet, a meme, a song, a virus, a news article, etc.}
}

%% file: Chapters/1_Introduction.tex
\chapter{Introduction}
\label{Chapter-Introduction}

\section{General considerations}
With the advent of the Internet, society realized that starting a manuscript as ``With the advent of the Internet'' is a banality at best, and obsolete at worst. It has now been thirty years that the amount of available online data grows exponentially. It has been twice as much that tools to automatically handle large corpora began to be developed. The impact of the Internet on societies is now a well-established fact. We know that large flows of data stream through it every second. We know the importance of developing automated means to make sense of these massive datasets. We know how crucial the understanding of the underlying mechanisms from which data emerges is. Or do we?

\subsection{About these large flows of data}
Most Internet users have at least an idea of how little of the total information flows appears on their usual platforms, be it Facebook, Reddit, Twitter, or any user generated \gls{content} platform. Understanding the voices of 
5 billion Internet users is not a human task. What most Internet users do not know, however, is how much data this represents. As an analogy, think of this child's dream of reading every book on the planet –which was eventually doomed after the invention of the press. Taking the 200,000 million books stored at the British Library as a good estimate of the available literature, such a task would imply reading roughly 5,000 books a day for 80 straight years. This amount of information represents a rough estimate of 150 terabytes of textual data. The task seems colossal, even by automated means -- the largest textual model to date GPT-3 is trained on 45TB of text data. However, confronting today's reality, the same amount of textual data gets published on Twitter over the course of a year and a half.
It remains a pretty long time, given the same amount of text gets sent over Whatsapp \textit{every single day}. Moreover, these numbers are for text only and do not account for other types of \gls{content}, such as audio, video, or images.

\subsection{About the automated means to make sense out large corpora}
To ``make sense out of the data'' can have various interpretations depending on the studied object. It is often used as a shortcut in scientific articles' abstracts. A quick query on any scientific search engine shows this expression systematically refers to a different aspect of data modelling: understanding tumour growth from medical reports, boosting a company's value from utility data, identifying depression in a pile of text messages, provoking Internet buzzes, etc. The common feature of these examples is that data are used as a means of an application. To ``make sense'' out of large datasets is to be understood as ``make them usable'' or ``describe them in-depth''. Being able to gather 150TB of Whatsapp messages a day is useless unless we have an application for it --be it descriptive or applied, medical or commercial, etc. Making sense of datasets boils down to extracting elements of interest from them. A lower-level description of ``making sense'' of it implies defining numerical quantities to compare data elements to each other. One of such quantities that will be extensively discussed in this manuscript is \glspl{entity}' group \gls{membership}; a favoured way to describe datasets is to group its elements into \glspl{cluster} and analyse data at a more tractable level. \Glspl{cluster} containing different data elements are likely to tell different stories.

\subsection{About understanding underlying data-generation mechanisms}
This point makes use of the two previous ones. We would like to develop models that explain the emergence of data. For instance, I tweeted this news about the last Disney movie because I wanted to, as the result of a free choice. However, there likely is more to it. I may have tweeted because I heard a Disney song in a supermarket, because a Facebook friend talked about it a few days earlier, because I felt like eating popcorn in the dark after seeing an ad on TV and went to see a random movie, or because \textit{someone} got influenced on her side and asked me to go. This \gls{decision} has certainly been influenced by these previous exposures and their \glspl{interaction} (see Fig.~\ref{fig-freewill}) --the extent of this influence appeals to philosophical notions about free will that are not discussed in the manuscript. Either way, the \gls{decision} could be \textit{explained} by underlying influence mechanisms --advertisement, social relations, hunger, social relations again. At the scale of individuals, understanding these processes is intractable. However, at the scale of 5 billion Internet users, statistics may be sufficient to accurately model what happens under the hood. This is what we call \textit{information \gls{spread}}. A piece of influence travels from one person to the other (people talking), from a media to a person (people browsing the Internet), or from a media to a media (news replication). Understanding the underlying mechanisms of information \gls{spread} can be considered from the angle of individual transmissions: who/what \gls{spread} what to who/what. On this matter, many works have modelled how \glspl{piece of information} \gls{spread} from one \gls{entity} to another. Understanding these mechanisms is crucial for moderation purposes. For instance: how to surgically stop the \gls{spread} of fake news by blocking a few spreader accounts, or by diffusing a denial from strategic spreaders?; how to nudge populations towards healthier behaviours by broadcasting the right \gls{content} on the right platform at the right time --which was a central objective of Obama's Nudge Unit?; how to advertise a product using selected spreaders (or \textit{influencers} in this case) that maximize buys?; etc. 

\begin{figure}
    \centering
    \includegraphics[width=\textwidth]{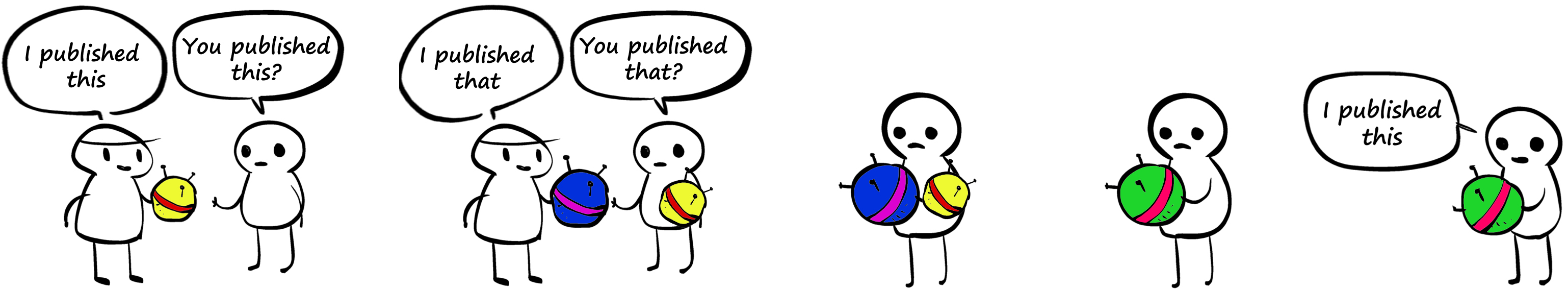}
    \caption[Intro - Information \gls{spread} and \gls{interaction}]{A cartoon illustrating how the \gls{interaction} between previous events can condition present events.}
    \label{fig-freewill}
\end{figure}

\section{Motivations}
The point about understanding underlying data-generation mechanisms is the one driving the present work. As stated in the title, we propose to explore how \glspl{interaction} play a role in information \gls{spread}. A rigorous definition of the \glspl{interaction} considered in this manuscript will be given further in the introduction (Section~\ref{Intro-definitions}); the key point here is that ``\textit{\gls{interacting}}'' is opposed to ``\textit{independent}''. That said, it appears that most works on information \gls{spread} consider spreading \glspl{entity} that are \gls{independent} of each other. This is poorly illustrated in Fig.~\ref{fig-indepcasc}.

\begin{figure}
    \centering
    \includegraphics[width=\textwidth]{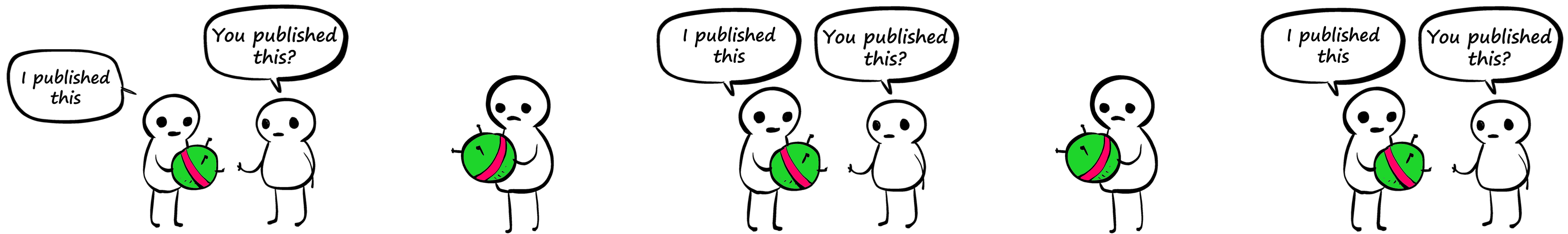}
    \caption[Intro - \Gls{independent} \gls{spread}]{\textbf{\Gls{independent} \gls{spread} assumption} --- Pieces of information \gls{spread} and replicate independently from each other.}
    \label{fig-indepcasc}
\end{figure}

\subsection{Most existing models do not consider information \gls{interaction}}
\paragraph{\Gls{independent} \Gls{cascade} model}
A seminal work that illustrates this paradigm is called the \Gls{independent} \Gls{cascade} model \citep{Kempe2003IndendendentCascade}. In this literature, a node (or user) is said to be \textit{\gls{infected}} when she acts on a \textit{\gls{piece of information}} (retweeting a tweet, liking a post, commenting on news, etc.).
In this work, an initial set of users is \textit{\gls{infected}}, meaning they act as initial spreaders for a given \gls{piece of information} (e.g., a virus, a news, a tweet, etc.). Each spreader has a single chance to infect each of its neighbours in a network. After deciding whether each node gets \gls{infected} by its contagious neighbour(s), the time goes forward, and the process repeats. Another model proposed in \citep{Kempe2003IndendendentCascade} is the Linear Threshold model. The process is roughly the same, except that a node gets \gls{infected} after repeated exposures to contagious neighbours. At each time step, the viral charge of a node increases according to its neighbour’s infection status, and to the strength of the link between them. In these processes, each node is \gls{infected} conditionally to its links to other nodes in the network, but not according to the \textit{\gls{cascade}} of infections flowing through the network. Two \glspl{cascade} spreading simultaneously on a network are assumed to be \gls{independent}.

\paragraph{Extensions and other \gls{independent}-spreading models}
Several more recent works could have illustrated that most research is oriented towards considering \gls{independent} \glspl{cascade} spreading on networks. In \citep{Saito2009ContinuousTimeIndependentCascade}, the authors propose to extend \citep{Kempe2003IndendendentCascade} to model the \gls{diffusion} process in continuous time (instead of considering time steps). In \citep{Larremore2012AvalancheCascade}, the authors propose the avalanche \gls{cascade} model, where one \gls{piece of information} spreads from a single node. This work has been extended in \citep{Poux2020InfluentialSpreaders} to highlight the importance of network structure in the description of \gls{independent} \gls{diffusion} processes. In \citep{Bourigault2016ApprentissageDR}, the authors propose to embed spreaders in a latent space and to model information \gls{diffusion} as heat \gls{diffusion} in this latent space. More elaborate models have been proposed to model \gls{diffusion} processes. Some propose to consider nodes' metadata in the modelling \citep{Saito2011SpreadUserAttr}, other to model the spreading processes according to the spreading \gls{content} \citep{Barbieri2012TopicawareSI,Du2013TopicCascade}, or to do both jointly \citep{Lagnier2013SpreadContentAndUsrProfile}. In all these works, spreading processes are \gls{independent} of each other.

\paragraph{Inferring the network from \gls{independent} \glspl{cascade}}
Reversing the problem, several works proposed to infer the underlying spreading network from the infection \glspl{cascade}. In \citep{GomezRodriguez2011NetRate}, the authors develop the NetRate model that considers infection timestamps to recover the underlying \gls{diffusion} network using this single \gls{piece of information}. With InfoPath \citep{GomezRodriguez2013InfoPath}, they extend their previous method to model time-varying \gls{diffusion} networks, and finally in \citep{GomezRodriguez2013SurvivalAnalysis} they generalize their previous works NetRate and InfoPath in a single survival analysis framework, as well as several ulterior works based on their methodology --MoNet \citep{Wang2012Monet}, KernelCascade \citep{Du2012KernelCascade}. The method has then been extended to consider the spreading \glspl{entity}' \gls{content} \citep{Du2013TopicCascade,Barbieri2017SurvivalFactorization}. However, here again, each spreading process is modelled independently from the others. Two \glspl{cascade} occurring jointly will not affect each other i.e., there is no \gls{interaction} between pieces of information.

\subsection{Should we consider information \gls{interaction}?}
Before devoting --likely significant-- efforts to consider information \gls{interaction} in all the models introduced in this section, we should first conclude on its importance in spreading processes. Answering this question is the guiding thread of the whole work described in this manuscript. 

Studying the role of information \gls{interaction} can be broken down into three main parts --discussed at the very beginning of this manuscript. 
Firstly, we must \textbf{define and characterize information \glspl{interaction} in the large flows of data}. How to spot them, how to measure them, are there particular challenges to solve before being able to study them, etc.?
Secondly, we can develop usable models \textbf{to make sense of large datasets}. Do we improve results on various tasks by considering \glspl{interaction}, do they have a significant role in our corpora, etc.?
Finally, once information \glspl{interaction} have been characterized and shown to influence models' results, come the \textbf{understanding of the unveiled mechanisms} at stake. Providing \glspl{interaction} exist, where do they occur, can we unveil unexpected \gls{interaction} patterns, do these results improve our understanding of spreading processes as a whole, etc.?

Providing insights on \glspl{interaction} through the lens of those guiding questions is the overall motivation of this manuscript. If \glspl{interaction} indeed play a significant role in information \gls{spread}, the impact could be broad, as a whole class of information \gls{spread} literature would have to be rethought from the perspective of \glspl{interaction} modelling. If \glspl{interaction} are shown to have lesser importance, our work would spare ulterior efforts on this problem. However, as it is the norm in research, we do not expect our conclusions to be so definitive and general. Therefore, our motivation is mainly about providing a methodology and models for \gls{interaction} investigation studies, as well as about glimpsing a global conclusion on \glspl{interaction} in spreading processes with a case study. 

\section{Landscape of information \gls{interaction} modelling}
At the end of his --seminal-- PhD thesis on the dynamics of \gls{diffusion} networks, M. Gomez-Rodriguez states ``\textit{we have assumed contagions to propagate independently. However, this is over-simplistic, as noticed recently \citep{Prakash2012WinnerTakesAll,Myers2012CoC}. It would be interesting to relax this assumption.}'' \citep{Rodriguez2013StructureAD}. 

In this section, we review the main efforts that have been made for modelling \glspl{interaction} in information \gls{spread}. Our goal here is to brush a global landscape of what is done on the topic without entering the technical specifics. This section does not substitute to the more technical state-of-the-art that inaugurates each chapter of this manuscript.

\subsection{Theoretical studies}
Several models have been proposed to investigate how information competes for user attention. This first section will treat theoretical works. Authors essentially set up the supposed rules for \gls{diffusion} processes, simulate them, and compare them to ground-truth observations; there is no \textit{learning} from the data.

\paragraph{Information overload as the consequence of micro-\glspl{interaction}}
An early work on the topic \citep{Weng2012CompetitionMemes} develops the concept of \textit{information overload}. The overload is characterized by the entropy of the \textit{meme} topics a user has retweeted. A meme is a \gls{piece of information} that carries a semantic meaning on online platforms. Here, \glspl{interaction} are considered on a global scale: we do not know which are the \gls{interacting} \glspl{piece of information}, but we observe the overall effect of the \gls{interaction}, which is a user's information overload.
The authors consider users have a given affinity regarding certain memes and are exposed to them due to their ties to other people in a network. Affinity is defined as the Maximum Information Path measure \citep{Markines2009PathSimMeasure}. Finally, the authors simulate a \gls{diffusion} process accounting for both user interests and information overload on a synthetic Erd\"os-Renyi network \citep{Erdos1960ERgraph}. In their experiments, they achieve to manually tune their model parameters and recover aggregated measures that are similar to those observed on Twitter.
This work has been followed by a large-scale quantitative study to describe the role of information overload \citep{GomezRodriguez2014InformationOverload}. The authors find that the maximum information processing rate for tweeter users (in 2010) is 30 tweets per hour, beyond which a user is said to be overloaded. Overloaded users restrict their attention to specific information sources. The number of sources a user is likely to retweet reaches a threshold once the number of followees gets large. This study confirms the need for considering information overload in \gls{diffusion} processes. However, considering information overload does not allow to explain how the overload happens. \Glspl{interaction} here are considered on a global scale, but the underlying agent-based \gls{interaction} mechanisms remain unknown.

\paragraph{Modelling micro-\glspl{interaction}}
Other works proposed to tackle information \gls{interaction} modelling from a microscopic perspective. The idea is to consider each \gls{piece of information} individually and observe how it relates to every other spreading \gls{entity}. In \citep{Beutel2012InteractingViruses}, the authors propose a \gls{diffusion} model to infer information outbreaks under the assumption of pair \gls{interaction} between \glspl{piece of information}. The proportion of nodes \gls{infected} by information A, by information B, and by both A and B at time $t$ is described using a set of differential equations --similar to SIR-like models. A node \gls{infected} by a given \gls{piece of information} can inhibit this node's sensitivity to the other one. They illustrate it using web browser (e.g., Firefox, Opera, Chrome, etc.) adoption: in most cases, one user will use one main browser at the time and will be much less likely to download another one once \gls{infected}.
The competition between the spreading pieces of information is accounted for on a global scale --no network is considered. The authors derive some elementary properties of the process they defined and then show their equations approach real-world users’ behaviour on the adoption of either Firefox or Chrome as web browsers. To obtain these results, the \gls{interaction} parameter was tuned manually, as well as several other hyperparameters.

\paragraph{Modelling complex micro-\glspl{interaction}}
Recently, the authors in \citep{Zhu2020CompetitionInformationSocialNet} proposed a complex model to account for cooperation and competition
among information on social networks. This model considers \glspl{interaction} at a microscopic scale, meaning that the influence of each \gls{piece of information} on the others is considered. The authors jointly model several attributes: user affinity, information complexity, bot spreaders (nodes that \gls{spread} every information given to them), user memory, and social reinforcement. The authors define ad-hoc rules to model a \gls{spread} that could consider all these parameters and specifically study the results of the simulations. The proposed modelling reproduces some emergent effects from this micro-model. However, no extensive comparison to real-world data is done.

\paragraph{We should \textit{learn} models from the data}
Note that we presented only a snapshot of current research on theoretical \gls{interaction} modelling which is sufficient for our demonstration. Several other models tackle this problem in an analogous way \citep{Wang2019CoevolutionSpreadNet}. The fact that most works are done in a theoretical framework is decisive in the motivation of the present work. The models we just presented in this section first define the rules for \gls{interacting} processes, and then tune parameters to reproduce global-scale observations on synthetic networks. However, a single effect can arise from several causes. The global observed retweeting behaviour could arise from information \gls{interacting} processes, but it could also arise from other underlying processes --such as hidden ties from other media sources for instance \citep{Myers2012ExternalInfluence}. If most of those studies develop interesting models, we believe it is fundamental to follow the opposite process: first, we should learn the model from the data, and then analyse its characteristics, instead of first defining an ad-hoc model and only then comparing its results to the data. In the next section, we present works that tackle the problem of information \gls{interaction} modelling from a machine learning perspective.

\subsection{Data-driven studies}

\paragraph{Clash of the Contagion}
To our knowledge, Clash of the contagions \citep{Myers2012CoC} is the earliest attempt to learn the \gls{interaction} intensity from the data. This work extensively refers to \citep{Beutel2012InteractingViruses}, as it proposes to infer the pair-\gls{interaction} parameter described in the previous section. Their main motivation is the prediction of retweets on Twitter. It estimates the probability of retweeting a piece of information given the last tweets a user has been exposed to, according to their position in the Twitter feed. To do so, they define a block-model trained on quadruplets \textit{(tweet A, tweet B, $\Delta t$, tweet B retweeted?)} where $\Delta t$ represents the time separation between A and B. Tweets are first grouped into \glspl{cluster} --one different \gls{cluster} at each time slice--, and \glspl{cluster} interact with each other to determine whether yes or no the tweet B has been retweeted by the exposed user. The authors conclude that most \glspl{interaction} between tweets are weak, but that their overall effect cannot be neglected. 

However, the method suffers various flaws. Most importantly, it is based on a questionable hypothesis on the prior probability of a retweet (in the absence of \glspl{interaction}). The probability of retweets in the absence of \glspl{interaction} is defined as equal to the frequency of retweets. \Glspl{interaction} are defined on this ground. We show later in this manuscript that this assumption does not hold, and thus makes conclusions about information \glspl{interaction} sloppy. Other technical problems have been encountered during the replication of their results. Typically, the approach does not have convergence guarantees and is ill-defined -- typically because the model's ``probabilities'' are not constrained to be between 0 and 1. A note on our implementation to alleviate some of these problems is provided in Appendix, Section~\ref{InterRate-implemCoC}.

\paragraph{Correlated \glspl{cascade}}
Another work that tackles \gls{interaction} modelling from a machine learning point of view proposed the Correlated \Gls{cascade} model  \citep{Zarezade2017CorrelatedCascade}. In this work, the authors define a marked Hawkes process to model how existing \glspl{piece of information} condition the appearance of ulterior \glspl{piece of information} in a network. Explicitly, it models the rate at which each of two \glspl{piece of information} flow through the network's edges depending on the (non-)presence of each information type. The Hawkes process models the intensity of the flows between each user, and the \gls{interaction} term is tuned by a hyper-parameter $\beta$. 

The final aim is to infer the latent \gls{diffusion} network of an \gls{interacting} spreading process. If the \gls{interaction} parameter is not directly inferred as in \citep{Myers2012CoC}, its tuning indirectly relies on observation from the data, as it plays a role in the inference of the \gls{diffusion} network. The authors show that their model allows recovering global processes on real-world spreading processes on Twitter. In the conclusion, the authors formulate the open problem of learning several kernels and the \glspl{interaction} intensity parameter $\beta$, that we address in our work.

\paragraph{Further learning-based studies?}
In the previous section, we only presented a snapshot of the available works on information \gls{interaction} modelling from a theoretical perspective. However, to the best of our knowledge, the two models introduced in this section are quite an exhaustive list of machine learning efforts investigating this problem. This concurs with another recent survey on coevolution in information \gls{spread} \citep{Wang2019CoevolutionSpreadNet}, which only cites \citep{Myers2012CoC,Zarezade2017CorrelatedCascade} as examples of learning \glspl{interaction} in information \gls{spread}.

On one hand, \citep{Myers2012CoC} proposes to model the \gls{interaction} term between pair \glspl{interaction}, does not consider continuous-time modelling, suffers various formulation flaws and arguable assumptions. On the other hand, \citep{Zarezade2017CorrelatedCascade} models continuous-time processes and considers pair \glspl{interaction} but does not infer the \gls{interaction} parameter, which must be tuned. This leaves plenty of space for our work to fit in. As we will discuss in the next section, we will fill these gaps by developing a framework that models the \gls{interaction} parameters, over continuous times, and that can consider not only pair \glspl{interaction} but n-order \glspl{interaction} for a reasonable computational cost.

\subsection{Definitions}
\label{Intro-definitions}
In the works introduced in the section above, words such as ``\gls{interaction}'' or ``information'' can have many different meanings. For instance, \glspl{interaction} can take place between users, \glspl{piece of information}, a user and its environment, etc. In this section, we give a strict definition of key concepts present we use throughout this manuscript. In doing so, we frame our work into a specific research area. These concepts are illustrated in Fig.~\ref{fig-illustr-defs}. %

\begin{figure}
    \centering
    \includegraphics[width=\textwidth]{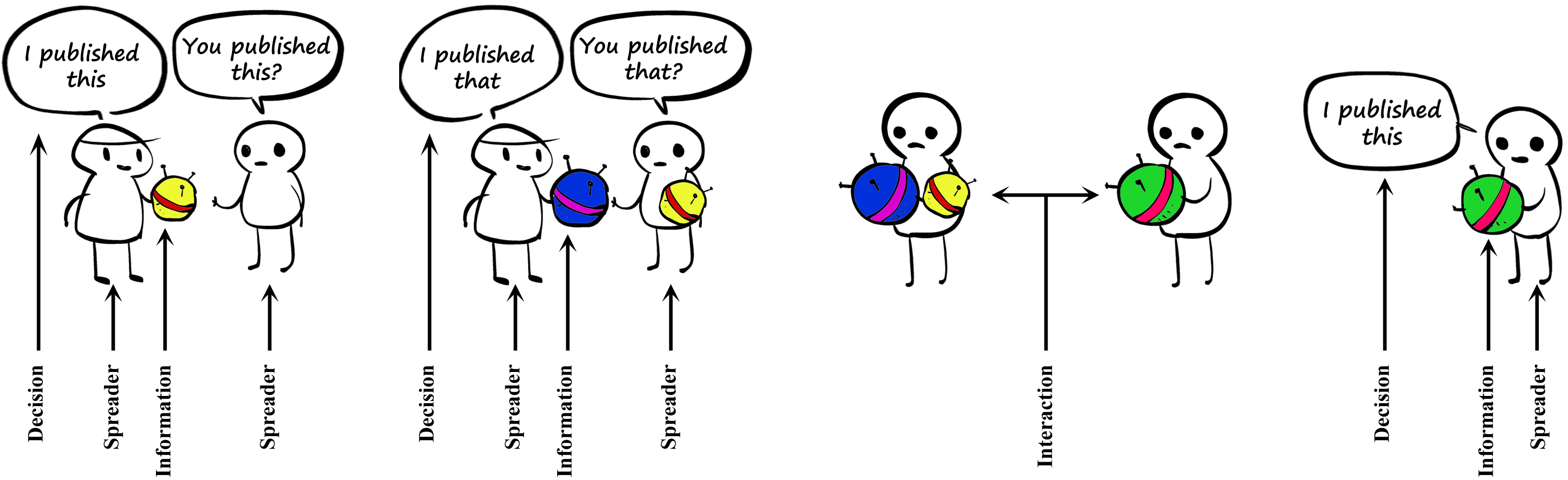}
    \caption[Intro - Illustration of the definitions]{\textbf{Illustration of the definitions} --- Spreader, Information, \Gls{decision}, and \Gls{interaction}.}
    \label{fig-illustr-defs}
\end{figure}

\textit{\textbf{Spreader}} --- An agent that is likely to \gls{spread} a \gls{piece of information} to other spreaders. In the context of social media, spreaders are often referred to as \textit{users}. For instance: a tweeter user, a member of a friends/family group, a journalist, an influencer, etc.

\textit{\textbf{Information, or \Gls{entity}}} --- Any object that is susceptible to have an influence on spreaders. It carries a semantic meaning. For instance: a tweet, a meme, a song, a virus, a news article, etc.

\textit{\textbf{\Gls{decision}}} --- The action of a spreader when facing an \gls{entity}. The \gls{decision} can be endogenous (the spreader acts on the \gls{entity} alone, e.g. by sharing it, liking it, clicking on it, reacting to it, etc.) or exogenous (the spreader creates a new \gls{entity} as a consequence of the first one, e.g. a denial, an answer to a mail, to a tweet, etc.). A \gls{diffusion} is a set of individual \glspl{decision}. For instance: a retweet \gls{cascade}, an internet buzz, a Youtube trend, etc.

\textit{\textbf{\Gls{interaction}}} --- When the joint effect of several \glspl{entity} does not equal the sum of their individual effect on spreaders. For instance, imagine a spreader that retweets (\gls{decision} of retweeting noted $x$) a tweet A given only A with 10\% chance, retweets the same tweet given only tweet B with 50\% chance, but also retweets A given tweet A \textit{and} tweet B with a 15\% chance. We say there is \gls{interaction} between A and B, because ${P(x \vert A)P(x \vert B) = 5\% \neq P(x \vert A,B) = 10\%}$. In this example, this \gls{interaction} raises the \gls{decision} probability by 5\%. We sketch an illustration for this definition in Fig.~\ref{fig-illustr-inter-intro}.

\begin{figure}
    \centering
    \includegraphics{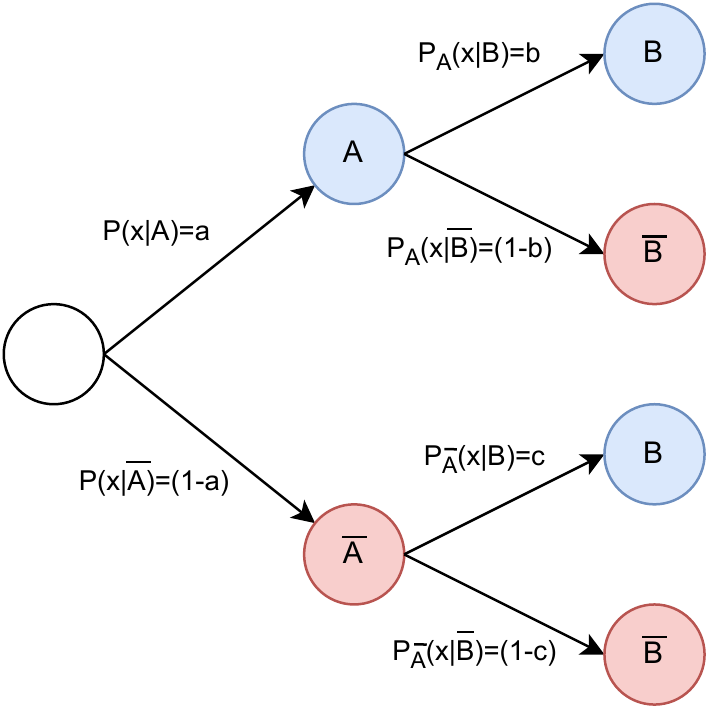}
    \caption[Intro - Definition of an \gls{interaction}]{\textbf{Definition of an \gls{interaction}} --- Probability tree for action $x$ given the presence of $A$ and $B$. We always have the following relation $P(x \vert A,B) = P(x \vert A) P_A(x \vert B)$. 
    Now if $b=c$, we have $P(x \vert B) =  P(x \vert A) P_A(x \vert B) + P(x \vert \bar A) P_{\bar A}(x \vert B) = P_A(x \vert B) = P_{\bar A}(x \vert B)$. The independence relation follows: $P(x \vert A,B) = P(x \vert A) P(x \vert B)$. If $b \neq c$, the independence relation does not hold; there is an \gls{interaction} between $A$ and $B$ regarding $x$.}
    \label{fig-illustr-inter-intro}
\end{figure}

\section{About this manuscript}
\subsection{Problematic}
We formulate our problem in the following terms:
\vspace{0.5cm}

\noindent
\textit{\textbf{\Large How to model \glspl{interaction} between pieces of information so as to evaluate their impact on spreading mechanisms?}}
\vspace{0.cm}

\vspace{0.2cm}
\noindent
We decompose this guiding thread into four more specific questions:
\vspace{0.2cm}

\noindent
\underline{\textbf{Question 1 --- How frequent are \glspl{interaction}?}}\\
Do significant \glspl{interaction} occur between each spreading \gls{entity} individually? Or are they rare? Can we define groups of \glspl{interaction}? How to model possibly rare phenomena?
\vspace{0.2cm}

\noindent
\underline{\textbf{Question 2 --- How persistent are \glspl{interaction}?}}\\
Do \glspl{interaction} last in time? How long does the influence of a \gls{piece of information} remain significant? How to unveil generic \gls{interaction} temporal trends?
\vspace{0.2cm}

\noindent
\underline{\textbf{Question 3 --- Can we efficiently model \glspl{interaction}?}}\\
Given an answer to our two first questions, can we develop a global model that would account for possibly rare \textit{and} brief \glspl{interaction}? Could such a model handle large corpora? What challenges arise in its development, and how to overcome them?
\vspace{0.2cm}

\noindent
\underline{\textbf{Question 4 --- Do \glspl{interaction} play a significant role in spreading processes?}}\\
Provided we develop a way to efficiently model \glspl{interaction}, we can conclude on the importance of \glspl{interaction} in real-world \gls{diffusion} processes. In which corpus is it relevant to consider information \gls{interaction}? Do \glspl{interaction} play a significant role in it? Do we unveil unsuspected \gls{interaction} patterns?
\vspace{0.2cm}

\subsection{Plan and contributions}
Our contributions in this manuscript are multiple. Overall, we develop a methodology to investigate the role of \glspl{interaction} in spreading processes. In doing so, we derive models that are based on radically different fields of machine learning. These models are applied to \gls{interaction} modelling but can serve much more global purposes. 
All our contributions in terms of models are summarized in Table~\ref{table-contrib}.

\begin{table}
    \caption[Intro - Contributions]{\textbf{Contributions presented in this manuscript} --- Our contributions are listed below in the order of their appearance in the text.}
    \label{table-contrib}
	\centering
    \noindent\makebox[\textwidth]{\resizebox{\textwidth}{!}{
	\begin{tabular}{|c|c|c|c|c|c|c|c|}
	\cline{1-8}
	& \hyperref[SIMSBM]{SIMSBM} & \hyperref[IMMSBM]{IMMSBM}& \hyperref[SDSBM]{SDSBM} & \hyperref[Chapter-InterRate]{InterRate} & \hyperref[PDP]{PDP} & \hyperref[PDHP]{PDHP} & \hyperref[MPDHP]{MPDHP}\\
	& Chap.~\ref{Chapter-SBMs} & Chap.~\ref{Chapter-SBMs} & Chap.~\ref{Chapter-SBMs} & Chap.~\ref{Chapter-InterRate} & Chap.~\ref{Chapter-DHPs} & Chap.~\ref{Chapter-DHPs} & Chap.~\ref{Chapter-DHPs} \\
	& Q1 & Q1 & Q1 & Q2 & Q3 & Q2,Q3 & Q1-Q4 \\
	\cline{1-8}
	Self-\glspl{interaction}    & x & x & x & x & x & x & x \\
	\cline{1-8}
	Pair-\glspl{interaction}    & x & x & x & x &   & x & x \\
	\cline{1-8}
	N-\glspl{interaction}       & x &  & x &   &   & x & x \\
	\cline{1-8}
	Clustering           & x & x & x &   & x & x & x \\
	\cline{1-8}
	Discrete time    &   &   & x & x &   & x & x \\
	\cline{1-8}
	Continuous time  &   &   &   & x &   & x & x \\
	\cline{1-8}
	Online inference     &   &   &   &   & x & x & x \\
	\cline{1-8}
	\end{tabular}
	}}
\end{table}

\subsubsection{Chapter~\ref{Chapter-SBMs} -- Stochastic Block Models (Question 1)}
Publications --- Section~\ref{IMMSBM}: \citep{Poux2021IMMSBM}
\vspace{0.2cm}

\noindent
In this chapter, we observe if there are regularities in the way \glspl{interaction} may happen by using recent advances in Stochastic Block Modelling (\acrshort{SBM}). We assume that subsets of \glspl{piece of information} exhibit similar \gls{interaction} patterns.

\vspace{0.3cm}

Section~\ref{SIMSBM} --- We derive a global framework that generalizes several existing Stochastic Block Models. It allows to consider an arbitrarily high number of labels as an input, and an arbitrarily high order of \gls{interaction} between them, to make predictions on multi-label outputs. 

Section~\ref{IMMSBM} --- We consider a special case of this global framework. We investigate the role of \glspl{interaction} on several social media datasets (Twitter, Reddit, and Spotify).
Our conclusions in this section are that most \glspl{interaction} are weak (or that significant \glspl{interaction} are rare), but that they play a non-negligible role on global scales. We emphasize the need for clustering \glspl{piece of information} to unveil global \gls{interaction} patterns.

Section~\ref{SDSBM} --- The models developed in the previous section are highly flexible but do not allow to model time. In this section, we will develop a simple way to incorporate time in the models discussed previously in the form of a Dirichlet prior on observations.

\subsubsection{Chapter~\ref{Chapter-InterRate} -- Temporal \gls{diffusion} networks (Question 2)}
Publications --- Section~\ref{InterRate-model}: \citep{Poux2021InterRate}
\vspace{0.2cm}

\noindent
In this chapter, we consider the temporal aspect of information \gls{interaction}.

\vspace{0.3cm}

Section~\ref{Chapter-InterRate} --- We develop a convex model that infers the temporal evolution of \glspl{interaction} influence. Typically, given two \glspl{entity} at different points in time, our model allows us to investigate how influential was the first \gls{entity} concerning a possible action on the second one. 
The model, baptized InterRate, is convex and can be run in parallel. 
We specifically apply InterRate to \gls{interaction} modelling on Twitter. We find that the most significant \glspl{interaction} happen immediately after the appearance of an influential \gls{piece of information}. 
The overall conclusion of this section is that \glspl{interaction} in spreading processes are brief. Typically, we find that the intensity of \glspl{interaction} on Twitter decreases exponentially as time goes forward. Besides, we find once again that most \glspl{interaction} are weak.

\subsubsection{Chapter~\ref{Chapter-DHPs} -- Dirichlet-Hawkes Processes (Question 3 and Question 4)}
Publications --- Section~\ref{PDHP}: \citep{Poux2021PDHP}
\vspace{0.2cm}

\noindent
In this chapter, we report our steps towards a model that answers the challenges raised in the two previous chapters. 

\vspace{0.3cm}

Section~\ref{PDP} --- We first take the Dirichlet Process (\acrshort{DP}) as a base, and generalize it as a special case of the Powered Dirichlet Process (\acrshort{PDP}). Specifically, it allows to control the importance of the ``\gls{rich-get-richer}'' assumption.

Section~\ref{PDHP} --- We then use the PDP and the Dirichlet-Hawkes Process (\acrshort{DHP}) as a base to develop the Powered Dirichlet-Hawkes Process (\acrshort{PDHP}). This model can model intra-\gls{cluster} temporal \gls{interaction}, meaning that information groups self-interact. It can also handle challenging situations where text is rare or where publication dynamics are intricate.

Section~\ref{MPDHP} --- We extend PDHP to the Multivariate Powered Dirichlet-Hawkes Process (\acrshort{MPDHP}) to model intra- \textit{and} extra-\gls{cluster} temporal \glspl{interaction}. The main output of this model is a temporal \gls{interaction} network between \glspl{cluster} of documents.

Section~\ref{MPDHP-Reddit} --- Finally, we apply the MPDHP to a large-scale real-world Reddit news dataset and conclude on the role of information \gls{interaction} in its underlying spreading processes.

\begin{table}[h]
    \caption[Intro - Codes and datasets]{\textbf{Codes and datasets}}
    \label{table-codesandDS}
	\centering
    \noindent\makebox[\textwidth]{\resizebox{\textwidth}{!}{
	\begin{tabular}{|c|l|c|}
	\cline{1-3}
	Model & Link & External datasets \\
	\cline{1-3}
    \shortstack{\hyperref[SIMSBM]{SIMSBM} \\ \vspace{0.08cm}} & \shortstack{\url{https://github.com/GaelPouxMedard/SIMSBM}\\\vspace{0.08cm}} & \shortstack{\citep{Harper2015ImdbDataset}\\\citep{Guttieres2016MrBanks}} \\
	\cline{1-3}
	\shortstack{\hyperref[IMMSBM]{IMMSBM} \\ \vspace{0.05cm}} & \shortstack{\url{https://github.com/GaelPouxMedard/IMMSBM} \\ \vspace{0.05cm}} & \shortstack{\citep{Hodas2014DataSetTwitter} \\ \citep{Baumgartner2020PushshiftRedditDataset} }\\
	\cline{1-3}
	\shortstack{\hyperref[SDSBM]{SDSBM} \\ \vspace{0.06cm}}  & \shortstack{\url{https://github.com/GaelPouxMedard/SDSBM} \\ \vspace{0.06cm}} & \shortstack{\citep{Kumar2019jodie}\\\citep{ClaussSlabyDataset}}\\
	\cline{1-3}
	\shortstack{\hyperref[Chapter-InterRate]{InterRate} \\ \\ \vspace{0.25cm}}& \shortstack{\url{https://github.com/GaelPouxMedard/InterRate} \\ \\ \vspace{0.2cm}} & \shortstack{\citep{Bereby2006PrisonerDilemaDataset}\\ \citep{Hodas2014DataSetTwitter}\\ \citep{Cao2019AdsDataset}}\\
	\cline{1-3}
	\hyperref[PDP]{PDP} & \url{https://github.com/GaelPouxMedard/PDP} & \citep{ClaussSlabyDataset}\\
	\cline{1-3}
	\hyperref[PDHP]{PDHP} & \url{https://github.com/GaelPouxMedard/PDHP} & \citep{Baumgartner2020PushshiftRedditDataset} \\
	\cline{1-3}
	\hyperref[MPDHP]{MPDHP} & \url{https://github.com/GaelPouxMedard/MPDHP} & \citep{Baumgartner2020PushshiftRedditDataset} \\
	\cline{1-3}
	\end{tabular}
	}}
\end{table}

\subsection{Reproducible research}

\noindent
\emph{``Non-reproducible single occurrences are of no significance to science''}, in \citetitle{Popper1935ReproductibleResearch} \citep{Popper1935ReproductibleResearch} 

\vspace{0.1cm}

\noindent
\emph{``We may say that a phenomenon is experimentally demonstrable when we know how to conduct an experiment which will rarely fail to give us statistically significant results''}, in \citetitle{Yates1935ReproductibleResearch} \citep{Yates1935ReproductibleResearch}

\vspace{0.1cm}

\noindent
\emph{``Improving the reliability and efficiency of scientific research will increase the credibility of the published scientific literature and accelerate discovery.''}, in \textit{A manifesto for reproducible science} \citep{Munafo2017ReproductibleResearch}

\vspace{0.3cm}

All the experiments presented in this manuscript are available on GitHub. The repositories contain the datasets along with the Python implementation used for running the experiments. Note that datasets which are not referenced here have been gathered by us --Spotify, PubMed, Reddit. We reference the material used for each of the models presented throughout this manuscript in Table~\ref{table-codesandDS}.

%% file: Chapters/2_SBMs.tex
\chapter{Stochastic Block Models -- \Glspl{interaction} are rare} %
\label{Chapter-SBMs} %

\begin{chapabstract}
A straightforward way to represent \glspl{interaction} is to embed them as a network. \Glspl{entity} are nodes of this network, and \glspl{interaction} between these \glspl{entity} are link between these nodes. These links can be of diverse types; they are labelled. The last years have seen a regain of interest in stochastic block modelling of such labelled networks, which canonical decomposition-based methods are not fit to tackle. 

\underline{Section~\ref{SBMs-intro}}, we first consider static stochastic block models and their use for \gls{interaction} modelling.
Existing models are not fit for modelling \glspl{interaction}: the number of \glspl{entity} that can interact is limited and higher-order \glspl{interaction} are not allowed, which is extremely restrictive when it comes to modelling highly \gls{interacting} systems. 

\underline{Section~\ref{section-static-interaction}}, to answer these limitations, we show in Section~\ref{SIMSBM} that most of the state-of-the-art models are all special cases of a global framework, the Serialized \Gls{interacting} Mixed \gls{membership} Stochastic Block Model (\acrshort{SIMSBM}). This generalization now allows modelling an arbitrarily large context as well as an arbitrarily high order of \glspl{interaction}.

We then consider a special case of SIMSBM in Section~\ref{IMMSBM} to model \glspl{interaction} between \glspl{entity}. 
This particular iteration is denoted by SIMSBM(2), or \Gls{interacting} Mixed \Gls{membership} SBM (\acrshort{IMMSBM}). We investigate the role of \glspl{interaction} between \glspl{entity} (hashtags, words, memes, etc.) and quantify their importance in several real-world datasets.

\underline{Section~\ref{SDSBM}}, we introduce a simple way to incorporate dynamic modelling into SIMSBM in the form of a temporal prior on the model's parameters. The proposed approach relies on the single assumption that dynamics are not abrupt. We demonstrate the interest of our method on several synthetic experiments and four real-world datasets.

\underline{Section~\ref{SBMs-conclusion}}, we report our first conclusion: \textbf{\glspl{interaction} are sparse}. Significant \glspl{interaction} take place only between a limited fraction of \gls{cluster} pairs, and between an even smaller fraction of \gls{entity} pairs. However, their overall impact increases the predictive accuracy of the models. This underlines the necessity of considering \glspl{cluster} of \glspl{entity} to efficiently model such sparse \glspl{interaction}.

\vfill
\noindent
\textit{
Published works: 
\begin{itemize}[noitemsep,topsep=0pt]
    \item \citep{Poux2022SIMSBM} in Section~\ref{SIMSBM}
    \item \citep{Poux2021IMMSBM} in Section~\ref{IMMSBM}
\end{itemize}
}

\end{chapabstract}

\pagebreak

\section{Introduction}
\label{SBMs-intro}
\subsection{Motivation}
Social networks such as Facebook, Reddit or WhatsApp let individuals share and compare ideas. Modelling the mechanisms of these exchanges can help us understand why and how various \glspl{piece of information} (e.g., hashtags, memes, ideas, etc.) flow through communities. We refer to these \glspl{piece of information} as \textit{\glspl{entity}}. In particular, we suppose that \glspl{entity} can interact with each other. Someone's opinion on a given \gls{entity} might be influenced by previous \glspl{entity} this person has been exposed to; we say there is an \gls{interaction} between \glspl{entity}. For instance, a customer that bought a smartphone might be interested in side accessories such as headphones or selfie sticks, but a customer that bought a smartphone \textit{and} a camera lens extension might be more interested in buying a professional camera. 
An approach without \glspl{interaction} would be less successful here, since each product can lead to a greater number of different recommendations. Considering \glspl{interaction} would narrow the field of possible \glspl{outcome}. The same line of reasoning can be applied to the prediction of retweets (user exposures to tweet A and tweet B affects the retweet probability of C), music playlist building (which song should be added to a playlist given the last songs a user listened to), detection of controversial posts (which combinations of words trigger which answers), etc. Our goal is to infer such underlying relations between \glspl{entity}.

Up to now, little work has been done on investigating the role of \gls{interacting} \glspl{entity} in users' \glspl{decision} (retweet, share, comment answer, etc.), or \textit{\glspl{outcome}}. Several previous works on information \gls{diffusion} theory consider a user acting on an isolated \gls{piece of information} \citep{PastorSatorras2015,Poux2020InfluentialSpreaders}. On some occasions, theoretical frameworks have been developed to investigate how the presence of concurrent \glspl{piece of information} affects the action a user exerts on them in a network \citep{Beutel2012InteractingViruses}. However, a fundamental question that remains unanswered is \textit{how} \glspl{piece of information} interact in the informational landscape.

\subsection{Overview of the proposed approaches}
\paragraph{Representing \glspl{interaction} as a network}
A straightforward way to represent \glspl{interaction} is to embed them as a network. \Glspl{entity} are nodes of this network, and \glspl{interaction} between these \glspl{entity} are link between these nodes. Given \glspl{interaction} can give rise to various \glspl{outcome} (e.g., retweet of A, like of B, etc.), there must be one link per possible \gls{outcome} -- the links are labelled.

In real-world applications, the number of \gls{interacting} \glspl{entity} can be considerable. Modelling each individual labelled link between all of these \glspl{entity} is infeasible in practice. However, the task becomes tractable if we suppose that sets of nodes behave similarly with respect to other sets of \glspl{entity}. The idea is to infer such sets, or \textit{\glspl{cluster}}, and model only the labelled ties between those. A schematic representation of these assumptions is proposed in Fig.~\ref{fig:SIMSBM-illustration-SBMs}.

\begin{figure}
    \centering
    \includegraphics[width=0.99\textwidth]{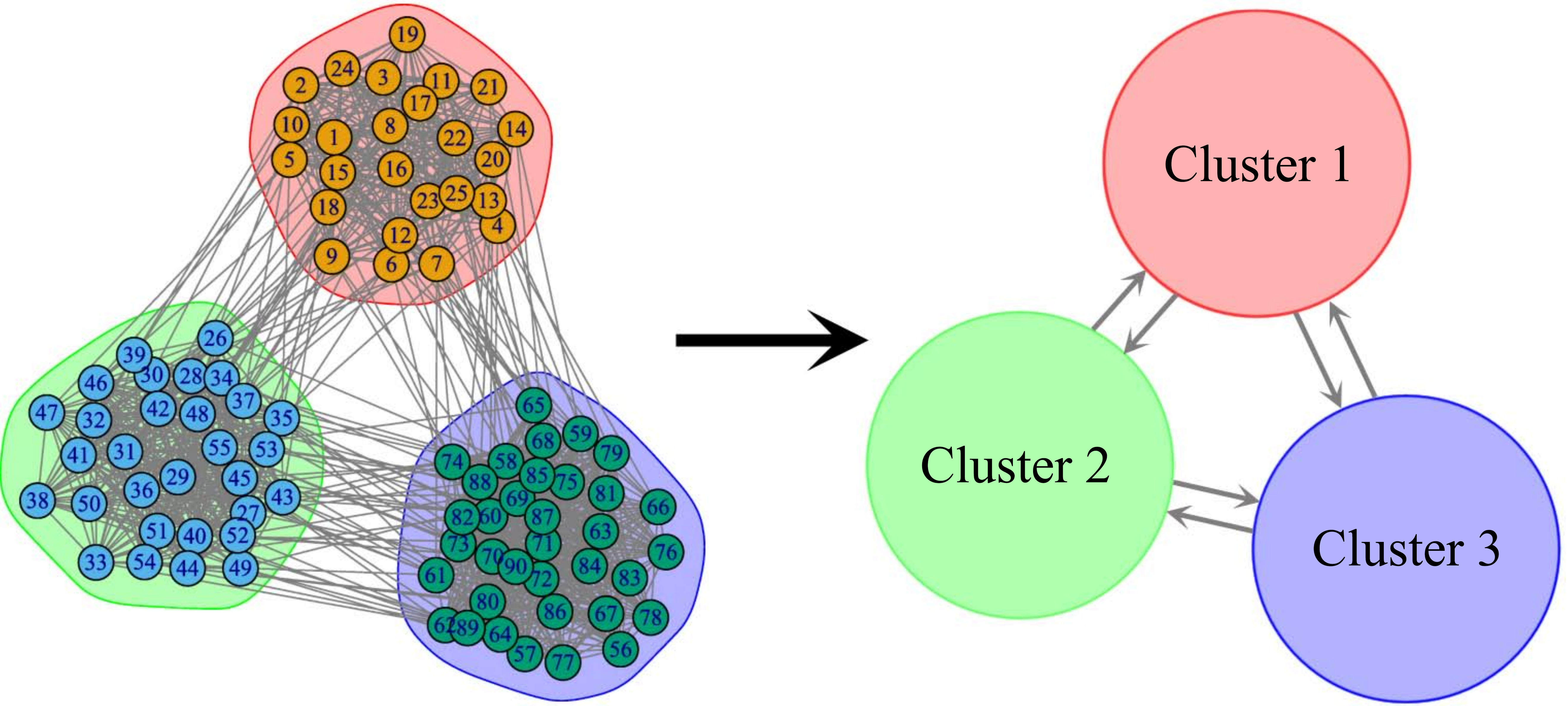}
    \caption[SBM - Illustration of the stochastic block modelling]{\textbf{Illustration of the stochastic block modelling} --- (Left) A set of 100 nodes is densely connected. There is a total of 100x100 links between these nodes to model. The SBM assumption is that sets of nodes behave similarly -- here the nodes that share the same background colour (Right) The SBM approach proposes to model only the relations between these sets of nodes instead of each pair. In this case, it leaves us with 3x3 ties to model (we count the self-ties), as well as the node’s \gls{membership}.}
    \label{fig:SIMSBM-illustration-SBMs}
\end{figure}

\paragraph{Spotting regularities in the \glspl{interaction}}
In general, clustering is a core concept of machine learning. Among other applications, it has proven to be especially fit to tackle real-world recommendation problems. A recommendation consists in guessing an \gls{outcome} given a certain context. This context can often be represented as a high dimensional set of input \glspl{entity}. On retail websites, for instance, the context could be the ID of a user, the last product she bought, the last visited page, the current month, and so on; the \gls{outcome} would be whether this user buys a given product. Clustering algorithms look for regularities in these datasets to reduce the dimensionality of the input context to its most defining characteristics. Continuing the online retail example, a well-designed algorithm would spot that a mouse, a keyboard, and a computer screen are somehow related buys and that the next buy is likely to be another computer device. Besides, as stated before, subsets of users are likely to share a similar interest in a given product if their buying history is similar. We can define such groups of people that share similar behaviours using clustering algorithms; this is called \textit{\gls{collaborative filtering}}. One of the most widely used approaches to perform this task efficiently relies on tensor decomposition.

\paragraph{Tensor decomposition approaches}
Tensor decomposition approaches provide a variety of efficient mathematical tools for breaking a tensor into a combination of smaller components. One of the most popular tensor decomposition methods is Non-negative Matrix Factorization (\acrshort{NMF}). Its application to recommender systems has been proposed in 2006 on Simon Funk's blog for an open competition on movie recommendation \citep{Koren2009MFReco}. The underlying idea is to approximate a target 2-dimensional real-valued observations tensor $D \in {\mathbb{R}^+}^{I \times J}$ (a positive real matrix in this case) 
as the product of two lower dimensional matrices $W \in {\mathbb{R}^+}^{I \times K}$ and $H  \in {\mathbb{R}^+}^{K \times J}$ such that $D = WH$.
This approach has seen numerous developments, such as an algorithm allowing its online optimization \citep{Fund2007OnlineNMF1,Cao2007OnlineNMF2}. Thanks to its low computational cost, this method is still today at the core of many real-world large scale recommender systems. However, a major drawback of NMF is that it can only consider two-dimensional data. Several extensions have been proposed to consider n-dimensional data. A straightforward generalization of NMF is the Tensor Factorization, that generalizes NMF to infer an n-dimensional matrix $D \in \mathbb{R}^{I \times J \times ...}$ as the product of a core tensor $C \in \mathbb{R}^{K \times L \times ...}$ with n smaller matrices $M_1 \in \mathbb{R}^{I \times K}$, $M_2 \in \mathbb{R}^{J \times L}$, and so on, such that $D = C M_1 M_2...M_n$ \citep{Karatzoglou2010TF}. This approach allows to consider a larger context as input data. Several variants have been proposed based on similar ideas \citep{Hidasi2012Itals,Bhargava2015}.

Another popular decomposition method is the Singular Value Decomposition (\acrshort{SVD}) \citep{Klema1980SVD}. It generalizes the concept of eigenvalue decomposition to non-square matrices. As NMF, it has been used to a great extent in recommender systems. The idea is to approximate a positive real-valued matrix $D \in {\mathbb{R}^+}^{I \times J}$ by the product of three matrices $U \in {\mathbb{R}^+}^{I \times K}$, $S \in {\mathbb{R}^+}^{K \times K}$ and $V \in {\mathbb{R}^+}^{K \times J}$. Geometrically, SVD can be interpreted as the decomposition of the initial transformation matrix $D$ as the composition of 3 elementary operations: one scaling $S$ and two rotations $U$ and $V$.

Yet another class of tensor decomposition is called Tensor rank or Canonical Polyadic (\acrshort{CP}) decomposition. It is at the base of several popular decomposition methods that consider a sum of rank-one matrices instead of a product decomposition. In this case, an n-dimensional tensor is approximated as the sum of rank-one tensors \citep{Harshman1970CP,Carroll1970CP}. Several extensions based on CP have been proposed \citep{Acar2010CP1,Filipovi2015CP2}. 

\paragraph{Limitations of tensor decomposition methods}
All the decomposition methods introduced in the previous paragraph are based on the linear decomposition of a real-valued tensor $D$, which makes them unfit to tackle discrete problems. These methods can efficiently infer continuous outputs (the rating of a movie as in \citep{Koren2009MFReco} for instance, or the number of buys of a product) but must be tweaked to consider discrete outputs (the next buy on an online retail website for instance). In this case, a possible approach consists in mapping all possible discrete outputs as a continuous variable. This is straightforward as in the case of movie ratings, because the set of possible ratings (1, 2, 3, ...) can be ordered on a continuous scale. However, for recommending one of several products (mouse, keyboard, computer, ...), the mapping of possible outputs to a continuous value is not trivial.
To consider discrete data in decomposition methods, we can add another dimension to the input tensor, whose size equals the number of possible outputs. Then, the algorithm optimizes the model based on the frequency of each of those items. This trick induces a strong bias and increases the complexity of the algorithm. 

\paragraph{Bayesian network modelling}
To answer this problem, recent years have seen a growing interest in the literature on Stochastic Block Modelling (\acrshort{SBM}). The core idea is to represent the data in the form of a network and apply Bayesian network inference methods to the resulting graph. In particular, these methods assume a block structure. Each node of the network is associated with a block, and blocks relate to each other through a latent block \gls{interaction} network, which can be labelled. The labels correspond to each possible output, and must not be mapped on a continuous scale, as in \citep{Carroll1970CP,Klema1980SVD,Koren2009MFReco}. This makes the model fit to explicitly consider problems such as movie recommendation without the need for mapping movies as an additional dimension in the observation tensor.
Besides, learning does not rely on linear algebra decomposition methods, but on Bayesian learning. It provides greater interpretability of the results and relies on solid statistical foundations; the model actually learns from the data. Besides, Bayesian modelling allows us to incorporate \textit{a priori} knowledge on the data --and we will make extensive use of it in Section~\ref{SDSBM}.

We represent a schematic example of the type of Stochastic Block Model we consider in this chapter in Fig.~\ref{fig:SIMSBM-illustration-IMMSBM}. In this case, we have two input \glspl{entity} that are embedded as nodes. They are linked together by labelled edges, which represent the probability of an output given the \glspl{entity}' combination. Instead of directly modelling such links, we assume a block structure to this network. Nodes belong to \glspl{cluster} to a certain extent, and only the labelled links between these \glspl{cluster} are modelled.

We extensively review this class of models in the sections below and present our proposed improvements to tackle \gls{interaction} modelling. Explicitly, we first generalize the existing labelled stochastic block models so that they can model high order \glspl{interaction}, arbitrary large input contexts, and provide predictions on as many output labels as wanted (Section~\ref{section-static-interaction}). In a second time, we design a temporal plug-in for our model. It allows to model temporal variations of the block model's parameters with high accuracy at a minimal cost --in terms of data needed, simplicity of use and numerical complexity (Section~\ref{SDSBM}).

\begin{figure}
    \centering
    \includegraphics[width=\textwidth]{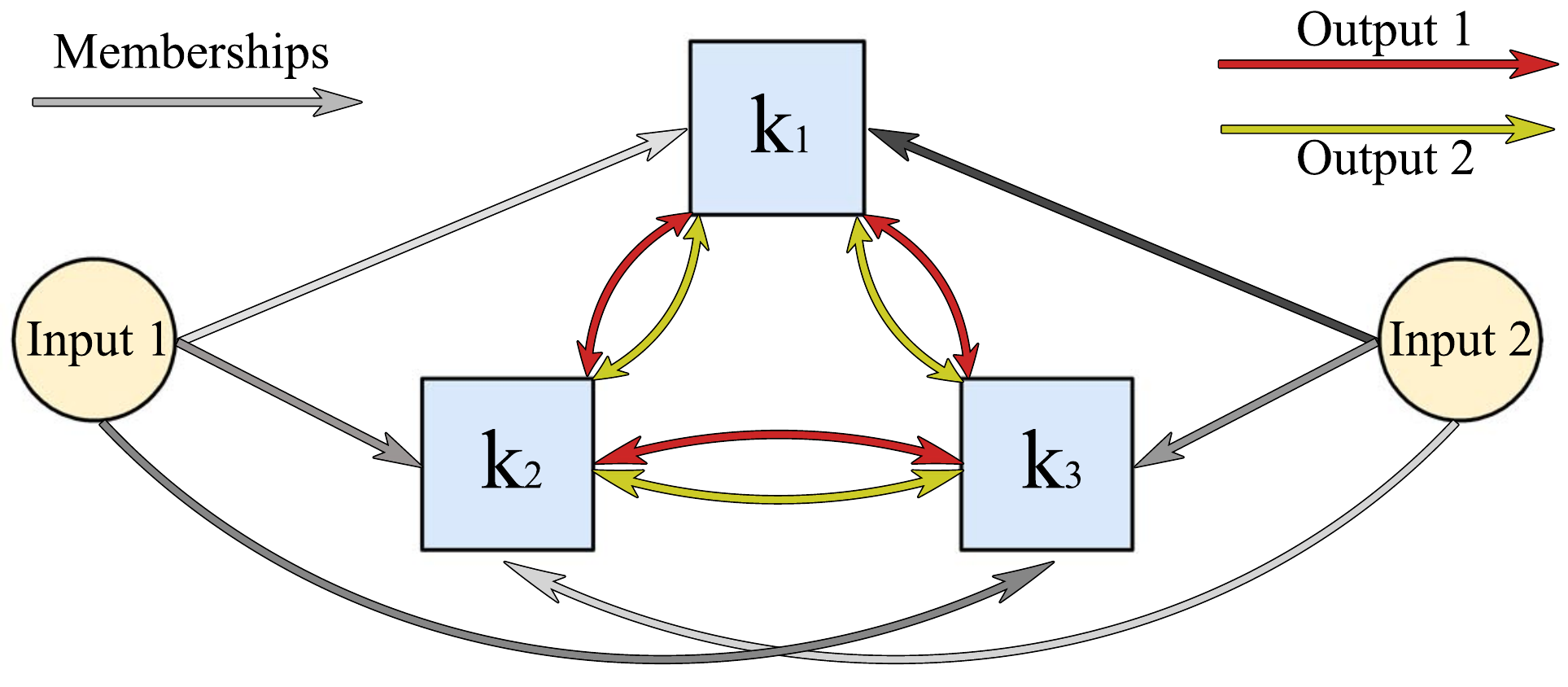}
    \caption[SBM - Illustration of the principle]{\textbf{Illustration of the SIMSBM principle} --- Two input \glspl{entity} are embedded as nodes, linked together by labelled edges. Edges represent the probability of an output given a combination of nodes. Instead of directly modelling such links, we assume a block structure to this network. Nodes belong to \glspl{cluster} to a certain extent, and only the labelled links between these \glspl{cluster} are modelled.}
    \label{fig:SIMSBM-illustration-IMMSBM}
\end{figure}

\section{Static \glspl{interaction}}
\label{section-static-interaction}

\subsection{State of the art, limitations, and contributions}
\label{SBM-SotA-lim-contrib}

\subsubsection{Modelling static \glspl{interaction}}
Accounting simultaneously for multiple \glspl{piece of information} is motivated by numerous descriptive studies on multimodal networks structure \citep{YizhouHan2012,Chuan2016,Huan2016,Rashed2020}. Typically, in \citep{YizhouHan2012}, the authors study \gls{interaction} between multiple \gls{entity} \glspl{type} \textit{via} a heterogeneous network representation and define \glspl{cluster} of \glspl{entity} based on the structural properties of the resulting graph. 
However, as pointed out by the authors, this method is heavily influenced by the structural clustering method used --in this case a meta-path-based clustering \citep{Yizhou2011}. Moreover, defining edge weights in heterogeneous graphs is subjective and requires additional learning algorithms.

As seen in the introduction, several works proposed to model the \gls{diffusion} of information as the result of an \gls{interaction} between \glspl{entity} \citep{Beutel2012InteractingViruses}. Following a similar idea, \citep{Myers2012CoC} investigated \glspl{interaction} between contagions on Twitter. The authors aim to find the \gls{interaction} factors between different tweets in activity feeds. Their findings suggest that \glspl{interaction} between tweets play a determinant role in their retweet probability. The authors assume that there is an inherent \gls{virality} for every tweet (that is an inherent probability to be retweeted) computed from the frequency of retweets, to which is added a small \gls{interaction} term. 

\citep{Myers2012CoC} paves the way to model \glspl{interaction} in information \gls{spread} but presents various limitations that we are alleviated by using Stochastic Block Models.
Firstly, the method proposed by the authors makes predictions solely based on tweets that have been observed in the feed of a given user. It therefore limits the application range of the model uniquely to systems based on the retweets (or share) concept, where information has to appear first to be \gls{spread}. This model is hardly applicable to systems that are based on exogenous reactions (e.g., online forums, playlist building and recommender systems) where information can appear as a consequence of different \glspl{entity} (``Capital'' + ``Netherlands'' $\rightarrow$ ``Amsterdam''). An SBM-based modelling allows \glspl{outcome} that are different from the input \glspl{entity}.
Secondly, \gls{interaction} is defined as a modulation of the frequency of retweets of a given tweet in any context (i.e., $P(X) = f_X + \Delta P_{inter}(X)$). We argue this can lead to false conclusions about \glspl{interaction}. Imagine that \glspl{interaction} lead to a shift of $\Delta P(X)$ on the actual \gls{virality} $V_X$ of a retweet. This \gls{interaction} happens in a fraction $s$ of all observations of a given tweet being retweeted. The \gls{virality} $V_{X, hyp}$ as defined in \citep{Myers2012CoC} equals the frequency of retweets: $V_{X, hyp} = V_X(1-s) + (V_X+\Delta P(X)) s = V_X + s \Delta P(X)$, which is by construct larger than the actual \gls{virality} of a retweet by a term $s \Delta P(X)$. Therefore, defining \gls{interaction} according to this quantity is wrong. \Gls{virality} needs to be inferred by the model at the same time as the contribution of \glspl{interaction} to be properly defined, which the proposed SBM-based modelling allows.

\subsubsection{Stochastic Block Models}
\label{SBM-SotA}
\paragraph{Stochastic Block Models}
As already stated, recent years saw a growing interest for Stochastic Block Models (\acrshort{SBM}) to tackle \gls{collaborative filtering} problems in recommender systems \citep{Antonia2016AccurateAndScalableRS,Tarres2019TMBM,Poux2021MMSBMMrBanks}. These models first \gls{cluster} input \glspl{entity} together and then analyse how \glspl{cluster} relate to each other. Each input \gls{entity} can either be associated to one \gls{cluster} (Single-\Gls{membership} SBM, or SBM) \citep{Holland1983SBM,Guimera2012HumanPrefSBM,Funka2019ReviewSBM} or to a distribution over available \glspl{cluster} (Mixed \Gls{membership} SBM, or MMSBM) \citep{Airoldi2008MMSBM}. While the single \gls{membership} SBM has been successfully applied to a wide range of problems \citep{Guimera2012HumanPrefSBM,Guimera2013DrugdrugSBM,CoboLopez2018SocialDilemma,Funka2019ReviewSBM}, inference is done using greedy algorithms, typically simulated annealing, making it unfit for large scale real-world applications \citep{Antonia2016AccurateAndScalableRS}. 

\paragraph{Mixed \Gls{membership} Stochastic Block Models}
The Mixed-\Gls{membership} SBM (\textbf{MMSBM}) is a major extension of Single-\Gls{membership} SBM that has been proposed in the seminal work \citep{Airoldi2008MMSBM}. 
In the frame of \gls{collaborative filtering} for recommender systems, \citep{Antonia2016AccurateAndScalableRS} proposed the \textbf{\acrshort{Bi-MMSBM}} extension. This formulation generalizes the Matrix Factorization model. In this approach, \glspl{entity} of different \gls{type} (e.g., users and movies in this case) are grouped into distinct \glspl{cluster} whose \gls{interaction} result in an \gls{outcome} (e.g., a rating on a movie). This model has later been extended as the \textbf{\acrshort{T-MBM}} to consider triples of input \glspl{entity} instead of pairs \citep{Tarres2019TMBM}. It assumes that all the \glspl{entity} in a given triplet are linked together by a given label. This boils down to assuming data can be represented in the form of a tripartite network instead of a bipartite network.

\paragraph{The need for a more global framework}
However, the existing literature does not answer all the challenges inherent to \gls{interaction} modelling. 

\textbf{First}, these models allow considering only \gls{interaction} pairs, which are arguably not exhaustive enough to correctly model the \gls{interaction} processes at stake \citep{Steck2021Interactions}. As stated in the introduction, \glspl{interaction} can involve an arbitrarily large number of \glspl{entity}. For instance, a Twitter user might want to retweet an \gls{entity} based on its \gls{content}, but also on the previous tweets she has been exposed to, the hour of the day, the last tweet she published, etc. No existing MMSBM variation proposes a solution for considering an arbitrarily large number of \gls{interacting} \glspl{entity}.

\textbf{Second}, none of the existing models consider the case where \gls{interacting} \glspl{entity} are of the same \gls{type}. \Glspl{entity} of the same \gls{type} carry the same semantic meaning (e.g., two movies are of the same \gls{type} ``movie'', but a movie and a director are of different \glspl{type}). When considering input \glspl{entity} of the same \gls{type}, permutation-invariant clustering must be accounted for: the probability of an output given \glspl{entity} $(A,B)$ should not differ from the probability of this output given $(B,A)$. This may not hold when the order of \glspl{entity}’ apparition is considered, which is not the case here. We consider as input an unordered set of \glspl{entity}, hence the need for permutation-invariant modelling.

Considering the \gls{interaction} between \glspl{entity} of the same \gls{type} is novel and important. When ignoring it, a user on Netflix is predicted to like a given movie because she is partially part of the group that liked the movie A and partially part of the group that disliked movie B, but all those groups are \gls{independent} of one another \citep{Koren2009MFReco,Antonia2016AccurateAndScalableRS}. The friendship between individuals is determined based on the \gls{independent} groups of friends they belong to, but not based on the joint belonging to various groups \citep{Airoldi2008MMSBM,Jamali2011}. A drug might interact with another one, but the joint \gls{interaction} of two drugs on a third one cannot be investigated \citep{Guimera2013DrugdrugSBM}. 
Typically, in \citep{Antonia2016AccurateAndScalableRS}, embedding \glspl{piece of information} of the same \gls{type} in a bipartite graph seems irrelevant: an \gls{entity} should not interact differently with another one (or belong to different \glspl{cluster}) because it is on the left or the right side of this bipartite graph. Instead, we need to enforce a permutation-invariant \gls{interaction}.

\subsubsection{Contributions}
In this section, we present the Serialized Interacting Mixed Membership Stochastic Block Model (\acrshort{SIMSBM}). The SIMSBM is a global framework that generalizes several state-of-the-art models which tackle the problem of discrete recommendation by a Bayesian network approach --including but not limited to \citep{Airoldi2008MMSBM,Koren2009MFReco,Antonia2016AccurateAndScalableRS,Tarres2019TMBM} and the IMMSBM model presented in Section~\ref{IMMSBM}. 

First, we present the proposed framework and develop an Expectation Maximization (\acrshort{EM}) optimization procedure (Section~\ref{SIMSBM-model}). Then, we review previous works on Stochastic Block Models that are used in \gls{collaborative filtering} and detail for each one how to recover it as a special case of our framework (Section~\ref{SIMSBM-generalizes}). Experiments are then conducted on 6 real-world datasets and compared to standard baselines of the literature (Section~\ref{SIMSBM-XP}). We show our formulation allows us to obtain better recommendations than existing methods either by adding layers to the modelled multipartite graph or by adding higher-order \gls{interaction} terms in the modelling.

Then, we proceed to an in-depth study of a special case of the SIMSBM applied to \glspl{interaction} modelling, named IMMSBM. We consider the case where two input \glspl{entity} of the same \gls{type} interact with each other and influence the probability of an output (Section~\ref{IMMSBM-model}). We consider four real-world datasets that are expected to comprise underlying \gls{interaction} processes: Twitter (\glspl{entity} are tweets and \glspl{outcome} are retweets), Reddit (\glspl{entity} are words in a post and \glspl{outcome} are words in an answer), Spotify (\glspl{entity} are songs in a user-generated playlist and \glspl{outcome} are songs added next) and PubMed (\glspl{entity} are symptoms and \glspl{outcome} are diagnosed diseases) (Section~\ref{IMMSBM-XP}). We compare to similar works that do not consider \glspl{interaction}; we also discuss and discard a core assumption made in \citep{Myers2012CoC} on \gls{interaction} modelling in information \gls{spread} (Section~\ref{IMMSBM-discussion}). We investigate in detail the role of \glspl{interaction} in the proposed corpora. Finally, we conclude that significant \glspl{interaction} between spreading \glspl{entity} are rare (Section~\ref{IMMSBM-conclusion}).

\subsection{SIMSBM -- A global MMSBM framework}
\label{SIMSBM}
\textit{This work has been published, see \citep{Poux2022SIMSBM}}

\subsubsection{SIMSBM -- Serialized \Gls{interacting} MMSBM}
\label{SIMSBM-model}
\paragraph{Overall idea}
The goal of the SIMSBM is to recommend an \textit{output \gls{entity}}, that is one $o$ in $O$ possible outputs \glspl{entity}, given a context --for instance, $o$ can take any label in $O = \{ 1, 2, 3 \}$. To do so, it considers data in the form of a multipartite network. A multipartite network is a generalization of the bipartite network (two layers of nodes with no intra-connection within layers) to any number $N$ of layers.

The network's nodes are the context elements. They are called \textit{input \glspl{entity}} $f_n$, each of which is one of $F_n$ possible input \glspl{entity} --for instance, $f_1$ can take one value in $F_1 = \{ A,B,C \}$, $f_2$ can take one value in $F_2 = \{ D,E,F \}$, etc. One hyper-edge between the $N$ input \glspl{entity} (or nodes) of a context represents the probability of a given output \gls{entity} $o$, that is ${P(o \in O \vert f_1 \in F_1, ..., f_N \in F_N)}$. 

The SIMSBM \glspl{cluster} input \glspl{entity} (or nodes) together and infers edges between these \glspl{cluster} to form a smaller multipartite network. Each node $f_n$ is associated with a certain extent to each \gls{cluster} $k_n$ among $K_n$ possible \glspl{cluster} for layer $n$.
The adjacency matrix of this network is noted ${p(o \in O \vert k_1 \in K_1, ..., k_N \in K_N)}$ --note the lowercase ${p}$ for the \glspl{cluster}-\gls{interaction} network. This projection is illustrated in Fig.~\ref{fig-schema}. The notations used throughout this section are given in Table.~\ref{table-not}.

\paragraph{Input data}
Formally, input data of the SIMSBM is noted ${R^{\circ}}$. Its entries are tuples of form ${(\vec{f}, o)}$. The vector ${\vec{f}=(f_1, ..., f_N)}$ represents a context, whose entries $f_n$ are the input \glspl{entity} that can be one of $F_n$ possible input \glspl{entity}; outputs are designated by $o$ and can be any of the $O$ possible outputs. Each $F_n$ and $O$ are unordered sets containing the labels each $f_n$ and $o$ can take. Each input \gls{entity} is represented as a node in a layer of the multipartite network (circles in Fig.~\ref{fig-schema}). The number of layers of the multipartite graph is then $N = \vert \vec{f} \vert$; each layer $n$ comprises $F_n$ nodes. For each node in each layer, the SIMSBM infers a vector ${\vec{\theta_i} \in \left[ 0,1 \right]^{K}}$ that represents its probability to belong to each of $K$ possible \glspl{cluster} (i.e., associate each circle to a distribution over the squares in Fig.~\ref{fig-schema}). 

Input \glspl{entity} are of a given \textit{\gls{type}}; \glspl{entity} of the same \gls{type} carry the same semantic meaning and are drawn from the same set of possible values. We note $a(f)$ the function that associates an input \gls{entity} $f$ to its given \gls{type}, and $K_{a(f_n)}$ the number of available \glspl{cluster} for this \gls{type}. \Glspl{entity} of the same \gls{type} can interact with each other but are forced to share the same \gls{cluster} \gls{membership} matrix. For instance, consider a user rating a movie according to the pair of actors starring in it: the context vector takes the shape $(user, actor 1, actor 2)$. In this example, one \gls{entity} is of \gls{type} ``user'', and the two others are of \gls{type} ``actor'', and each of the three \glspl{entity} is embedded in its own layer of a multipartite graph. When creating \glspl{cluster}, the SIMSBM enforces that the two layers of \gls{type} ``actor'' share the same \gls{membership} matrix, while the layer accounting for the \gls{type} ``user'' has its own \gls{membership} matrix. This structure is exactly the same as the one depicted in Fig.~\ref{fig-schema}. This is needed to get results that are permutation \gls{independent}, meaning in our example that $P(o \vert (user, actor 1, actor 2)) = P(o \vert (user, actor 2, actor 1))$. Our formulation as a multipartite network allows us to consider higher-order \glspl{interaction}, in contrast to \citep{Antonia2016AccurateAndScalableRS} which only considers pair \glspl{interaction}.
\begin{figure}[h]
    \centering
    \includegraphics[width=0.99\columnwidth]{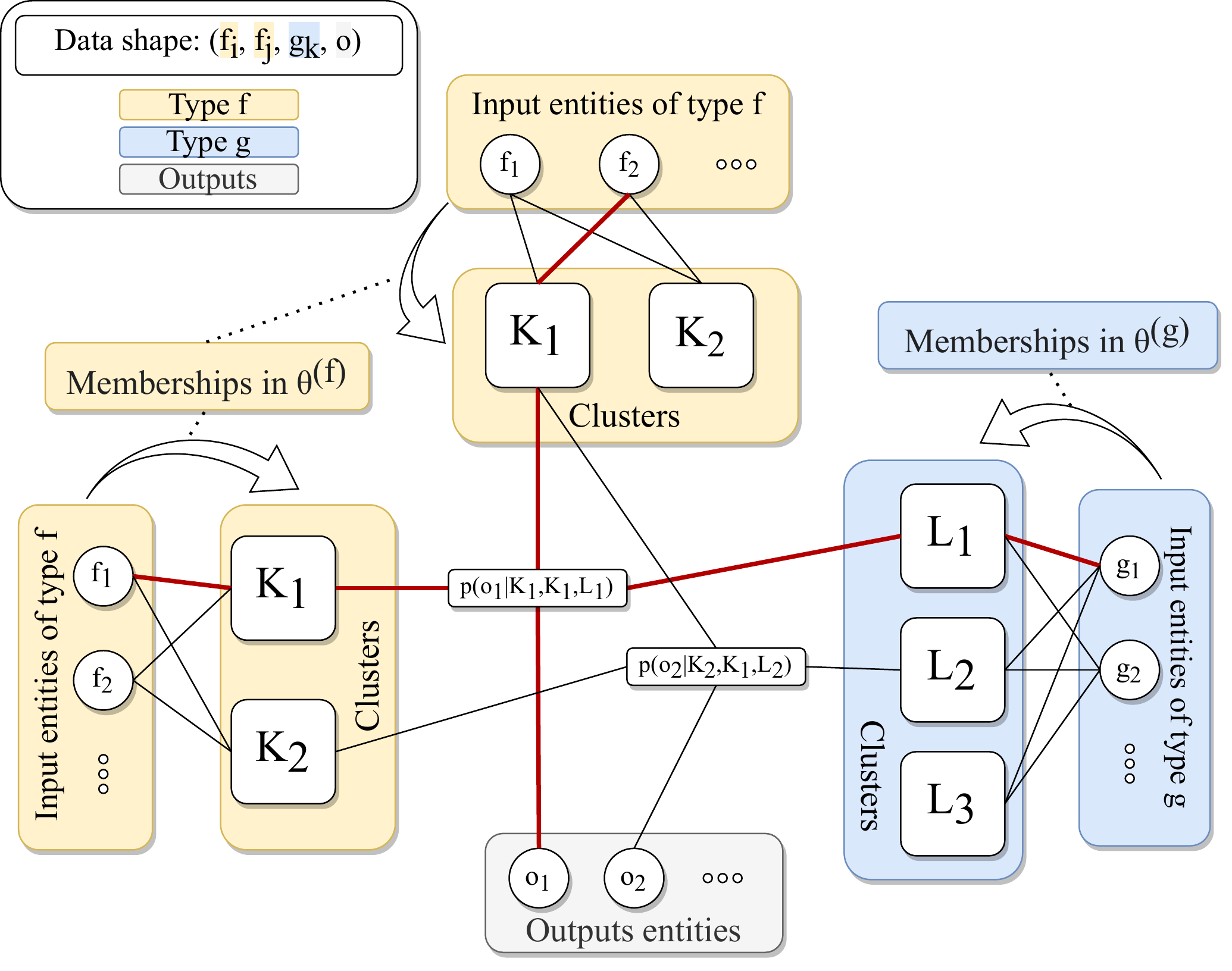}
    \caption[SIMSBM - Illustration of the SIMSBM]{\textbf{Illustration of the SIMSBM} --- For input \glspl{entity} of \gls{type} $f$ and $g$, where \glspl{entity} of \gls{type} $f$ interact with each other as pairs and \glspl{entity} of \gls{type} $g$ do not interact with each other. The \gls{membership} of an \gls{entity} $f_i$ of \gls{type} $f$ to a \gls{cluster} $K_n$ is encoded into the \gls{membership} matrix entry $\theta^{(f)}_{f_i,K_n}$. The \gls{interaction} between \glspl{cluster} is embedded in a multipartite network, whose adjacency tensor is $\mathbf{p}$. A weighted edge between several \glspl{cluster} and one output represents the probability of this output given the context \glspl{cluster} --only two such edges are represented here. \textbf{In red}, we represent the probability of $o_1$ given $f_1$ belonging to $K_1$, $f_2$ belonging to $K_1$ and $g_1$ belonging to $L_1$, which is equal to $\theta_{f_1,K_1}^{(f)}\theta_{f_2,K_1}^{(f)}\theta_{g_1,L_1}^{(g)}p(o_1|K_1,K_1,L_1)$. 
    }
    \label{fig-schema}
\end{figure}

\paragraph{Model parameters}
The \gls{membership} of an \gls{entity} must sum to 1 over all the $K_{a(f_n)}$ available \glspl{cluster}, hence the following constraint:
\begin{equation}
\label{eq:normTheta}
    \sum_k^{K_{a(f_n)}} \theta_{f_n,k}^{(a(f_n))} = 1 \,\,\forall f_n
\end{equation}
Once the node \gls{membership} is known, the SIMSBM infers the \glspl{cluster} multipartite network, whose weighted hyper-edges stand for the probability of an output $o$ given a combination of \glspl{cluster}. 

Note that there are two possible views on this: we can consider that each possible $o$ is associated with its own multipartite network, or that each $o$ is a node in an additional output layer, as in Fig.~\ref{fig-schema}. These views are strictly equivalent. 

The hyper-edge corresponding to the clustered context ${\vec{k}=(k_1, ..., k_N)}$ associated with the output $o$ is written ${p_{\vec{k}}(o) \in \mathbb{R}^{K_{a(f_1)} \times ... \times K_{a(f_N)} \times O}}$ --it is one entry of the multipartite network's adjacency matrix. 
Note that the clustered context $\vec{k}$ can take any value among all the possible permutations ${\vec{K} = \{K_1, ..., K_N\}_{K_1=1, ..., K_{a(f_1)}; ...; K_N=1,...,K_{a(f_N)}}}$.
As we want SIMSBM to infer a distribution over possible outputs in a given context, the edges of the multiparted graph are related by the following constraint:
\begin{equation}
    \sum_o p_{\vec{k}}(o) = 1 \,\,\,\, \forall \vec{k} \in \vec{K}
\end{equation}
Finally, the probability of an output $o$ given a context of input \glspl{entity} $\vec{f}$ can be written as:
\begin{equation}
    \label{eq-probfinale}
    P(o \vert\vec{f}) = \sum_{\vec{k} \in \vec{K}} p_{\vec{k}}(o) \prod_{n \in N} \theta_{f_n,k_n}^{(a(f_n))}
\end{equation}
From Eq.~\ref{eq-probfinale}, we can define the log-likelihood of the model as:
\begin{equation}
\label{eq:likelihood}
   \ell := \log P(R^{\circ} \vert \theta, p) = \sum_{(\vec{f},o) \in R^{\circ}} \log \left( \sum_{\vec{k} \in \vec{K}} p_{\vec{k}}(o) \prod_{n \in N} \theta_{f_n,k_n}^{(a(f_n))} \right)
\end{equation}
\myequations{\ \ \ \ SIMSBM - Likelihood}

\begin{table}
    \caption[SIMSBM - Notations]{Notations for the SIMSBM}
    \label{table-not}
	\centering
	\begin{tabular}{|l|l|}
	$f_n$ & An input \gls{entity}, can take any value in $F_n$\\
	$F_n$ & Set of possible input \glspl{entity} for layer $n$\\
	$\vec{f}$ & Context vector $(f_1, ..., f_N)$\\
	$o$ & Output \gls{entity}, can take any value in $O$ \\
	$N$ & Number of input layers $\vert \vec{f} \vert$ \\
	$R^{\circ}$ & Data, a list of (N+1)-plets $(f_1, ..., f_N, o)$\\
	$a(f_n)$ & \gls{type} of \gls{entity} $f_n$\\
	$K_{a(f_n)}$ & Number of available \glspl{cluster} for \gls{type} $a(f_n)$\\
	$\theta^{(a(f_n))}$ & \Gls{membership} matrix for \glspl{entity} of \gls{type} $a(f_n)$ \\
	$\mathbf{p}(o)$ & \Glspl{cluster}' multipartite network for output $o$\\
	$\vec{K}$ & Every possible \glspl{cluster} permutation $\{(k_1,..., k_N)\}_{k_1=1 \text{ to } K_1;...;k_N=1 \text{ to } K_N}$\\
	$\vec{k}$ & One permutation of \glspl{cluster} indices $(k_1,..., k_N) \in \vec{K}$ \\
	$c_v(x)$ & Count of element $x$ in vector $\vec{v}$\\
	$C_{f_n}$ & Total count of $f_n$ in $R^{\circ}$\\
	\end{tabular}
\end{table}

\subsubsection{Inference}
In this section, we derive an Expectation-Maximization algorithm for inferring the model's parameters $\mathbf{p}, \theta$. Such an algorithm guarantees the convergence towards a local maximum of the likelihood function \citep{Neal1998ConvergenceEM}.

\paragraph{E-step}
\label{SIMSBM-Estep}
The derivation presented in this paragraph follows a well-known general derivation of the EM algorithm, which can be found in \citep{Bishop2006} for instance.

We recall that one entry of the dataset $R^{\circ}$ takes the form of a tuple $(\vec{f}, o)$, where $\vec{f}$ is the features input vector in which an output $o$ has been observed. Each feature in $\vec{f}$ is associated with a \gls{cluster}. We note $\vec{k} \in \vec{K}$ the set of latent variables accounting for each \gls{cluster} allocation of a given feature vector among the space of $\vec{K}$ possible allocation combinations. The probability of an output given a combination of \glspl{cluster} is given by $p_{\vec{k}}(o)$

The total log-likelihood (Eq.~\ref{eq:likelihood}) is expressed as the sum of each observation’s individual log-likelihood: $\log P(R^{\circ} \vert \theta, p) = \sum_{(\vec{f}, o) \in R^{\circ}}\log P(\vec{f}, o \vert \theta, p)$. Without loss of generality, we focus here on a single observation for clarity.

We assume a distribution $Q(\vec{k})$ on the latent variables associated with one observation in the dataset $R^{\circ}$; this distribution is yet to be defined. Because $\vec{k}$ takes values in $\vec{K}$, we have $\sum_{\vec{k} \in \vec{K}} Q(\vec{k}) = 1$. Given this normalization condition, we can decompose Eq.\ref{eq:likelihood} for any distribution $Q(\vec{k})$ as:

\begin{align}
    \label{SIMSBM-eq-decomp}
    \log P(\vec{f}, o \vert \theta, p) =& \underbrace{\log P(\vec{f}, o, \vec{k} \vert \theta, p) - \log P(\vec{k} \vert \vec{f}, o, \theta, p)}_{\text{This difference does not depend on $\vec{k}$}} \notag \\
    = &\left( \log P(\vec{f}, o, \vec{k} \vert \theta, p) - \log P(\vec{k} \vert \vec{f}, o, \theta, p) \right) \underbrace{\sum_{\vec{k} \in \vec{K}} Q(\vec{k})}_{=1}  \notag \\
    = &\sum_{\vec{k} \in \vec{K}} Q(\vec{k}) \log P(\vec{f}, o, \vec{k} \vert \theta, p) - \sum_{\vec{k} \in \vec{K}} Q(\vec{k}) \log P(\vec{k} \vert \vec{f}, o, \theta, p) \notag \\
    = &\sum_{\vec{k} \in \vec{K}} Q(\vec{k}) \log \frac{P(\vec{f}, o, \vec{k} \vert \theta, p)}{Q(\vec{k})} - \sum_{\vec{k} \in \vec{K}} Q(\vec{k}) \log \frac{P(\vec{k} \vert \vec{f}, o, \theta, p)}{Q(\vec{k})}
\end{align}

In this equation, the first term of the last line is proportional to the expectation of the complete log-likelihood $P(\vec{f}, o, \vec{k} \vert \theta, p)$ with respect to the latent variables $\vec{k}$ --minus the entropy of the distribution over $\vec{k}$. This explains the name ``Expectation step''.

We note that the second term in the last line of Eq.\ref{SIMSBM-eq-decomp} is the Kullback-Leibler (KL) divergence between $P(\vec{k} \vert \cdot)$ and $Q(\vec{k})$, noted $KL(P \vert \vert Q)$. The KL divergence obeys $KL(P \vert \vert Q) \geq 0$, and is null iif $P$ equals $Q$. Therefore, the expectation of the complete log-likelihood is interpreted as a lower bound on the log-likelihood $\log P(\vec{f}, o \vert \theta, p)$.

The goal of the E-step is to find the expression of $Q(\vec{k})$ that maximizes the expectation of the complete log-likelihood according to the latent variables $\vec{k}$, which is the same as maximizing the lower bound of the log-likelihood $\log P(\vec{f},o \vert \theta, p)$. Given that the log-likelihood does not depend on $Q(\vec{k})$ and that $KL(P \vert \vert Q) \geq 0$, it is maximized when $KL(P \vert \vert Q) = 0$, which occurs for $Q(\vec{k}) = P(\vec{k} \vert \vec{f}, o, \theta, p)$. In this case, the expectation of the complete log-likelihood equals the log-likelihood itself. Its expression reaches a global maximum with respect to the latent variables $\vec{k}$ for fixed parameters $\theta$ and $p$.

Given the definition of the SIMSBM, the derivation of $P(\vec{k} \vert \vec{f}, o, \theta, p)$ is straightforward (see Eq.\ref{eq:likelihood}). The probability of one combination of \glspl{cluster} $\vec{k}$ among $\vec{K}$ possible combinations given an input features vector $\vec{f}$ and an output $o$ is proportional to ${p_{\vec{k}}(o) \prod_{n \in N} \theta_{f_n,k_n}^{(a(f_n))}}$. Therefore, we can write the now-known distribution $Q(\vec{k})$, noted $\omega_{\vec{f}, o}(\vec{k})$, as:

\begin{align}
    \label{eq:expectation}
     \omega_{\vec{f}, o}(\vec{k}) := Q(\vec{k}) &=  P(\vec{k} \vert \vec{f}, o, \theta, p) \notag \\
     &= \frac{p_{\vec{k}}(o) \prod_{n \in N} \theta_{f_n,k_n}^{(a(f_n))}}{\sum_{\vec{k'} \in \vec{K}} p_{\vec{k'}}(o) \prod_{n \in N} \theta_{f_n,k'_n}^{(a(f_n))}}
\end{align}

By substituting $Q(\vec{k})$ and $P(\vec{k} \vert \vec{f}, o, \theta, p)$ by $\omega_{\vec{f}, o}(\vec{k})$ in Eq.~\ref{SIMSBM-eq-decomp}, we get an expression for the log-likelihood which is maximized with respect to the latent variables. Explicitly:
\begin{equation}
    \label{SIMSBM-eq-lik-estep}
    \ell = \sum_{(\vec{f},o) \in R^{\circ}} \sum_{\vec{k} \in \vec{K}} \omega_{\vec{f},o}(\vec{k}) \cdot \log  \left( \frac{p_{\vec{k}}(o) \prod_{n \in N} \theta_{f_n,k_n}^{(a(f_n))}}{\omega_{\vec{f},o}(\vec{k})} \right)
\end{equation}
\myequations{\ \ \ \ SIMSBM - E-step}

The maximization step follows by maximizing Eq.\ref{SIMSBM-eq-lik-estep} with respect to the parameters $\theta$ and $p$, holding $\omega (\cdot)$ constant.

\paragraph{M-step}
We take back Eq.~\ref{eq:likelihood} and add Lagrangian multipliers $\phi$ to account for the constraints on $\theta$. We maximize of the resulting constrained likelihood $\ell_c$ according to each latent variable:
\begin{align}
\label{eq:m-theta}
    &\frac{\partial \ell_c}{\partial \theta_{mn}^{(a(m))}} = \frac{\partial}{\partial \theta_{mn}^{(a(m))}} \left[ \ell - \sum_i \phi_i^{(a(i))} \left( \sum_k \theta_{ik}^{(n)} - 1 \right) \right] \notag\\
    &\Leftrightarrow \ \ \phi_{m}^{(a(m))} = 
     \sum_{(\vec{f},o) \in \partial m} \sum_{\vec{k} \in \vec{K}} \frac{ c_{\vec{k_{i_m}}}(n) \omega_{\vec{f},o}(\vec{k})}{\theta_{mn}^{(a(m))}}  \notag \\
    &\Leftrightarrow \ \ \theta_{mn}^{(a(m))} \phi_{m}^{(a(m))} = 
     \sum_{(\vec{f},o) \in \partial m} \sum_{\vec{k} \in \vec{K}} c_{\vec{k_{i_m}}}(n) \omega_{\vec{f},o}(\vec{k})
\end{align}
The term $c_{\vec{k_{i_m}}}(n)$ arises because of the non-linearity induced by the \gls{interaction} between input \glspl{entity} of the same \gls{type}.
Let $i_m$ be the indices where \gls{entity} $m$ appears in the input vector $\vec{f}$. The corresponding entries of the permutation vector $\vec{k}$ are noted $\vec{k_{i_m}}$. Then, ${c_{\vec{k_{i_m}}}(n)=\vert \{ 1 \vert \vec{k}_{i}=n \}_{i \in i_m} \vert}$ is the count of $n$ in $\vec{k}_{i_m}$.
When $n$ appears $c_{\vec{k_{i_m}}}(n)$ times in a permutation comprising $\vec{k_{i_m}}$, so will a term $\log \theta_{nm}^{c_{\vec{k_{i_m}}}(n)}$, whose derivative is $\frac{c_{\vec{k_{i_m}}}(n)}{\theta_{nm}}$, hence this term arising. 
Note that $c_{\vec{k_{i_m}}}(n)=0$ nullifies the contribution of permutations $\vec{k}$ where $n$ does not appear in $\vec{k_{i_m}}$. Therefore, we can restrict the sum over $\vec{k}$ in Eq.~\ref{eq:m-theta} to the set ${\partial n = \{ \vec{k} \vert \vec{k}\in \vec{K}, n \in \vec{k}_{i_m} \}}$. We also defined the set ${\partial m = \{ (\vec{f}, o) \vert (\vec{f}, o) \in R^{\circ}, m \in \vec{f}_{i_m} \}}$. 

Using Eq.~\ref{eq:normTheta} and Eq.~\ref{eq:m-theta}, we compute $\phi_m^{(a(m))}$:
\begin{align}
\label{eq:phi}
    \sum_n^{K_{a(m)}} \phi_{m}^{(a(m))}& \theta_{mn}^{(a(m))}  = \sum_{(\vec{f},o) \in \partial m} \sum_n^{K_{a(m)}} \sum_{\vec{k} \in \vec{K}}  c_{\vec{k_{i_m}}}(n) \omega_{\vec{f},o}(\vec{k})  \notag \\
    = \phi_{m}^{(a(m))} &= \sum_{(\vec{f},o) \in \partial m}  \underbrace{\sum_{\vec{k} \in \vec{K}} \omega_{\vec{f}, o}(\vec{k})}_{=1 \text{ (Eq.~\ref{eq:expectation})}} \underbrace{\sum_n^{K_{a(n)}} c_{\vec{k_{i_m}}}(n)}_{=c_{\vec{f}}(m)}  \notag \\
    &= \sum_{(\vec{f},o) \in \partial m}  c_{\vec{f}}(m) = C_m
\end{align}
When summing over $n$, $c_{\vec{k_{i_m}}}(n)$ successively counts the number of times each $n$ appears in $\vec{k_{i_m}}$, which equals the length of $\vec{k_{i_m}}$. Therefore $\sum_n c_{\vec{k_{i_m}}}(n) = \vert \vec{k_{i_m}} \vert = c_{\vec{f}}(m)$ is the number of times input \gls{entity} $m$ appears in the entry $(\vec{f},o)$ considered, which does not depend on $\vec{k}$. $C_m$ is the total count of $m$ in the dataset. Note that this differs from the derivation proposed in \citep{Tarres2019TMBM}, where nonlinear terms are not accounted for. 

The derivation of the maximization equation for $\mathbf{p}$ is very similar. Let the set $\partial s = \{ (\vec{f}, o) \vert (\vec{f}, o) \in R^{\circ}, o = s \}$. We solve:
\begin{align}
    \label{eq:max-p}
    &\frac{\partial \ell_c}{\partial p_{\vec{r}}(s)} = \frac{\partial}{\partial p_{\vec{r}}(s)} \left[ \ell - \sum_{\vec{k}} \psi_{\vec{k}} \left( \sum_o p_{\vec{k},o} - 1 \right) \right]=0 \notag\\
    &\Leftrightarrow \psi_{\vec{r}} = \sum_{(\vec{f}, o) \in \partial s}  \frac{\omega_{\vec{f},o}(\vec{r})}{p_{\vec{r}}(s)} \notag\\
    &\Leftrightarrow \sum_{n}\psi_{\vec{r}} p_{\vec{r}}(s) = \psi_{\vec{r}} = \sum_{(\vec{f}, o) \in R^{\circ}}  \omega_{\vec{f},o}(\vec{r})
\end{align}
Finally, combining Eq.~\ref{eq:m-theta} with Eq.~\ref{eq:phi}, and the two last lines of Eq.~\ref{eq:max-p}, the maximization equations are:
\begin{equation}
    \begin{cases}
        \label{eq:maxTheta}
        \theta_{mn}^{(a(m))} = \frac{\sum_{(\vec{f},o) \in \partial m} \sum_{\vec{k} \in \partial n} c_{\vec{k_{i_m}}}(n) \omega_{\vec{f},o}(\vec{k})}{C_m}\\
        p_{\vec{r}}(s) = \frac{\sum_{(\vec{f}, o) \in \partial s} \omega_{\vec{f},o}(\vec{r})}{\sum_{(\vec{f}, o) \in R^{\circ}} \omega_{\vec{f},o}(\vec{r})}
    \end{cases}
\end{equation}
\myequations{\ \ \ \ SIMSBM - M-step}
From Eq.~\ref{eq:maxTheta} we see that for a given $(K_{a(f_1)}, ..., K_{a(f_N)})$, one iteration of the EM algorithm runs in $\mathcal{O}(\vert R^{\circ} \vert)$.

\subsubsection{SIMSBM generalizes several state-of-the-art models}
\label{SIMSBM-generalizes}
The formulation of the SIMSBM allows a great flexibility on modelling choices. Now, we develop how our formulation allows to recover several state-of-the-art models. Throughout this section, we denote input \glspl{entity} of distinct \glspl{type} by different letters (e.g., $f_1$ is not of the same \gls{type} as $g_1$), and the model's output as $o$. The set of corresponding \gls{membership} matrices for each \gls{type} is noted as $\Theta = \{ \theta^{(f)}, \theta^{(g)}, ... \}$ and one edge of the multipartite \glspl{cluster}-\gls{interaction} tensor is noted $(p_{k(f_1), k(f_2), ...}(o))$ where $k(f_i)$ is one of the possible \gls{cluster} indices for an input \gls{entity} of \gls{type} $f$.

\paragraph{Nomenclature}
We must define a nomenclature to refer to each special case of the SIMSBM --what input \gls{entity} \glspl{type} are considered, and how many \glspl{interaction} for each \gls{type}. We use the notation SIMSBM(number \glspl{interaction} \gls{type} 1, number \glspl{interaction} \gls{type} 2, ...). For instance, SIMSBM(2,3) represents a case where the SIMSBM considers two \glspl{type} of input \glspl{entity}, with the first one \gls{interacting} as pairs with other \glspl{entity} of the same \gls{type}, and the second one \gls{interacting} as triples with \glspl{entity} of the same \gls{type}. The corresponding data has a shape $(f_1, f_2, g_1, g_2, g_3, o)$ where $f$ and $g$ are the two considered \glspl{type}.

\paragraph{MMSBM}
The historical \acrshort{MMSBM} has been proposed in \citep{Airoldi2008MMSBM} and is at the base of most models discussed in this section. MMSBM takes pairs $(f_1, o)$ as input data. We can recover this model with our framework by setting $\Theta = \{ \theta^{(f)} \}$. The multipartite network then becomes ``unipartite'', which is a simple one-layer clustering of \glspl{entity}. The probability of an output is defined by \glspl{entity}' \gls{cluster} \gls{membership} only. The tensor $p$ then takes the shape $p_{f_1}(o)$. Using the SIMSBM notation, this corresponds to SIMSBM(1).

\paragraph{Bi-MMSBM}
The Bi-MMSBM has first been proposed in \citep{Antonia2016AccurateAndScalableRS} and has since been applied on several occasions \citep{Tarres2019TMBM}. In this modelling, data is made of triplets $(f_1, g_1, o)$. Each \gls{entity} is associated with a node on a side of a bipartite graph ($f_i$'s on one side, $g_i$'s on the other) and edges represent the probability of an output $o$.
We recover the Bi-MMSBM with our model by setting $\Theta = \{ \theta^{(f)}, \theta^{(g)} \}$ and the bipartite \glspl{cluster} network tensor $(p_{k(f_1), k(g_1)}(o))$. This corresponds to SIMSBM(1,1).

\paragraph{T-MBM}
The T-MBM is a model proposed in \citep{Tarres2019TMBM} that goes a step further than \citep{Antonia2016AccurateAndScalableRS} by adding a layer to the bipartite network used to model quadruplet data $(f_1, f_2, g_1, o)$. This model considers \glspl{interaction} between \glspl{entity} of the same \gls{type} (as in Section~\ref{IMMSBM}, see below) by clustering $f_1$ and $f_2$ using the same \gls{membership} matrix, but does not account for nonlinear terms. 
We recover the T-MBM modelling by setting $\Theta = \{ \theta^{(f)}, \theta^{(g)} \}$ and $(p_{k(f_1), k(f_2), k(g_1)}(o))$. Our formulation allows to go further by adding an arbitrary number of layers to the multipartite networks proposed in \citep{Antonia2016AccurateAndScalableRS,Tarres2019TMBM}. This corresponds to SIMSBM(2,1).

\paragraph{IMMSBM}
The \acrshort{IMMSBM} we propose in Section~\ref{IMMSBM} models \glspl{interaction} between \glspl{entity} of the same \gls{type} to predict an output. The data takes the form $(f_1, f_2, o)$. Each input \gls{entity} is still associated with one node on either side of a bipartite graph, except that here the \gls{membership} matrix is shared between the two layers. The links between each pair of \glspl{cluster} represent the probability of an output $o$.
We recover the IMMSBM with our model by setting $\Theta = \{ \theta^{(f)} \}$ and the bipartite \glspl{cluster} network tensor $(p_{k(f_1), k(f_2)}(o))$. Importantly, our formulation allows to consider \glspl{interaction} between $n$ input \glspl{entity} instead of simply pair \glspl{interaction}. This corresponds to SIMSBM(2).

\paragraph{Indirect generalizations}
We did not detail the generalization of other families of block models because our algorithm does not readily fit these cases. However, it is worth mentioning that MMSBM has been developed as an alternative to Single \Gls{membership} SBM \citep{Yuchung1987} that allows more flexibility \citep{Airoldi2008MMSBM}. Our model reduces to most existing SBM by modifying the definition of the entries of $\theta^{(n)}$. In the Single \Gls{membership} SBM, Eq.~\ref{eq:normTheta} reads $\theta_{f_n, k}^{(n)}=\delta_{k, x}$ where $x$ is one of the $K_n$ possible values for $k$ and $\delta$ is the Kronecker's delta. This means the \gls{membership} vector of each input \gls{entity} equals 1 for one \gls{cluster} only, and 0 anywhere else. Therefore, the optimization process is not the same as we described. In practice, optimization is done with greedy algorithms such as the simulated annealing \citep{CoboLopez2018SocialDilemma,Poux2021MMSBMMrBanks}.

It has also been shown in \citep{Antonia2016AccurateAndScalableRS}-Eq.7 that the Bi-MMSBM model generalizes matrix factorization. Therefore, it follows that SIMSBM also generalizes it. The underlying idea is to remove the multipartite network tensor $p$ and define \glspl{cluster} that are shared by both sides of the bipartite network. This way, \glspl{cluster} do not interact with each other because they are not embedded into a multipartite network; input \glspl{entity} on one side of the bipartite network belonging to one \gls{cluster} are solely linked to \glspl{entity} on the other side belonging to this same \gls{cluster}.

\subsubsection{Experiments}
\label{SIMSBM-XP}

\begin{table*}
\scriptsize
    \setlength{\lgCase}{2.0cm}
    \caption[SIMSBM - Datasets considered]{Datasets considered. The number of discrete values each input or output \gls{entity} \gls{type} can take is given between parentheses. We also provide the number of \glspl{cluster} for each \gls{entity} \gls{type} used in the experiments.}
    \label{table-DS}
	\centering
	\noindent\makebox[\textwidth]{\resizebox{\textwidth}{!}{
	\begin{tabular}{|p{0.7\lgCase}|p{3.05\lgCase}|p{0.85\lgCase}|p{0.90\lgCase}|p{0.3\lgCase}|p{0.6\lgCase}|}
    \cline{1-6}
    & Type of the input \glspl{entity} & \# \glspl{interaction} & Type outputs & $\vert R^{\circ} \vert$ & \# \glspl{cluster}\\
    \cline{1-6}
    
    MrBanks 1 & \{Player (280), Situation (7), Gender (2), Age (6)\} & Situations: 3 & User guess (2) & 16k & \{5,5,3,3\} \\
    MrBanks 2 & \{Player (280), Situation (7)\} & Situations: 3 & User guess (2) & 16k & \{5,5\} \\
    Spotify & \{Artists (143)\} & Songs: 3 & Artist (740) & 50k & \{20\} \\
    PubMed & \{Symptoms (13)\} & Symptoms: 3 & Disease (280) & 2M & \{20\} \\
    Imdb 1 & \{User (2502), Casting (809)\} & Casting: 2 & Rating (10) & 1M & \{10,8\} \\
    Imdb 2 & \{User (2502), Director (255), Casting (809)\} & None & Rating (10) & 700k & \{10,10,10\} \\

    \cline{1-6}
	\end{tabular}
	}}
\end{table*}

\paragraph{Range of application}
As shown in the previous section, our formulation generalizes several existing models from the state-of-the-art. Therefore, it is readily applicable to any of the datasets considered in these works. This includes recommendation datasets (movies \citep{Antonia2016AccurateAndScalableRS}, songs), medical datasets (symptoms-disease networks, drug \gls{interaction} networks \citep{Guimera2013DrugdrugSBM,Tarres2019TMBM}) and social behaviour datasets (social dilemmas \citep{CoboLopez2018SocialDilemma,Poux2021MMSBMMrBanks}, e-mail networks \citep{Godoy2016EmailNetwork,Tarres2019TMBM}). In general, it applies to datasets where there is a given number of input \glspl{entity} leading to a set of possible outputs. In this section, we propose to illustrate an application of our model on 6 different datasets. 

\input{Tables/Chapter_2/table-res-SIMSBM}

\paragraph{Datasets}
The \textbf{MrBanks datasets} has been gathered from a social experiment led in Barcelona in 2013, detailed in \citep{Guttieres2016MrBanks}. The experiment takes the form of a game where a player must guess whether a stock market curve will go up or down at the next time step. To do so, she can access various \glspl{piece of information}, from which we selected the most relevant subset based on the description in \citep{Guttieres2016MrBanks,Poux2021MMSBMMrBanks}: direction of the market on the previous day (up/down), whether she guessed right (yes/no), and an expert's advice who is correct 60\% of the time (up/down/not consulted). Those are the 7 \gls{interacting} \glspl{piece of information} that define a situation. If the model considers pair \glspl{interaction} for instance, a situation can be defined as ``market went down and user guessed wrong'', or ``market went up and expert advised up''. A triplet \gls{interaction} allows one to get the full picture according to the selected \glspl{piece of information}. In addition, we have access to the player’s age and gender. The goal is to predict whether the user will guess up or down given the available information.

For the \textbf{Spotify dataset} we collected user-made playlists on Spotify using the Spotipy python API. Our goal is to predict which next artist the user will add to the playlist, given the previous artists he already added. We consider the last 4 artists added by the user and their \gls{interaction} to guess the next one. Note that it often happens for an artist to be added several times in a row.

The \textbf{PubMed dataset} is made of medical reports we gathered using the PubMed API. We use provided keywords to isolate symptoms and diseases in the text, as in \citep{Zhou2014DiseaseSymptomNet}. Our goal is to guess which diseases are discussed in the report given the symptoms that are listed in the document. Our guess is that a combination of symptoms helps narrow the set of possible diagnoses. 

Finally, the \textbf{Imdb datasets} are provided and discussed in \citep{Harper2015ImdbDataset}. The original dataset comes with information about movies such as the lead actors starring in them and the movie's director. It also provides a list of users' ratings on movies. We aim to predict which rating a movie will get according to several combinations of parameters and their \glspl{interaction}: who directed the movie, who played in it, who gave the rating, etc.

The datasets we consider here are summarized in Table~\ref{table-DS}. 90\% of each dataset's documents are used as a training set, and the other 10\% are used as an evaluation set. Each iteration of the SIMSBM is run 100 times on every dataset. The EM algorithm stops once the relative variation of the likelihood falls below $10^{-4}$ for 30 iterations in a row. We present the average results over all the runs. The number of \glspl{cluster} has been chosen based on the existing literature on similar datasets (Imdb \citep{Antonia2016AccurateAndScalableRS}, MrBanks \citep{Poux2021MMSBMMrBanks}) and heuristic methods (Section~\ref{IMMSBM-XP}) for demonstration purposes; dedicated work would be needed to infer their optimal number for every dataset.

Finally, when SIMSBM is evaluated on a dataset containing more \glspl{interaction} than it is designed to consider, the model is trained on the lower-order corresponding dataset. For instance, imagine a dataset considering one \gls{type} \gls{interacting} three times. This dataset is made of one observation only $(1, 2, 3, o)$. A SIMSBM iteration that considers pair \glspl{interaction} will then be trained on triplets $(1,2,o), (1,3,o)$ and $(2,3,o)$, and evaluation will be performed accordingly. 

\paragraph{Baselines and evaluation}
Evaluation is done by considering test set entries $\left( \vec{f}, o_{true} \right)$. We ask the SIMSBM to provide a probability for each possible output $o \in O$ given the unobserved input context $\vec{f}$. We then compare the resulting probability distribution to $o_{true}$. We evaluate our results according to the maximum \textbf{F1} score, the precision at 1 (\textbf{\acrshort{P@}1}), the area under the ROC curve (\textbf{\acrshort{AUCROC}}), the area under the Precision-Recall curve (\textbf{\acrshort{AUCPR}}, or Average Precision). Since the problem is about multi-label classification, we consider the weighted version of these metrics --metrics are computed individually for each class, and averaged with each class’s weight being equal to the number of true instances in the class. The presented results are averaged over all 100 runs. We also consider the rank average precision (\textbf{\acrshort{RAP}}) and the normalized covering error (\textbf{\acrshort{NCE}}, only here lower is better).\footnote{\label{note1}We used the sklearn Python library implementation.}
We compare to several standard baselines:
\begin{itemize}
    \item \acrshort{BL}: the most naive baseline, where each output is predicted according to its frequency in the training set, without any context.
    \item \acrshort{NB}$^{\text{\ref{note1}}}$: the Naive Bayes baseline assumes conditional independence between the input \glspl{entity} and updates the posterior probability according to Bayes law.
    \item \acrshort{KNN}$^{\text{\ref{note1}}}$: K-nearest-neighbours. The output probabilities for a given \glspl{entity} array are defined according to a majority vote among the most similar \glspl{entity} arrays.
    \item \acrshort{NMF}$^{\text{\ref{note1}}}$ and \acrshort{TF}: Tensor Factorization baselines. For TF, we use the implementation provided by authors of \citep{Karatzoglou2010TF}. As discussed in the introduction, to make these methods fit our problem, we have to define a continuous quantity to train the model. Instead of requiring an additional model to map possible outputs into a continuous space, we train the model on the frequency of outputs in a given context. Since NMF can only consider one \gls{entity} as an input, we consider each different context as an \gls{independent} \gls{entity}. Outputs are added as an additional dimension to the data matrix instead of being a proper objective --their frequency is now the objective. The TF baseline is run for the same number of \glspl{cluster} as for the SIMSBM.
    \item \acrshort{MMSBM} \citep{Airoldi2008MMSBM}, \acrshort{Bi-MMSBM} \citep{Antonia2016AccurateAndScalableRS}, \acrshort{IMMSBM} (Section~\ref{IMMSBM}), \acrshort{T-MBM} \citep{Tarres2019TMBM}: as discussed before, each of these models are particular cases of SIMSBM. For presentation purpose, for each model, we keep the SIMSBM notation and indicate in superscript which of these it reduces to in this context. MMSBM=$^{\text{a}}$; Bi-MMSBM=$^{\text{b}}$; IMMSBM=$^{\text{c}}$; T-MBM=$^{\text{d}}$.
\end{itemize}

\subsubsection{Discussion}
We present our main results in Table~\ref{table-res-SIMSBM}. In this table, we see that our formulation systematically outperforms the proposed baselines, as well as the ones it generalizes. In most cases, the possibility to add a layer or to consider higher-order \glspl{interaction} improves the performance over the existing baselines (MMSBM, Bi-MMSBM, IMMSBM and T-MBM). About the Spotify dataset, as stated before, artists are often added to a playlist in a row, leading to the probability of the next artist being the same as the one immediately before her to be higher. In this context, adding \gls{interaction} terms adds noise to the modelling. This explains why the triple \glspl{interaction} version of SIMSBM does not perform better than its pair-\glspl{interaction} \citep{Antonia2016AccurateAndScalableRS} or no-\gls{interaction} \citep{Airoldi2008MMSBM} iterations.

We also ran SIMSBM on the datasets considered in \citep{Antonia2016AccurateAndScalableRS} and \citep{Poux2021MMSBMMrBanks} in the corresponding model's specifications and obtained comparable results as theirs (provided in Appendix, Section~\ref{SI-SIMSBM-Additional-Results}). In the case of \citep{Poux2021MMSBMMrBanks}, the authors propose to describe a given situation in form of a unique key, where each key is \gls{independent} of the others. Interestingly, the formulation with triple \glspl{interaction} improves the results on the same dataset the authors provided. This is because the constituents of a situation are not \gls{independent} anymore but instead behave as elementary \gls{interacting} pieces of context, which provides a more accurate description of reality: a situation is not considered as a whole anymore, but instead as the combination of several \glspl{piece of information}.

\subsubsection{Conclusion}
In this section, we developed a global framework, SIMSBM, that generalizes several models from the literature as particular cases, such as MMSBM, Bi-MMSBM, IMMSBM and T-MBM. Our derivation accounts for the nonlinearity induced by the \glspl{interaction}, which has not been taken into account in state-of-the-art works.

This results in a highly flexible model that can be applied to a broad range of problems, as shown through systematic evaluation of the proposed formulation on several real-world datasets. In particular, we cited throughout the text several experimental studies conducted in medicine, social behaviour and recommendation using special cases of our model; using alternative iterations of the SIMSBM framework may help further improve the description and understanding of the \gls{interacting} processes at stake between an arbitrary greater number of input \glspl{entity}.

To summarize, we developed a framework that allows to consider MMSBMs for any number of \gls{entity} \glspl{type} that can have arbitrarily high-order \glspl{interaction}. This answers the two challenges raised in Section~\ref{SBM-SotA}. We now propose to restrict our study of \glspl{interaction} in information \gls{spread} to a special case of the SIMSBM --the IMMSBM.

\subsection{IMMSBM -- A study of pair \glspl{interaction}}
\label{IMMSBM}
\textit{This work has been published, see \citep{Poux2021IMMSBM}.}

\paragraph{About this section}
As stated at the end of Section~\ref{SIMSBM}, the \acrshort{IMMSBM} introduced here is a special case of the SIMSBM, noted SIMSBM(2). As such, the motivations and derivations detailed in Section~\ref{SIMSBM-model} are readily applicable to the case at hand. 

In particular, Section~\ref{IMMSBM-model} and Section~\ref{IMMSBM-inference-params} may seem redundant at first glance. However, these sections' argumentation is specifically focused on \gls{interaction} modelling. Besides, we present an alternative, simpler, derivation of the EM algorithm for SIMSBM using Jensen's inequality in Section~\ref{IMMSBM-inference-params}.

The novel results of SIMSBM(2), or IMMSBM, applied to \glspl{interaction} modelling are discussed from Section~\ref{IMMSBM-XP} to Section~\ref{IMMSBM-conclusion}.

\subsubsection{IMMSBM -- \Gls{interacting} MMSBM}
\label{IMMSBM-model}

\paragraph{MMSBM to model \glspl{interaction}}
In this section, we consider a specific iteration of the \acrshort{SIMSBM} framework, referred to as \acrshort{IMMSBM} (for \Gls{interacting} Mixed \Gls{membership} Stochastic Block Model). Studying \glspl{interaction} using a model built on the mixed-\gls{membership} SBMs allows us to assume that each \gls{entity} does not have only one \gls{membership}. This is in line with what is expected for real-life situations, where \glspl{entity} can belong to several \glspl{cluster} simultaneously (a song can be rock \textit{and} blues, a word can be used in two topics, etc.).
The IMMSBM requires no prior information on the system and its numerical implementation is possible via the scalable Expectation-Maximization algorithm detailed in the previous section. The EM complexity scales linearly with the size of the dataset. Our method offers a better predictive power than non-\gls{interacting} baselines.

The goal of the model is to predict the most likely result of an \gls{interaction} between two \glspl{entity} ($i$ and $j$ in Fig.~\ref{figModel}). Typically, IMMSBM yields the probability of a tweet getting retweeted given a user's previous exposure to pairs of tweets, or the product the most likely to be bought given item pairs in a user's purchase history. Another example about assisted medical diagnosis: observing the words ``fatigue'' and ``cough'' in a medical report is more likely to imply the observation of ``flu'' than ``anaemia'', despite ``anaemia'' being often associated with the ``fatigue'' symptom. The model will group data into \glspl{cluster} (each user's \gls{membership} is encoded in the matrix $\theta$, see Fig.~\ref{figModel}) that interact symmetrically with each other (\glspl{interaction} tensor \textbf{p}), resulting in a probability over the possible outputs to appear (histograms Fig.~\ref{figModel}-left). We have no prior knowledge of the \gls{content} of the groups, and we need to set the number of \glspl{cluster} $K$.
\begin{figure}
\centering
    \noindent\makebox[\textwidth]{
    \includegraphics[width=.45\textwidth]{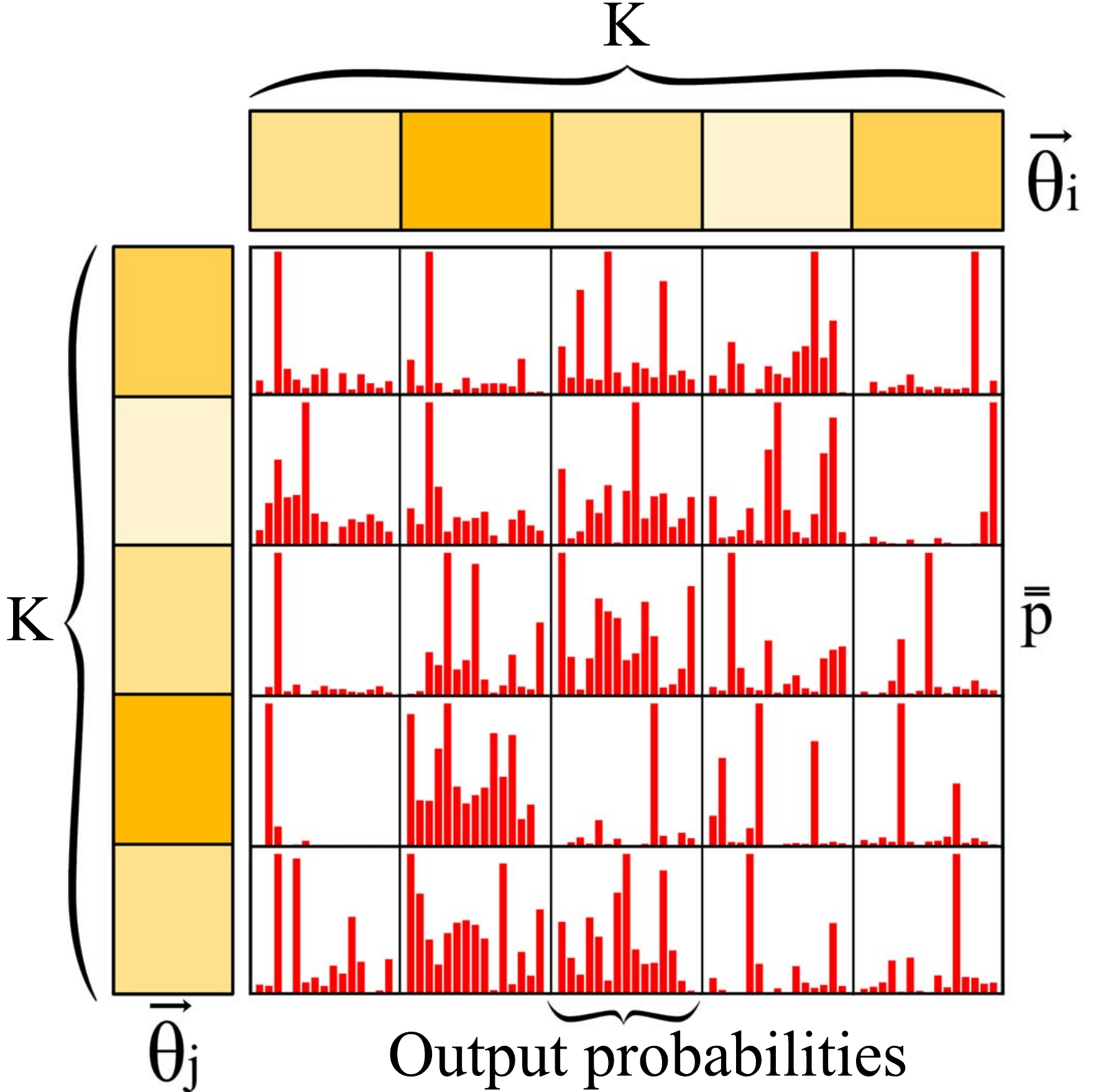}
    \includegraphics[width=.45\textwidth]{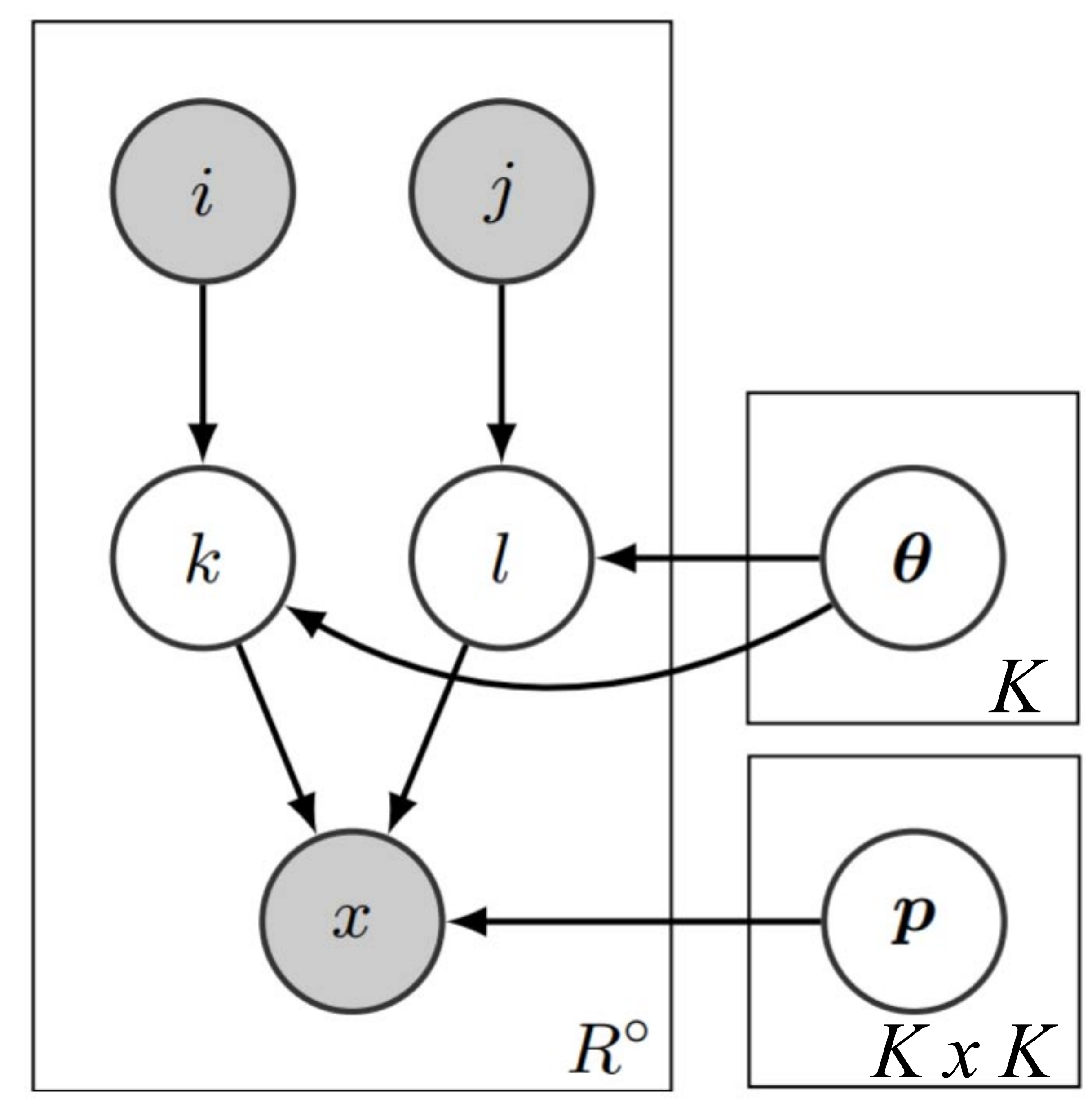}
    }
    \caption[IMMSBM - Graphical representation of the model]{\textbf{Illustration of the model} --- \textbf{(Left)} Schema of IMMSBM for a single pair of \glspl{entity} $(i,j)$ (which could be ``fever'' and ``cough'' for instance). Input \glspl{entity} are grouped into K \glspl{cluster} in different proportions; the proportion to which they belong to each \gls{cluster} is quantified by a $\theta$ matrix (dimension $[I \times K]$ where $I$ is the input space). The \glspl{cluster} then interact to generate a probability distribution over the output \glspl{entity} defined by the \glspl{interaction} tensor \textbf{p} (dimension $[K \times K \times O]$ where $O$ is the output space). \textbf{(Right)} Alternative representation of the IMMSBM as a graphical model. To generate each output, for each observation ($i$, $j$, $x$) in the set $R^{\circ}$, a \gls{cluster} ($k$ and $l$) is drawn for each input \gls{entity} ($i$, $j$) from a distribution encoded in the matrix $\theta$. The generated output $x$ is drawn from a multinomial distribution conditioned by the previously drawn \glspl{cluster} k and l encoded in $\Vec{p}$.} \label{figModel}
\end{figure}

\paragraph{Model}
\label{likelihoodSec}
We refer to the \gls{interacting} \glspl{entity} as input \glspl{entity} $(i,j) \in I^2$, and to an output as $x \in O$. $I$ is the input space (the \glspl{entity} that interact: products, symptoms, or songs for instance) and $O$ is the output space (the \glspl{entity} resulting from the \gls{interaction} in an answer, diagnosed diseases for instance). We illustrate in Fig.~\ref{figModel}-left how the input space and output space are related according to IMMSBM: input \glspl{entity} are clustered, and the \gls{interaction} between those \glspl{cluster} yields a probability distribution over the possible output. Note that $I$ and $O$ can be disjoint, unlike in \citep{Myers2012CoC}. In the case of retweet prediction, $I$ accounts for tweets in a user's history, and $O$ for retweeted tweets --that do not necessarily appear in the user's feed.
As an alternative visualization, we present the graphical generative model of the IMMSBM in Fig.~\ref{figModel}-right. The observed data then takes the form of triplets ($i$, $j$, $x$) signifying that the combined presence of input \glspl{entity} i and j leads to the output \gls{entity} x.

We assume there are regularities in the studied datasets. Given subsets of inputs may exhibit a similar behaviour. In the medical dataset example, if symptoms such as ``fever'' and ``pallor'' often come in pairs, they are considered as similar regarding the diagnosis. They belong to the same \gls{cluster}. We define the \gls{membership} matrix $\theta$ associating each input \gls{entity} with \glspl{cluster} in different proportions, such that $\theta_i$ is a $[1 \times K]$ vector with $\sum_k^K \theta_{i,t}=1$. Note that unlike in single \gls{membership} stochastic block models an \gls{entity} does not have to belong to only one group \citep{Funka2019ReviewSBM}. Given the possible semantic variation of \glspl{entity} (polysemy of words in natural languages -e.g., ``like'', ``swallow''-, symptoms with various causes in medicine -``headache'', ``fever''-, etc.), an approach \textit{via} a mixed-\gls{membership} clustering is more in line with reality.

We then define the \gls{cluster} \glspl{interaction} tensor $p_{k,l}(X_{k,l}=x)$ whose dimensions are ${[K \times K \times O]}$ as the probability that the \gls{interaction} between \glspl{cluster} $k \in K$ and $l \in K$ gives rise to the output $x \in O$.
By definition, $\sum_{x} p_{k,l}(X_{k,l}=x) = 1 \ \forall k,l$. The role of those two quantities is schematized in Fig.~\ref{figModel}.

We choose to consider only one \gls{membership} matrix $\theta$ for all of the inputs, instead of one per input entry as in \citep{Antonia2016AccurateAndScalableRS}. It enforces the model to be permutation-independent with respect to the input \glspl{entity}. Observation ($i$, $j$, $x$) is equivalent to ($j$, $i$, $x$). This follows the idea of \citep{Beutel2012InteractingViruses} where it is assumed that the \gls{interaction} between two viruses is symmetric, meaning the \gls{interaction} influences both viruses with the same magnitude. There is no need to consider a different clustering for inputs $i$ and $j$, which motivates the use of a single \gls{membership} matrix $\theta$. This differs from other recent works on \gls{interaction} modelling that do not consider permutation-invariant clustering \citep{Christakopoulou2014HOSLIM,Steck2021Interactions} or the non-linearity induced by symmetric \glspl{interaction} \citep{Myers2012CoC,Tarres2019TMBM}.

We now propose to define the \glspl{entity} \glspl{interaction} tensor $P_{i,j}(X_{i,j}=x)$, representing the probability that the \gls{interaction} between inputs i and j implies the output x as:
\begin{equation}
    P_{i,j}(X_{i,j}=x) = \sum_{k,l} \theta_{i,k}\theta_{j,l}p_{k,l}(X_{k,l}=x)
\end{equation}
For the sake of brevity, from now on we will refer to $p_{k,l}(X_{k,l}=x)$ as $p_{k,l}(x)$. We define the likelihood of the observations given the parameters as:
\begin{equation}
    \label{eqLikelihood}
    P(R^{\circ} \vert \theta, p) = \prod_{(i,j,x) \in R^{\circ}} \sum_{k,l} \theta_{i,k}\theta_{j,l}p_{k,l}(x)
\end{equation}
\myequations{\ \ \ \ IMMSBM - Likelihood}
where $R^{\circ}$ denotes the set of triplets in the training set (input, input, output). Note that the remaining triplets $R \setminus R^{\circ}$ are used as test set.

From the definitions above, it is straightforward to show that IMMSBM corresponds to the special case SIMSBM(2), detailed in the previous section.

\subsubsection{Inference of the parameters}
\label{IMMSBM-inference-params}
\paragraph{Expectation step}
\label{IMMSBM-Estep}
We detailed a general derivation of the E-step of EM algorithm for mixture models in Section~\ref{SIMSBM-model}. Here, we present an alternative (and quicker) derivation of the E-step using Jensen's inequality. The two methods yield identical results.
Taking the logarithm of the likelihood as defined in Eq.~\ref{eqLikelihood}, denoted $\ell$, we have:
\begin{equation}
\label{eqJensen}
    \begin{split}
         \ell &= \sum_{(i,j,x) \in R^{\circ}} \ln \sum_{k,l} \theta_{i,k}\theta_{j,l}p_{k,l}(x) \\
         &= \sum_{(i,j,x) \in R^{\circ}} \ln \sum_{k,l} \omega_{i,j,x}(k,l) \frac{\theta_{i,k}\theta_{j,l}p_{k,l}(x)}{\omega_{i,j,x}(k,l)} \\
         & \geq \sum_{(i,j,x) \in R^{\circ}} \sum_{k,l} \omega_{i,j,x}(k,l) \ln \frac{\theta_{i,k}\theta_{j,l}p_{k,l}(x)}{\omega_{i,j,x}(k,l)}
    \end{split}
\end{equation}
\myequations{\ \ \ \ IMMSBM - E-step}
We used Jensen's inequality to go from the 2nd to 3rd line.
The inequality in Eq.~\ref{eqJensen} becomes an equality for: \begin{equation}
\label{eqOmega}
    \omega_{i,j,x}(k,l) = \frac{\theta_{i,k}\theta_{j,l}p_{k,l}(x)}{\sum_{k',l'} \theta_{i,k'}\theta_{j,k'}p_{k', l'}(x)}
\end{equation}
where $\omega_{i,j,x}(k,l)$ is interpreted as the probability that the observation $(i,j,x)$ is due to i belonging to the group $k$ and j to $l$, that is the expectation of the likelihood of the observation $(i,j,x)$ with respect to the latent variables $k$ and $l$. Therefore, Eq.~\ref{eqOmega} is the formula for the expectation step of the EM algorithm. As stated earlier, an alternative derivation of the expectation step is discussed in \citep{Bishop2006} and in Section~\ref{SIMSBM-Estep}.

\paragraph{Maximization step}
\label{maximization}
This step consists in maximizing the likelihood using the parameters of the model $\theta$ and p, independently of the latent variables.
To take into account the normalization constraints, we introduce the Lagrange multipliers $\phi$ and $\psi$. Following this, the constrained log-likelihood reads:
\begin{equation}
\label{eqLagMult}
    \ell_c = \ell - \sum_i (\phi_i \sum_{t} \theta_{i,t} - 1) - \sum_{k,l} (\psi_{k,l} \sum_{x} p_{k,l}(x) - 1)
\end{equation}
We first maximize $\ell_c$ with respect to each entry $\theta_{mn}$.
\begin{equation}
\label{eqDerivTheta}
    \begin{split}
        \frac{\partial \ell_c}{\partial \theta_{mn}} =& \frac{\partial \ell}{\partial \theta_{mn}} - \phi_m = 0\\
        =& \sum_{\partial m} \left( \sum_{l} \frac{\omega_{m,j,x}(n,l)}{\theta_{mn}} + \sum_{k} \frac{\omega_{j,m,x}(k,n)}{\theta_{mn}} \right) - \phi_m\\
        =& \frac{1}{\theta_{mn}}\sum_{\partial m} \left( \sum_{t} \omega_{m,j,x}(n,t) + \omega_{j,m,x}(t,n) \right) - \phi_m\\
        \Leftrightarrow\ \theta_{mn} =& \frac{1}{\phi_m}\sum_{\partial m} \sum_{t} \omega_{m,j,x}(n,t) + \omega_{j,m,x}(t,n)
    \end{split}
\end{equation}
where $\partial m = \{ (j,x) \vert (m,j,x) \in R^{\circ} \}$ stands for the set of observations in which the entry $m$ appears. Note that summing over this set in line 2 of Eq.~\ref{eqDerivTheta} implies that the relation between two inputs \glspl{entity} is symmetric --if ${(i,j,x) \in R^{\circ}}$ implies ${(j,i,x) \in R^{\circ}}$. Summing the last line of Eq.~\ref{eqDerivTheta} over $n \in K$ then multiplying it by $\phi_m$, we get an expression for $\phi_m$:
\begin{equation}
        \phi_m\sum_n \theta_{mn} = \phi_m = \sum_{\partial m} \sum_{t,n} \omega_{m,j,x}(n,t) + \omega_{i,m,x}(t,n) = \sum_{\partial m} 2 = 2 \cdot n_m
\end{equation}
Where $n_m$ is the total number of times that m appears as an input in $R^{\circ}$. Finally, plugging back this result in Eq.~\ref{eqDerivTheta}, we get:
\begin{equation}
    \label{eqMaxTheta}
    \theta_{mn} = \frac{\sum_{\partial m} \left( \sum_{t} \omega_{m,j,x}(n,t) + \omega_{i,m,x}(t,n) \right)}{2.n_m}
\end{equation}

Following the same line of reasoning for p, we get:
\begin{equation}
    \label{eqMaxp}
    p_{r,s}(l) = \frac{\sum_{\partial l} \omega_{i,j,l}(r,s)}{\sum_{(i,j,l') \in R^{\circ}} \omega_{i,j,l'}(r,s)}
\end{equation}
\myequations{\ \ \ \ IMMSBM - M-step}
The set of Eq.~\ref{eqMaxTheta} and Eq.~\ref{eqMaxp} constitutes the maximization step of the EM algorithm. They hold only if the input \glspl{entity} \glspl{interaction} are symmetric (e.g., when $\{ (j,x) \vert (m,j,x) \in R^{\circ} \} = \{ (i,x) \vert (i,m,x) \in R^{\circ} \}$), which is what we aimed to do.
It is worth noting that the proposed algorithm offers linear complexity with the size of the dataset $\mathcal{O}(\vert R^{\circ} \vert)$ provided the number of \glspl{cluster} is constant, while guaranteeing convergence to a local maximum.

\paragraph{Instantaneous derivation as a special case of SIMSBM}
The EM equations we just derived could be directly retrieved, given that IMMSBM is corresponds to the special case SIMSBM(2). We define the \gls{membership} matrices set $\Theta = \{ \theta^{(f)} \}$, the \gls{interaction} tensor $p_{k(f_1),k(f_2)}(o)$ and consider data shaped as $(f_1, f_2, o)$. Plugging these into Eq.~\ref{eq:expectation} and Eq.~\ref{eq:maxTheta} yields identical expressions as Eq.~\ref{eqOmega}, Eq.~\ref{eqMaxTheta} and Eq.~\ref{eqMaxp}.

\subsubsection{Experiments}
\label{IMMSBM-XP}
\paragraph{Datasets and evaluation protocol}
We assess the performance of IMMSBM on 4 different datasets. The \textbf{PubMed} dataset is built with 15,809,271 medical reports collected from the PubMed database as a good approximation for human-disease network \citep{Zhou2014DiseaseSymptomNet}. This dataset is not explicitly about recommender systems but provides an intuitive way to understand how our recommendation approach works by suggesting likely diseases given a collection of symptoms. 
The \textbf{Twitter} dataset with 139,098 retweets gathered in October 2010 associated with the 3 last tweets in the feed preceding each retweet \citep{Hodas2014DataSetTwitter}. The task is to infer which tweets a user is the most likely to retweet. A possible application would be a personalized recommendation of such tweets in the ``Trends for you'' Twitter section. 
The \textbf{Reddit} dataset with the entirety of posts in the subreddit r/news in May 2019 (225,485 message-answer relationships in total). We aim at predicting the \gls{content} of the answer given the incoming message. A possible application would be similar to what Gmail does when suggesting automated answers to an email given keywords present in the message.
Finally, \textbf{Spotify} dataset is built with 2,000 music playlists associated with keywords ``English'' and ``rock'' of random Spotify users. We predict the next song a user will add to a playlist given this user's history. A RS application would suggest a ranked list of songs to the user is likely to add to a playlist.
Each dataset is formed by associating every pair of inputs in a \textit{message} (i.e., a list of symptoms, a user's feed, a Reddit post, and a playlist's last artists) with an \textit{answer} (i.e., a disease, a retweet, a Reddit answer, an artist added to a playlist). This is illustrated in Fig.~\ref{IMMSBM-fig-corpus-pairs}. The building process of datasets is further detailed in Appendix, Section~\ref{SI-IMMSBM-datasets}, together with direct links to access them for possible replication studies.

\begin{figure}
    \centering
    \includegraphics[width=\textwidth]{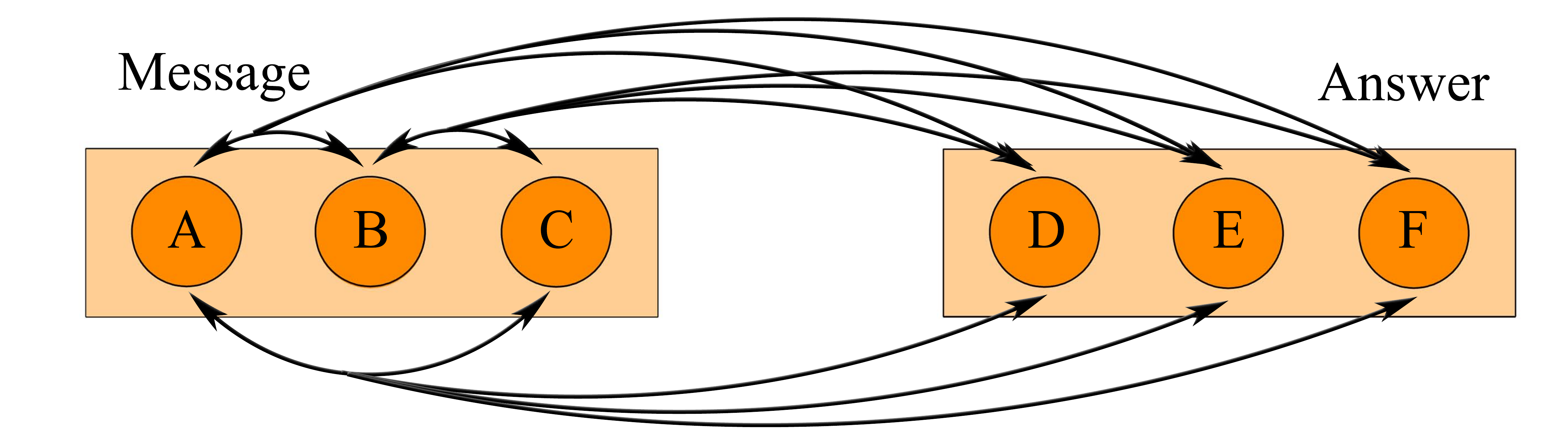}
    \caption[IMMSBM - Illustration of dataset generation]{\textbf{Illustration of dataset generation} --- Each dataset is organized as a list of message-answer pairs. Each message and each answer can comprise several \glspl{entity} (words, URLs, etc.). We consider every possible pair of \glspl{entity} in the message, and link each of them to each output in the answer to create our triplets dataset $\left( i,j,x \right)$}
    \label{IMMSBM-fig-corpus-pairs}
\end{figure}

From the raw datasets, we form the test set by randomly sampling 10\% of the (message$\rightarrow$answer) data entries. The 90\% entries left are used as a training set. The number of \glspl{cluster} is determined using the elbow method (see Fig.~\ref{figCalib}). The number of \glspl{cluster} considered for each corpus: 30 for PubMed, 15 for Twitter, 30 for Reddit, and 15 for Spotify. We perform 100 \gls{independent} runs, each with \gls{independent} random initialization of the parameters $\theta$ and $p$. The EM loop stops once the relative variation of the likelihood between two iterations is less than 0.001\%.

\paragraph{Baselines}
\label{baselines}

\begin{itemize}
    \item \textbf{Naive baseline} -- The naive baseline is simply the frequency of the outputs in the test set. This naive classifier predicts the value of every output independently of the inputs.

    \item \textbf{\acrshort{MMSBM}} -- We use the classical MMSBM as a second baseline. In this formulation, \glspl{interaction} are not taken into account \citep{Airoldi2008MMSBM}. Instead of considering triplets (input, input, output), we consider pairs (input, output). We then train this MMSBM baseline whose log-likelihood is defined as $\ell_{BL} = \sum_{(i,x) \in R^{\circ}} \ln \sum_{k} \theta_{ik} p_k(x)$ on the same datasets as in the main experiments. We make 100 \gls{independent} runs with random initialization and discuss the results of the highest likelihood run. This baseline provides a way to quantify the importance of \glspl{interaction} by comparing to the case where these are taken into account. We expect the baseline model to find results that are equivalent to the diagonal probabilities of the main model (i.e., similar to $P_{i,i}(x) = \sum_{k,l}\theta_{i,k}\theta_{i,l}p_{k,l}(x)$). This is because the diagonal $P_{i,i}(x)$ is supposed to account for the apparition of x given only the presence of i. Furthermore, this model provides insight into the generalization of the assumption made in \citep{Myers2012CoC}, stating that the probability of an output is mostly equal to the frequency of this output.

    \item \textbf{Perfect modelling} - Upper limit to prediction -- We compare our results to an upper limit to predictions. In most situations, the dataset simply does not allow for perfect performances. Consider as an example a case where the test set contains twice the triplet (``fever'', ``pallor'', ``influenza'') and once the triplet (``fever'', ``pallor'', ``anaemia''): a model yielding a single value for the input pair (``fever'', ``pallor'') cannot make a prediction better than 66\%. In Appendix, Section~\ref{SI-IMMSBM-upperLim}, we develop a general method to derive this upper limit to predictions.
\end{itemize}

\paragraph{Results}
\begin{figure}
    \includegraphics[width=1.0\textwidth]{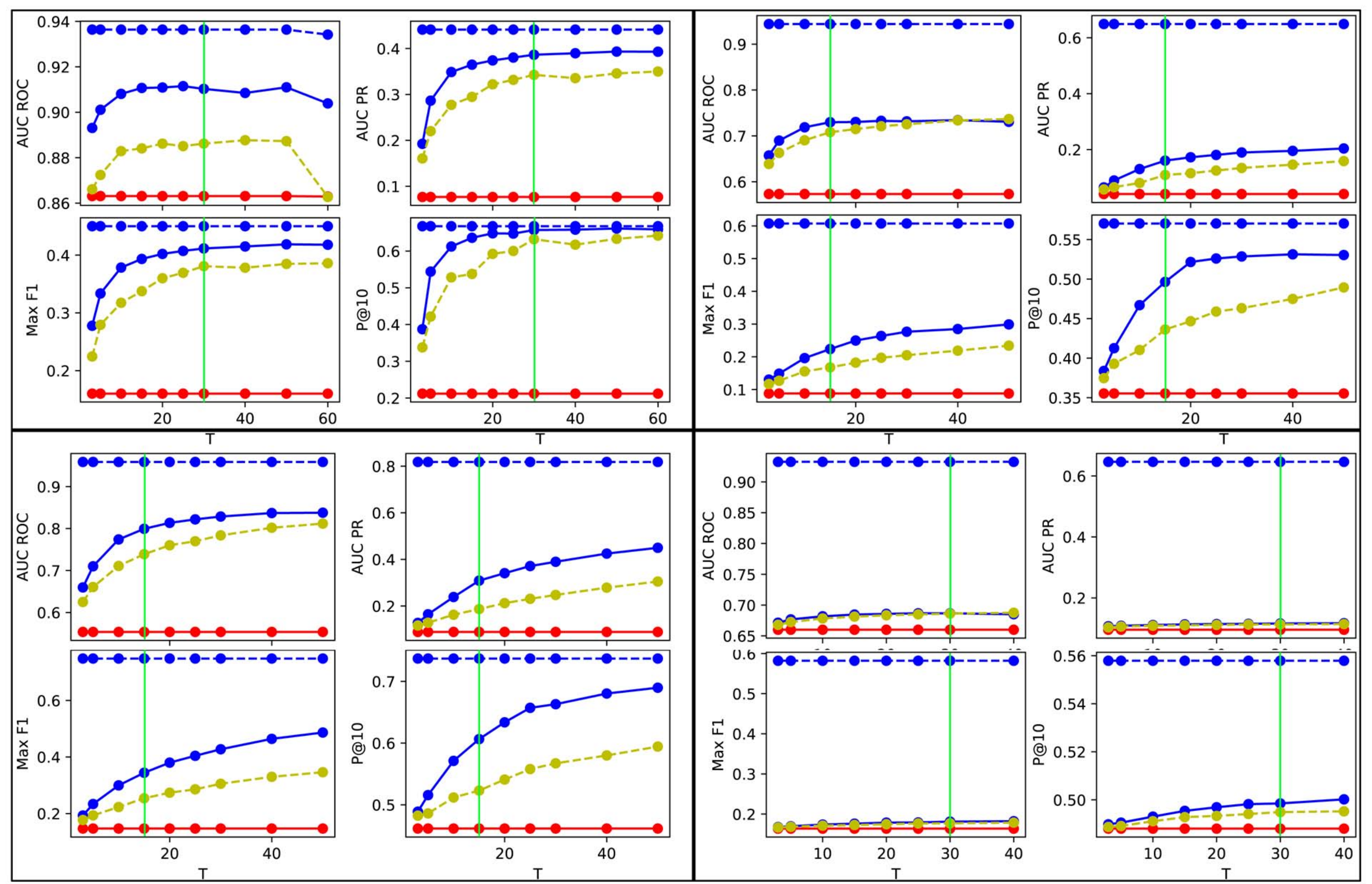}
    \caption[IMMSBM - Choosing the number of \glspl{cluster}]{\textbf{Choosing the number of \glspl{cluster}} --- Performance variations on all the metrics for every dataset considered. Dashed blue line: upper limit to performances ; blue line: IMMSBM ; yellow dashed line: MMSBM ; red line: naive baseline. Top left: PubMed ; top right: Spotify ; bottom left: Twitter ; bottom right: Reddit. The vertical green line shows the selected number of \glspl{cluster}; it is chosen using the AIC criterion, which matches with the elbow of the various metrics considered.} \label{figCalib}
\end{figure}
The metrics we use to assess the performance of our model are the \textbf{max-F1} score, the area under the receiver operating characteristic curve (\textbf{\acrshort{AUCROC}}) curve and the precision@10 (\textbf{P@10}). For all of these quantities, the closer to 1 they are, the better the performance.

For evaluation, we adopt a guessing process as follows. For every pair of inputs, we compute the probability vector for the presence of every possible output. Then we predict all of the probabilities larger than a given threshold to be ``present'', and all the others to be ``absent''. Comparing those predictions with the observations in the test set, we get the confusion matrix for the given threshold. We then lower the threshold and repeat the process to compute the various metrics. 

We recall that we do not compare our approach to \citep{Myers2012CoC} because its formulation does not allow us to make a prediction on exogenous outputs -- when the output is not part of the input pair.

\input{Tables/Chapter_2/table-res-IMMSBM}

In Table~\ref{tabMetrics-IMMSBM}, we show the performances of our model compared with the baselines introduced in Section~\ref{baselines}. We see that IMMSBM outperforms the proposed baselines in most cases. As expected, taking \glspl{interaction} between \glspl{entity} into account leads to improved accuracy in the prediction of the test outputs. This correlates to the conclusions drawn in \citep{Myers2012CoC}, stating the importance of \glspl{interaction} in real-world phenomenon modelling.

We notice however that accounting for \glspl{interaction} did not lead to a significant improvement in performance over the two baselines on the Reddit corpus. The lack of improvement can be imputed to the dataset itself. The Reddit dataset contains few observations for every possible pair, due to the wide range of available vocabulary of natural language \citep{Loreto2016}. Likely, the model has not been trained with enough data to learn significant regularities in pair \glspl{interaction}. This can also be seen during the building of the test set: approximately one-half of the pairs have never been observed in the training set. As future perspectives, it might be interesting to answer this problem by considering a corpus of pre-clustered \glspl{entity} instead of \gls{independent} named \glspl{entity}, hence reducing the vocabulary range and adding to the regularity of the dataset.

Our results show that taking \glspl{interaction} between \glspl{entity} into account is particularly relevant in the case of the PubMed corpus (98.2\% of the maximum reachable precision@10 vs 93.8\% for the non-\gls{interacting} baseline). It seems reasonable to consider that a diagnosis is better determined by joint observation of given symptoms, and not only by the sum of their individual effect. The \gls{interaction} aspect is especially relevant given the small number of observed symptoms (322) used to predict the possible diseases (4,442).

For the Twitter corpus, we confirm the results of \citep{Myers2012CoC} on the importance of \glspl{interaction} between URLs in their \gls{spread}. Our model consequently outperforms the non-\gls{interacting} baseline.

Finally, our model performs better than the baselines on the Spotify dataset. In particular, it achieves better prediction for the top-10 artists one would listen to (+7.6\%). A good P@10 precision is of key interest in the application of any model to playlist building and recommender systems in general. Taking into account artists’ \gls{interaction} clearly added to the level of prediction details.

\subsubsection{Discussion}
\label{IMMSBM-discussion}
\comment{
\paragraph{A good estimate of intrinsic \gls{virality}}
First of all, we need to confirm our hypothesis that the diagonal elements i=j of $P_{i,j}(x)$ account for the intrinsic \gls{virality} of an input on an output. To confirm it, we compare the results of our model with the ones yielded by the non-\gls{interacting} baseline. By design, it only considers the probability of an output given a single input (without any \gls{interaction}), which corresponds to the definition of \gls{virality} (that is the probability of observing an output given only one input, see Section~\ref{baselines} - MMSBM). We therefore have access to the true \gls{virality} of an \gls{entity}, noted $P_{i,BL}(x)$. 

To assess whether the diagonal elements i=j of $P_{i,j}(x)$ account for the intrinsic \gls{virality} on an output, we compare them with $P_{i,BL}(x)$. More precisely, we focus on the average of their absolute difference for all the inputs i, noted $\Delta_i (x)=\vert P_{i,i}(x)-P_{i,BL}(x) \vert$. The closer to 0 the value, the better our model accounts for intrinsic \gls{virality}. The results are presented in Table~\ref{tabQtyInter} (line 1). We see that the \gls{virality} yielded by our model is close to the actual \gls{virality} (0 to 5\% of error depending on the corpus). Therefore, we conclude that our model accurately accounts for the intrinsic \gls{virality} of \glspl{entity}.

We also compare the actual \gls{virality} to the hypothesis made in \citep{Myers2012CoC}, that is: intrinsic \gls{virality} is equal to the fraction of times an input gives rise to itself in the output. We denote this quantity $P^{H}(i)$ and compare it to the non-\gls{interacting} baseline. We focus once again on the average of the absolute difference $\Delta_i^H (x)=\vert P_{i,i}(i)-P^{H}(i) \vert$. Note that due to the definition of $P^{H}(i)$, this quantity cannot be calculated for the PubMed corpus, since input symptoms do not give rise to symptoms as output (but only to diseases). The results testing the hypothesis made in \citep{Myers2012CoC} are presented Table~\ref{tabQtyInter} (line 2). Here we see that the \gls{virality} as defined in this work differs from the actual \gls{virality} up to 40\%, and is not defined for datasets where the output space is disjoint from the input space. Therefore, we conclude that it lacks generality.

\begin{table}
\caption[IMMSBM - \Gls{virality} estimation]{Results on the accuracy of the \gls{virality} as defined in our model and in \citep{Myers2012CoC}. Lower is better. Average $\Delta_{i}^H(x)$ is not defined for the PubMed corpus because the definition proposed in \citep{Myers2012CoC} assumes that the input space is the same as the output space. \textbf{Line 1}: average of the absolute difference between the \gls{virality} as defined in our model and the actual \gls{virality}. \textbf{Line 2}: average of the absolute difference between the \gls{virality} as defined in \citep{Myers2012CoC} and the actual \gls{virality}. \label{tabQtyInter}}
\centering
\setlength{\lgCase}{2.6cm}
\begin{tabular}{ |p{1.3\lgCase}|p{1.\lgCase}|p{0.65\lgCase}|p{0.65\lgCase}|p{0.65\lgCase}|p{0\lgCase}}
 \cline{2-5}
 \multicolumn{0}{c|}{\rotatebox[origin=c]{0}{}} & \centering\rotatebox[origin=c]{0}{Pubmed} & \centering\rotatebox[origin=c]{0}{Twitter} & \centering\rotatebox[origin=c]{0}{Reddit} & \centering\rotatebox[origin=c]{0}{Spotify} & \\
 \cline{1-5}
 \centering Average $\Delta_{i}(x)$ & \centering 3.0\% & \centering 4.4\% & \centering 0.1\% & \centering 2.5\% & \\
 \centering Average $\Delta_{i}^H(x)$ & \centering Not defined & \centering 40.8\% & \centering 2.0\% & \centering 16.1\% & \\
 
 \cline{1-5}
\end{tabular}
\end{table}

In summary, Table~\ref{tabQtyInter} shows that our model correctly accounts for the \gls{virality} of an \gls{entity}, while the hypothesis made in \citep{Myers2012CoC} does not in every case. In particular, it does not work well for Twitter URLs. This last statement makes the conclusions of \citep{Myers2012CoC} debatable. On the other hand, our model allows for a correct quantification of the \gls{virality} with no need for a strong prior assumption.

}

\paragraph{Global impact of \glspl{interaction}}
IMMSBM infers \gls{virality} along with \gls{interaction} terms and yields better results than state-of-the-art methods (see Table~\ref{tabMetrics-IMMSBM}), which provides solid ground for analysing the effect of information \gls{interaction}. Close analysis of \glspl{interaction} between \glspl{piece of information} has been little considered in literature --what lexical fields, groups of symptoms, musical genres, and kinds of tweets interact with each other.
We can evaluate the importance of \glspl{interaction} between \glspl{entity} based on inferred \gls{virality}. We study two quantities for each corpus: the overall relative impact of the \glspl{interaction} on the probability of an \gls{outcome} and the contribution of each pair of \glspl{cluster} in the modification of \gls{outcome} probabilities.

\begin{figure}
    \centering
    \includegraphics[width=\textwidth]{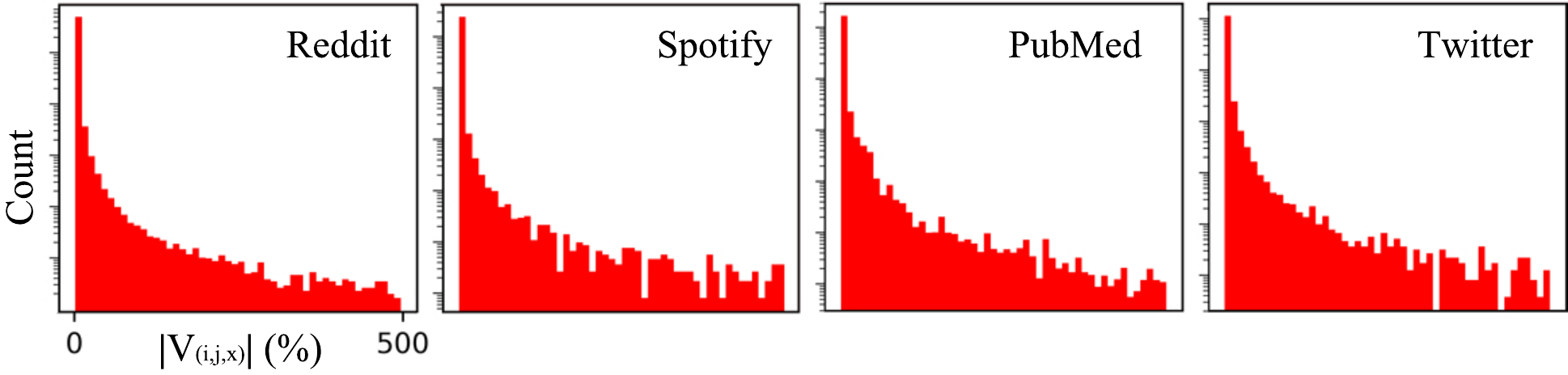}
    \caption[IMMSBM - Histogram of relative impact of \glspl{interaction}]{\textbf{Relative impact of \glspl{interaction}} --- Histogram of the relative impact of \glspl{interaction} on the base \gls{virality} of outputs. Overall, most \glspl{interaction} do not lead to any notable change in the probability of an output.}
    \label{fig-IMMSBM-histRelInter}
\end{figure}

To evaluate the global impact of the \glspl{interaction}, we compute the relative change of probability according to the inferred \gls{virality} for each triplet $\left( i,j,x \right)$, noted $V_{i,j,x}$ and average this quantity over all the triplets in the corpus. We note this quantity $\bar V$:
\begin{equation}
\label{IMMSBM-eq-relImportInter}
\bar V = \frac{1}{\vert R^{\circ} \vert} \sum_{(i,j,x) \in R^{\circ}} \underbrace{\frac{\vert P_{i,i}(x) - P_{i,j}(x) \vert}{P_{i,i}(x)}}_{V_{i,j,x}}
\end{equation}
where $P_{i,j}(x) = \sum_{k,l}\theta_{i, k} \theta_{j, l} p_{k,l}(x)$ denotes the probability of \gls{outcome} $x$ given the \glspl{entity} $i$ and $j$; as shown in the previous section, the diagonal elements $P_{i,i}(x)$ account for the \gls{virality} of $i$ on $x$. 

\begin{figure}[h]
    \centering
    \includegraphics[width=1.\textwidth]{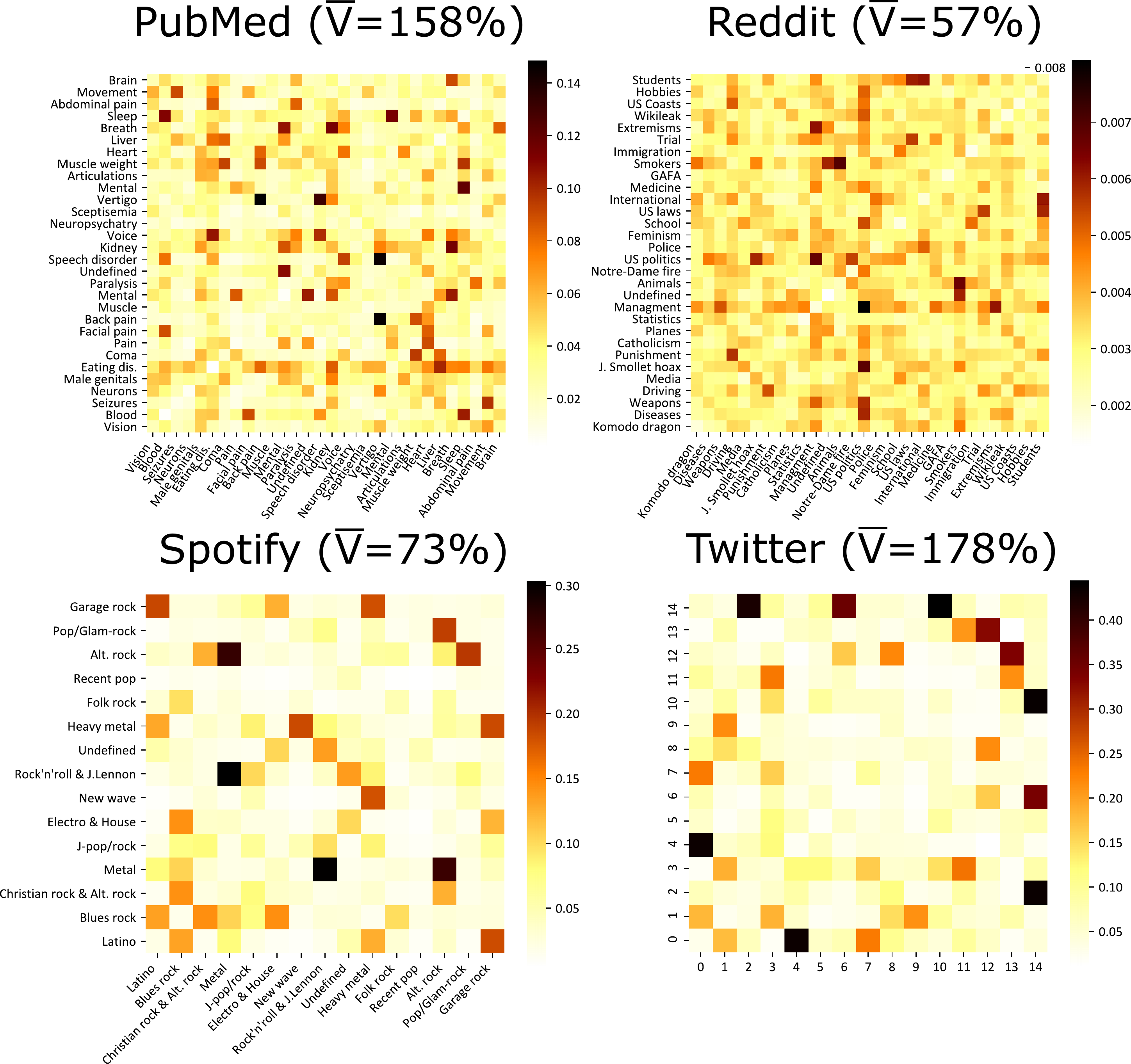}
    \caption[IMMSBM - Importance of \glspl{interaction}]{\textbf{Importance of \glspl{interaction}} - Contribution of each pair of \glspl{cluster} $V_{k,l}$ and average impact of the \glspl{interaction} $\bar V$ (on the right) in \gls{outcome} probabilities for each corpus. \Glspl{cluster} typically interact with a limited number of others; these \glspl{interaction} still play a significant role in \glspl{outcome} probabilities. The \gls{cluster} have been annotated manually.}
    \label{figGpesInter}
\end{figure}

First, we report in Fig.~\ref{fig-IMMSBM-histRelInter} the distribution $V_{i,j,x}$ over all triplets in every dataset. In this figure, we plot the histogram of the relative impact of \glspl{interaction} on the base \gls{virality} of outputs. Overall, we confirm the conclusions of \citep{Myers2012CoC} that most \glspl{interaction} do not lead to any notable change in the probability of an output. However, a non-negligible part of them leads to changes in probability up to 500\% the base \gls{virality}. The overall impact of \glspl{interaction} if the weighted average of this histogram, calculated in Eq.~\ref{IMMSBM-eq-relImportInter}. These results are shown in Fig.~\ref{figGpesInter}.

For every corpus, the \gls{interaction} between \glspl{entity} exerts a non-negligible overall influence $\bar V$ on the probability of outputs. Those results confirm previous work done on \glspl{interaction} modelling, stating the importance of taking \glspl{interaction} into account when analysing real-world datasets \citep{Myers2012CoC}. \Glspl{interaction} increase the \gls{virality} of an output by a factor of 2.58 in the PubMed corpus, 2.78 in the Twitter corpus, 1.73 in the Spotify corpus and 1.57 in the Reddit corpus. \Glspl{interaction} have a greater effect on output probabilities for PubMed and Twitter corpora, and a lesser role for the Spotify and Reddit ones. Besides, our model applied to a dataset where \glspl{interaction} do not play any role ($\bar V = 0$) reduces to the non-\gls{interacting} MMSBM baseline. This metric therefore allows us to assert the importance of the \glspl{interaction} in a given corpus. 

\paragraph{Which \glspl{cluster} interact}
To evaluate \glspl{cluster} pair-\gls{interaction}, we consider the following quantity:
$$
V_{k,l} = \frac{\sum_{(i,j,x) \in R^{\circ}} \theta_{i,k} \theta_{j,l} \vert p_{k,l}(x) - P_{i,i}(x) \vert }{\sum_{(i,j,x) \in R^{\circ}} \theta_{i,k} \theta_{j,l}}
$$
This quantity is the weighted average of the absolute change in output probability with respect to \gls{virality} due to the \gls{interaction} between every pair of \glspl{cluster} $(k,l)$ for each possible pair of \glspl{entity}. The results are shown in Fig.~\ref{figGpesInter}. \Glspl{cluster} have been annotated manually.

We see that most of the \glspl{cluster} do not interact with each other; the \glspl{interaction} essentially take place between a limited number of \glspl{cluster}. Typically, a \gls{cluster} interacts significantly with only one or two other \glspl{cluster} in every corpus (``Vertigo'' and ``Speech disorder'' in PubMed, ``Students'' and ``Schools'' in Reddit, etc.). We also notice that in each corpus, the model forms some non-\gls{interacting} \gls{cluster} with low values of $V_{k,l}$ (``Neuropsychiatry'' in PubMed, ``Recent pop'' in Spotify, etc.); for those, the probability of an output is essentially equal to the \gls{virality} of this output. We also notice that the diagonal of the $V_{k,l}$ matrices comprises low values. The \gls{interaction} of a group with itself leads to an output probability close to its \glspl{entity}' \gls{virality}. To picture how this makes sense, we can imagine diagnosing a disease based on two ear-related symptoms (``earache'' and ``hearing disorders''): the diagnosis is likely to be related to the ear as we would have guessed with only one symptom (its probability equals the \gls{virality}). Now imagine two symptoms of different kinds (``earache'' and ``speech disorder''): the diagnosis is then likely to be related to the brain and less to the ear, so the \gls{interaction} lowers the base probability (\gls{virality}) of the ``ear disease'' output and increases the one of the brain disease.

Being able to see in detail the extent to which \glspl{interaction} exert an influence in a corpus and between which categories they take place opens new perspectives in research. Models that allow explaining the underlying mechanisms are of interest for applied social sciences \citep{Guimera2012HumanPrefSBM,CoboLopez2018SocialDilemma,Poux2021MMSBMMrBanks}.

\paragraph{Entropy of \gls{membership}}
We now consider the \gls{membership} entropy of the \glspl{entity}. It measures how \glspl{entity}’ \gls{membership} is \gls{spread} over the \glspl{cluster}. When it is low, \glspl{entity} belong to a small number of \glspl{cluster} to a great extent; when it is high, it means \glspl{entity} belong to every \gls{cluster} to roughly the same extent. We use the normalized Shannon entropy of \glspl{membership} of user $i$ as a metric, noted $S_{i}^{(m)}$:
\begin{equation}
    S_{i}^{(m)} = \frac{1}{\log_2 \frac{1}{K}}\sum_k^K \theta_{i,t} \log_2 \theta_{i,t}
\end{equation}
Here the lowest entropy reachable is 0, which corresponds to an \gls{entity} belonging to only one group ; the largest is 1 corresponding to belonging to every \gls{cluster} evenly (with probability $\frac{1}{K}$).

Overall, the entropy of \glspl{membership} is low. The average entropy values per corpus are: 0.320 for PubMed (equivalent to belonging on average to 2-3 \glspl{cluster}), 0.324 for Twitter (2 \glspl{cluster}), 0.561 for Reddit (6-7 \glspl{cluster}) and 0.364 for Spotify (2-3 \glspl{cluster}). The small number of \glspl{entity} \gls{spread} among \glspl{cluster} means that the clustering done by our model is easy to interpret --which eased the manual annotation of the \glspl{cluster} presented in the previous section.

\subsubsection{Conclusion}
\label{IMMSBM-conclusion}
In most previous approaches to information \gls{spread}, the effect of \glspl{interaction} between diffusing \glspl{entity} has been neglected. Here, we proposed an in-depth study of the IMMSBM (corresponding to the special case SIMSBM(2)) that allows us to investigate the detail of \glspl{interaction} strength. 
Note that the aim of IMMSBM includes but is not restricted to \gls{interaction} modelling. Throughout this section, we also illustrated the interest of IMMSBM for recommender systems (Spotify and PubMed datasets). 

Our conclusions specific to \glspl{interaction} in information \gls{spread} come from the Twitter dataset. We show that their effect is not trivial (average relative change of 178\%) and that taking them into account increases predictive performance (+0.06 AUCROC over the non-\gls{interacting} baseline). However, these \glspl{interaction} appear to be sparse: only 5 pairs of \glspl{cluster} out of 105 possible pairs seem to have significant \glspl{interaction}. In most cases, \gls{virality} seems to be enough to predict an output with good accuracy. 

However, all the models discussed in this section exhibit a major flaw: they are all static. The data collected on Twitter and Reddit used in our experiments spans approximately over a month. There is no a priori reason for the underlying \gls{interaction} mechanisms to remain the same over this period. Slicing this data into time intervals would provide deeper insights into the \gls{interaction} mechanisms at stake and their evolution over time.

\section{Dynamic \glspl{interaction}}
\label{SDSBM}

\subsection{Introduction}
Dynamic networks are powerful tools to visualize and model \glspl{interaction} between different \glspl{entity} that can evolve over time. The network's nodes represent the \gls{interacting} \glspl{entity}, and ties between these nodes represent an \gls{interaction}. In many real-world situations, the strength of the ties can vary over time --on music streaming websites for instance, users' affinity with various musical genres can vary greatly over time \citep{Kumar2019jodie,RothVillermet2021MusicConsoSpotify}. Such network is said \textbf{dynamic}.

Now, every \gls{interaction} does not have the same importance. A music listener can like both Rock and Jazz, but might prefer one over the other. This person's tie to each musical genre does not have the same intensity; each tie is associated with a number, representing the strength of the \gls{interaction}. The network is said to be \textbf{dynamic and valued}. 

However, in many real-world situations, valued networks are not enough to fully represent a given situation. The same music listener as before can have different opinions on musical genres; she can like it, dislike it, be bored of it, like to listen to it only in the morning, or at night, etc. Each of these relations can be represented by its own tie in the network, each having its own value. The network is said to be \textbf{dynamic, valued and labelled}.

\paragraph{Inferring dynamic, valued and labelled networks}
Networks are high-dimensional objects, whose exact inference is a difficult problem --as stated in the previous section. Several ways to achieve this task have been proposed, the Stochastic Block Models (\acrshort{SBM}) family being one of the most popular approaches, see Chapter~\ref{Chapter-SBMs} and \citep{Holland1983SBM,Guimera2013DrugdrugSBM,CoboLopez2018SocialDilemma}. The underlying assumption is that certain sets of nodes behave similarly. They can be described using a single mathematical object instead of individual ones: the so-called \glspl{cluster}. Instead of modelling every edge for every node in a network, only the edges between these sets are modelled, which makes the task much more tractable. Each \gls{cluster} is then associated with a labelled edge, and each node is associated with a \gls{cluster}. A variant of SBM that allows more expressive power is called the Mixed-\Gls{membership} SBM (\acrshort{MMSBM}), where each node can belong to several \glspl{cluster} in different proportions at the same time \citep{Yuchung1987,Airoldi2008MMSBM,Antonia2016AccurateAndScalableRS,Tarres2019TMBM}. A major advantage of this model’s family is that it yields readily interpretable results (interpretable \glspl{cluster}), unlike most existing neural-network-based modelling of labelled networks \citep{Fan2021InterpretabilityNeuralNet}.

\paragraph{Overview of the proposed approach}
The goal of this section is therefore to provide a way to infer networks that are dynamic, valued and labelled by assuming a block structure -- by using a mixed-\gls{membership} SBM. After a careful review of dynamic network inference literature, it seems that no prior work tackles such a task. Although some previous works handle similar problems, their adaptation to the case at hand is not trivial. Besides, the solution we propose here is conceptually much simpler than those present in the literature. Last but not least, our method is readily pluggable into most existing MMSBM for labelled and valued networks (such as SIMSBM and its special cases, comprising the IMMSBM and each of \citep{Airoldi2008MMSBM,Antonia2016AccurateAndScalableRS,Tarres2019TMBM}) as their temporal extension.

We develop an EM optimization algorithm that scales linearly with the size of the dataset, demonstrate the effectiveness of our approach on both synthetic and real-world datasets, and further detail a possible application scenario.

\begin{figure}
    \centering
    \includegraphics[width=0.99\columnwidth]{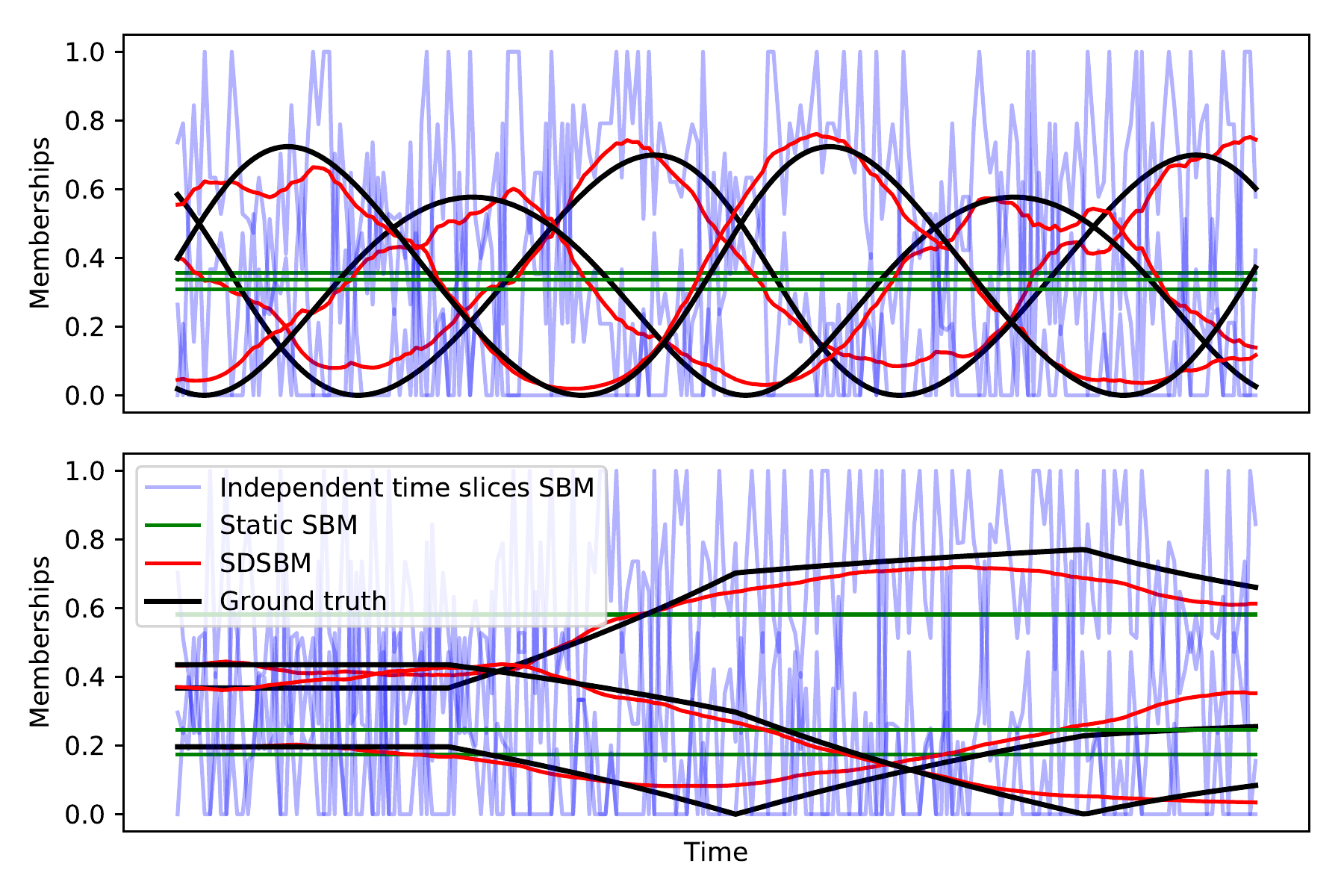}
    \caption[SDSBM - Users' attachment to groups can vary over time]{\textbf{Users' attachment to groups can vary over time} --- A music listener could cyclically prefer Rock, Jazz, or Pop music (top), or listen to either of these without any specific pattern (bottom). For 200 epochs each containing only 5 observations, our approach (in red) infers any smooth dynamic \gls{membership} pattern and does it more accurately than static models (in green \citep{Antonia2016AccurateAndScalableRS}) and models that consider each time slice independently (in blue \citep{Tarres2019TMBM}).
    }
    \label{fig-illustration}
\end{figure}

\subsection{State of the art and limitations}

\subsubsection{Notations}
We consider a network of $I$ nodes and $O$ labels. All the clustering models discussed in this section can be represented using a matrix $\theta \in \mathbb{R}^{I \times K}$ accounting for \glspl{membership} over a set of $K$ possible \glspl{cluster} for each of $I$ nodes, and a block-\gls{interaction} matrix $p \in \mathbb{R}^{K \times K \times O}$ linking each of $K$ \glspl{cluster} to every of $O$ labels. Both $\theta$ and $p$ can vary over time. The network is said to be unlabelled when $O=1$, and binary (as opposed to valued) if an edge can only exist or not exist.

\subsubsection{Dynamic unlabelled networks - Single-\gls{membership}}
Single-\gls{membership} SBM considers a \gls{membership} matrix such as $\theta \in \{0;1\}^{I \times K}$: each \gls{membership} vector equals $1$ for one \gls{cluster}, and $0$ everywhere else. Literature also speaks of ``hard'' clustering.
The authors in \citep{Xu2014DynamicSBM,Xu2014DynamicSBM2} proposed to model a binary unlabelled dynamic network using a label-switching inference method along with a Sequential Monte Carlo algorithm \citep{Jin2021ReviewSBMs}. Both the \gls{membership} and the \gls{interaction} matrices can vary over time, thus supposing two \gls{independent} underlying Markov processes \citep{Jin2021ReviewSBMs}.

In \citep{Yang2010DetectingCA,Tang2014}, the authors propose to model a binary dynamic and unlabelled network. The \gls{cluster} \gls{interaction} matrix $p$ can vary over time while keeping the \glspl{membership} $\theta$ fixed over time. The entries of $p$ are drawn from a Dirichlet distribution and expressed as a Chinese Restaurant Process. This process converges to a Dirichlet distribution and allows to infer a potentially infinite number of \glspl{cluster}. This model is therefore non-parametric and inferred using an MCMC algorithm. 

A novel way of tackling the problem has been proposed in \citep{Matias2017DynSBM,Matias2018DynSBM}, where the authors propose to model the \gls{cluster} \gls{interaction} and \gls{membership} matrices as Poisson processes, that explicitly model the temporal dependency without slicing the time dimension into episodes. The method allows to infer varying \gls{membership} \textit{and} \gls{interaction} matrices for dynamic binary or valued networks, but their results have shown that allowing both to vary simultaneously leads to identifiability and label switching issues \citep{Funka2019ReviewSBM}. This conclusion seems reasonable, given none of these SBM algorithms can reach a global optimum in the likelihood function. During optimization, a model with both \gls{membership} $\theta$ and \gls{interaction} $p$ matrices is all the more likely to get stuck into a local maximum if both can vary over time.

Finally, we can mention the existence of SBM variants that account for dynamic degree-correction \citep{Wilson2019ModelingDynDegCorr} or that enforce a scale-free characteristic \citep{Wu2019ScaleFreedSBM}. However, all these methods consider unlabelled networks and consider a hard clustering which does not allow for as much expressive power as the Mixed-\Gls{membership} approaches.

\subsubsection{Dynamic unlabelled networks - Mixed-\gls{membership}}
Mixed-\gls{membership} SBM considers a \gls{membership} matrix such as $\theta \in \mathbb{R}^{I \times K}$, where each \gls{membership} vector is normalized to $1$. Literature also refers to it as ``soft'' clustering.

Similar to \citep{Yang2010DetectingCA,Tang2014}, a method for inferring dynamical binary unlabelled networks has been proposed in \citep{Fan2015DynInfMMSBM}. The \gls{membership} vector of each \gls{piece of information} is drawn from a Chinese Restaurant Process (\acrshort{CRP}) according to the number of times a node has already been associated with each \gls{cluster} before. The resulting process yields a distribution over an infinity of available \glspl{cluster}. The formulation as a CRP arises naturally because the prior on \gls{membership} vectors is typically a Dirichlet distribution; CRP naturally converges to this distribution. The block-\gls{interaction} matrix $p$ does not evolve with time. The article shows a complexity analysis that suggests the methods run with a complexity of $\mathcal{O}(N^2)$ which makes it unfit for large-scale real-world applications.

The work the most closely related to ours is \citep{Xing2010dMMSBM}. This seminal work proposed the dMMSB as a way to model dynamic binary unlabelled networks using a variational algorithm \citep{Lee2019ReviewSBMs}. To do so, the authors modify the original MMSBM \citep{Airoldi2008MMSBM} to consider a logistic normal distribution as a prior on the \gls{membership} vectors $\vec{\theta_i}$. This choice allows to model correlations between \gls{membership} vectors' evolution \citep{Ahmed2007LogisticNormal}. The \gls{membership} vectors are then embedded in a state-space model, which is a space where we can define a linear transition between two time points for a given variable. The authors define such a trajectory for the \gls{membership} vectors as a linear function of the previous time point. The trajectory is estimated and smoothed using a Kalman Filter. This approach is the most closely related to ours, as it proposes to consider the temporal dependency directly in a prior distribution over \glspl{membership}, noted $P(\theta)$. 

However, this model is not fit for the task at hand. It considers unlabelled and binary networks \citep{Lee2019ReviewSBMs}, and extension to labelled and valued networks is not trivial. The proposed optimization algorithm requires a loop until convergence at each EM iteration, making it unable to handle large datasets. It is not designed to let the \glspl{cluster} \gls{interaction} matrix $p$ depend on time, which we alleviate here. And most importantly, it assumes a linear transition between time steps in the state space, while we do not assume any kernel function in our proposed approach. \citep{Xing2010dMMSBM} has been extended to consider $P(\theta)$ as a logistic normal mixture prior \citep{Ho2011dM3SB}, which improves its expressive power. However, it does not address the aforementioned points.

\subsubsection{Static labelled networks - Mixed-\gls{membership}}
Recent years saw a rise of Bayesian methods for inferring static valued labelled networks using MMSBM variants \citep{Antonia2016AccurateAndScalableRS,Tarres2019TMBM,Poux2021MMSBMMrBanks}. Note that in \citep{Tarres2019TMBM}, the authors consider a temporal slicing of the data and consider each slice as \gls{independent} from the others; a time slice is considered as a node in a tripartite network. We will compare our approach to this modelling later. We do not present the details of these works here, as we will develop their in-depth functioning in Section~\ref{section-model}.

The method we propose here uses these works as a base. This section focuses on making the prior probability of both $\theta$ and $p$ time-dependent to model these parameters’ dynamics. We provide a ready-to-use temporal plug-in for each of the works we have just presented in this section and in Section~\ref{section-static-interaction}. It applies to dynamical, valued, and labelled networks in a mixed-\gls{membership} context, and inference is conducted with a scalable variational EM algorithm.

\subsection{SDSBM -- Simple Dynamic labelled MMSBM}
\label{section-model}

\subsubsection{Base model}
In this section, we present the simplest form of a labelled MMSBM, or SIMSBM(1), for demonstration purposes. We illustrate the structure of this model in Fig.~\ref{fig-schema_SIMSBM1}. Its derivation trivially extends to \citep{Antonia2016AccurateAndScalableRS,Tarres2019TMBM,Poux2021MMSBMMrBanks}, the IMMSBM, and any of the SIMSBM iterations discussed Section~\ref{SIMSBM} (we detail our work's inclusion in these more complex models in Appendix, Section~\ref{SI-SDSBM-inclusion-SotA}).

\begin{figure}
    \centering
    \includegraphics[width=\textwidth]{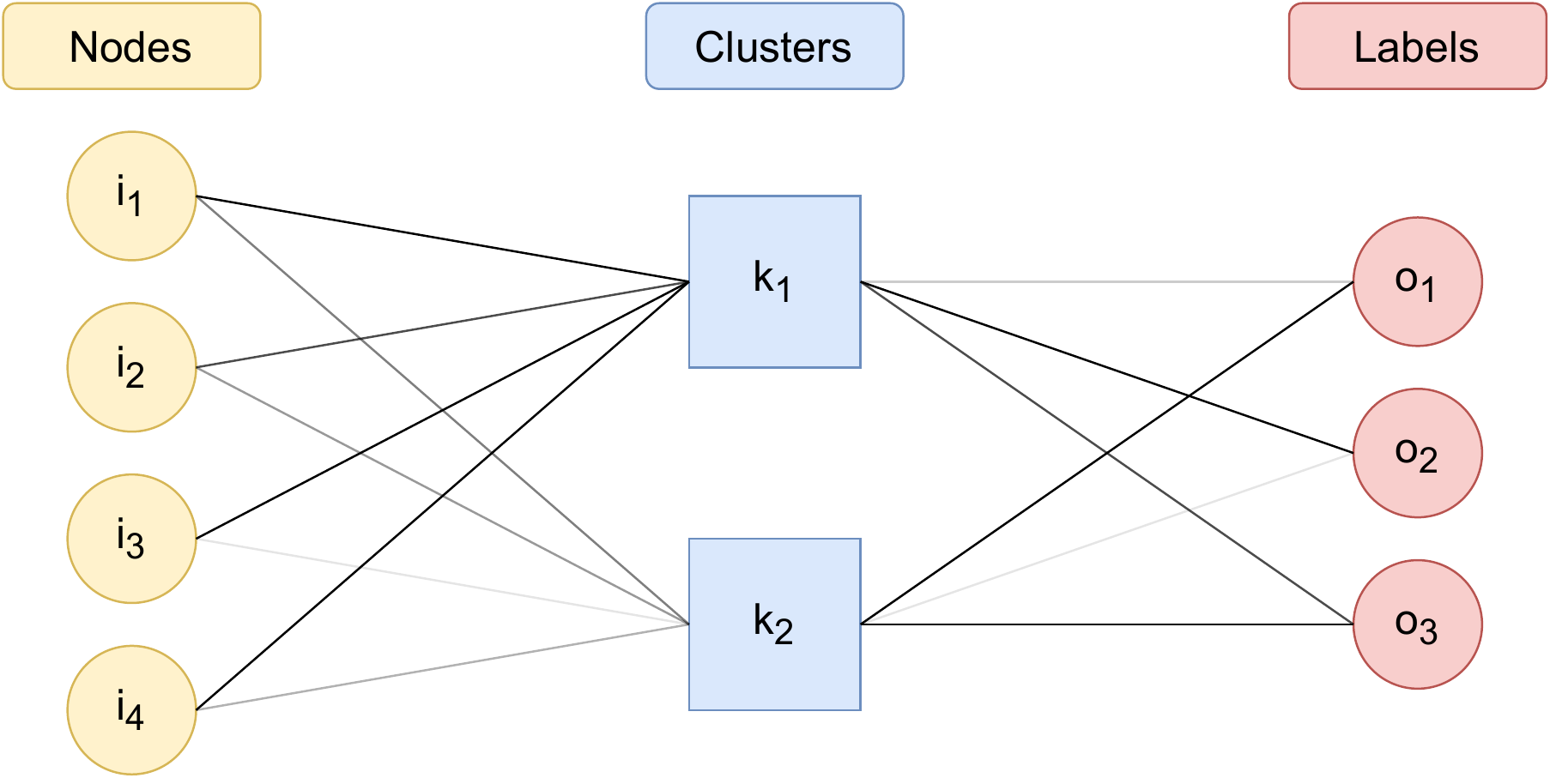}
    \caption[SDSBM - Illustration of the base model SIMSBM(1)]{Illustration of the SIMSBM(1), which is the base model coupled to the SDSBM prior. Nodes are associated with \glspl{cluster}, and \glspl{cluster} are associated with labels. Ties between each \gls{entity} represent a \gls{membership} and can take values between 0 and 1.}
    \label{fig-schema_SIMSBM1}
\end{figure}

We consider a set of $I$ nodes that can be associated with $O$ possible labels on a discrete time interval, or epoch, written $t$. We assume that nodes can be efficiently represented as a mixture of $K$ available \glspl{cluster} at each time step, each of which is in turn linked to the labels. The \gls{membership} of each of $I$ nodes into each of the $K$ possible \glspl{cluster} at time $t$ is encoded in the \gls{membership} matrix $\theta^{(t)} \in \mathbb{R}^{I \times K}$. One vector $\theta_i^{(t)}$ represents the probability that $i$ belongs to any of the $K$ \glspl{cluster} at time $t$, and is normalized as:
\begin{equation}
    \label{eq-normtheta}
    \sum_{k \in K} \theta_{i,k}^{(t)} = 1 \ \forall i, t
\end{equation}
The probability of each \gls{cluster} to be associated with each label at time $t$ is encoded in the matrix $p^{(t)} \in \mathbb{R}^{K \times O}$. An entry $p_k^{(t)}(o)$ represents the probability that \gls{cluster} $k$ is associated with label $o$ at time $t$, and thus is normalized as:
\begin{equation}
    \label{eq-normp}
    \sum_{o \in O} p_k^{(t)}(o) = 1 \ \forall k, ^{(t)}
\end{equation}
Finally, the probability that a node $i$ is associated with label $o$ at time $t$ (this can be seen as the probability an edge between $i$ and $o$ exists at time $t$) is written:
\begin{equation}
    P^{(t)}(i \rightarrow o) = \sum_{k \in K} \theta_{i,k}^{(t)} p_k^{(t)}(o)
\end{equation}
Given a set $R^{\circ}$ of observed triplets $(i,o,t)$, the model's posterior distribution can be written \citep{Antonia2016AccurateAndScalableRS,Tarres2019TMBM}:
\begin{align}
    \label{eq-L}
    P(\theta, p \vert R^{\circ}) &\propto P(R^{\circ} \vert \theta, p) \prod_t P(\theta^{(t)})P(p^{(t)}) \\ 
    = \prod_{(i,o,t) \in R^{\circ}} &\sum_{k \in K} \theta_{i,k}^{(t)} p_k^{(t)}(o) \prod_t \left( \prod_i P(\theta_i^{(t)}) \prod_k P(p_k^{(t)}) \right) \notag 
\end{align}
Now, before we describe the optimization procedure, we must choose the priors $P(\theta^{(t)})$ and $P(p^{(t)})$.

\subsubsection{Simple Dynamic prior}
We formulate the prior distribution over $\theta^{(t)}$ and $p^{(t)}$ following a simple assumption: the parameters at a given time are unlikely to vary abruptly at small time scales --- \textit{the apple does not fall far from the tree}. It means an entry $\theta^{(t_1)}_{ik}$ should not differ so much from $\theta^{(t_2)}_{ik}$ for every $t_2$ close enough to $t_1$. Such entries close to a reference time are called \textbf{temporal neighbours}.

Our \textit{a priori} knowledge on each entry $\theta_i^{(t)}$ and $p_k^{(t)}$ is that they should not differ significantly from their temporal neighbours. This is a fundamental difference with \citep{Xing2010dMMSBM}, where the next parameters values are estimated using a Kalman Filter that only considers the previous time step. Moreover, the authors assume a linear transition function, while we do not make such a hypothesis. An illustration of the proposed approach is given in Fig.~\ref{fig-dirichlet}, where the prior probability of a \gls{membership} vector depends on its temporal neighbours (in white).

\paragraph{Dirichlet distribution}
Since each entry $\theta_i^{(t)}$ and $p_k^{(t)}$ is normalized to 1, we consider a Dirichlet distribution as a prior, which naturally yields normalized vectors such that $\sum_n x_n = 1$. It reads:
\begin{equation}
    \label{eq-dir}
    Dir(x \vert \alpha) = \frac{1}{B(\alpha)}\prod_n x_n^{\alpha_n-1}
\end{equation}
where $B(\cdot)$ is the multivariate beta function. In Eq.~\ref{eq-dir}, the vector $\alpha$ is called the concentration parameter and must be provided to the model. Importantly to the model introduced in this section, it allows to control the mode and the variance of the Dirichlet distribution.

We consider a concentration parameter such as $\alpha = 1 + \beta \alpha_0$, so that for $\beta=0$ we recover a uniform prior over the simplex, and $\alpha>1$ so that the prior has a unique mode.
The most frequent value drawn from Eq.~\ref{eq-dir} (or mode) is $\mathbb{M}(x_n) = \frac{\alpha_n-1}{\sum_n' (\alpha_n' - 1)} = \frac{\alpha_{0,n}}{\sum_n' \alpha_{0,n}}$. We recover a uniform prior for $\beta=0$; the variance vanishes with $\beta \gg 1$ as $\frac{1}{\beta}$. The effect of various values of $\beta$ on the prior distribution is illustrated in Fig.~\ref{fig-dirichlet}.

\begin{figure*}
    \centering
    \includegraphics[width=\textwidth]{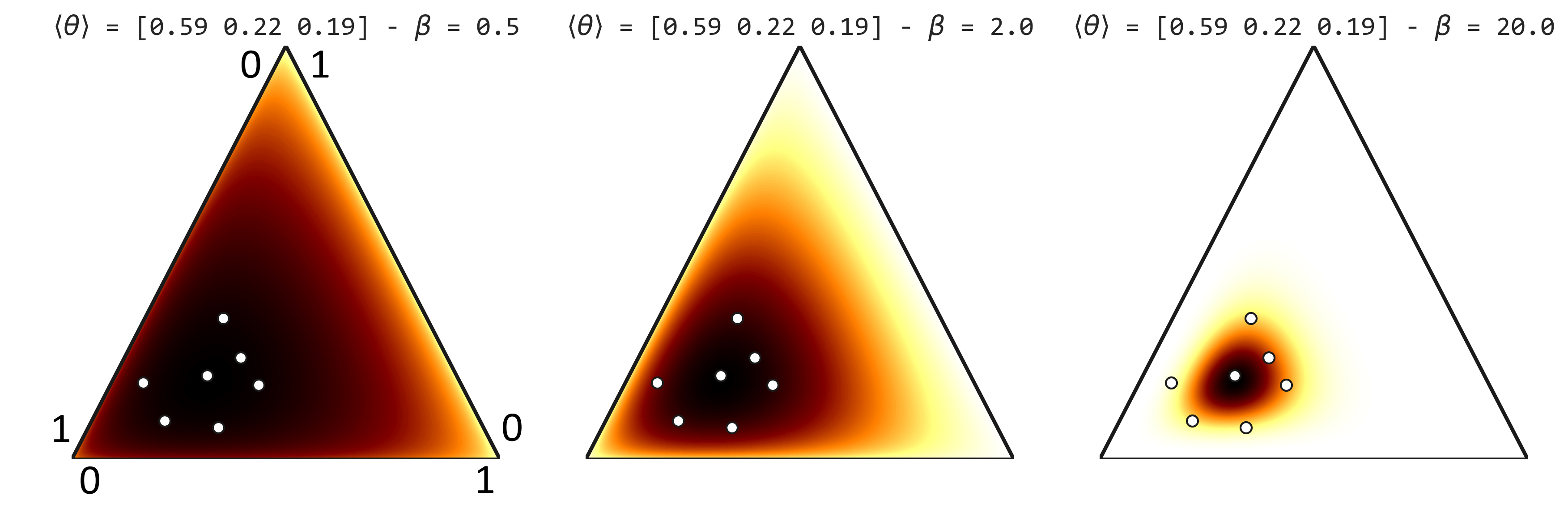}
    \caption[SDSBM - Illustration of dynamic prior probability]{\textbf{Prior probability on a \gls{membership} vector for various values of $\beta$ according to temporal neighbourhood} --- Darker means higher probability. Projected on a simplex tri-space (each of 3 axes ranges from 0 to 1). The white dots represent the temporal neighbours of the considered 3D vector. Their average is given as $\langle \theta \rangle$ using a uniform weight function $\kappa(t,t')$ for illustration purpose. $\beta$ controls the variable's prior variance around its neighbours.}
    \label{fig-dirichlet}
\end{figure*}

\paragraph{Prior's mode}
Our main assumption states that $\theta_i^{(t)}$ and $p_k^{(t)}$ do not vary abruptly. To enforce this, we define their prior probability mode with respect to their close temporal neighbours. The hyper-parameter $\beta$ controls the variance of the prior --that is, how much it should impact the inference procedure. We express the Simple Dynamic prior parameters for $\theta_i^{(t)}$ as:
\begin{equation}
    \label{eq-prior}
    \alpha_{i,k}^{(t,\theta)} = 1+ \beta \underbrace{\left(\frac{\sum_{t' \neq t} \kappa(t,t') \theta_{i,k}^{(t')}}{\sum_{t' \neq t} \kappa(t,t')}\right)}_{\langle\theta_{i,k}^{(t)}\rangle}
\end{equation}
\myequations{\ \ \ \ SDSBM - Temporal prior}
\noindent 
where $\kappa(t,t')$ is a weight function, and $\alpha^{(t,\theta)}$ corresponds to the concentration parameter for $\theta$ at time $t$. In following experiments, we define the weight function as $\kappa(t,t') = \frac{N_{t'}}{\vert t-t' \vert}$, where $N_{t'}$ is the number of observations made at time $t'$, so that temporal neighbours' influence decrease as the inverse of temporal distance. We illustrate the influence of this particular kernel function on the prior probability on \gls{membership} at all times in Fig.~\ref{fig-illustration-kernel}. In particular we see that with this expression, the prior probability variance collapses to 0 when the considered time is very close to a temporal neighbour.

\begin{figure}[h]
    \centering
    \includegraphics[width=\textwidth]{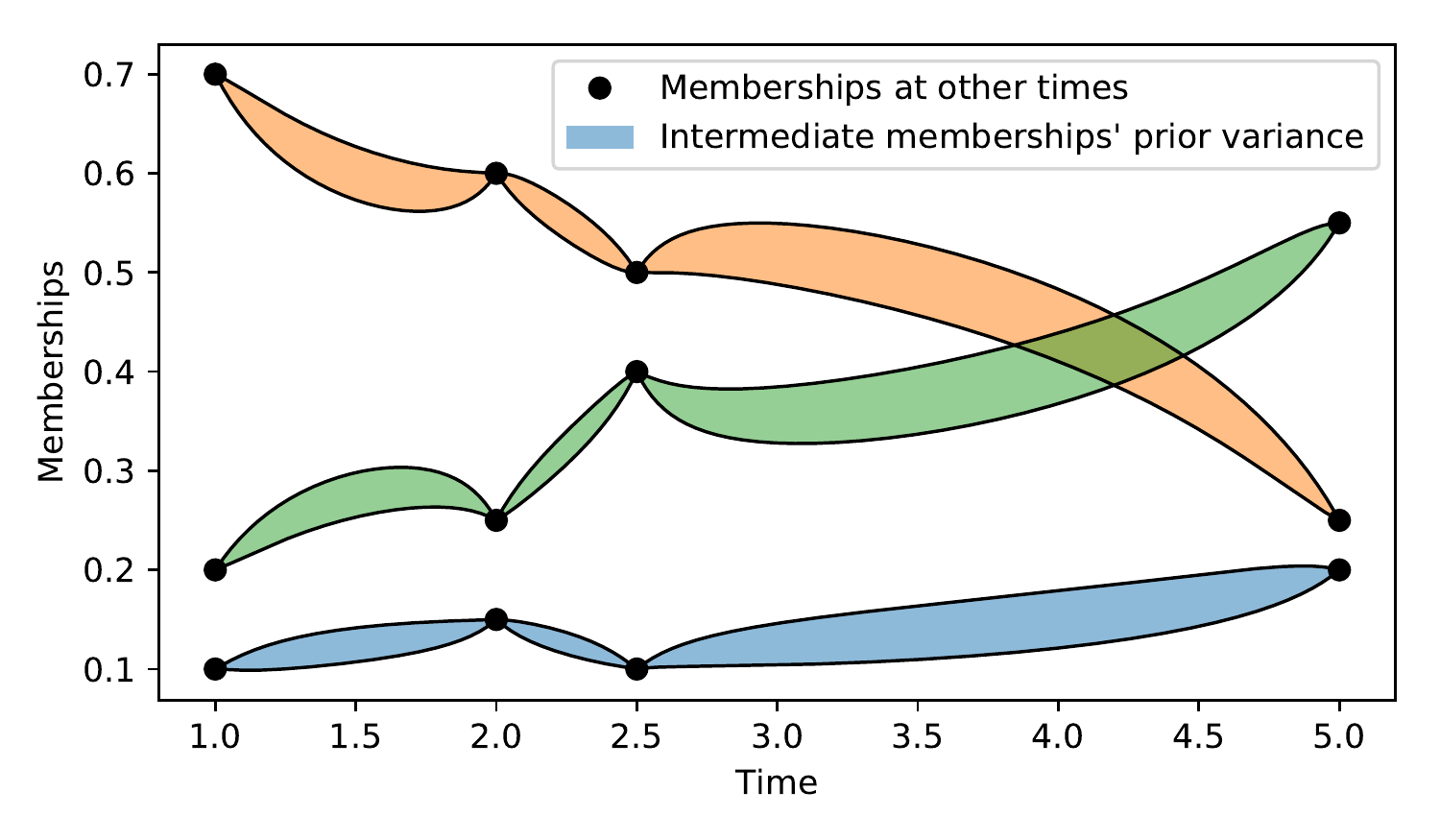}
    \caption[SDSBM - Prior probability's variance according to time]{\textbf{Prior probability's variance on \glspl{membership} at all times according to the temporal neighbourhood} --- Variance of the prior over a \gls{membership} entry (filled curves, we represented 3 such entries as illustration) as a function of time, given some temporal neighbours (black dots). This illustration considers an averaging kernel as $\kappa(t,t') = \frac{1}{\vert t-t' \vert}$. When inferring a parameter $x^{(t)}$ at a time $t$, the variance of its prior probability $P(x^{(t)})$ depends on $t$ relative to the temporal neighbours. Here for instance, the variance is null at $t=2$ because $\kappa(t,t')$ diverges, and so does $\alpha^{(t)}$, hence the variance collapsing to 0.
    }
    \label{fig-illustration-kernel}
\end{figure}

The mode of the prior over a variable is then the average value of its the temporal neighbours weighted by $\kappa(t,t')$, noted $\langle\theta_{i,k}^{(t)}\rangle$. Note that this holds because ${\sum_k \langle\theta_{i,k}^{(t)}\rangle = 1 \ \forall i,t}$. Besides, the prior variance is a decreasing function of $\beta$; when $\beta=0$ the prior is uniform over the simplex, and when $\beta \rightarrow \infty$ the variance goes to $0$, as illustrated Fig.~\ref{fig-dirichlet}. The same reasoning holds for $p_k^{(t)}$, with prior parameters $\alpha_{k,o}^{(t,p)} = 1 + \beta \langle p_{k}^{(t)}(o)\rangle$.

\paragraph{Priors expression}
Finally, we give the final log-priors on $\theta_i^{(t)}$ and $p_k^{(t)}$:
\begin{align}
    \label{eq-priors}
    P(\theta_i^{(t)} \vert \{\theta_{i,k}^{(t')}\}_{t' \neq t}) &\propto \prod_k {\theta_{i,k}^{(t)}}^{\beta \langle\theta_{i,k}^{(t)}\rangle}\\
    P(p_k^{(t)}(o) \vert \{p_{k}^{(t')}(o)\}_{t' \neq t}) &\propto \prod_o {p_{k}^{(t)}(o)}^{\beta \langle p_{k}^{(t)}(o)\rangle} \notag 
\end{align}
We omitted the normalisation factor for clarity --it does not influence the inference procedure. Note that the dependence of each entry of $\theta_i^{(t)}$ and $p_k^{(t)}$ raises additional questions from a statistical perspective. It is unsure whether the specified conditional distributions would allow to retrieve a global, joint probability distribution for all entries of those vectors. Therefore, our approach works as a penalization of the likelihood function, but does not imply the existence of a statistical law on $\theta_i^{(t)}$ and $p_k^{(t)}$.

\subsubsection{Inference}

\paragraph{E step}
We develop an EM inference procedure for maximizing the log-posterior distribution defined in Eq.~\ref{eq-L}. The expectation step computes the expected probability of a latent variable (here a \gls{cluster} $k$) being chosen given each entry of $R^{\circ}$. Since such latent variables do not appear in the priors expressions, the expectation step remains unchanged by the introduction of the Simple Dynamic Priors; in general, prior distributions do not intervene in the computation of the expectation step \citep{Bishop2006}. The E step for such labelled networks has already been derived on many occasions. Therefore, we give the expectation step equation without explicit derivation (full derivation is however provided in Appendix, Section~\ref{SI-SDSBM-derivationEstep}, and in Section~\ref{SIMSBM-Estep} and Section~\ref{IMMSBM-Estep}).

The expectation of the latent variable $k$ given an observation $(i,o,t) \in R^{\circ}$, written $\omega_{i,o}^{(t)}(k)$, is defined as:
\begin{equation}
    \label{eq-omega}
    \omega_{i,o}^{(t)}(k) = \frac{\theta_{i,k}^{(t)} p_k^{(t)}(o)}{\sum_{k'} \theta_{i,k'}^{(t)} p_{k'}^{(t)}(o)}
\end{equation}
Using this expression, we can rewrite the log-likelihood $\log P(R^{\circ} \vert \theta, p)$ as \citep{Bishop2006,Antonia2016AccurateAndScalableRS,Tarres2019TMBM}:
\begin{equation}
    \label{eq-newlogl}
    \log P(R^{\circ} \vert \theta, p) = \sum_{(i,o,t) \in R^{\circ}} \sum_{k \in K}  \omega_{i,o}^{(t)}(k) \log \frac{\theta_{i,k}^{(t)} p_k^{(t)}(o)}{\omega_{i,o}^{(t)}(k)}
\end{equation}

\paragraph{M step}
Taking back the first line of Eq.~\ref{eq-L} and substituting with Eq.~\ref{eq-priors} and Eq.~\ref{eq-newlogl}, we get an unconstrained expression of the posterior distribution. We introduce Lagrange multipliers to account for the constraints of Eq.~\ref{eq-normtheta} ($\phi_i^{(t)}$) and Eq.~\ref{eq-normp} ($\psi_i^{(t)}$), and finally compute the maximization equations with respect to the model's parameters. Starting with the \gls{membership} matrix entries $\theta_{i,k}^{(t)}$:
\begin{align}
    &\frac{\partial \left(\log P(\theta, p \vert R^{\circ}) - \sum_{i',t'} \phi_{i'}^{(t')}(\sum_{k'}\theta_{i',k'}^{(t')}-1)\right)}{\partial \theta_{i,k}^{(t)}} = 0 \notag \\
    &\Leftrightarrow \sum_{o \in \partial(i,t)} \frac{\omega_{i,o}^{(t)}(k)}{\theta_{i,k}^{(t)}} + \frac{\beta\langle\theta_{i,k}^{(t)}\rangle}{\theta_{i,k}^{(t)}} - \phi_{i}^{(t)} = 0 \notag \\
    &\Leftrightarrow \sum_{o \in \partial(i,t)} \omega_{i,o}^{(t)}(k) + \beta\langle\theta_{i,k}^{(t)}\rangle = \phi_{i}^{(t)}\theta_{i,k}^{(t)} \notag \\
    &\Leftrightarrow \sum_{o \in \partial(i,t)} \underbrace{\sum_k \omega_{i,o}^{(t)}(k)}_{=1 \text{ (Eq.~\ref{eq-omega})}} + \beta\underbrace{\sum_k \langle\theta_{i,k}^{(t)}\rangle}_{=1 \text{ (Eq.~\ref{eq-normtheta})}} = \phi_{i}^{(t)} \notag \\
    &\Leftrightarrow \frac{\sum_{o \in \partial(i,t)} \omega_{i,o}^{(t)}(k) + \beta\langle\theta_{i,k}^{(t)}\rangle}{N_{i,t}+\beta} = \theta_{i,k}^{(t)}
\end{align}
where $\partial(i,t)=\{ o \vert (i, \cdot, t) \in R^{\circ} \}$ is the subset of labels associated with both $i$ and $t$, and $N_{i,t} = \vert \partial(i,t) \vert$ is the size of this set. Note that for $\beta = 0$ we recover the M-step of standard static MMSBM models, see \citep{Antonia2016AccurateAndScalableRS,Tarres2019TMBM} and Section~\ref{section-static-interaction}.

The derivation of the M-step for the entries $p_{k}^{(t)}(o)$ is identical and yields (see Appendix, Section~\ref{SI-SDSBM-derivationMstep}, for details):
\begin{align}
    p_{k}^{(t)}(o) = \frac{\sum_{(i,t) \in \partial o} \omega_{i,o}^{(t)}(k) + \beta\langle p_{k}^{(t)}(o)\rangle}{\sum_{(i,o,t) \in R^{\circ}} \omega_{i,o}^{(t)}(k)+\beta}
\end{align}

\subsubsection{Discussion}
\paragraph{Easy to use}
We briefly review some key points of the Simple Dynamic prior.
Its introduction induces minor changes in the existing works on MMSBM for labelled networks. Compared to \citep{Antonia2016AccurateAndScalableRS,Tarres2019TMBM} and Section~\ref{section-static-interaction}, its introduction boils down to simply adding a term $\beta\langle x \rangle$ to the numerator of maximization equations, and the corresponding normalizing term $\beta$ to the denominator. This way, our approach is ready-to-use to make these models, and variants built on them, dynamic --explicit derivations are provided in Appendix, Section~\ref{SI-SDSBM-inclusion-SotA}.

\paragraph{Flexible dynamic modelling}
The prior allows us to consider that some parameters are dynamic and that others are not. For instance, when several \gls{membership} matrices are involved, as in \citep{Antonia2016AccurateAndScalableRS,Tarres2019TMBM}), setting $\beta=0$ for some makes them time-invariant (or universal). The SD prior also allows choosing whether the block-\gls{interaction} tensor $p$ is dynamic. 

In general, $\beta$ does not have to be identical for every \gls{membership} matrix, or even every entry $i$ of each of them. Moreover, it is not mandatory for $\beta$ to be constant over time. A dynamic parameter $\beta(t)$ is especially relevant when epochs are not evenly spaced over time; $\beta(t)$ would typically lower the variance (by increasing) when temporal neighbours are closer (right plot in Fig.~\ref{fig-dirichlet}).

To summarize, $\beta$ allows controlling the time scale over which variables may vary. This allows greater modelling flexibility, allowing to jointly model universal variables ($\beta=0$) and dynamical ones ($\beta \neq 0$).

\paragraph{Tuneable temporal dependence}
Finally, the choice of the averaging kernel function $\kappa(t,t')$ is important. It allows choosing the range over which the inference of a variable should rely on its temporal neighbours. A formulation as the inverse of time difference seems relevant: the weight of a neighbour appearing at a time $\delta t$ later should diverge as $\delta t \rightarrow 0$, so that continuity is ensured. Besides, it allows controlling the smoothness of the curve with respect to time by tuning the weight function as $\kappa(t,t')=\frac{N_{t'}}{\vert t-t' \vert^{a}}$ where $a=1,2,...$ for instance, where $N_{t'}$ is the number of observations in the time slice $t'$.

Overall, the Simple Dynamic prior works by inferring the variables using both microscopic and mesoscopic temporal scales. If a time slice $t$ has very few observations but some of its neighbours have a greater number of them for instance; learning the parameters at $t$ is helped mostly by the population of its closest ($\frac{1}{\vert \Delta t \vert}$) and most populated ($N_{t'}$) epochs, and less influenced by further and less populated epochs. This is what is illustrated in Fig.~\ref{fig-illustration-kernel}.

\subsubsection{Experiments}

\paragraph{Synthetic data}
In this section, we address several situations in which our method (abbreviated SDSBM for Simple Dynamic MMSBM) could be useful.
Experiments are run for $I=100$, $K=3$ and $O=3$, which are standard testing parameters in the literature of dynamic networks inference \citep{Fan2015DynInfMMSBM,Matias2017DynSBM}. We choose to infer a dynamic \gls{membership} matrix $\theta^{(t)}$ and to provide a universal block-\gls{interaction} matrix $p\ \forall t$. Note that the model yields good performances when $p$ also has to be inferred, but due to identifiability and label-switching issues raised in \citep{Matias2017DynSBM}, there is no unbiased way to assess the correctness of the inferred \glspl{membership} values. An experiment where $p$ is inferred jointly to $\theta$ is provided and discussed in Appendix, Section~\ref{SI-SDSBM-bothdynmatrices}. The expression we use for this matrix $p$ is given Eq.~\ref{eq-xpp}. We systematically test two variation patterns for $\theta$: a sinusoidal pattern (Fig.~\ref{fig-illustration}-top) and a broken-line pattern (Fig.~\ref{fig-illustration}-bottom). Each pattern is generated with a different coefficient for each item; the \glspl{membership} still sum to 1 at all times. 

\begin{equation}
    \label{eq-xpp}
    p = \begin{bmatrix}
        1-s & s & 0\\
        0 & 1-s & s\\
        s & 0 & 1-s
    \end{bmatrix}
\end{equation}

To the best of our knowledge, the only attempt to model dynamic parameters in labelled valued and dynamic networks using a MMSBM is \citep{Tarres2019TMBM}. In this work, each epoch is modelled independently from the others. We refer to this baseline as the ``No coupling'' or ``\textbf{NC}'' baseline. For reference, we also compare to a baseline that does not consider the temporal dimension and infers a single universal value for each variable (classical \textbf{MMSBM}) \citep{Antonia2016AccurateAndScalableRS} and the models Section~\ref{section-static-interaction}.

We systematically perform a 5-fold cross-validation. The model is trained on 80\% of the data, $\beta$ is tuned using 10\% as a validation set, and the model is evaluated on the 10\% left. We choose as metrics the \acrshort{AUCROC} and \acrshort{RMSE} on the real values of $\theta$ (black line in Fig.~\ref{fig-illustration}). The procedure is repeated 5 times; the error bars reported in the experimental results represent the standard error over these folds.

\begin{figure}
    \centering
    \includegraphics[width=\columnwidth]{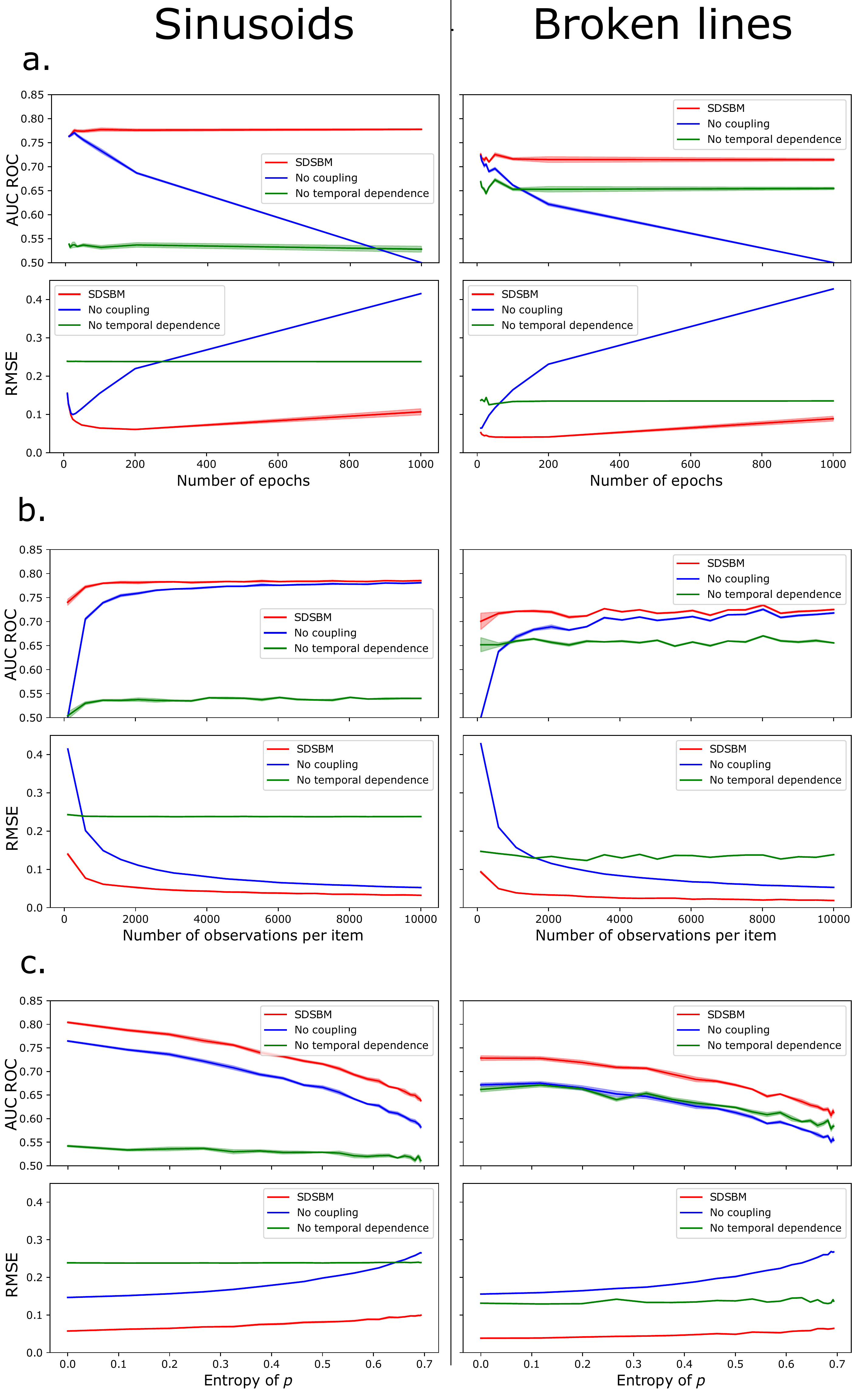}
    \caption[SDSBM - Results on synthetic data]{\textbf{Results on synthetic data} --- \textbf{(a.)} SDSBM retrieves the correct dynamic \glspl{membership} and is little influenced by the data slicing. \textbf{(b.)} SDSBM works well on tiny datasets. \textbf{(c.)} SDSBM retrieves correct dynamic \glspl{membership} in challenging situations.}
    \label{fig-resSynth}
\end{figure}

\paragraph{SDSBM unveils complex temporal patterns}
In Fig.~\ref{fig-resSynth}a., we consider 1,000 observations for each item $i \in I$ and vary the number of epochs from 10 to 1,000. In the expression of $p$, $s$ is set to 0.05. For both the sinusoidal and line-broken \glspl{membership}, the model shows better predictive performances (in terms of AUCROC) than the proposed baselines. Interestingly, the SDSBM performances remain stable as the number of epochs increases unlike the NC baseline, which means it alleviates a bias of the temporal modelling proposed in \citep{Tarres2019TMBM}. The RMSE with respect to the true parameters remains low over the whole range of the tested number of epochs. The RMSE increases as the number of epochs grows because the number of parameters to estimate increases with it; this makes the inference more subject to local variations, which in turn mechanically increases the RMSE. Overall, SDSBM recovers dynamic variations of the \gls{membership} vectors with superior performance; a sample of the inferred dynamic \glspl{membership} is shown in Fig.~\ref{fig-illustration}.

\paragraph{SDSBM works with little data}
A major problem that arises when considering temporal data is the scarcity of observations, because slicing implies reducing the number of observations in each slice. This concern mostly arises in social sciences, where data retrieval cannot be automated and requires long-lasting human labour. Here, we demonstrate that our method works in challenging conditions when data is scarce. In Fig.~\ref{fig-resSynth}b., we vary the number of observations available for each item from 100 to 10.000, distributed over 100 epochs. Thus, in the most challenging situation, there is only one observation per epoch used to determine $I$ dynamic \glspl{membership} over 3 \glspl{cluster}. In the expression of $p$, $s$ is set to 0.05. We see Fig.~\ref{fig-resSynth}b. that for both patterns, the predictive power of SDSBM remains high in such conditions. Moreover, the RMSE on the true dynamic \glspl{membership} in this case is fairly low and decreases rapidly as the number of observations increases. When the number of observations is high, the ``no coupling'' baseline \citep{Tarres2019TMBM} reaches the performances of SDSBM. This is because as the number of observations in each slice goes to infinity, the models need to rely less on temporal neighbours. However, even for 10.000 observations per item (100 observations per epoch), SDSBM yields better results than the proposed baselines. As an illustration, the results presented in Fig.~\ref{fig-illustration} have been inferred using only 5 observations per epoch.

\paragraph{SDSBM handles highly stochastic \gls{interaction} patterns}
Finally, we control the deterministic character of the block-\gls{interaction} matrix $p$ by varying $s$. We express such character as the mean entropy of $p$ $\langle S(p) \rangle$ with respect to its possible outputs: $\langle S(p) \rangle = \frac{1}{K}\sum_{k \in K}\sum_o p_k(o)\log p_k(o)$. The maximum entropy for the proposed expression of $p$ is reached for $s=0.5$. We consider 1,000 observations \gls{spread} over 100 epochs. We show in Fig.~\ref{fig-resSynth}c. that the predictive performance of all three methods drops as the entropy increases. This is expected, as the observations are generated from the true model with a higher variance; each observation becomes less informative about the underlying generative structure as $s$ grows. However, the RMSE on the real parameters inferred using SDSBM remains low even at the maximum entropy, meaning the model recovers the correct \gls{membership} parameters.

\input{Tables/Chapter_2/table-res-SDSBM}
\paragraph{Real-world data}
Finally, we demonstrate the validity of our approach on real-world data to argue for its usefulness and scalability. SDSBM builds on previous works on labelled MMSBM and shares the same linear complexity $\mathcal{O}(\vert R^{\circ}\vert)$ with $\vert R^{\circ}\vert$ the size of the dataset \citep{Antonia2016AccurateAndScalableRS}. For our experiments, we consider the recent and documented datasets from \citep{Kumar2019jodie}, namely the \textbf{Reddit} dataset (10.000 users, 984 labels, $\sim$670k observations), the \textbf{LastFm} dataset (980 users, 1000 labels, $\sim$1.3M observations) and the Wikipedia (\textbf{Wiki}) dataset (8227 users, 1000 labels, $\sim$157k observations). The Reddit and Wikipedia datasets contain 1 month of data; we slice them in 1 day long temporal intervals. The LastFm dataset spans over approximately 5 years; we slice it into periods of 3 days each. In addition, we build an additional dataset (\textbf{Epi}) about historical epigraphy data \citep{ClaussSlabyDataset}. The dataset is made of 117.000 Latin inscriptions comprising one or several of 18 social statuses (slave, soldier, senator, etc.) and its location as one of 62 possible regions, along with an estimated dating spanning from 100BC to 400AD. The goal is to guess the region where a status has been found, with respect to time. The goal is to recover statuses \gls{diffusion} in territories newly conquered by the Roman Empire. We slice this dataset in epochs of one year each. 

Evaluation is again conducted using a 5-fold cross-validation with 80\% of training data, 10\% of validation data and 10\% of testing data for each fold. For each pair $(i, o_{true})$ in the test set, we query the probability for every output $o$ given $i$ and build the confusion matrix by comparing them to $o_{true}$.
In Table~\ref{table-res-SDSBM}, we present the results of our method compared to the proposed baselines for various metrics: Area under the ROC curve (\textbf{\acrshort{AUCROC}}), Average Precision (\textbf{\acrshort{AP}}) and Normalized Coverage Error (\textbf{\acrshort{NCE}}). The first two metrics evaluate how well models assign probabilities to observations, and the latter evaluates the order in which possible outputs are ranked. Overall, we see that our method exhibits a greater predictive power, except for the Reddit dataset where the static SBM performs as well as SDSBM. We explain this by the lack of significant temporal variation over the considered interval. This could be expected, since the dataset comprises roughly 80\% of repeated actions \citep{Kumar2019jodie}, meaning that users do not significantly explore new communities over a month. This result shows that SDSBM also works well in the static case. On the other datasets, SDSBM performs better often by a large margin, especially for the AUCROC, meaning that SDSBM is efficient at distinguishing classes from each other. We recall that the model used here is deliberately simplistic for demonstration purposes; low metrics do not mean the Simple Dynamic prior does not work, but instead that it should be coupled to a more complex model.

\begin{figure*}
    \centering
    \includegraphics[width=\textwidth]{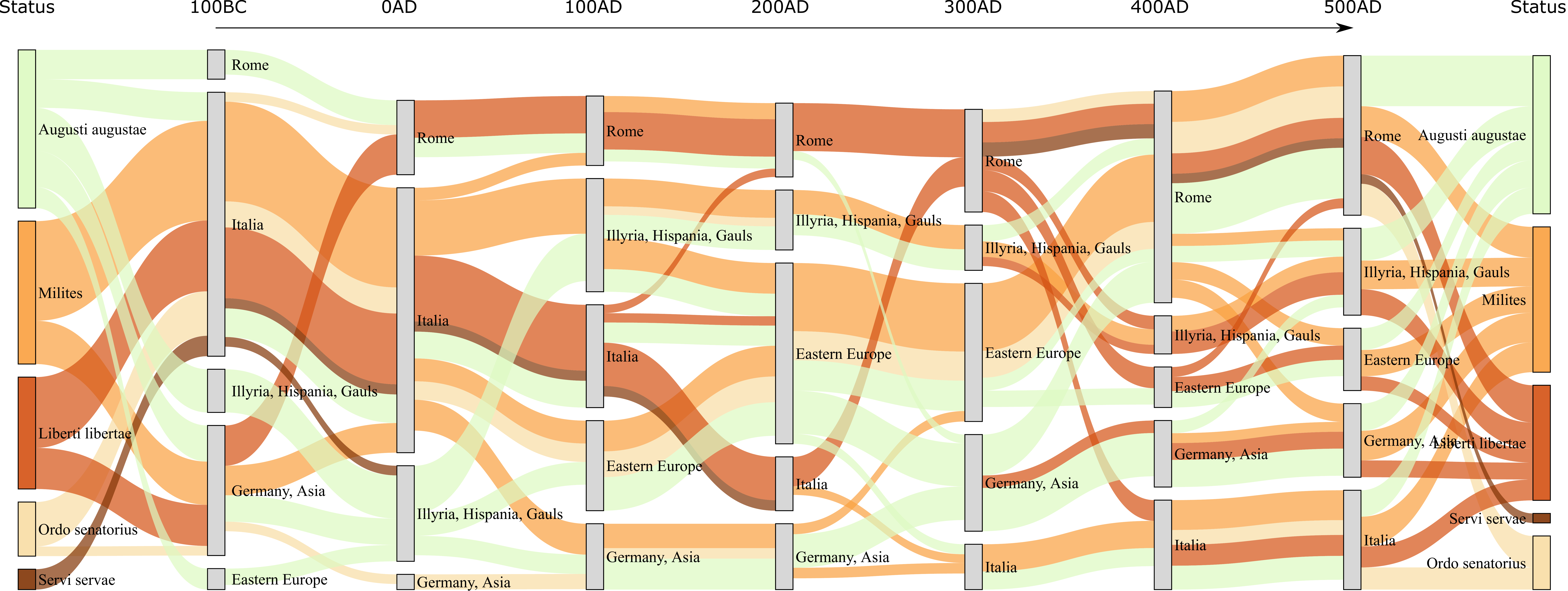}
    \caption[SDSBM - Status geographic distribution (Europe, 100BC - 500AD)]{\textbf{Geographic evolution of status distribution from Latin graves (100BC - 500AD)} --- We applied the SDSBM to the Epigraphy dataset. We recall that our goal is to predict a roman region (e.g., Illyria, Hispania, etc.) given a status (e.g., Slave, Senator, etc.) and a year. We plot the temporal evolution of statuses \gls{membership} to the five manually labelled \glspl{cluster} (in grey). For clarity, we removed small \gls{membership} transfers from the data, which explains why the total \gls{cluster}'s population may vary from one time to another. This plot allows us to visualize some global historical trends about the evolution of the Roman Empire (e.g., 3rd-century crisis, the \gls{spread} of military presence in Europe, Italy’s demilitarization, etc.).}
    \label{fig-epigraphy}
\end{figure*}

As an illustration of what SDSBM has to offer, we plot in Fig.~\ref{fig-epigraphy} a possible visualization of the \gls{membership}’s evolution over time for the epigraphy dataset. On the left and the right, we show the items that are considered in the visualization. The time goes from left (100BC) to right (500AD), and the flows represent the \gls{membership} transfers between epochs. The grey bars represent the \glspl{cluster}. We manually annotated them by looking at their composition --explicit \glspl{cluster} composition is given in Appendix, Section~\ref{SI-SDSBM-clusComp}. From this figure, we can recover several historical facts: military presence in Rome was scarce for most of the times considered; Italy concentrates less military presence as time goes (due to its \gls{spread} over the now extended empire), until the 3rd-century crisis that led to its re-militarization; most of the slaves that have been accorded an inscription are located in Italia throughout time; the religious functions (Augusti) are evenly \gls{spread} on the territory at all times; the libertii (freed slaves) inscriptions are essentially present in Rome and Italy, etc. Obviously, dedicated works are needed to support these illustrative claims, and we believe SDSBM can provide such extended comprehension of the processes at stake.

\subsubsection{Conclusion}
We introduced a simple way to model time in dynamic valued and labelled networks by assuming a dynamic Mixed-\Gls{membership} SBM. Our method defines the Simple Dynamic prior, ready to plug into any iteration of \acrshort{SIMSBM} --see Section~\ref{SIMSBM-model}. Time is considered under the single assumption that a network's ties do not vary abruptly over time. 

We assessed the performance of the proposed method by defining the SDSBM and evaluating it in several controlled situations on synthetic datasets. In particular, we show that our prior shows stable performances with respect to the dataset slicing, and that it works well under challenging conditions (small amounts of data or high entropy blocks \gls{interaction} matrix). We also evaluated SDSBM on large scale real-world datasets and showed how accounting for time increases yields better results than the two proposed baselines. Finally, we illustrate an application interest on a dataset of Latin inscriptions that indirectly narrates the evolution of the Roman Empire.

In the discussion section, we argued that our temporal prior offers great modelling flexibility: uneven slicing of observations over time, heterogeneous dynamic time scales for items (or \glspl{cluster}), time-dependent blocks-\gls{interaction} matrix, informativeness of the prior, and temporal neighbours’ dependence with respect to the averaging function. Future works exploring these directions on real-world data may help retrieve meaningful \glspl{cluster} for useful applications. On a further note, we believe that a key interest of our approach is the amount of data needed to get satisfactory predictive performances. As mentioned above, this point is fundamental to several social sciences, and we believe our approach could ease the incorporation of automated learning methods in these fields.

\section{Conclusions}
\label{SBMs-conclusion}
\paragraph{Global SIMSBM framework}
In this section, we first developed a global framework, SIMSBM, that generalizes several models from the literature as particular cases, such as MMSBM, Bi-MMSBM, IMMSBM and T-MBM. 

This results in a highly flexible model that can be applied to a broad range of problems. In particular, we cited throughout the text several experimental studies conducted in medicine, social behaviour and recommendation using special cases of our model.

Our framework answers the two challenges raised in Section~\ref{SBM-SotA}: it extends MMSBMs to any number of \gls{entity} \glspl{type} that can have higher-order \glspl{interaction}. 

\paragraph{Case study with SIMSBM(2)}
We then proposed to study \glspl{interaction} in information \gls{spread} by considering a special case of the SIMSBM: SIMSBM(2), or IMMSBM. 
We demonstrate that SIMSBM(2) allows us to investigate the detail of \glspl{interaction} strength in several datasets.

Our conclusions on \glspl{interaction} in information \gls{spread} show that their effect is not trivial and that taking them into account increases predictive performance. However, these \glspl{interaction} are sparse: only 5\% of significantly \gls{interacting} pairs of \glspl{cluster}. In most cases, non-\gls{interacting} models seem to provide an accurate enough description of output predictions -- typically a spreading action.

\paragraph{Modelling dynamic \glspl{membership} of \gls{interacting} \glspl{entity}}
The previous conclusions have been made by supposing the underlying \gls{interaction} mechanisms to remain stable over one month. While this assumption holds on some datasets, most of the time there is a need to model their temporal evolution. In this perspective, we introduced a simple temporal prior, ready to plug into any of the SIMSBM iterations. The time is considered under the single assumption that the network's ties do not vary abruptly over time. 

The performance of the proposed method is evaluated on several real-world datasets and shows that taking the time into account improves the results on each of them over the equivalent static model. On a further note, a key interest of this approach is the small amount of data needed to get satisfactory performances. This point is fundamental to several social sciences and extends beyond the scope of \glspl{interaction} modelling.

\paragraph{\Glspl{interaction} are sparse}
The overall conclusion of this section is the following: \textbf{\glspl{interaction} are sparse}. Significant \glspl{interaction} take place only between a limited fraction of \gls{cluster} pairs, and between an even smaller fraction of \gls{entity} pairs. This underlines the necessity of considering \glspl{cluster} of \glspl{entity} to efficiently model such sparse \glspl{interaction}.

\paragraph{Time range of \glspl{interaction}?}
However, all the models discussed consider static \glspl{interaction}. The dynamic version of SIMSBM proposed in Section~\ref{SDSBM} allows us to consider static \glspl{interaction} whose expression evolves with time, but the \glspl{interaction} themselves are not dynamic. To illustrate, a reasonable assumption is that the strength of \glspl{interaction} is not constant over time; a tweet does not have the same influence on a user if she saw it ten minutes or ten days earlier. The SDSBM approach accounts for the long-term changes in the evolution of the \gls{interacting} mechanisms, but not for their immediate temporal evolution. 
A Twitter user might be influenced by tweets about sports at some point, and about politics at some other point, but being always influenced in the same way by either of them; the \gls{membership} might change but not the \gls{interaction} dynamics of each tweet. Modelling the temporal evolution of \gls{interaction} strength using MMSBMs might be possible. However, other methods that consider time as a continuous variable instead of slicing it into episodes seem more relevant to the task. In the next section, we propose to investigate how \glspl{entity}’ \gls{interaction} strength evolves with time.

%% file: Tables/Chapter_2/table-res-SIMSBM.tex
\begin{table}
    \small
    \caption[SIMSBM - Experimental results on real-world datasets]{Results for every dataset presented. The letters in superscript represent the model SIMSBM generalizes in this particular configuration: MMSBM \citep{Airoldi2008MMSBM}=$^{\text{a}}$; Bi-MMSBM \citep{Antonia2016AccurateAndScalableRS}=$^{\text{b}}$; IMMSBM (Section~\ref{IMMSBM})=$^{\text{c}}$; T-MBM \citep{Tarres2019TMBM}=$^{\text{d}}$. The standard error on the last digits over all 100 runs is indicated between parenthesis -- $0.123(12) \Leftrightarrow 0.123 \pm 0.012$. The models presented in this chapter are \underline{underlined}. Overall, we see that our formulation allows to improve results on every dataset.}
    \label{table-res-SIMSBM}
	\centering
	\noindent\makebox[\textwidth]{\resizebox{\textwidth}{!}{
	\begin{tabular}{|l|l|l|S|S|S|S|S|S|S}

		\cline{1-9}
		& & & \text{F1} & \text{P@1} & \text{AUCROC} & \text{AUCPR} & \text{RAP} & \text{NCE} \\ 

		\cline{2-9}
		\multirow{8}{*}{\rotatebox[origin=c]{90}{\footnotesize \text{\textbf{MrBanks 1}}}}

		& \multirow{8}{*}{\rotatebox[origin=c]{90}{\footnotesize \text{\textbf{Ply, Sit (3), Gen, Age}}}}
		&  \underline{SIMSBM(1,1,1,1)} &  \num{ 0.7124 +- 0.0002 } &  \num{ 0.6549 +- 0.0003 } &  \num{ 0.7071 +- 0.0002 } &  \num{ 0.7141 +- 0.0003 } &  \num{ 0.8274 +- 0.0001 } &  \num{ 0.1726 +- 0.0001 } \\ 
		& & \underline{SIMSBM(1,2,1,1)} &  \num{ 0.7107 +- 0.0002 } &  \num{ 0.6696 +- 0.0005 } &  \num{ 0.7120 +- 0.0004 } &  \num{ 0.7158 +- 0.0005 } &  \num{ 0.8348 +- 0.0003 } &  \num{ 0.1652 +- 0.0003 } \\ 
		& & \underline{SIMSBM(1,3,1,1)} &  \maxf{\num{ 0.7348 +- 0.0002 } } & \maxf{\num{ 0.7172 +- 0.0005 } } & \maxf{\num{ 0.7610 +- 0.0004 } } & \maxf{\num{ 0.7646 +- 0.0004 } } & \maxf{\num{ 0.8586 +- 0.0003 } } & \maxf{\num{ 0.1414 +- 0.0003 } } \\ 
		& & TF &  \num{ 0.6795 } &  \num{ 0.6037 } &  \num{ 0.4702 } &  \num{ 0.4967 } &  \num{ 0.8019 } &  \num{ 0.1981 } \\ 
		& & NMF &  \num{ 0.7178 } &  \num{ 0.6976 } &  \num{ 0.7232 } &  \num{ 0.7182 } &  \num{ 0.8409 } &  \num{ 0.1591 } \\ 
		& & KNN &  \num{ 0.7023 } &  \num{ 0.6648 } &  \num{ 0.6859 } &  \num{ 0.6623 } &  \num{ 0.8324 } &  \num{ 0.1676 } \\ 
		& & NB &  \num{ 0.6867 } &  \num{ 0.6382 } &  \num{ 0.6323 } &  \num{ 0.6250 } &  \num{ 0.8191 } &  \num{ 0.1809 } \\ 
		& & BL &  \num{ 0.6795 } &  \num{ 0.6037 } &  \num{ 0.5000 } &  \num{ 0.5215 } &  \num{ 0.8019 } &  \num{ 0.1981 } \\ 

		\cline{1-9}
		
		\multirow{8}{*}{\rotatebox[origin=c]{90}{\footnotesize \text{\textbf{MrBanks 2}}}}

		& \multirow{8}{*}{\rotatebox[origin=c]{90}{\footnotesize \text{\textbf{ Ply, Sit (3) }}}}
		& \underline{SIMSBM(1,1)$^{\text{b}}$} &  \num{ 0.7032 +- 0.0001 } &  \num{ 0.6700 +- 0.0003 } &  \num{ 0.7049 +- 0.0002 } &  \num{ 0.7018 +- 0.0002 } &  \num{ 0.8350 +- 0.0002 } &  \num{ 0.1650 +- 0.0002 } \\ 
		& & \underline{SIMSBM(1,2)$^{\text{d}}$} &  \num{ 0.7032 +- 0.0002 } &  \num{ 0.6679 +- 0.0005 } &  \num{ 0.7028 +- 0.0004 } &  \num{ 0.7010 +- 0.0004 } &  \num{ 0.8340 +- 0.0003 } &  \num{ 0.1660 +- 0.0003 } \\ 
		& & \underline{SIMSBM(1,3)} & \maxf{\num{ 0.7290 +- 0.0003 } } & \maxf{\num{ 0.7067 +- 0.0006 } } & \maxf{\num{ 0.7547 +- 0.0005 } } & \maxf{\num{ 0.7530 +- 0.0006 } } & \maxf{\num{ 0.8533 +- 0.0003 } } & \maxf{\num{ 0.1467 +- 0.0003 } } \\ 
		& & TF &  \num{ 0.6775 } &  \num{ 0.5953 } &  \num{ 0.5054 } &  \num{ 0.5259 } &  \num{ 0.7976 } &  \num{ 0.2024 } \\ 
		& & NMF &  \num{ 0.7137 } &  \num{ 0.6908 } &  \num{ 0.7246 } &  \num{ 0.7128 } &  \num{ 0.8397 } &  \num{ 0.1603 } \\ 
		& & KNN &  \num{ 0.7100 } &  \num{ 0.6699 } &  \num{ 0.7126 } &  \num{ 0.6856 } &  \num{ 0.8349 } &  \num{ 0.1651 } \\ 
		& & NB &  \num{ 0.6802 } &  \num{ 0.6512 } &  \num{ 0.6329 } &  \num{ 0.6225 } &  \num{ 0.8256 } &  \num{ 0.1744 } \\ 
		& & BL &  \num{ 0.6775 } &  \num{ 0.5953 } &  \num{ 0.5000 } &  \num{ 0.5181 } &  \num{ 0.7976 } &  \num{ 0.2024 } \\ 

		\cline{1-9}
		
		\multirow{8}{*}{\rotatebox[origin=c]{90}{\footnotesize \text{\textbf{Spotify}}}}

		& \multirow{8}{*}{\rotatebox[origin=c]{90}{\footnotesize \text{\textbf{ Artists (3)}}}}
		& \underline{SIMSBM(1)$^{\text{a}}$} &  \num{ 0.1741 +- 0.0004 } &  \num{ 0.2155 +- 0.0007 } & \maxf{\num{ 0.7908 +- 0.0006 } } &  \num{ 0.1603 +- 0.0003 } &  \num{ 0.3827 +- 0.0004 } &  \num{ 0.0786 +- 0.0003 } \\ 
		& & \underline{SIMSBM(2)$^{\text{c}}$} &  \num{ 0.3156 +- 0.0005 } & \maxf{\num{ 0.3348 +- 0.0004 } } &  \num{ 0.7661 +- 0.0005 } &  \num{ 0.2545 +- 0.0003 } & \maxf{\num{ 0.4528 +- 0.0003 } } &  \num{ 0.0938 +- 0.0006 } \\ 
		& & \underline{SIMSBM(3)} & \maxf{\num{ 0.3243 +- 0.0004 } } &  \num{ 0.3209 +- 0.0003 } &  \num{ 0.7384 +- 0.0006 } & \maxf{\num{ 0.2613 +- 0.0003 } } &  \num{ 0.4366 +- 0.0003 } &  \num{ 0.1079 +- 0.0007 } \\ 
		& & TF &  \num{ 0.0262 } &  \num{ 0.0042 } &  \num{ 0.4805 } &  \num{ 0.0159 } &  \num{ 0.0962 } &  \num{ 0.1550 } \\ 
		& & NMF &  \num{ 0.0371 } &  \num{ 0.0658 } &  \num{ 0.5650 } &  \num{ 0.0403 } &  \num{ 0.1762 } &  \num{ 0.2557 } \\ 
		& & KNN &  \num{ 0.3201 } &  \num{ 0.3009 } &  \num{ 0.7079 } &  \num{ 0.2400 } &  \num{ 0.3941 } &  \num{ 0.5212 } \\ 
		& & NB &  \num{ 0.0463 } &  \num{ 0.0846 } &  \num{ 0.7005 } &  \num{ 0.0576 } &  \num{ 0.2264 } & \maxf{\num{ 0.0763 } } \\ 
		& & BL &  \num{ 0.0262 } &  \num{ 0.0532 } &  \num{ 0.5000 } &  \num{ 0.0135 } &  \num{ 0.1879 } &  \num{ 0.0969 } \\ 

		\cline{1-9}
		
		\multirow{8}{*}{\rotatebox[origin=c]{90}{\footnotesize \text{\textbf{PubMed}}}}

		& \multirow{8}{*}{\rotatebox[origin=c]{90}{\footnotesize \text{\textbf{ Symptoms (3)}}}}
		& \underline{SIMSBM(1)$^{\text{a}}$} &  \num{ 0.2915 +- 0.0002 } &  \num{ 0.5576 +- 0.0004 } &  \num{ 0.7475 +- 0.0001 } &  \num{ 0.2658 +- 0.0001 } &  \num{ 0.4641 +- 0.0001 } &  \num{ 0.2033 +- 0.0001 } \\ 
		& & \underline{SIMSBM(2)$^{\text{c}}$} &  \num{ 0.3127 +- 0.0001 } &  \num{ 0.5704 +- 0.0001 } &  \num{ 0.7613 +- 0.0001 } &  \num{ 0.2840 +- 0.0001 } &  \num{ 0.4838 +- 0.0001 } &  \num{ 0.1991 +- 0.0001 } \\ 
		& & \underline{SIMSBM(3)} & \maxf{\num{ 0.3219 +- 0.0001 } } & \maxf{\num{ 0.5790 +- 0.0001 } } & \maxf{\num{ 0.7666 +- 0.0001 } } & \maxf{\num{ 0.2895 +- 0.0001 } } & \maxf{\num{ 0.4937 +- 0.0001 } } & \maxf{\num{ 0.1983 +- 0.0001 } } \\ 
		& & TF &  \num{ 0.1607 } &  \num{ 0.1003 } &  \num{ 0.5605 } &  \num{ 0.1777 } &  \num{ 0.1370 } &  \num{ 0.5118 } \\ 
		& & NMF &  \num{ 0.1606 } &  \num{ 0.0293 } &  \num{ 0.5368 } &  \num{ 0.2158 } &  \num{ 0.2321 } &  \num{ 0.2959 } \\ 
		& & KNN &  \num{ 0.2414 } &  \num{ 0.3251 } &  \num{ 0.6154 } &  \num{ 0.2324 } &  \num{ 0.2891 } &  \num{ 0.7730 } \\ 
		& & NB &  \num{ 0.2600 } &  \num{ 0.1618 } &  \num{ 0.7054 } &  \num{ 0.2389 } &  \num{ 0.2036 } &  \num{ 0.3058 } \\ 
		& & BL &  \num{ 0.1607 } &  \num{ 0.1003 } &  \num{ 0.5000 } &  \num{ 0.1026 } &  \num{ 0.2464 } &  \num{ 0.2834 } \\ 

		\cline{1-9}
		
		\multirow{7}{*}{\rotatebox[origin=c]{90}{\footnotesize \text{\textbf{Imdb 1}}}}

		& \multirow{7}{*}{\rotatebox[origin=c]{90}{\footnotesize \text{\textbf{ Usr, Cast (2) }}}}
		& \underline{SIMSBM(1,1)$^{\text{b}}$} & \maxf{\num{ 0.3212 +- 0.0001 } } & \maxf{\num{ 0.2434 +- 0.0001 } } & \maxf{\num{ 0.6265 +- 0.0001 } } & \maxf{\num{ 0.2502 +- 0.0001 } } &  \num{ 0.4360 +- 0.0001 } &  \num{ 0.3504 +- 0.0001 } \\ 
		& & \underline{SIMSBM(1,2)$^{\text{d}}$} & \num{ 0.2546 +- 0.0001 } &  \num{ 0.1006 +- 0.0038 } &  \num{ 0.4998 +- 0.0003 } &  \num{ 0.1509 +- 0.0001 } &  \num{ 0.2928 +- 0.0042 } &  \num{ 0.4527 +- 0.0051 } \\ 
		& & TF &  \num{ 0.2546 } &  \num{ 0.2300 } &  \num{ 0.4960 } &  \num{ 0.1485 } &  \num{ 0.4568 } &  \num{ 0.2702 } \\ 
		& & NMF &  \num{ 0.1329 } &  \num{ 0.0593 } &  \num{ 0.5007 } &  \num{ 0.1531 } &  \num{ 0.1549 } &  \num{ 0.8087 } \\ 
		& & KNN &  \num{ 0.2578 } &  \num{ 0.1899 } &  \num{ 0.5489 } &  \num{ 0.1679 } &  \num{ 0.3290 } &  \num{ 0.5328 } \\ 
		& & NB &  \num{ 0.2555 } &  \num{ 0.2351 } &  \num{ 0.5308 } &  \num{ 0.1607 } &  \maxf{\num{ 0.4619 } } &  \maxf{\num{ 0.2596 } } \\ 
		& & BL &  \num{ 0.2546 } &  \num{ 0.2300 } &  \num{ 0.5000 } &  \num{ 0.1508 } &  \num{ 0.4605 } &  \num{ 0.2586 } \\ 
		
		\cline{1-9}

		\multirow{6}{*}{\rotatebox[origin=c]{90}{\footnotesize \text{\textbf{Imdb 2}}}}

		& \multirow{6}{*}{\rotatebox[origin=c]{90}{\footnotesize \text{\textbf{Usr, Dir, Cast}}}}
		& \underline{SIMSBM(1,1,1)} & \maxf{\num{ 0.3896 +- 0.0001 } } & \maxf{\num{ 0.3437 +- 0.0002 } } & \maxf{\num{ 0.7593 +- 0.0001 } } & \maxf{\num{ 0.3293 +- 0.0002 } } & \maxf{\num{ 0.5705 +- 0.0001 } } & \maxf{\num{ 0.1654 +- 0.0001 } } \\ 
		& & TF &  \num{ 0.2547 } &  \num{ 0.2238 } &  \num{ 0.5039 } &  \num{ 0.1513 } &  \num{ 0.4549 } &  \num{ 0.2636 } \\ 
		& & NMF &  \num{ 0.1127 } &  \num{ 0.0483 } &  \num{ 0.5005 } &  \num{ 0.1529 } &  \num{ 0.1406 } &  \num{ 0.8319 } \\ 
		& & KNN &  \num{ 0.2596 } &  \num{ 0.1890 } &  \num{ 0.5501 } &  \num{ 0.1681 } &  \num{ 0.3268 } &  \num{ 0.5248 } \\ 
		& & NB &  \num{ 0.2558 } &  \num{ 0.2373 } &  \num{ 0.5362 } &  \num{ 0.1617 } &  \num{ 0.4632 } &  \num{ 0.2571 } \\ 
		& & BL &  \num{ 0.2547 } &  \num{ 0.2286 } &  \num{ 0.5000 } &  \num{ 0.1507 } &  \num{ 0.4598 } &  \num{ 0.2574 } \\

		\cline{1-9}
	\end{tabular}
    }}
\end{table}

%% file: Tables/Chapter_2/table-res-IMMSBM.tex
\begin{table*}
\caption[IMMSBM - Experimental results on real-world datasets]{Experimental results for the four metrics considered, from each model applied to each corpus. We see that our model outperforms the proposed baselines in every dataset for almost every evaluation metric -- the error bars overlap for the AUC (ROC) on the PubMed corpus. The given error corresponds to the standard deviation over the 10 runs -- $0.123(12) \Leftrightarrow 0.123 \pm 0.012$. The naive baseline and upper limit results are constant over the runs and therefore have no variance. The models presented in this chapter are \underline{underlined}. \label{tabMetrics-IMMSBM}}
\centering
\setlength{\lgCase}{3.15cm}
\noindent\makebox[\textwidth]{
\begin{tabular}{|p{0.08\lgCase}|p{0.6\lgCase}|p{\lgCase}|p{\lgCase}|p{\lgCase}|p{0\lgCase}}
\cline{3-5}
  \multicolumn{1}{c}{} & \multicolumn{1}{c|}{} & \centering P@10 & \centering Max-F1 & \centering AUC ROC & \\
 \cline{1-5}
    \centering\multirow{4}{*}{\centering\rotatebox[origin=c]{90}{\textbf{PubMed}}} & \centering Naive & \centering 0.212 & \centering 0.160 & \centering 0.863 & \\
     
     & \centering MMSBM & \centering 0.627(2) & \centering 0.393(2) & \centering \textbf{0.911(0)} & \\
     
     & \centering \underline{IMMSBM} & \centering \textbf{0.656(1)} & \centering \textbf{0.411(1)} & \centering \textbf{0.911(2)} & \\
     
     & \centering Up.lim. & \centering 0.668 & \centering 0.450 & \centering 0.936 & \\
 \cline{1-5}

     \centering\multirow{4}{*}{\centering\rotatebox[origin=c]{90}{\textbf{Twitter}}} & \centering Naive & \centering 0.462 & \centering 0.147 & \centering 0.554 & \\
    
     & \centering MMSBM & \centering 0.529(5) & \centering 0.254(5) & \centering 0.741(4) & \\
     
     & \centering \underline{IMMSBM} & \centering \textbf{0.610(4)} & \centering \textbf{0.349(6)} & \centering \textbf{0.800(1)} & \\
     
     & \centering Up.lim. & \centering 0.737 & \centering 0.748 & \centering 0.959 & \\
 \cline{1-5}
 
     \centering\multirow{4}{*}{\centering\rotatebox[origin=c]{90}{\textbf{Reddit}}} & \centering Naive & \centering 0.488 & \centering 0.164 & \centering 0.660 & \\
    
     & \centering MMSBM & \centering 0.495(0) & \centering 0.177(0) & \centering 0.686(0) & \\
     
     & \centering \underline{IMMSBM} & \centering \textbf{0.499(0)} & \centering \textbf{0.181(0)} & \centering \textbf{0.687(0)} & \\
     
     & \centering Up.lim. & \centering 0.558 & \centering 0.582 & \centering 0.933 & \\
 \cline{1-5}
 
     \centering\multirow{4}{*}{\centering\rotatebox[origin=c]{90}{\textbf{Spotify}}} & \centering Naive & \centering 0.355 & \centering 0.088 & \centering 0.573 & \\
    
     & \centering MMSBM & \centering 0.426(6) & \centering 0.167(3) & \centering 0.707(2) & \\
     
     & \centering \underline{IMMSBM} & \centering \textbf{0.502(6)} & \centering \textbf{0.228(5)} & \centering \textbf{0.723(2)} & \\
     
     & \centering Up.lim. & \centering 0.570 & \centering 0.607 & \centering 0.944 & \\
 \cline{1-5}

\end{tabular}
}

\end{table*}

%% file: Tables/Chapter_2/table-res-SDSBM.tex
\begin{table}
	\caption[SDSBM - Experimental results on real-world datasets]{\textbf{Numerical results on real-world datasets} --- Metrics abbreviations stand for the area under the ROC curve (ROC), the Average Precision (AP), the Normalized Coverage Error (NCE). Metrics for models stand for Simple Dynamic SDM (SDSBM), No Coupling baseline (NC) and the classical static mixed membership SBM (MMSBM). Overall, our approach allows for a higher predictive power. The standard error over the folds is given in standard notation -- $0.123(12) \Leftrightarrow 0.123 \pm 0.012$. The models presented in this chapter are \underline{underlined}.}
	\label{table-res-SDSBM}
	\centering
	\setlength{\lgCase}{3.2cm}
	\noindent\makebox[\textwidth]{
	\begin{tabular}{|l|l|S|S|S|S}
		\cline{3-5}
		\multicolumn{1}{c}{\parbox{\lgCase}} & \multicolumn{1}{c|}{\parbox{\lgCase}} & {\parbox{\lgCase}{ROC}} & {\parbox{\lgCase}{AP}} & {\parbox{\lgCase}{NCE}} \\ 

		\cline{1-5}
		\multirow{3}{*}{\rotatebox[origin=c]{90}{\footnotesize \text{\textbf{Epi}}}}
		& \underline{SDSBM} & \maxf{\num{ 0.9025 +- 0.0011 } }& \maxf{\num{ 0.37 +- 0.0017 } }& \maxf{\num{ 0.1151 +- 0.0011 } }\\
		& NC & \num{ 0.842 +- 0.0022 } & \num{ 0.3435 +- 0.0036 } & \num{ 0.1582 +- 0.0019 } \\
		& MMSBM & \num{ 0.8597 +- 0.0012 } & \num{ 0.2141 +- 0.0016 } & \num{ 0.1451 +- 0.0013 } \\

		\cline{1-5}
		\multirow{3}{*}{\rotatebox[origin=c]{90}{\footnotesize \text{\textbf{Lastfm}}}}
		& \underline{SDSBM} & \maxf{\num{ 0.8942 +- 0.0008 } }& \maxf{\num{ 0.0168 +- 0.0001 } }& \maxf{\num{ 0.1284 +- 0.0011 } }\\
		& NC & \num{ 0.8393 +- 0.0005 } & \num{ 0.0157 +- 0.0002 } & \num{ 0.1785 +- 0.0007 } \\
		& MMSBM & \num{ 0.8647 +- 0.0005 } & \num{ 0.0115 +- 0.0002 } & \num{ 0.1493 +- 0.0004 } \\

		\cline{1-5}
		\multirow{3}{*}{\rotatebox[origin=c]{90}{\footnotesize \text{\textbf{Wiki}}}}
		& \underline{SDSBM} & \maxf{\num{ 0.9759 +- 0.0002 } }& \maxf{\num{ 0.0659 +- 0.0009 } }& \maxf{\num{ 0.0459 +- 0.0003 } }\\
		& NC & \num{ 0.9092 +- 0.0007 } & \num{ 0.0608 +- 0.001 } & \num{ 0.1195 +- 0.0008 } \\
		& MMSBM & \num{ 0.9576 +- 0.0007 } & \num{ 0.0622 +- 0.0004 } & \num{ 0.0565 +- 0.0008 } \\

		\cline{1-5}
		\multirow{3}{*}{\rotatebox[origin=c]{90}{\footnotesize \text{\textbf{Reddit}}}}
		& \underline{SDSBM} & \maxf{\num{ 0.9803 +- 0.0003 } }& \maxf{\num{ 0.4295 +- 0.0054 } }& \maxf{\num{ 0.0312 +- 0.0003 } }\\
		& NC & \num{ 0.8508 +- 0.0005 } & \num{ 0.3598 +- 0.0017 } & \num{ 0.1846 +- 0.0007 } \\
		& MMSBM & \num{ 0.9798 +- 0.0002 } & \maxf{\num{ 0.4269 +- 0.004 } }& \num{ 0.0322 +- 0.0003 } \\
		\cline{1-5}
	\end{tabular}
	}
\end{table}

%% file: Chapters/3_InterRate.tex
\chapter{Temporal \gls{diffusion} networks -- \Glspl{interaction} are brief}
\label{Chapter-InterRate} %

\begin{chapabstract}
We saw in the previous chapter that despite being sparse, \glspl{interaction} between \glspl{piece of information} (the \textit{\glspl{entity}}) globally play a substantial role in individuals' actions: the adoption of a product, the \gls{spread} of news, a strategy choice, etc. However, the sparsity of significant \glspl{interaction} could be due to temporal considerations; \glspl{interaction} may fade with time, making them harder to spot for static models. 

\underline{Section~\ref{InterRate-intro}}, we discuss and justify the role of time in \glspl{interaction} modelling.

\underline{Section~\ref{InterRate-sota}}, we show that this aspect of the underlying \gls{interaction} mechanisms remains unexplored in the literature.

\underline{Section~\ref{InterRate-model}}, we introduce an efficient method to infer both the \glspl{entity}’ \gls{interaction} network and its evolution according to the temporal distance separating \gls{interacting} \glspl{entity}. We develop a convex model using multi-kernel inference, named InterRate. Here, time is considered as a continuous variable, unlike what we proposed in Section~\ref{SDSBM}. The temporal evolution of \gls{interaction} intensity is what we call the \textit{\gls{interaction profile}}. 

\underline{Section~\ref{InterRate-XP}}, we consider a timestamped sequence of exposures to \glspl{entity} (e.g., URL, ads, situations) and the actions a user exerts on them (e.g., share, click, \gls{decision}). We study how users exhibit different behaviours according to \textit{combinations} of \glspl{entity} they have been exposed to in the past. We show that the joint effect of two exposures on a user is more than the disjoint sum of their individual effect --there is an \gls{interaction}. InterRate allows for non-parametric convex optimization and can be solved in parallel. 

\underline{Section~\ref{InterRate-discussion}}, we show our method recovers state-of-the-art conclusions on \gls{interaction} processes on three real-world datasets. It outperforms the proposed baselines in the inference of the underlying data generation mechanisms. Finally, we show that \glspl{interaction profile} can be visualized intuitively, making the interpretation of the model easier.

The overall conclusion of this section is that \textbf{\glspl{interaction} are brief}. In real-world \gls{diffusion} processes, the \gls{interaction} between two \glspl{entity} is significant only when the two \glspl{entity} are close to each other in time. For instance, a tweet does not have the same influence on the last tweet a user saw if the first appeared ten minutes or ten days earlier in her news feed. Typically, the intensity of an \gls{interaction} decreases exponentially with the time separating the \gls{interacting} \glspl{entity}. Our conclusion emphasizes the necessity of modelling the temporal aspect of \glspl{interaction} in social networks.

\vfill
\noindent
\textit{Published works: \citep{Poux2021InterRate}}

\end{chapabstract}
\pagebreak

\section{Introduction}
\label{InterRate-intro}
\subsection{Temporal evolution of \glspl{interaction}}
In Chapter~\ref{Chapter-SBMs}, we showed that \glspl{interaction} play a significant role in information \gls{spread}. However, the models introduced there are all static: an \gls{interaction} between several \glspl{entity} has a constant intensity over time. Many everyday examples contradict this assumption. The writing of a scientific article is more influenced by last year’s publications than by last century's ones; a user on social media is more likely to answer a recent post than a years-old one; a Spotify user is more likely to listen to a band if she listened to it five minutes ago than if she listened to it five years ago. In general, \glspl{interaction} between \glspl{entity} last for a given time; eventually \glspl{entity} fade to a non-\gls{interacting} ``ground-state'' as time goes forward. In Chapter~\ref{Chapter-SBMs}, we called this ground-state \textit{\gls{virality}}, which is the probability of an \gls{outcome} in the absence of \glspl{interaction}.

The histograms presented in Fig.\ref{fig:FigDistrib} illustrate this assumption: the probability of an action on an \gls{entity} (like, share, comment, etc.) varies according to the temporal distance separating any two given \glspl{entity}. We refer to such figures as \textit{\glspl{interaction profile}}. They represent an action's probability evolution given the time separating the \gls{interacting} \glspl{entity}.

The study of this quantity is a novel perspective: \glspl{interaction} between \glspl{piece of information} have been little explored in the literature, and no previous work unveils trends in the information \gls{interaction} mechanisms.

\begin{figure}[h]
    \centering
    \includegraphics[width=1.\columnwidth]{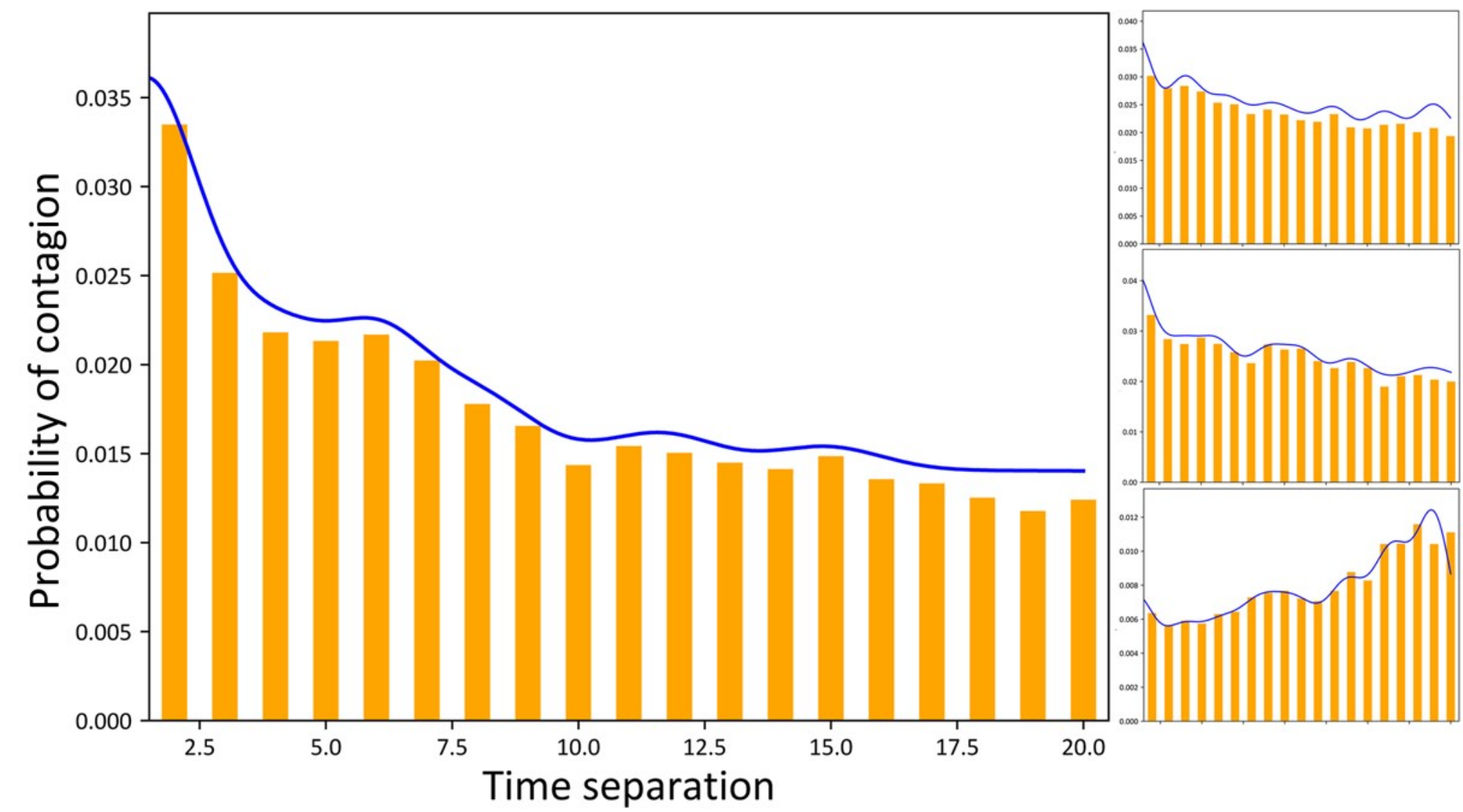}
    \caption[InterRate - \Glspl{interaction profile} between \glspl{entity} pairs]{\textbf{\Glspl{interaction profile} between pairs of \glspl{entity}} --- Examples of \glspl{interaction profile} on Twitter; here is shown the effect of URL shortening services migre.me (left), bit.ly (right-top), tinyurl (right-middle) and t.co (right-bottom) on the probability of tweeting a t.co URL and its evolution over time. This \gls{interaction profile} shows, for instance, that there is an increased probability of retweeting for a t.co URL when it appears shortly after a migre.me one (interaction). This increase fades when the time separation grows (no more \gls{interaction}). In blue, the \gls{interaction profile} inferred by our model.}
    \label{fig:FigDistrib}
\end{figure}

The study of \glspl{interaction} between \glspl{entity} has several applications in real-world systems. We can mention the fields of recommender systems (the probability of adoption is influenced by what a user saw shortly before it), news propagation and control (when to expose users to an \gls{entity} to maximize its spreading probability) \citep{Vosoughi2018SpreadNewsOnlineScience}, advertising (same reasons as before) \citep{Cao2019AdsDataset}, choice behaviour (what influences a choice and how) \citep{CoboLopez2018SocialDilemma}. 

\subsection{Proposed approach}
In this chapter, we propose to unveil the temporal mechanisms at stake within those \gls{interacting} processes; we infer their \glspl{interaction profile}. Imagine, for instance, that an internet user is exposed to a tweet at time $t_1$ and to another at time $t_2>t_1$. We suppose that the exposure to the first one influences the user's sensitivity (likeliness of a retweet) to the second one a time $t_2-t_1$ later. Modelling this process involves quantifying the influence a tweet exerts on the other and how this influence varies with the time separating the exposures. The representation of this probability evolution is what we call the \textit{\gls{interaction profile}} --illustrated Fig.\ref{fig:FigDistrib}). It represents the influence the exposure to an \gls{entity} exerts on an \gls{outcome} (click, buy, retweet, etc.) for another exposure to an \gls{entity} a given time later.

To perform this task, we introduce an efficient method to infer both the \glspl{entity}’ \gls{interaction} network and its evolution according to the temporal distance separating the \gls{interacting} \glspl{entity}. In the proposed InterRate model, nodes are \glspl{entity}, and edges between them represent the intensity of their \gls{interaction}; this intensity is a continuous-time function that depends on the time separating two exposures to the \glspl{entity} (or nodes). This intensity function is the \gls{interaction profile}.

\subsection{Workflow}
First, we review the state of the art in temporal \gls{interaction} between \glspl{entity} and underline the open challenges they rise in Section~\ref{InterRate-sota}.
Then, we answer them with InterRate in Section~\ref{InterRate-model}, a model for inferring all the \glspl{interaction profile} between every node pair in a given set. This is performed in a continuous-time setup using multi-kernel inference methods \citep{Du2012KernelCascade}. 
We show in Section~\ref{InterRate-XP} that the inference of the parameters boils down to a convex optimization problem for specific kernel families. Moreover, the problem can be subdivided into as many subproblems as \glspl{entity}, which can be solved in parallel. The convexity of the problem guarantees convergence to the likelihood's global optimum for each subproblem and, therefore, to the problem's optimal likelihood. 
In Section~\ref{InterRate-Res}, we use InterRate to investigate the role of \glspl{interaction profile} on synthetic data and in various corpora from different fields of research: advertisement (the exposure to an ad influences the adoption of other ads \citep{Cao2019AdsDataset}), social dilemmas (the previous actions of one influences an other's actions \citep{CoboLopez2018SocialDilemma}) and information \gls{spread} on Twitter (the last tweets read influence what a user retweets \citep{Myers2012CoC}).
Finally, in Section~\ref{InterRate-discussion}, we provide analysis leads and show that our method recovers state-of-the-art results on \gls{interaction} processes on each of the three considered datasets.

\subsection{Contributions}
The main contributions developed in this chapter are the following:
\begin{itemize}
    \item We introduce the \gls{interaction profile}. It represents the evolution of \gls{interaction} intensities as the time separating \gls{interacting} \glspl{entity} increases. The \gls{interaction profile} is a powerful tool to understand how \glspl{interaction} take place in a given corpus (see Fig.\ref{fig:FigAnalInter}) and it has not been developed in the literature up to now. Its introduction is the main contribution of the present section.
    
    \item We design a convex non-parametric algorithm that can be solved in parallel, baptized InterRate. InterRate automatically infers the \gls{interaction profile} for every node pair in a network. The collection of \glspl{interaction profile} can be interpreted as a temporal \gls{interaction} network.
    
    \item We show that InterRate yields better results than non-\gls{interacting} or non-temporal baseline models on several real-world datasets. Furthermore, our model can recover several conclusions about the datasets from state-of-the-art works.
    
    \item We discuss the temporal aspect of \glspl{entity} \glspl{interaction}. Specifically, we show that \glspl{interaction} are brief in real-world \gls{diffusion} processes. We also recover the \glspl{interaction} sparsity that has been observed in Chapter~\ref{Chapter-SBMs}.
\end{itemize}

\section{State of the art on temporal \gls{interaction} network inference}
\label{InterRate-sota}
\subsection{Temporal \glspl{interaction} in general}
Previous efforts in investigating the role of \glspl{interaction} in information \gls{diffusion} have shown their importance in the underlying spreading processes. Several works study the \gls{interaction} of information with users' attention \citep{Weng2012CompetitionMemes}, closely linked to information overload concepts \citep{GomezRodriguez2014InformationOverload}, but not the \gls{interaction} between the \glspl{piece of information} themselves. On the other hand, whereas most of the modelling of spreading processes is based on either no competition \citep{Senanayake2016InfluenzaSpread,Poux2020InfluentialSpreaders} or perfect competition \citep{Prakash2012WinnerTakesAll} assumption, intermediate competitions lead to a better description of their \gls{spread} \citep{Beutel2012InteractingViruses} --with the example of Firefox and Chrome web browsers, whose respective popularity are correlated but not perfectly as in \citep{Prakash2012WinnerTakesAll}. 

\subsection{Modelling \glspl{interaction}}
A significant effort has been put in elaborating complex processes to \textit{simulate} \gls{interaction} on real-world networks \citep{Prakash2012WinnerTakesAll,Zhu2020CompetitionInformationSocialNet}. However, fewer works have been developed to tackle \gls{interaction} in information \gls{spread} from a machine learning point of view. The correlated \gls{cascade} mode \citep{Zarezade2017CorrelatedCascade} infers an \gls{interacting} spreading process's latent \gls{diffusion} network. In this work, the \gls{interaction} is modelled by a hyper-parameter $\beta$. It tunes the intensity of \glspl{interaction} according to an exponentially decaying kernel. In their conclusion, the authors formulate the open problem of learning several kernels and the \gls{interaction} intensity parameter $\beta$, which we address in the present chapter.

To our knowledge, the attempt the closest to our task to model the \gls{interaction} intensity parameter $\beta$ is Clash of the contagions \citep{Myers2012CoC}. It predicts retweets on Twitter based on tweets seen by a user. This model estimates the probability of retweeting a \gls{piece of information}, given the last tweets a user has been exposed to, according to their relative position in the Twitter feed. The method suffers various flaws (scalability, non-convexity). It also defines \glspl{interaction} based on an arguable hypothesis made on the prior probability of a retweet (in the absence of \glspl{interaction}) that makes its conclusions about \glspl{interaction} sloppy (see Section~\ref{SBM-SotA-lim-contrib}). It is worth noting that in \citep{Myers2012CoC}, the authors outline the problem of the inference of the \gls{interaction profile} but do so without searching for continuous trends such as the one shown in Fig.\ref{fig:FigDistrib}. 

In Chapter~\ref{Chapter-SBMs}, we addressed the various flaws observed in \citep{Myers2012CoC} and suggested a more general approach to the estimation of the \gls{interaction} intensity parameters. However, we neglected the temporal aspect of \glspl{interaction}. To take back the Twitter case study, it implies that in the case of a retweet at time $t$, a tweet appearing at $t_1 \ll t$ in the news feed has the same influence on the retweet as a tweet that appeared at $t_2 \approx t$; the \gls{interaction profile} is constant over time.

\subsection{Temporal network inference}
For several years, temporal network inference has been a subject of interest. Significant advances have been made using survival theory modelling applied to partial observations of \gls{independent} \glspl{cascade} of contagions \citep{GomezRodriguez2011NetRate,GomezRodriguez2013InfoPath}. 
In this context, an \gls{infected} node tries to contaminate every other node at a rate that is tuned by $\beta$. 
While this work is not directly linked to ours, it has been a strong influence on the \gls{interaction profile} inference problem we develop here; the problems are different, but the methodology they introduce greatly helped building InterRate (development and convexity of the problem, analogy between \gls{interaction profile} and hazard rate). Moreover, advances in network inference based on the same works propose a multi-kernel network inference method that we adapted to the problem we tackle here \citep{Du2012KernelCascade}. Inspired by these works, we develop a flexible approach that allows for the inference of the best \gls{interaction profile} by using several candidate kernels.

\section{InterRate -- \Gls{interaction} dynamics}
\label{InterRate-model}
\textit{This work has been published, see \citep{Poux2021InterRate}}
\subsection{Problem definition}
We illustrate the process to model in Fig.~\ref{fig:FigProc}. It runs as follows: a user is first exposed to a \gls{piece of information} at time $t_0$. The user then chooses whether to act on it at time $t_0+t_{s}$ (an act can be a retweet, a buy, a booking, etc.); $t_s$ can be interpreted as the ``reaction time'' of the user to the exposure, assumed constant. The user is then exposed to the next \gls{piece of information} a time $\delta t$ later, at $t_1=t_0+\delta t$ and decides whether to act on it a time $t_s$ later, at $t_1+t_{s}$, and so on. Here, $\delta t$ is the time separating two consecutive exposures, and $t_{s}$ is the reaction time, separating the exposure from the possible contagion. In the remaining of this chapter, we refer to the user's action on an exposure (tweet appearing in the feed, exposure to an ad, etc.) as a contagion (retweet or the tweet, click on an ad, etc.).

\begin{figure*} %
    \centering
    \includegraphics[width=1.\textwidth]{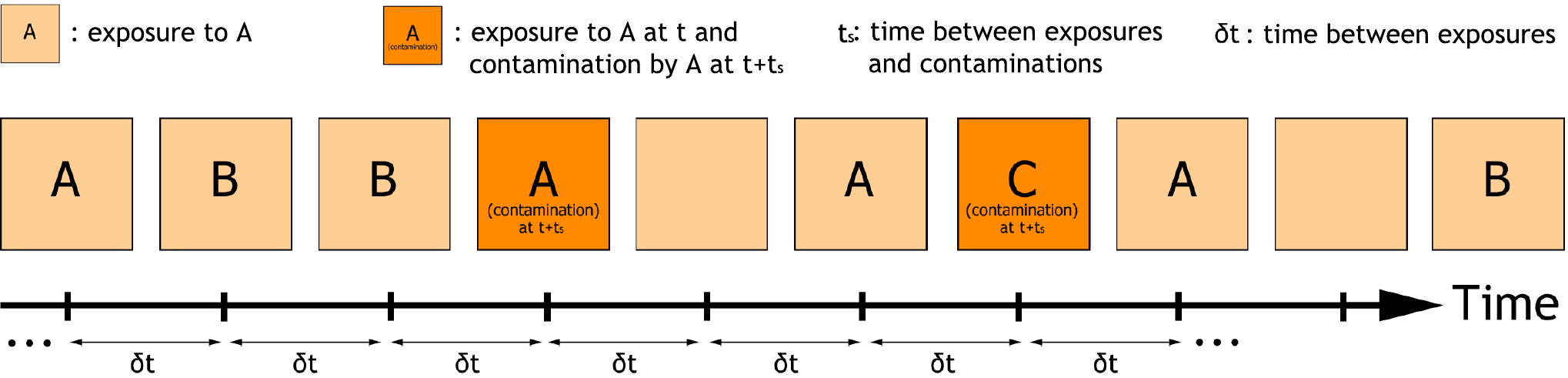}
    \caption[InterRate - Illustration of the \gls{interacting} process]{\textbf{Illustration of the \gls{interacting} process} --- Light orange squares represent the exposures, dark orange squares represent the exposures that are followed by contagions and empty squares represent the exposures to the information we do not consider in the datasets (they only play a role in the distance between exposures when we consider the order of appearance as a time feature). A contagion occurs at a time $t_s$ after the corresponding exposure. Each new exposure arrives at a time $\delta$t after the previous one. Contagion takes place with a probability conditioned by all previous exposures. In the example, the contagion by A at time $t+t_s$ depends on the effect of the exposure to A at times $t$ and $t-3\delta t$, and to B at times $t-\delta t$ and $t-2\delta t$.}
    \label{fig:FigProc}
\end{figure*}

This choice of modelling comes with several hypotheses. 
\textbf{First}, the \glspl{piece of information} a user is exposed to appear independently from each other. It is the main difference between our work and survival analysis literature: the pseudo-survival of an \gls{entity} is conditioned by the random arrival of \glspl{piece of information}. Therefore, users' actions cannot be modelled as a survival process. This assumption holds in our experiments on real-world datasets, where users have no influence on what information they are exposed to.
\textbf{Second} hypothesis, the user is contaminated solely based on the previous exposures in the feed \citep{Myers2012CoC,Zarezade2017CorrelatedCascade}.
\textbf{Third}, the reaction time separating the exposure to a \gls{piece of information} from its possible contagion, $t_s$, is constant (i.e., the time between a read and a retweet in the case of Twitter). Importantly, this hypothesis is a deliberate simplification of the model for clarity purposes; relaxing this hypothesis is straightforward by extending the kernel family, which preserves convexity and time complexity. Note that this simplification does not always hold, as shown in recent works concluding that response time can have complex time-dependent mechanisms \citep{Yu2017SurvivalSocialBehaviour}.

\subsection{Likelihood}
We now define the likelihood of the model whose process is described in Fig.~\ref{fig:FigProc}. Let $t_i^{(x)}$ be the exposure to $x$ at time $t_i$, and $t_i^{(x)} + t_s$ the time of its possible contagion.
Consider now the instantaneous probability of contagion (\textit{hazard function}) ${H(t_i^{(x)}+t_s \vert t_j^{(y)},\beta_{xy})}$, that is the probability that a user exposed to the \gls{piece of information} $x$ at time $t_i$ is contaminated by $x$ at $t_i+t_s$ given an exposure to $y$ at time $t_j \leq t_i$.
The matrix of parameters $\beta_{ij}$ is what the model infers. $\beta_{ij}$ is used to characterize the \gls{interaction profile} between \glspl{entity}.
We define the set of exposures preceding the exposure to $x$ at time $t_i$ (or history of $t_i^{(x)}$) as $\mathcal{H}_i^{(x)} := \{ t_j^{(y)} \leq t_i^{(x)} \}_{j,y}$. 
Let $\mathcal{D}$ be the whole dataset such as $\mathcal{D} := \{ ( \mathcal{H}_i^{(x)}, t_i^{(x)}, c_{t_i}^{(x)} ) \}_{i,x}$. Here, c is a binary variable that account for the contagion ($c_{t_i}^{(x)}=1$) or non-contagion ($c_{t_i}^{(x)}=0$) of $x$ at time $t_i+t_s$. The likelihood for one exposure in the sequence given $t_j^{(y)}$ is:

{\small
    \begin{equation*}
        \begin{split}
            &L(\beta_{xy} \vert \mathcal{D}, t_s) = P(\mathcal{D} \vert \beta_{xy}, t_s) =\\
            &\underbrace{H(t_i^{(x)}+t_s \vert t_j^{(y)},\beta_{xy})^{c_{t_i}^{(x)}}  }_{\text{contagion at $t_i^{(x)}+t_s$ due to $t_j^{(y)}$}}\cdot
            \underbrace{(1 - H(t_i^{(x)}+t_s \vert t_j^{(y)},\beta_{xy}))^{(1-c_{t_i}^{(x)})}}_{\text{Survival at $t_i^{(x)}+t_s$ due to $t_j^{(y)}$}}
        \end{split}
    \end{equation*}
}
The likelihood of a sequence (as defined in Fig.\ref{fig:FigProc}) is then the product of the previous expression over all the exposures that happened before the contagion event $t_i^{(x)} + t_s$ e.g. for all $t_j^{(y)} \in \mathcal{H}_i^{(x)}$.
Finally, the likelihood of the whole dataset $\mathcal{D}$ is the product of $L(\beta_{x} \vert \mathcal{D}, t_s)$ over all the observed exposures $t_i^{(x)}$. Taking the logarithm of the resulting likelihood, we get the final log-likelihood to maximize:
\begin{equation}
    \label{Eq:likelihood}
    \begin{split}
        \ell (\beta \vert \mathcal{D},t_s) =\ \ \ \ \ \ \ \ \ \ \ & \\
        \sum_{\mathcal{D}} \ \sum_{t_j^{(y)} \in \mathcal{H}_i^{(x)}}&
        c_{t_i}^{(x)} \log \left( H(t_i^{(x)}+t_s \vert t_j^{(y)},\beta_{xy}) \right)
        \\
        +\ (1-c_{t_j}^{(y)})&\log \left( 1-H(t_i^{(x)}+t_s \vert t_j^{(y)},\beta_{xy}) \right)
    \end{split}
\end{equation}
\myequations{\ \ \ \ InterRate - Likelihood}

\subsection{Proof of convexity}
The convexity of a problem guarantees to retrieve its optimal solution and allows using dedicated fast optimization algorithms.
\begin{prop}
The inference problem $\min_{\beta} -\ell (\beta \vert \mathcal{D},t_s) \  \forall \beta \geq 0$, is convex in all of the entries of $\beta$ for any hazard function that obeys the following conditions: 
\begin{equation}
    \label{Eq:ConvexCond}
    \begin{cases}
        &H'^2 \geq H''H  \\
        &H'^2 \geq -H''(1-H)\\
        &H \in \left]0;1\right[
    \end{cases}
\end{equation}
where $'$ and $''$ denote the first and second derivative with respect to $\beta$, and H is the shorthand notation for $H(t_i^{(x)}+t_s \vert t_j^{(y)},\beta_{xy}) \ \forall i,j,x,y$.
\end{prop}

\begin{proof}
The negative log-likelihood as defined in Eq.\ref{Eq:likelihood} is a summation of $-\log H$ and $-\log(1-H)$; therefore $H \in \left]0;1\right[$. The second derivative of these expressions according to any entry $\beta_{mn}$ (noted $''$) reads:
\begin{equation}
    \label{Eq:DemoConvex}
    \begin{cases}
    &\left( -\log H \right)'' = \left( \frac{-H'}{H} \right)' = \frac{H'^2-H''H}{H^2} \\
    &\left( -\log(1-H) \right)'' = \left( \frac{H'}{1-H} \right)' = \frac{H'^2 + H''(1-H)}{(1-H)^2}
    \end{cases}
\end{equation}
The convexity according to a single variable holds when the second derivative is positive, which leads to Eq.\ref{Eq:ConvexCond}. The convexity of the problem then follows from the composition rules of convexity.
\end{proof}

Several functions obey the conditions of Eq.\ref{Eq:ConvexCond}, such as the exponential ($e^{-\beta t}$), Rayleigh ($e^{-\frac{\beta}{2} t^2}$), power-law ($e^{-\beta \log t}$) functions, and any log-linear combination of those \citep{Du2012KernelCascade}. These functions are standard in survival theory literature \citep{GomezRodriguez2013SurvivalAnalysis}. 

The final convex problem can then be written $\min_{\beta \geq 0} -\ell (\beta \vert \mathcal{D},t_s)$. An interesting feature of the proposed method is that the problem can be subdivided into N convex subproblems that can be solved independently (one for each \gls{piece of information}). To solve the subproblem of the \gls{piece of information} $x$, that is to find the vector $\beta_x$, we need to consider only the subset of $\mathcal{D}$ where $x$ appears. Explicitly, each subproblem consists in maximizing Eq.\ref{Eq:likelihood} over the set of observations $\mathcal{D}^{(x)} := \{ ( \mathcal{H}_i^{(x)}, t_i^{(x)}, c_{t_i}^{(x)} ) \}_i$.

\section{Experiments}
\label{InterRate-XP}
\subsection{Experimental setup}
\subsubsection{Kernel choice}
\paragraph{Gaussian RBF kernel family (IR-RBF)}
Based on \citep{Du2012KernelCascade}, we consider a log-linear combination of Gaussian radial basis function (\acrshort{RBF}) kernels as hazard function. We also consider the time-independent kernel needed to infer the base probability of contagion discussed in the section ``Background noise in the data'' below. The resulting hazard function is then:

\begin{equation*}
\label{Eq:H}
\log H(t_i^{(x)}+t_s \vert t_j^{(y)},\beta_{ij}) = -\beta_{ij}^{(bg)} - \sum_{s=0}^S \frac{\beta^{(s)}_{ij}}{2} (t_i+t_s-t_j-s)^2
\end{equation*}

The parameters $\beta^{(s)}$ of Rayleigh kernels are the amplitude of a Gaussian distribution centered on time $s$. The parameter S represents the maximum time shift we consider. In our setup, we set S=20.
We think it is reasonable to assume that an exposition does not significantly affect a possible contagion 20 steps later. The parameter $\beta_{ij}^{(bg)}$ corresponds to the time-independent kernel --base probability of contagion by i, or \gls{virality}. The formulation allows the model to infer complex distributions from a reduced set of parameters whose interpretation is straightforward. 

\paragraph{Exponentially decaying kernel (IR-EXP)}
We also consider an exponentially decaying kernel.
We consider the following form for the hazard function and refer to this modelling as IR-EXP:%

\begin{equation*}
    \label{Eq:HHawkes}
    \log H(t_i^{(x)}+t_s \vert t_j^{(y)},\beta_{xy}) = -\beta_{ij}^{(bg)} -\beta_{ij} (t_i+t_s-t_j)
\end{equation*}
where $\beta_{ij}^{(bg)}$ once again accounts for the background noise in the data discussed further in this section.

\subsubsection{Parameters learning}
Datasets are made of sequences of exposures and contagions, as shown in Fig.\ref{fig:FigProc}. 
To assess the robustness of the proposed model, we apply a 5-fold cross-validation method. After shuffling the dataset, we use 80\% of the sequences as a training set and the 20\% left as a test set. We repeat this slicing five times, taking care that an interval cannot be part of the test set more than once. The optimization is made in parallel for each \gls{piece of information} via the convex optimization module for Python CVXPY.

We also set the time separating two exposures $\delta t$ as constant. It means that we consider only the order of arrival of exposures instead of their absolute arrival time. The hypothesis that the order of exposures matters more than the absolute exposure times has already been used with success in the literature \citep{Myers2012CoC}. Besides, in some situations, the exact exposure time cannot be collected, while the exposures' order is known.
For instance, in a Twitter corpus, we only know in what order a user reads her feed, unlike the exact time she read each of the posts. However, from its definition, our model works the same with non-integer and non-constant $\delta t$ in datasets where absolute time matters more than the order of appearance.

\subsubsection{Background noise in the data}
\begin{figure}
    \centering
    \includegraphics[width=0.8\textwidth]{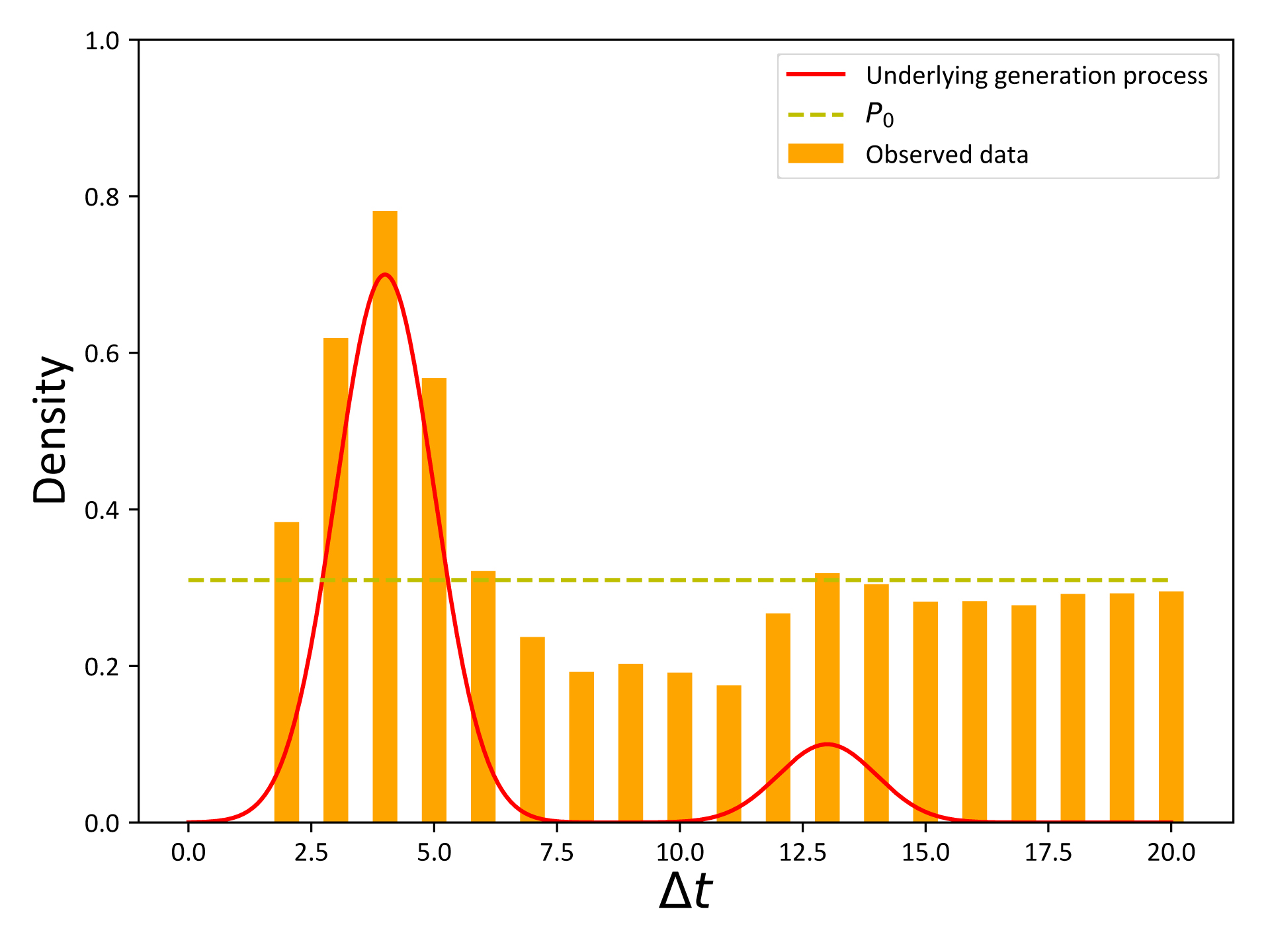}
    \caption[InterRate - Noise in the data]{\textbf{Underlying generation process vs observed data} --- The red curve represents the underlying probability of contagion by C given an exposure observed $\Delta t$ steps before C. The orange bars represent the observed probability of such events. We see that there is a noise $P_0(C)$ in the observed data. The underlying generation process can then only be observed in the dataset when its effect is larger than some threshold $P_0(C)$.}
    \label{fig:FigTrueVsObs}
\end{figure}
Because the dataset is built looking at all exposure-contagion correlations in a sequence, there is inherent noise in the resulting data. To illustrate this, we look at the illustrated example Fig.\ref{fig:FigProc} and consider the exposure to C leading to a contagion happening at time $t$. We assume that in the underlying \gls{interaction} process, the contagion by C at time $t+t_s$ only took place because C appeared at time $t$. However, when building the dataset, the contagion by C is also attributed to A appearing at times $t-\delta t$, $t-3\delta t$ and $t-6\delta t$, and to B appearing at times $t-4\delta t$ and $t-5\delta t$. It induces noise in the data. In general, for any contagion in the dataset, several observations (pair exposure-contagion) come from the random presence of \glspl{entity} unrelated to this contagion. 

We now illustrate how this problem introduces noise in the data.
In Fig.\ref{fig:FigTrueVsObs}, we see that the actual underlying data generation process (probability of a contagion by C given an exposure present $\Delta t$ step earlier) does not exactly fit the collected resulting data: the data gathering process induces a constant noise whose value is noted $P_0(C)$ --that is the average probability of contagion by C. Thus, the \gls{interaction} effect can only be observed when its associated probability of contagion is larger than $P_0(C)$. Consequently, the performance improvement of a model that accounts for \glspl{interaction} may seem small compared with a baseline that only infers $P_0(C)$. That is what we observe in the experimental section. However, in this context, a small improvement in performance shows an extended comprehension of the underlying \gls{interacting} processes at stake (see Fig.\ref{fig:FigTrueVsObs}, where the red line explains the data better than a constant baseline). Our method efficiently infers $P_0(C)$ \textit{via} a time-independent kernel function $\beta_{i,j}^{P_0(i)}$. Inferring $P_0(i)$ is in line with the discussion of Section~\ref{SBM-SotA-lim-contrib} on the necessity to infer the \gls{virality} along with \gls{interaction} terms.

\subsubsection{Evaluation criteria}
The main difficulty in evaluating these models is that \glspl{interaction} might occur between a small number of \glspl{entity} only. It is the case here, where many pairs of \glspl{entity} have little to no \gls{interaction} (see the Discussion section). This makes it difficult to evaluate how good a model is at capturing them. To this end, our principal metric is the residual sum of squares (\textbf{\acrshort{RSS}}). The RSS is the sum of the squared difference between the observed and the expected frequency of an \gls{outcome}. This metric is particularly relevant in our case, where \glspl{interaction} may occur between a small number of \glspl{entity}: any deviation from the observed frequency of contagion is accounted for, which is what we aim at predicting here. 
We also consider the Jensen-Shannon (\textbf{\acrshort{JS}}) divergence; the JS divergence is a symmetric version of the Kullback–Leibler divergence, which makes it usable as a metric \citep{Nielsen2020JSDiv}.

We finally consider the best-case F1-score (\textbf{BCF1}) of the models, that is, the F1-score of the best scenario of evaluation. It is not the standard F1 metric (that poorly distinguishes the models since few \glspl{interaction} occur), although its computation is similar. Explicitly, it generalizes F1-score for comparing probabilities instead of comparing classifications; the closer to 1, the closer the inferred and observed probabilities. It is derived from the best-case confusion matrix, whose building process is as follows: we consider the set of every information that appeared before information i at time $t_i$ in the interval, that we denote $\mathcal{H}_i$. We then compute the contagion probability of i at time $t_i+t_s$ to every exposure event $t_j^{(y)} \in \mathcal{H}_i$. Confronting this probability with the observed frequency f of contagions of i at time $t_i+t_s$ given $t_j^{(y)}$ among N observations, we can build the best-case confusion matrix. 
In the best-case scenario, if out of N observations the observed frequency is f and the predicted frequency is p, the number of True Positives is $N \times \min \{ p, f \}$, the number of False Positives is $N \times \max \{ p-f, 0 \}$, the number of True Negatives is $N \times \min \{ 1-p, 1-f \}$, the number of False Positives is $N \times \max \{ f-p, 0 \}$.

Finally, when synthetic data is considered, we also compute the mean squared error of the $\beta$ matrix inferred according to the $\beta$ matrix used to generate the observations, that we note \textbf{\acrshort{MSE}} $\beta$.

We purposely ignore evaluation in prediction because, as we show later, \glspl{interaction} influence quickly fades over time: probabilities of contagion at large times are mainly governed by the background noise discussed in previous sections. Therefore, it would be irrelevant to evaluate our approach's predictive power on the whole range of times where it does not bring any improvement over a naive baseline (see Fig.\ref{fig:FigDistrib}). A way to alleviate this problem would be to make predictions only when \gls{interaction} effects are above/below a certain threshold (at short times, for instance). However, such an evaluation process would be debatable. Here, we choose to focus on the descriptive aspect of InterRate.

\subsubsection{Baselines}
\paragraph{Naive baseline}
For a given \gls{piece of information} i, the contagion probability is defined as the number of times a user acts on it divided by the number of its occurrences.
\paragraph{Clash of the contagions}
We use the work presented in \citep{Myers2012CoC} as a baseline. In this work, the authors model the probability of a retweet given the presence of a tweet in a user's feed. This model does not look for trends in the way \glspl{interaction} take place (it does not infer an \gls{interaction profile}), considers discrete time steps (while our model works in a continuous-time framework), and is optimized via a non-convex SGD algorithm (which does not guarantee convergence towards the optimal model). More details on implementation are provided in Appendix, Section~\ref{InterRate-implemCoC}.

\paragraph{IMMSBM}
The Interactive Mixed-\Gls{membership} Stochastic Block Model (\acrshort{IMMSBM}) is a model that takes \glspl{interaction} between \glspl{piece of information} into account to compute the probability of a (non-)contagion --see Section~\ref{IMMSBM-model}. Note that this baseline does not take the position of the \gls{interacting} \glspl{piece of information} into account (time-independent) and assumes that \glspl{interaction} are symmetric (the effect of A on B is the same as B on A).

\paragraph{ICIR}
The \Gls{independent} \Gls{cascade} InterRate (ICIR) is a reduction of our main IR-RBF model to the case where \glspl{interaction} are not considered. We consider the same dataset, enforcing the constraint that off-diagonal terms of $\beta$ are null. The (non-)contagion of a \gls{piece of information} i is then determined solely by the previous exposures to i itself.

\subsection{Results}
\label{InterRate-Res}
\subsubsection{Synthetic data}

\input{Tables/Chapter_3/table-synth}

\paragraph{Data generation}
We generate synthetic data according to the process described in Fig.\ref{fig:FigProc} for a given $\beta$ matrix using the RBF kernel family. First, we generate a random matrix $\beta$, whose entries are between 0 and 1. A \gls{piece of information} is then drawn with uniform probability and can result in a contagion according to $\beta$, the RBF kernel family and its history. We simulate the \gls{outcome} by drawing a random number and finally increment the clock. The process then keeps on by randomly drawing a new exposure and adding it to the sequence. We set the maximum length of intervals to 50 steps and generate datasets of 20,000 sequences. 

\paragraph{Numerical results}
We present in Tab.~\ref{tabMetricsSynth} the results of the various models with generated \glspl{interaction} between 20 (Synth-20) and 5 (Synth-5) \glspl{entity}. The \glspl{interaction} are generated using the RBF kernel, hence the fact we are not evaluating the IR-EXP model --its use would be irrelevant.
The InterRate model outperforms the proposed baselines for every metric considered. It is worth noting that performances of non-\gls{interacting} and/or non-temporal baselines are good on the JS divergence and F1-score metrics due to the constant background noise $P_0$. For cases where \glspl{interaction} do not play a significant role, IMMSBM and Naive models perform well by fitting only the background noise. By contrast, the RSS metric distinguishes very well the models that are better at modelling \glspl{interaction}.

Note that while the baseline \citep{Myers2012CoC} yields good results when few \glspl{interaction} are simulated (Synth-5), it performs as bad as the naive baseline when this number increases (Synth-20). This is due to the non-convexity of the proposed model, which struggles to reach a global maximum of the likelihood even after 100 runs (see Appendix, Section~\ref{InterRate-implemCoC} for implementation details).

\subsubsection{Real data}
We consider 3 real-world datasets. For each dataset, we select a subset of \glspl{entity} that are likely to interact with each other. For instance, it has been shown that the \gls{interaction} between the various URL shortening services on Twitter is non-trivial \citep{Zarezade2017CorrelatedCascade}.

We provide details on the way datasets have been built from raw data. For each of the real-world datasets, we choose to consider only the order of the various \glspl{entity}' apparition instead of their absolute appearance times. It implies setting the time separating two successive exposures as constant, that we note $\delta t$. This choice is supported by state-of-the-art works \citep{Myers2012CoC}, and we observed in our experiments that it is more relevant than considering absolute times. Besides, we do not consider the first 10 \glspl{piece of information} of any sequence to avoid boundary effects (the first 5 steps for the PD dataset): the history of exposures is incomplete in this case and could lead to biased results.
For each dataset \glspl{entity} list, the number before the \gls{entity} name is the key used in Fig.~\ref{fig:FigAnalInter}. The \glspl{entity} subsets have been chosen by computing the co-occurrence matrix of all the \glspl{entity} and then selecting the ones that are part of a \gls{cluster} using a K-means algorithm. The datasets are:

\paragraph{Datasets}
\label{InterRate-datasets}
\begin{itemize}
    \item \textbf{Twitter} dataset \citep{Hodas2014DataSetTwitter}: a collection of all the tweets containing URLs that have been posted on Twitter during October 2010, with the associated followers’ network. A tweet read by a user in her feed is an exposition, and its possible retweet is a contagion. We consider only the URLs associated with the following URL shortening websites, the same as in \citep{Zarezade2017CorrelatedCascade}: \{0: migre.me, 1: bit.ly, 2: tinyurl, 3: t.co\}. The final dataset is made of 104,349 sequences of average length 53.5 steps (1 step = $t_s$), for 1,276,670,965 observed \glspl{interaction}.

    \item Prisoner's dilemma dataset (\textbf{PD}): contains ordered sequences of repeated Prisoner's dilemma games between two players. From the dataset introduced in \citep{Nay2016PrisonerDilemaDataset}, we consider the sub-dataset noted BR-risk 0 (first entry of Tab.2 in the reference); we choose this subset to have \glspl{decision} made in a homogeneous context, where players struggle in a dilemma that is hard to solve (which depends on the combination of the parameters T, R, S and P discussed further). Within each round, the two players can defect or cooperate. Each duel is made of 10 rounds. If both cooperate, the reward R is high; if both defect, the reward P is low; if one player cooperates while the other defects, this one gets a penalty S, while the other gets a reward T. To make the game a Prisoner dilemma, the variables have to obey T$>$R$>$P$>$S.
    We refer to the combination of players' actions ("the user cooperated, and the opponent defected at time t") as exposures and to the \textit{defect} actions of the player in the following round as a contagion. We defined the action of cooperating as a non-contagion. We therefore have 4 possible situations (\{0: Player cooperated, and opponent defected, 1: Both players defected, 2: Both players cooperated, 3: Player defected and opponent cooperated\}) and 2 possible \glspl{outcome} (Player cooperates or defects). The final dataset is made of 2,337 sequences of average length of 10.0 steps, for 189,297 observed \glspl{interaction}.
    
    \item Taobao dataset (\textbf{Ads}): contains all ads exposures for 1,140,000 randomly sampled users from the website of Taobao for 8 days (5/6/2017-5/13/2017) \citep{Cao2019AdsDataset}. Taobao is one of the largest e-commerce websites and is owned by Alibaba. Each exposure is associated with the corresponding timestamp and user's action (click on the ad or not). A click is considered a contagion. The subset of ads we consider is: \{0: 4520, 1: 4280, 2: 1665, 3: 4282\}. The resulting dataset is made of 87,500 sequences of average length of 23.9 steps, for 240,932,401 observed \glspl{interaction}.
\end{itemize}

\input{Tables/Chapter_3/table-rw}

\paragraph{Numerical results}
The results on real-world datasets are presented in Tab.\ref{tabMetricsRW}. We see that the IMMSBM baseline performs poorly on the PD dataset: either considering the time plays a consequent role in the probability of contagion, or \glspl{interaction} are not symmetric. Indeed, as we saw in the previous chapter, a core hypothesis of the IMMSBM is that the effect of exposition A on B is the same as B on A, whichever is the time separation between them. In a prisoner's dilemma game setting, for instance, we expect that a player does not react in the same way to defection followed by cooperation as to cooperation followed by defection, a situation for which the IMMSBM does not account. When there are few \glspl{entity}, the CoC baseline performs as good as IR, but it fails when this number increases; this is mainly due to the non-convexity of the problem that does not guarantee convergence towards the optimal solution.
Overall, the InterRate models yield the best results on every dataset.

\section{Discussion}
\label{InterRate-discussion}

\subsection{Exponential \glspl{interaction profile}}
In Fig.\ref{fig:FigAnalInter}, we represent the \gls{interaction} intensity over time for every pair of information considered in every corpus fitted with the RBF kernel model. The intensity of the \glspl{interaction} is the inferred probability of contagion minus the base contagion probability in any context: $P_{ij}(t)-P_0(i)$. We recall that $P_0(i)$ is inferred along with the other parameters; it is in line with the discussion of Section~\ref{SBM-SotA-lim-contrib} on the necessity to infer the \gls{virality} along with \gls{interaction} terms.
Therefore, we can determine the characteristic range of \glspl{interaction}, investigate recurrent patterns in \glspl{interaction}, whether the \gls{interaction} effect is positive or negative, etc. 

Overall, we understand why the EXP kernel performs as good as the RBF on the Twitter and Ads datasets: \glspl{interaction} tend to have an exponentially decaying influence over time. However, this is not the case in the PD dataset: the effect of a given \gls{interaction} is very dependent on its position in the history (pike on influence at $\Delta t=3$, shift from positive to negative influence, etc.). 

\begin{figure}[t!]
    \centering
    \includegraphics[width=1.\columnwidth]{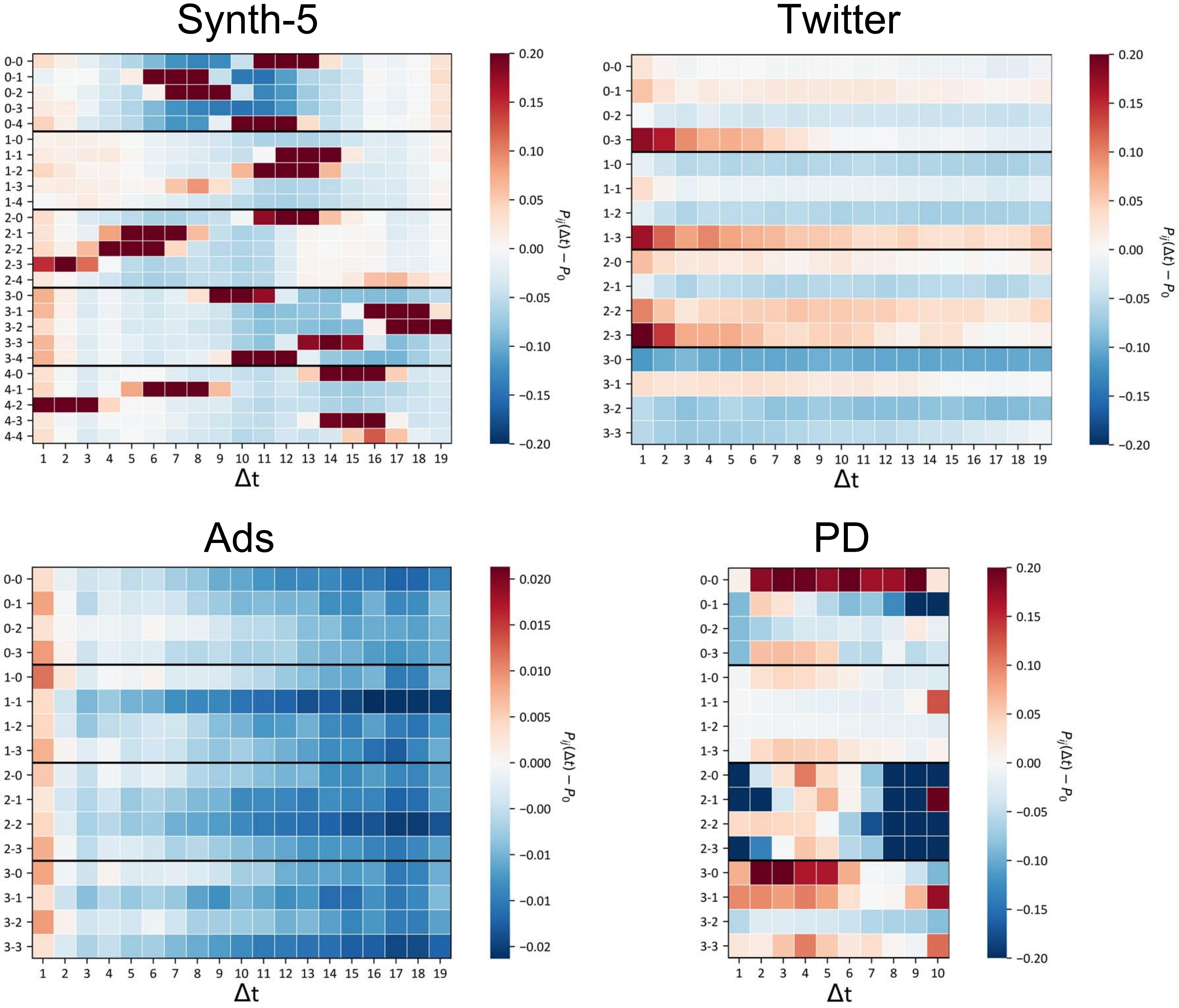}
    \caption[InterRate - Real-world \glspl{interaction profile}]{\textbf{Visualization of the \glspl{interaction profile}} --- Intensity of the \glspl{interaction} between every pair of \glspl{entity} according to their time separation (one line is one pair's \gls{interaction profile}, similar to Fig.\ref{fig:FigDistrib} seen "from the top"). A positive intensity means that the \gls{interaction} helps the contagion, while a negative intensity means it blocks it. Note that we represented discontinuous times for visualization proposes, but the inferred kernel is continuous, as in Fig.~\ref{fig:FigDistrib}. The key linking numbers on the y-axis to names for each dataset is provided in the main text, Section~\ref{InterRate-datasets}.}
    \label{fig:FigAnalInter}
\end{figure}

\subsection{Recovering state of the art conclusions}
In the Twitter dataset, the most substantial positive \glspl{interaction} occur before $\Delta$t=3. This finding agrees with previous works, which stated that the most informative \glspl{interaction} within the Twitter URL dataset occur within the 3 time steps before the possible retweet \citep{Myers2012CoC}. We also find that the vast majority of \glspl{interaction} are weak, matching with previous study's findings see Chapter~\ref{Chapter-SBMs} and \citep{Myers2012CoC}. However, it seems that tweets still exert influence even a long time after being seen, but with lesser intensity.

In the Prisoner's Dilemma dataset, players' behaviours are heavily influenced by the previous situations they have been exposed to. For instance, in the situation where both players cooperated in the previous round (pairs 2-x, 3$^{rd}$ section in Fig.\ref{fig:FigAnalInter}-PD). The probability that the player defects is then significantly increased if both players cooperated or if one betrayed the other exactly two rounds before but decreased if it has been two rounds that players both cooperate. 

Finally, we find that the \glspl{interaction} play a lesser role in the clicks on ads. We observe a slightly increased probability of clicking on every ad after direct exposure to another one. We also observe a globally decreasing probability of click when two exposures are distant in time, which agrees with previous work's findings \citep{Cao2019AdsDataset}. Finally, the \gls{interaction profile} is very similar for every pair of ads; we interpret this as a similarity in users' ads perception.

We showed that for each of the considered corpus, considering the \gls{interaction profile} provides an extended comprehension of choice adoption mechanisms and retrieves several state-of-the-art conclusions. The proposed graphical visualization also provides an intuitive view of how the \gls{interaction} occurs between \glspl{entity} and the associated trends, hence supporting its relevance as a new tool for researchers in a broad meaning.

\section{Conclusions}
\paragraph{Modelling the temporal aspect of \glspl{interaction}}
We showed in previous works that \glspl{interaction} between \glspl{entity} play a substantial role in individuals' actions. However, the temporal aspect of \glspl{interaction} has been little explored in the literature, despite their importance in modelling real-world processes \citep{Myers2012CoC,Zarezade2017CorrelatedCascade}. In this chapter, we filled the gap and introduced an efficient convex model that investigates the temporal aspect of \glspl{interaction}.

Unlike previous models, our method accounts for both the \gls{interaction} effects and their influence over time (the \gls{interaction profile}). We showed InterRate yields better results on synthetic and real-world datasets. Therefore, taking the temporal aspect of \glspl{interaction} provides a finer comprehension of the processes at stake.

However, the method introduced in this section presents a major flaw. As it has been stated in the previous chapter, most \glspl{entity} do not interact with each other --\glspl{interaction} are rare. Therefore, in this chapter, we restricted to the study of small subsets of nodes which we knew were strongly \gls{interacting} together. Our method shows that time is of importance in the modelling of \glspl{interaction} in spreading processes, but it cannot generalize these \gls{interaction} patterns to a class of \glspl{entity}; it lacks clustering. It makes the results less interpretable, as the model does not look for a trend in individual \glspl{entity}' behaviour, in which we are interested.

\paragraph{Recovering conclusions on temporal \glspl{interaction}}
We showed that InterRate manages to recover several state-of-the-art conclusions that have been made on the same datasets. On Twitter, most significant \glspl{interaction} happen at small times, and the majority of \glspl{interaction} are weak. On the advertisement corpus, our model highlighted the effects of marketing fatigue, which is a short increase in the probability of the click on an ad immediately after a first exposition, and a decreasing probability when expositions get more distant. On the prisoner's dilemma dataset, we showed our model recovers clues to a deterministic circular reasoning between players.

\paragraph{\Glspl{interaction} are short}
The overall conclusion of this section is that \textbf{\glspl{interaction} in spreading processes are brief}. Typically, the intensity of an \gls{interaction} on Twitter decreases exponentially with the time separating the \gls{interacting} \glspl{entity}.

\paragraph{Towards proper \glspl{interaction} modelling}
All the models proposed in both Chapter~\ref{Chapter-SBMs} and Chapter~\ref{Chapter-InterRate} make strong assumptions about the type of \glspl{entity} considered. Typically, an \gls{entity} has been identified by a link on Twitter, independently from its \textit{\gls{content}}. However, two different links may be about the same topic and thus be regarded identically by other \glspl{entity} they interact with. It seems reasonable to assume that \glspl{entity} carrying identical semantic meanings are regarded in the same way by other \glspl{entity} they interact with. 

As such, a more refined model should comprise three crucial elements to \gls{interaction} modelling. \textbf{\Glspl{entity} should be given a more complete definition that includes their semantic meaning} rather than only their identifier. These \textbf{\glspl{entity} must be clustered together} according to their dynamics to be able to recover significant \glspl{interaction}. \textbf{\Glspl{interaction} should be dynamic}, meaning their value should depend on the time separating the \gls{interacting} \glspl{entity}. In the next chapter, we derive such refined model that answers all these challenges.

%% file: Tables/Chapter_3/table-synth.tex
\setlength{\lgCase}{1.8cm}
\begin{table}
\centering
\caption[InterRate - Experimental results on synthetic data]{\textbf{Experimental results on synthetic data} --- Our model outperforms all of the baselines in almost every dataset for every evaluation metric. The standard deviations of the 5 folds cross-validation are negligible. The models presented in this chapter are \underline{underlined}. \label{tabMetricsSynth}}

\noindent\makebox[\textwidth]{
\begin{tabular}{|p{0.4\lgCase}|p{1.5\lgCase}|S[round-precision = 2,round-mode = places,table-text-alignment=right]|S[round-precision = 4,round-mode = places]|S[round-precision = 4]|S[round-precision = 4,round-mode = places]|}
\cline{3-6}
    \multicolumn{1}{c}{} & \multicolumn{1}{c|}{} & {\parbox{\lgCase} {RSS}} & {\parbox{\lgCase} {JS div.}} &  {\parbox{\lgCase} {BCF1}} & {\parbox{\lgCase} { MSE\,$\beta$}} \\
   
 \cline{1-6}
 
    {\multirow{4}{*}{\rotatebox[origin=c]{90}{\textbf{Synth-20}}}} &  \underline{IR-RBF} &  \maxf{\num{18.415145}} &  \maxf{\num{0.0022842}} &  \maxf{\num{0.918832}} & \maxf{\num{0.00050327}} \\
    
     &  \underline{ICIR} &  \num{139.5926} &  \num{0.0099801} & \num{0.8270477} & \num{0.0158811} \\
     
     \cdashline{2-6}
     
     &  {Naive} & \num{145.5132} &  \num{0.0103785} & \num{0.822139} &  \\
     
     &  {CoC} &  \num{123.0583} &  \num{0.0093838} &  \num{0.8220157} &  \\
     
     &  {IMMSBM} &  \num{222.055495} &  \num{0.0172875} & \num{0.7265413} &  \\
     
 \cline{1-6}
 
    {\multirow{4}{*}{\rotatebox[origin=c]{90}{\textbf{Synth-5}}}} & \underline{IR-RBF} &  \maxf{\num{0.1169267}} &  \maxf{\num{0.0002174}} & \num{0.9742133} &   \maxf{\num{0.0052977}} \\
    
     &  \underline{ICIR} &  \num{8.2660744} &  \num{0.0081174} & \num{0.8498637} & \num{0.0192061} \\
    
     \cdashline{2-6}
     
     &  {Naive} & \num{10.0264254} &  \num{0.0099556} & \num{0.8214128} & \\
     
     &  {CoC} & \maxf{\num{0.1154297}} &  \maxf{\num{0.0001974}} & \maxf{\num{0.9762872}} & \\
     
     &  {IMMSBM} & \num{11.6936415} &  \num{0.0136223} & \num{0.7692580} & \\
     
 \cline{1-6}
\end{tabular}
}
\end{table}

%% file: Tables/Chapter_3/table-rw.tex
\setlength{\lgCase}{2.3cm}
\begin{table}
\centering
\caption[InterRate - Experimental results on real-world data]{\textbf{Experimental results on real-world data} --- Our model outperforms all of the baselines in almost every dataset for every evaluation metric. The standard deviations of the 5 folds cross-validation are negligible. The models presented in this chapter are \underline{underlined}. \label{tabMetricsRW}}

\noindent\makebox[\textwidth]{
\begin{tabular}{|p{0.4\lgCase}|p{1.5\lgCase}|S[round-precision = 3,round-mode = places, table-text-alignment=right]|S[round-precision = 5,round-mode = places]|S[round-precision = 4]|}
\cline{3-5}
   \multicolumn{1}{c}{} & \multicolumn{1}{c|}{} & {\parbox{\lgCase}{RSS}} & {\parbox{\lgCase}{JS div.}} & {\parbox{\lgCase}{BCF1}} \\
 \cline{1-5}
 
    {\multirow{4}{*}{\rotatebox[origin=c]{90}{\textbf{Twitter\,\,\,\,}}}} &  \underline{IR-RBF} &  \maxf{\num{0.0014676}} &  \num{0.0000582} & \num{0.983202} \\
     
     &  \underline{IR-EXP} &  \maxf{\num{0.0011359}} &  \maxf{\num{0.0000488}} & \maxf{\num{0.986201}} \\
     
     &  \underline{ICIR} &  \num{0.0137100} &  \num{0.0006293} & \num{0.961401} \\
     
     \cdashline{2-5}
     
     &  {Naive} &  \num{0.0160866} &  \num{0.0007252} & \num{0.9379499} \\
     
     &  {CoC} &  \num{0.0016765} &  \num{0.0000672} & \num{0.9572230} \\
     
     &  {IMMSBM} &  \num{0.0147305} &  \num{0.0006829} & \num{0.9542923} \\
     
 \cline{1-5}
 
     {\multirow{4}{*}{\rotatebox[origin=c]{90}{\textbf{PD\,\,\,\,}}}} &  \underline{IR-RBF} &  \maxf{\num{1.12679}} &  \maxf{\num{0.007583}} & \maxf{\num{0.978903}} \\
     
     &  \underline{IR-EXP} &  \num{1.552556} &  \num{0.0086686} & \num{0.966056} \\
     
     &  \underline{ICIR} &  \num{3.53585} &  \num{0.018225} & \num{0.93812} \\
     
     \cdashline{2-5}
     
     &  {Naive} &  \num{3.6527} &  \num{0.0191466} & \num{0.94545} \\
     
     &  {CoC} &  \num{1.2408573} &  \num{0.0080876} & \num{0.9736453} \\
     
     &  {IMMSBM} &  \num{20.3773} &  \num{0.0870104} & \num{0.767153} \\
     
 \cline{1-5}

      {\multirow{4}{*}{\rotatebox[origin=c]{90}{\textbf{Ads\,\,\,\,\,\,}}}} &  \underline{IR-RBF} & \num{0.004338} &  \num{0.0000432} & \num{0.98143} \\
     
     &  \underline{IR-EXP} &  \maxf{\num{0.003016}} &  \maxf{\num{0.0000297}} & \maxf{\num{0.98524}} \\
     
     &  \underline{ICIR} &  \num{0.098337} &  \num{0.0008481} & \num{0.96588} \\
     
     \cdashline{2-5}
     
     &  {Naive} &   \num{0.1453111} &  \num{0.0012629} & \num{0.91258} \\
     
     &  {CoC} &  \num{0.0045146} &  \num{0.0000451} & \num{0.9741153} \\
     
     &  {IMMSBM} & \num{0.015465} &  \num{0.0001531} & \num{0.95427} \\
     
 \cline{1-5}
\end{tabular}
}
\end{table}

%% file: Chapters/4_DHPs.tex
\chapter{Dirichlet-Hawkes Processes - Modelling rare and brief \glspl{interaction}} %

\label{Chapter-DHPs} %

\begin{chapabstract}

The conclusions drawn from Chapter~\ref{Chapter-SBMs} and Chapter~\ref{Chapter-InterRate} are that \glspl{interaction} are sparse (we need \glspl{cluster} to model them) and that \glspl{interaction} are brief (we need to consider time). Besides, we discussed the fact that \gls{interacting} \glspl{entity} may have a short lifespan, and that their semantic \gls{content} must be taken into account.
In this last chapter, we present the steps paving our way to a single model that answers all of these challenges. 

\underline{Section~\ref{DHP-Introduction}}, we frame our approach as an in-depth modification of an existing Bayesian prior, the Dirichlet-Hawkes Process. 

\underline{Section~\ref{DHP-sota}}, we detail how to get this model by merging Dirichlet Processes with Hawkes processes, and highlight its main limits. 

\underline{Section~\ref{PDP}}, as an indirect way to overcome these limits, we explore alternative forms of Dirichlet Processes and end up with an expression that alleviates a major feature of the Dirichlet Processes, their ``\gls{rich-get-richer}'' property. We show that such a hypothesis is not always a relevant modelling choice. We propose the Powered Dirichlet Process as a way to directly control the ``\gls{rich-get-richer}'' assumption.

\underline{Section~\ref{PDHP}}, we then incorporate the newly proposed Powered Dirichlet Process into the standard Dirichlet-Hawkes process to create the Powered Dirichlet-Hawkes Process. We show our formulation yields significantly better results than state-of-the-art models when temporal information or textual \gls{content} is weakly informative and alleviates the hypothesis that textual \gls{content} and temporal dynamics are always perfectly correlated. 
Our approach eventually allows us to correctly model self-\glspl{interaction} within a \gls{cluster}.

\underline{Section~\ref{MPDHP}}, we extend our approach to the multivariate case, which allows us to explore not only self-\glspl{interaction}, but \glspl{interaction} between all \glspl{cluster}. We present the challenges that arise from such an extension and how we overcome them. 

\underline{Section~\ref{MPDHP-Reddit}}, finally, we perform large-scale experiments on a real-world Reddit dataset, made of all of the headlines published on news subreddits throughout the whole year 2019. We confirm the findings of previous chapters and conclude that \glspl{interaction} play a minor role in this particular dataset.

\vfill
\noindent
\textit{
Published works: 
\begin{itemize}[noitemsep,topsep=0pt]
    \item \citep{Poux2021PDHP} in Section~\ref{PDHP}
    \item \citep{Poux2022PropInterNewsReddit} in Section~\ref{MPDHP-Reddit}
\end{itemize}
}

\end{chapabstract}
\pagebreak

\section{Introduction}
\label{DHP-Introduction}
\subsection{How to properly model \glspl{interaction}}
\label{DHP-introduction-recapneedsmodel}
In the two previous chapters, we underlined several challenges that arise when modelling \glspl{interaction} between \glspl{entity} in information \gls{spread}. In Chapter~\ref{Chapter-SBMs}, we showed that \glspl{interaction} between pairs of \glspl{entity} are sparse --few \glspl{entity} interact significantly. In Chapter~\ref{Chapter-InterRate}, we showed that pair-\gls{interaction} between \glspl{entity} are brief --\glspl{entity} have to be close in time for significant \glspl{interaction} to happen.

These conclusions invite us to rethink the definition of what we consider an \gls{entity}. In previous chapters, an \gls{entity} is identified using a unique ID. It can be a link, a word, or a song, but it take multiple forms. This is a strong assumption since several of such \glspl{entity} can carry an identical semantic meaning. Let \href{https://www.youtube.com/watch?v=dQw4w9WgXcQ}{link A}\footnote{\url{https://www.youtube.com/watch?v=dQw4w9WgXcQ}} and \href{https://www.youtube.com/watch?v=oHg5SJYRHA0}{link B}\footnote{\url{https://www.youtube.com/watch?v=oHg5SJYRHA0}} point to the same \gls{content}. The previous definition of \glspl{entity} would consider them as distinct despite their semantic meaning being strictly identical. This could become problematic in contexts where such \glspl{entity} can express the same thing in several manners. Typically on Twitter or Reddit, opinions on recent news are expressed in various forms but may express the same thing. Another advantage of a more complete definition is the case of repeated information. Imagine a Twitter account that provides a daily weather forecast; each of the reports would be considered as a different \gls{entity} whereas users are expected to react in similar ways to this or that forecast (sunny, rainy, cloudy, etc.), but not to every one of them.

Therefore, we need to account for \glspl{entity}' \textit{\gls{content}}. This is even more crucial in situations where \glspl{piece of information} appear and disappear at a high rate, such as on online media platforms. The half-life of a Tweet is around 18 minutes, around 30 minutes for a Facebook post, 19h on Instagram, 6 days on Youtube, etc., which are all short in some perspective.
Short lifespans provide less information on which to learn the \gls{interacting} processes. However, aggregating these \glspl{piece of information} together using both their \gls{content} and dynamics would provide enough data to study \glspl{interaction}.

To summarize, we need to develop a model that creates \glspl{cluster} of \glspl{entity} based on both their \gls{content} and temporal \glspl{interaction}.

\subsection{Objective}
Our final goal is to model the temporal \gls{interaction} between \glspl{entity} in real-world large-scale datasets. Nowadays, online information is generated at an unprecedented rate. At the time of writing, every minute, 500,000 comments are posted on Facebook, 400 hours of videos are uploaded on Youtube, and 500,000 tweets are published on Twitter. 

As we stated in Chapter~\ref{Chapter-SBMs}, we need to \gls{cluster} this mass of \glspl{entity} together to make sense of this mass of information. As we saw in the introduction, Section~\ref{DHP-introduction-recapneedsmodel}, this clustering must take \glspl{entity}’ \gls{content} into account --typically their textual \gls{content}.
Many clustering algorithms are based on text similarity, that is, how similar the words of two published documents are \citep{Blei2003LDA,Bahdanau2015NeuralMT,Rathore2018}. 

However, in Chapter~\ref{Chapter-InterRate}, we saw that these \glspl{cluster} must also include a temporal dimension to reflect the \gls{interacting} processes at stake. Another variable to account for is thus the time of publication \citep{Blei2006DynamicTopicModel,Du2012KernelCascade}.

\subsection{Proposed approach}
Many clustering models that claim to model dynamic \glspl{cluster} do not in fact explicitly account for time. At a given time $t$, they sample a subset of recent observations according to a temporal sampling function and then learn a static model for this time-step using the selected data only \citep{Blei2006DynamicTopicModel,Amr2008RCRP, Yin2018ShortTextDHP}. However, sampling observations over time implies defining a sampling function that might not correctly model the temporal dynamics at stake. In general, time is not used by the model but instead only reduces the data provided to static models. It has been argued that such modelling is not fit to account for the arrival of documents in continuous-time settings \citep{Du2015DHP}. 

In \citep{Du2015DHP}, the authors combine techniques of standard textual clustering with temporal point processes. The idea is to infer the time-sampling functions for data selection jointly with the clustering model with which they are associated. They derive the Dirichlet-Hawkes process (\acrshort{DHP}) prior for clustering document streams by using jointly textual and temporal information in the \gls{cluster} inference. In this model, \glspl{cluster} are self-stimulated. It means that each \gls{entity} they comprise gets associated with a temporal function representing the probability of a new \gls{entity} appearing at all times. This can be interpreted as the diagonal of the \gls{interaction} matrix in Fig.~\ref{figGpesInter}, but in its temporal version --with each case being associated with an \gls{interaction profile}, see Section~\ref{InterRate-discussion}.

This model \citep{Du2015DHP} seems to fit our task very well. However, we cannot use it as such, as it suffers from some limitations and assumptions. For instance, it has been argued that this method cannot handle limit cases where text is less informative --e.g., short texts, overlapping vocabularies \citep{Yin2018ShortTextDHP}.

\subsection{Workflow}
In this chapter, we will first detail the Dirichlet-Hawkes Process (DHP) prior introduced in \citep{Du2015DHP} and point out its limits (Section~\ref{DHP-sota}). In particular, we will show that their inclusion of the temporal dimension results from an arbitrary choice, and that \textit{tuning} the influence of time allows us to recover better results on challenging datasets. 

To answer this problem, we will first question the Dirichlet Process (\acrshort{DP}) on which DHP is built (Section~\ref{PDP}). We demonstrate that alternative, more flexible variations of DP are possible and that they allow for better modelling performances; we call this alternative process the Powered Dirichlet Process (\acrshort{PDP}). 

We then develop the Powered Dirichlet-Hawkes Process (\acrshort{PDHP}) as the reformulation of the DHP model in terms of PDP (Section~\ref{PDHP}). We show this novel, more flexible formulation yields significantly better results than the original DHP and alleviate several hypotheses the original model made. 

Finally, in Section~\ref{MPDHP}, we extend the PDHP to the multivariate case (\acrshort{MPDHP}): we allow \glspl{cluster} to have temporal \glspl{interaction} with each other, and not only with themselves. This is equivalent to a temporal version of every matrix presented in Page~\pageref{figGpesInter}-Fig.~\ref{figGpesInter} --and not only their diagonal as for DHP. 

The final form of the proposed approach is able to model \glspl{interaction} that are:
\begin{itemize}
    \item sparse --by clustering \glspl{entity} together (Chapter~\ref{Chapter-SBMs})
    \item temporal --by associating each \gls{cluster} an \gls{interaction profile}, (Chapter~\ref{Chapter-InterRate})
    \item \gls{content} sensitive --by considering documents instead of \glspl{entity} identifiers
\end{itemize}

We finally conduct a large-scale study of the multivariate temporal \glspl{interaction} between \glspl{cluster} of documents (\glspl{entity} with a \gls{content}) in Section~\ref{MPDHP-Reddit}. We use 12 months of Reddit data specific to news subreddits, and conclude that the impact of \glspl{interaction} is small in this specific dataset. Reassuringly, we also confirm the findings of previous chapters on the sparsity and persistence of \glspl{interaction} in social media.

\section{State of the art and limits}
\label{DHP-sota}

\subsection{A brief overview of temporal clustering of textual documents}
The use of temporal dimension in document clustering has been studied on many occasions; a notable spike of interest happened in 2006. Many authors tackled the problem of inferring time-dependent \glspl{cluster} from models based on LDA \citep{Blei2006DynamicTopicModel,Wang2006TopicsOverTime,Iwata2009}. However, most of these models are parametric, meaning the number of \glspl{cluster} is fixed at the beginning of the algorithm. Depending on the considered time range and the dataset, the number of \glspl{cluster} needs to be fine-tuned with several \gls{independent} runs, making them hardly usable for many real-world applications. In all three references cited, the authors mention that a non-parametric version of the model might be derivable.

In 2008, A. Ahmed \textit{et al.} proposed the Recurrent Chinese Restaurant Process (RCRP) as an answer to this problem \citep{Amr2008RCRP}. Instead of considering a fixed-size dataset, this model can handle a stream of documents arriving in chronological order, and the number of \glspl{cluster} is automatically updated. In this model, time is split into episodes to capture the temporal aspect of \gls{cluster} formation; it considers an integer count of publications within a given time window. A later version of the model from 2010, the Distance-Dependent Chinese Restaurant Process (DD-CRP), tries to alleviate this approximation by replacing fixed-time episodes with a continuous-time sampling function \citep{Blei2010DDCRP}. However, the model still considers integer counts with only their distribution over time changing. The temporal dimension is not explicitly modelled, but instead used as a filter for the data fed to the model. Such slicing (even based on continuous sampling functions) can induce strong bias in the temporal modelling. One of these biases is illustrated in Fig.~\ref{fig:slicingbias}, where observations in the same time slice can be further in time than observations in different ones. Thus, the model is not designed to consider every temporal information in a continuous-time setting.

\begin{figure}
    \centering
    \includegraphics[width=\textwidth]{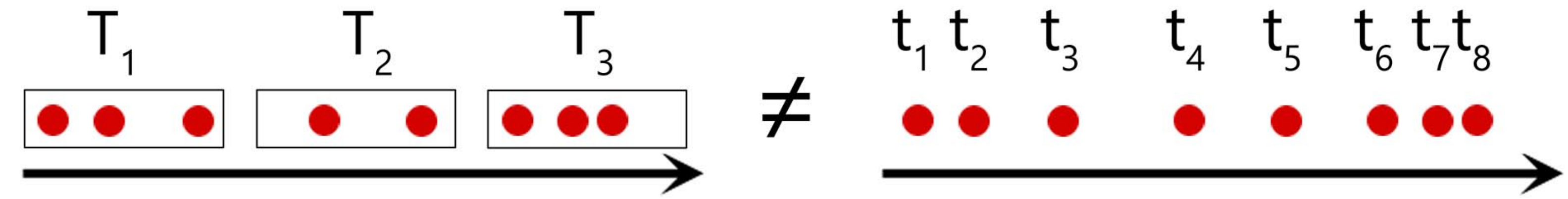}
    \caption[SotA - Time slicing bias]{\textbf{Slicing the data to consider time can introduce bias} -- Here, events in a same time slice can be further in time than observations in different ones. Explicitly modelling time using temporal processes helps getting rid of this bias.}
    \label{fig:slicingbias}
\end{figure}

In 2015, N. Du \textit{et al.} answered this problem by combining the Dirichlet process with the Hawkes process, used to model the appearance of events in a continuous-time setting. The key idea is to replace the counts of a Dirichlet process with the intensity function of the Hawkes process. The resulting Dirichlet-Hawkes process (\acrshort{DHP}) is then used as a prior for clustering documents appearing in a continuous-time stream. The inference is realized with a Sequential Monte-Carlo (\acrshort{SMC}) algorithm. Following DHP, two articles have been published extending the idea: the Hierarchical Dirichlet-Hawkes process \citep{Valera2017HDHP} and the Indian Buffet Hawkes process \citep{Tan2018IBHP}. Another work proposed an \acrshort{EM} algorithm for the inference \citep{Xu2017EMDHP} (it uses a heuristic method to update the number of \glspl{cluster} and cannot handle a stream of documents).

\subsection{Dirichlet-Hawkes Process}
\subsubsection{Dirichlet Process}
\label{sota-DHP-DP}
The Dirichlet Process is typically used in clustering models. It naturally yields a partition over a possibly infinite number of \glspl{cluster}. It is used as a prior on \glspl{entity}' \gls{membership} --such as the ones introduced in Section~\ref{SBM-SotA}.

A well-known metaphor for the Dirichlet process is referred to as ``Chinese restaurant''. The corresponding process is named ``Chinese Restaurant Process'' (\acrshort{CRP}). It can be illustrated as follows: when a $n^{th}$ client arrives in a Chinese restaurant, she will sit at one of the $K$ already occupied tables with a probability proportional to the number of persons already sitting at this table. She can also go to a new table in the restaurant and be the first client to sit there with a probability inversely proportional to the total number of clients already sitting at other tables.
It can be written formally as:
\begin{equation}
\label{eq-CRP-DHP}
    CRP (C_i = c \vert \alpha, C_1, C_2, ..., C_{i-1}) = 
    \begin{cases}
    \frac{N_c}{\alpha + N} \text{ if c = 1, 2, ..., K}\\
    \frac{\alpha}{\alpha + N} \text{ if c = K+1}
    \end{cases}
\end{equation}
\myequations{\ \ \ \ Dirichlet - Dirichlet process}
where $c$ is the \gls{cluster} chosen by the $n^{th}$ customer, $N_k$ is the population of \gls{cluster} $k$, $K$ is the number of already occupied tables and $\alpha\in \mathcal{R}^+$ the concentration parameter.
When the number of clients goes to infinity, this process is equivalent to a draw from a Dirichlet distribution over an infinite number of \glspl{cluster} with a uniform concentration parameter $\alpha$. The form of Eq.\ref{eq-CRP-DHP} is helpful to understand the underlying dynamics of the process and the contribution of seminal works we will detail now. It can be shown that the expected number of \glspl{cluster} after $N$ observations evolves as $\log N$ \citep{Arratia1992}.

The two best-known variations of the regular Dirichlet process that address the ``\gls{rich-get-richer}'' property control are the seminal Pitman-Yor process \citep{PitmanYor1997} and the Uniform process \citep{Wallach2010UnifP}. Each of them can be expressed in a similar form as Eq.\ref{eq-CRP-DHP}, and will be detailed in Section~\ref{PDP}.

\subsubsection{Hawkes Process}
\label{sota-DHP-Hawkes}
\begin{figure}
    \centering
    \includegraphics[width=\textwidth]{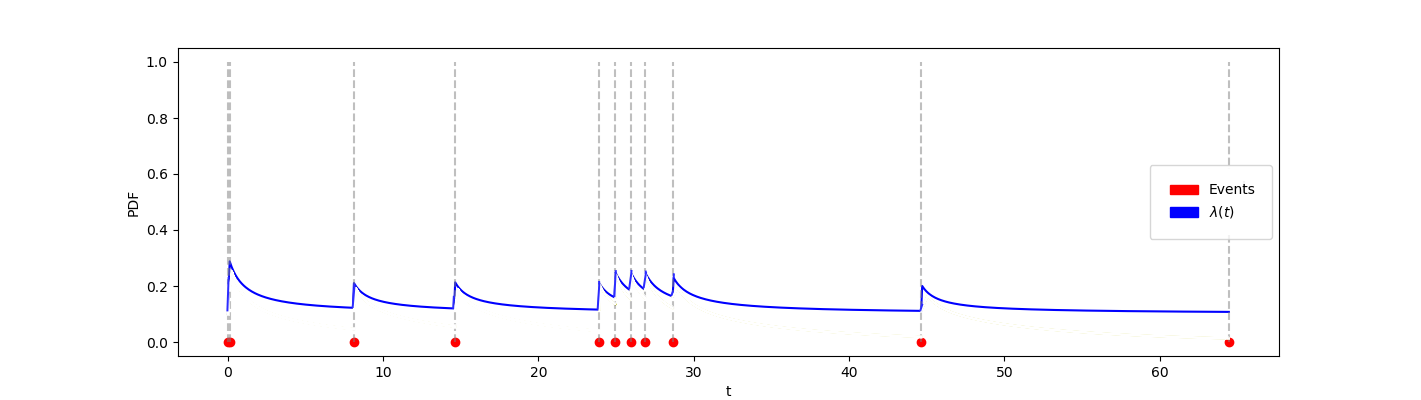}
    \caption[Hawkes - A realization of a Hawkes process]{A realization of a Hawkes process}
    \label{PDHP-Hawkes}
\end{figure}
Hawkes processes are typically used to model sequences of events in continuous time, under the assumption that the probability that new events happen is conditioned by the realization of earlier events. Typically, it can be used to model retweet \glspl{cascade}, as the more retweets there are, the most likely other retweets are to appear \citep{Chen2018HawkesTwitter}.

A Hawkes process is defined as a self-stimulating temporal point process. 
Point processes are fully characterized by an intensity function $\lambda(t)$, which is related to the probability $P(t_{events} \in [t;t+\Delta t])$ of an event happening between $t$ and $t+\Delta t$ by ${\lambda(t) = \lim_{\Delta t \rightarrow 0} \frac{P(t_{events} \in [t;t+\Delta t])}{\Delta t}}$. In the case of Hawkes processes, $\lambda(t)$ is defined conditionally on all the events that happened at times lower than $t$, that it the history up to $t$, noted $\mathcal{H}_{<t}$. We provide a synthetic realization of a Hawkes process as an illustration in Fig.~\ref{PDHP-Hawkes}.
In our setup, we define one Hawkes process for each \gls{cluster}, \gls{independent} from the others. The intensity of the Hawkes process associated with \gls{cluster} $c$ is defined as:
\begin{equation}
\label{eq-DHP-HawkesClus}
    \lambda_c(t \vert \mathcal{H}_{<t, c}) = \sum_{\mathcal{H}_{<t, c}} \vec{\alpha_c}^T \cdot \vec{\kappa}(t_{i,c})
\end{equation}
\myequations{\ \ \ \ Hawkes - Hawkes process}
where $t_{i,c}$ is the time of the $i^{th}$ observed in \gls{cluster} $c$, $\mathcal{H}_{<t, c} = \{t_{i,c} \vert t_{i,c}<t\}_{i=1,2,...}$ is the history of events in \gls{cluster} $c$ up to $t$, $\vec{\alpha_c}$ is a vector of coefficients, $\vec{\kappa}(t)$ is a vector of kernel functions with the same dimension as $\vec{\alpha}$ and $\cdot$ represents the dot product. The kernel functions are defined at the beginning of the algorithm and are not modified afterwards. We will later infer the weights vector $\vec{\alpha}$ to determine which entries of the kernel vector are the most relevant for a given situation. This technique has become standard in Hawkes processes modelling and used in several occasions \citep{Du2012KernelCascade,Yu2017}. 

The likelihood of a combination of $C$ \gls{independent} Hawkes processes can be written:
\begin{equation}
    \begin{split}
        \mathcal{L}(&\vec{\lambda} \vert \mathcal{H}_{<T, c}) = \prod_{c \in C} \mathcal{L}_c(\lambda_c \vert \mathcal{H}_{<T, c})\\
        &= \prod_c e^{-\int_0^T \lambda_c(t)dt}\prod_{t_{i,c}} \lambda_c(t_i \vert \mathcal{H}_{<t_{i,c}, c})\\
        &= e^{-\sum_c \int_0^T \lambda_c(t \vert \mathcal{H}_{<t, c})dt}\prod_{t_{i,c'}, c'=c} \lambda_c(t_{i,c'} \vert \mathcal{H}_{<t_{i,c'}, c})
    \end{split}
\end{equation}
where $T$ is the upper time of the considered observation window, going from $0$ to $T$. From now on, the dependence of Hawkes intensity functions on the history will be implicit for clarity of presentation; we define $\lambda(t) := \lambda(t \vert \mathcal{H}_{<t})$. 

\subsubsection{Dirichlet-Hawkes Process -- Expression}
As the merging of Dirichlet processes and Hawkes processes, Dirichlet-Hawkes processes aim to model self-stimulating \glspl{cluster} in continuous time. 
In their definition of the DHP, the authors of \citep{Du2015DHP} substitute the counts $N_k$ of the DP (Eq.~\ref{eq-CRP-DHP}) with the inferred Hawkes intensities (Eq.~\ref{eq-DHP-HawkesClus}), resulting in the following form for the Dirichlet-Hawkes prior:
\begin{equation}
    \label{eq-DHP}
    P(C_i = c\vert t_i, \lambda_0, \mathcal{H}_{<t_i,c}) = 
    \begin{cases}
    \frac{\lambda_c(t_i)}{\lambda_0 + \sum_{c'} \lambda_{c'}(t_i)} \text{ if c$\leq$C}\\
    \frac{\lambda_0}{\lambda_0 + \sum_{c'} \lambda_{c'}(t_i)} \text{ if c=C+1}
    \end{cases}
\end{equation}
\myequations{\ \ \ \ DHP - Dirichlet-Hawkes process}
where $t_i$ is the arrival time of document $i$. The expression in Eq.~\ref{eq-DHP} is used as a prior that accounts for the publication dynamics in Bayesian modelling. It is used as an \textit{a priori} for a model that explicitly consider the textual \gls{content} of a document.

In Eq.~\ref{eq-DHP}, the authors consider a time-independent intensity function $\lambda(t) = \lambda_0$. This process is used as the Dirichlet-Hawkes equivalent of the concentration parameter $\alpha_0$ in a Dirichlet process (see Eq.~\ref{eq-CRP-PDHP}). It corresponds to a background Poisson process which determines the rate at which new Hawkes processes are launched --that is, at which new dynamic \glspl{cluster} are created.

Given the existence of the underlying Poisson process of parameter $\lambda_0$, the temporal likelihood of the processes associated with all \glspl{cluster} can be written:
\begin{equation}
    \label{eq-likHawkes}
    \begin{split}
        \mathcal{L}(&\vec{\lambda} \vert \mathcal{H}_{<T, c}) = \mathcal{L}(\lambda_0) \prod_c \mathcal{L}_c(\lambda_c)\\
        &= e^{-\int_0^T \lambda_0 dt} \prod_c e^{-\int_0^T \lambda_c(t)dt}\prod_{t_{i,c}} \lambda_c(t_i)\\
        &= e^{-\lambda_0 T - \sum_c \int_0^T \lambda_c(t)dt}\prod_{t_{i,c'}, c'=c} \lambda_c(t_{i,c'})
    \end{split}
\end{equation}
Note that $\mathcal{L}(\lambda_0) = e^{-\int_0^T \lambda_0 dt}$ because no event is ever assigned to the Poisson process; the product over the history of events runs over 0 events, which equals 1 by convention. We recall that the Hawkes intensity dependence on the history of events is implicit; $\lambda(t) := \lambda(t \vert \mathcal{H}_{<t})$. 

\subsubsection{Textual modelling}
\label{sota-DHP-text}
We choose to model the textual \gls{content} of documents as the result of a Dirichlet-Multinomial distribution. This model is purposely simple to ease the understanding, but can easily be replaced by a more complex one. A more complete textual modelling is out of the scope of this presentation, which focuses on the definition of the DHP prior. A document will be associated with a given \gls{cluster} according to word count in every \gls{cluster} and words count in the document only. The generative process is as follows:
\begin{equation}
    \theta_i \sim Dir(\theta_0) \ \ \ \ ;\ \ \ \ \omega_{v,i} \sim Mult(\theta_i)
\end{equation}
where $\theta_i$ is the \gls{cluster} of document $i$, and $\omega_{v,i}$ is the $v^{th}$ word of document $i$. Let $\mathcal{L}_{txt}(\vec{C}_{<i, c} \vert N_{<i,c}, \theta_0)$ be the marginal joint distribution of every document's \gls{cluster} allocation up to the $i^{th}$ one. The likelihood of the $i^{th}$ document belonging to \gls{cluster} $c$ can then be expressed as:
\begin{equation}
\label{eq-likModelLg}
    \begin{split}
        \mathcal{L}(C_i=c \vert N_{<i,c}, n_i, \theta_0) &= P(n_i \vert C_i=c, N_{<i,c}, \theta_0)\\
        &= \frac{\mathcal{L}_{txt}(\vec{C}_{<i, c} \vert N_{<i,c}, \theta_0)}{\mathcal{L}_{txt}(\vec{C}_{<i-1, c} \vert N_{<i,c}, \theta_0)} \\
        &= \frac{\frac{\cancel{\Gamma(\theta_0)}}{\Gamma(N_c+n_i+\theta_0)} \prod_v \frac{\Gamma(N_{c,v} + n_{i,v} + \theta_{0,v})}{\cancel{\Gamma(\theta_{0,v})}}}
        {\frac{\cancel{\Gamma(\theta_0)}}{\Gamma(N_c+\theta_0)} \prod_v \frac{\Gamma(N_{c,v} + \theta_{0,v})}{\cancel{\Gamma(\theta_{0,v})}}}\\
        &= \frac{\Gamma(N_c+\theta_0)}{\Gamma(N_c+n_i+\theta_0)} \prod_v \frac{\Gamma(N_{c,v} + n_{i,v} + \theta_{0,v})}{\Gamma(N_{c,v}+\theta_0)}
    \end{split}
\end{equation}
\myequations{\ \ \ \ DHP - Dirichlet-Multinomial language model}
where $N_c$ is the total number of words in \gls{cluster} $c$ from observations previous to $i$, $n_i$ is the total number of words in document $i$, $N_{c,v}$ the count of word $v$ in \gls{cluster} $c$, $n_{i,v}$ the count of word $v$ in document $i$ and $\theta_0 = \sum_v \theta_{0,v}$.

\subsection{Limits}
\label{DHP-sota-limits}
A common feature of all the models we mentioned in this section is that they use a non-parametric Dirichlet process (\acrshort{DP}) prior (describe in Eq.~\ref{eq-CRP-DHP}) or variations built on it, such as DHP (Eq.~\ref{eq-DHP}) and HDHP. Yet, on several occasions, it has been pointed out that there are no specific reasons to use this process in particular and that alternative forms might work better depending on the dataset. In \citep{Welling2006AlterDP}, the author relaxes several conditions associated with DP and shows that alternative priors are an equally valid choice in Bayesian modelling. In \citep{Wallach2010UnifP}, the authors derive the Uniform process (\acrshort{UP}) and show that it performs better on a document clustering task. In Section~\ref{PDP}, we generalize UP and DP within a more general framework: the Powered Dirichlet process (\acrshort{PDP}). We show it performs better than DP on several datasets. As we show in Section~\ref{PDP} and Section~\ref{PDHP}, considering alternative definitions of the DP significantly leads to significantly different results.

As for DHP itself, it does not work well when the textual information within documents conveys little information, that is when the text is short \citep{Yin2018ShortTextDHP} or when vocabularies overlap significantly. To answer this problem, the authors develop an approach based on Dirichlet process mixtures, which is not designed for continuous-time document streams -- the temporal aspect comes from a sampling function as in \citep{Amr2008RCRP,Blei2010DDCRP}. 
There are other limiting cases for DHP, for instance when temporal information conveys little information (few observations, overlapping temporal intensities) or when documents within textual \glspl{cluster} do not follow the same temporal dynamics. To overcome those limitations, we develop the Powered Dirichlet-Hawkes process in the next section.

In addition, we uncover in Section~\ref{PDHP} new limiting cases in which DHP fails, typically when publication times convey little information (overlapping Hawkes intensities, few observations). We also show there are cases where different documents generated from the same textual \gls{cluster} do not follow the same temporal dynamics (they are associated with a different $\lambda (t)$), which the DHP is not designed to handle. For instance, an article published by a popular newspaper is unlikely to have the same influence on subsequent similar articles (temporal dynamics) as the same article published by a less popular newspaper. Textual \gls{content} is not perfectly correlated to publication times.

\section{Powered Dirichlet Process -- Alleviate the ``\gls{rich-get-richer}'' assumption}
\label{PDP}

\subsection{Introduction}
The limits of DHP raised in Section~\ref{DHP-sota-limits} can be overcome by redefining the Dirichlet Process (\acrshort{DP}) it is built on. From a broader perspective, most existing works based on Dirichlet Processes can be revisited by considering alternative forms for Dirichlet-like priors. These priors have been used extensively in Bayesian clustering over the last years. A non-exhaustive list of application includes medicine, \citep{Guimera2013DrugdrugSBM}, natural language processing \citep{Blei2003LDA,Yin2014}, genetics \citep{Qin2003,Jensen2008,McDowell2018}, recommender systems \citep{Airoldi2008MMSBM,Antonia2016AccurateAndScalableRS}, sociology \citep{Guimera2012HumanPrefSBM,CoboLopez2018SocialDilemma}, etc. 

The key idea of Bayesian clustering is to simulate a corpus of \gls{independent} observations by drawing them from a set of latent variables (\glspl{cluster}). Those \glspl{cluster} are each associated with a probability distribution on the observations, whose parameters are drawn from a prior distribution, as we will formulate mathematically later. Now, an often desirable property of Bayesian models is to make them nonparametric. In our case, it means that both the number of \glspl{cluster} and their associated distributions are inferred. A very popular prior on \glspl{cluster} distributions that allows this is the Dirichlet process. It includes a chance for a new \gls{cluster} to be created in the prior probability of a distribution (often when an observation is not likely to be explained by existing \glspl{cluster}). Otherwise, the observation is associated \textit{a priori} with an existing \gls{cluster} with a probability proportional to that \gls{cluster}'s population. Note that the model the prior is associated with might use this prior information, but might as well ignore it by design.

However, the Dirichlet process (and the related Pitman-Yor process) prior comes with a strong hypothesis on the way observations are allocated to various \glspl{cluster}: the \textit{\gls{rich-get-richer}} property \citep{Ferguson1973}. A new observation \textit{a priori} belongs to a \gls{cluster} with a probability proportional to the number of observations already present in the \gls{cluster} (see Eq.\ref{eq-CRP-DHP}); large \glspl{cluster} have a greater chance to get associated with new observations. This modelling implies a strong assumption on the way data is generated. It has already been pointed out \citep{Welling2006AlterDP} that there is a need for more flexible priors.

\subsection{Motivation}
\label{motivation}

The need for alternative priors is particularly relevant in the case of imbalanced data and scale-dependent clustering. A \gls{cluster} made of fewer \glspl{entity} might go unnoticed due to a \gls{rich-get-richer} prior. As an example of the imbalance problem, consider a case where data is processed sequentially --which is often the case when it comes to the Dirichlet process prior. The first observation from a new \gls{cluster} would then have a much larger \textit{a priori} probability to belong to a populated but irrelevant \gls{cluster}, than to open a new one (this probability decreases as $\frac{1}{N_{obs}}$). This typically happens when sampling topics from news streams \citep{Wallach2010UnifP,Xu2021ClusteringShortTexts}. In the case of scale-dependent clustering, a similar problem arises. Consider clustering people pinpointed on a map. Tiny \glspl{cluster} (at the scale of cities, for instance) might go unnoticed at larger scales (countries, for instance). To spot city \glspl{cluster} on a world map, the ``\gls{rich-get-richer}'' assumption becomes irrelevant and a ``rich-get-no-richer'' prior would be preferred \citep{Wallach2010UnifP}; the optimal solution might as well be in-between these two priors, as in Fig.\ref{fig-Antoni}. We design a method to bridge the variety of possible priors between the Dirichlet process and the Uniform process in a continuous fashion. By generalizing existing works, our method shows there exist Dirichlet-based priors that exhibit a yet unexplored class of behaviours, such as ``rich-get-less-richer'', ``rich-get-more-richer'' and ``poor-get-richer''.

Little effort has been put into exploring alternative forms of priors for nonparametric Bayesian modelling. In this section, we address this problem by deriving a more general form of the Dirichlet process that explicitly controls the importance of the ``\gls{rich-get-richer}'' assumption. Explicitly, we derive the Powered Chinese Restaurant Process (PCRP) that generalizes state-of-the-art works such as \acrshort{UP} \citep{Wallach2010UnifP} and \acrshort{DP}. We show that controlling the ``\gls{rich-get-richer}'' prior of simple models yields better results on synthetic and real-world datasets.

This work is motivated by the need to control the importance of the ``\gls{rich-get-richer}'' assumption in Dirichlet process (\acrshort{DP}) priors. Developing such a more permissive prior would impact many works based on the vanilla Dirichlet Process. In our case, we are particularly interested in the implication of DHP for the reasons discussed in Section~\ref{DHP-sota-limits}.

The ``\gls{rich-get-richer}'' property of the DP may not always be the most suitable prior for modelling a given dataset. The usual motivation for using a DP prior is that a new observation has a probability of being assigned to any \gls{cluster} proportional to its population (or intensity function, as in \citep{Du2015DHP}) in the absence of external information (such as inter-points distance in case of spatial clustering, for instance). However, this assumption might be flawed in several cases. 

Typically, most state-of-the-art works rely on tuning a parameter $\alpha$ (see Eq.\ref{eq-CRP-DHP}) to get the ``right'' number of \glspl{cluster} (this parameter shifts the distribution of the number of \glspl{cluster} as $\mathbb{E}(K \vert N) \propto \alpha \log N$ with $K$ the number of \glspl{cluster} and $N$ the number of observations). However, we argue this is a bad practice. Imagine sampling topics in a news stream: there is no specific reason for topics to appear at a rate $\alpha \log N$ as in the regular DP. The logarithmic dependence on $N$ cannot be tuned using $\alpha$. Such prior is then unfit to describe the data correctly as the number of observations grows; worst, it can lead the model it is coupled with on the wrong track.
Moreover, when considering observations streams \citep{Wallach2010UnifP,Xu2021ClusteringShortTexts}, there is usually no specific \textit{a priori} reason for a new observation to belong to a \gls{cluster} with a probability depending linearly on the \gls{cluster}'s size, as in the regular Dirichlet and Pitman-Yor processes (and variants). The order of appearance then plays too important of a role. Typically, newer observations might be associated with an existing large \gls{cluster}, despite being radically different from the points it comprises, due to a too influential ``\gls{rich-get-richer}'' assumption. %
To alleviate those assumptions, we develop a more general form of the DP process allowing a natural control of the ``\gls{rich-get-richer}'' property.

\subsection{Background}
\label{SotA-PDP}

\subsubsection{Previous works}

\paragraph{Dirichlet process}
The standard Dirichlet process has already been detailed in Section~\ref{sota-DHP-DP}. For completeness, we recall its expression as a Chinese Restaurant Process:
\begin{equation}
    CRP (C_i = c \vert \alpha, C_1, C_2, ..., C_{i-1}) = 
    \begin{cases}
    \frac{N_c}{\alpha + N} \text{ if c = 1, 2, ..., K}\\
    \frac{\alpha}{\alpha + N} \text{ if c = K+1}
    \end{cases}
\end{equation}

\paragraph{Uniform process}
A first process that breaks the ``\gls{rich-get-richer}'' property is the Uniform process. It has been used on some occasions \citep{Qin2003,Jensen2008} without focus on the prior itself. It has later been formalized and studied in comparison with the regular Dirichlet and Pitman-Yor processes \citep{Wallach2010UnifP}. It can be written as follows:
\begin{equation}
\label{eq-UP-PDP}
    UP (C_i = c \vert \alpha, C_1, C_2, ..., C_{i-1}) = 
    \begin{cases}
    \frac{1}{\alpha + K} \text{ if c = 1, 2, ..., K}\\
    \frac{\alpha}{\alpha + K} \text{ if c = K+1}
    \end{cases}
\end{equation}
This formulation completely gets rid of the ``\gls{rich-get-richer}'' property. The probability of a new client joining an occupied table is a uniform distribution over the number of occupied tables; it does not depend on the tables' population. In \citep{Wallach2010UnifP}, it has been shown that the expected number of tables evolves with $N$ as $\sqrt{N}$. Removing the ``\gls{rich-get-richer}'' property leads to a flat prior. As we show later, our formulation allows to retrieve such flat priors and thus generalizes the Uniform Process.

\paragraph{Pitman-Yor process}
Following the Chinese Restaurant process metaphor, the Pitman-Yor process \citep{PitmanYor1997,Ishwaran2003} proposed to incorporate a \textit{discount} parameter when a client opens a new table. 
Mathematically, the process can be formulated as:
\begin{equation}
\label{eq-PY}
    PY (C_i = c \vert \alpha, \beta, C_1, C_2, ..., C_{i-1}) = 
    \begin{cases}
    \frac{N_c - \beta}{\alpha + N} \text{ if c = 1, 2, ..., K}\\
    \frac{\alpha + \beta K}{\alpha + N} \text{ if c = K+1}
    \end{cases}
\end{equation}

The introduction of the parameter $\beta>0$ increases the probability of creating new \glspl{cluster}. A table with a small number of customers has significantly fewer chances to gain new ones, while the probability of opening a new table increases significantly. It can be shown that the number of tables evolves with the number of clients $N$ as $N^{\beta}$ \citep{Sudderth2009,Wallach2010UnifP,Goldwater2011}. However, this process does not control the arguable ``\gls{rich-get-richer}'' hypothesis \citep{Welling2006AlterDP}, since the relation to the population of a table remains linear; it only shifts this dependence of a value $\beta$. It makes so by creating \glspl{cluster} based on the number of existing \glspl{cluster} and the total number of observations, but not according to the population of already existing \glspl{cluster}. Those play the same role in the Pitman-Yor process as in the DP. The Pitman-Yor process thus comes with two limitations. Firstly, since $\beta>0$, it cannot modify the process to generate fewer \glspl{cluster}. Secondly, the discount parameter does not modify the linear dependence on previous observations for \gls{cluster} allocations --- rich still get richer; the prior is as peaky on large \glspl{cluster} as before. The present section offers to address those two limitations.

\paragraph{Other extensions}
Another similar prior, the Power-law Indian Buffet Process, has been proposed so that a realization would yield a number of \glspl{cluster} obeying a power-law as the number of observations increases \citep{Teh2009IBDwithPL}. This formulation can be seen as a generalization of the Pitman-Yor process; it adds an additional parameter that sums with $N$ in the denominator of Eq.~\ref{eq-PY}. However, the posterior probability for a new customer to belong to a \gls{cluster} depends linearly on each \gls{cluster}'s size, and the ``\gls{rich-get-richer}'' hypothesis is preserved.

Finally, the Generalized Gamma Process proposed a similar discount idea to increase the probability of opening new \glspl{cluster} in \citep{Lijoi2017GeneralizedGamma}. The proposed prior (\citep{Lijoi2017GeneralizedGamma}-Eq.4) modifies a \gls{cluster}'s probability to get chosen by subtracting a constant term from each \gls{cluster}'s population. Thus, the ``\gls{rich-get-richer}'' property is not alleviated in their approach either, since the dependence on the \gls{cluster}'s population is still linear. As for the \acrshort{PY} process, this formulation only allows to increase the number of \glspl{cluster} and does not alleviate the ``\gls{rich-get-richer}'' hypothesis.

\subsubsection{Contributions}
In the next section, we derive the Powered Dirichlet Process (\acrshort{PDP}) that allows controlling the ``\gls{rich-get-richer}'' property. The process is also referred to as the Powered Chinese Restaurant Process (PCRP) due to its formulation being close to the vanilla metaphor. This new process generalises state-of-the-art works. Such generalization allows to define unexplored classes of \textit{a priori} hypotheses: poor-get-richer, rich-get-no-richer (Uniform process), rich-get-less-richer, \gls{rich-get-richer} (DP), and rich-get-more-richer. In doing so, we define the Powered Dirichlet-Multinomial distribution. We detail some key properties of the Powered Dirichlet Process (convergence, expected number of \glspl{cluster}). Finally, we show that controlling the ``\gls{rich-get-richer}'' prior of simple models yields better results on synthetic and real-world datasets.

\subsection{The model}

\subsubsection{The Dirichlet-Multinomial distribution}
As explained earlier, the Dirichlet distribution yields a collection of positive variables whose sum equals 1. The Multinomial distribution yields a sum of counts in each of the $K$ ``boxes'' (here \glspl{cluster}) following $N$ draws from an identical distribution over those boxes. 
We recall the definition of the Dirichlet distribution and of the Multinomial distribution:
\begin{equation}
\label{eq-Mult}
    Dir(\vec{p} \vert \vec{\alpha}) = \frac{\prod_k p_k^{\alpha_k - 1}}{B(\vec{\alpha})} \ \ \ \, \ \ \ \ Mult(\vec{N} \vert N, \vec{p}) = \frac{\Gamma(\sum_k N_k + 1)}{\prod_k \Gamma(N_k + 1)} \prod_k p_k^{N_k} 
\end{equation}
with $\vec{N} = (N_1, N_2, ..., N_K)$ where $N_k$ is the integer number of draws assigned to \gls{cluster} $k$, $N = \sum_k N_k$ the total number of draws, $\Gamma(x)=(x-1)!$ is the gamma function, and ${B(\vec{x}) = \prod_k \Gamma(x_k)/\Gamma(\sum_k x_k)}$ is the beta function.

The Dirichlet-Multinomial distribution merges both distributions: the probabilities for each ``box'' in the Multinomial distribution are sampled once from a Dirichlet distribution. As we will show later, the Dirichlet process can be derived from this distribution. The Dirichlet-Multinomial distribution is defined as follows:
\begin{equation}
\label{eq-DirMult}
\begin{split}
    p(\vec{N} \vert \vec{\alpha}, n) &= \int_{\vec{p}} p(\vec{N}\vert \vec{p}, n)p(\vec{p} \vert \vec{\alpha}) d\vec{p}\\
    &= \frac{(n!)\Gamma(\sum_k \alpha_k)}{\Gamma(n+\sum_k \alpha_k)} \prod_{k=1}^{K}\frac{\Gamma(N_k + \alpha_k)}{(N_k!)\Gamma(\alpha_k)}\\
    \text{where }&\Vec{p} \sim Dir(\vec{p} \vert \vec{\alpha}) \, ; \, \Vec{N} \sim Mult(\vec{N} \vert n, \Vec{p})
\end{split}
\end{equation}
In Eq.\ref{eq-DirMult}, we sample $n$ values over a space of $K$ distinct \glspl{cluster} each with probability $\vec{p}=(p_1, p_2, ..., p_K)$, using a Dirichlet prior with parameter $\vec{\alpha} = (\alpha_1, \alpha_2, ..., \alpha_K)$. To derive the Dirichlet process equation, we must compute a new observation's conditional distribution to belong to any \gls{cluster} given the allocation of all the previous random variables when $K \rightarrow \infty$. 

\subsubsection{Powered conditional Dirichlet prior}
In the derivation of the standard Dirichlet-Multinomial posterior predictive, we consider a single draw from the Multinomial distribution (i.e., a categorical distribution) with a Dirichlet prior on the parameter $\vec{p}$. Usually, this prior is linearly dependent on previous draws from the distribution. We propose to modify this assumption by using a Dirichlet prior that depends non-linearly on the history of draws as:
\begin{equation}
\label{eq-DirMultPrior}
    Dir_r(\vec{p} \vert \vec{\alpha}, \vec{N}) = \frac{1}{B(\vec{\alpha} + \vec{N}^r)} \prod_k p_k^{\alpha_k + N_k^r - 1}
\end{equation}
In Eq.\ref{eq-DirMultPrior}, the vector $\vec{N}^r$ shifts the parameter $\vec{\alpha}$ according to the count of draws allocated to each \gls{cluster} $k$ up to the n$^{th}$ draw. The parameter $r \in \mathbb{R}$ controls the intensity of this shift for each entry of $\vec{N}$.

We demonstrate that the Powered Dirichlet distribution is a conjugate prior of the Multinomial distribution, by writing Eq.\ref{eq-DirMultPrior} as:
\begin{equation}
\begin{split}
\label{eq-DirMultPrior2}
    Dir_r(&\vec{p} \vert \vec{\alpha}, \vec{N}) = \frac{1}{B(\vec{\alpha} + \vec{N}^r)} \prod_k p_k^{\alpha_k - 1} \prod_k p_k^{N_k^r}\\
    &\stackrel{\text{Eqs.\ref{eq-Mult}}}{=} \frac{B(\vec{\alpha}) \prod_k N_k^r!}{B(\vec{\alpha} + \vec{N}^r) (\sum_k N_k^r)!} Dir(\vec{p} \vert \vec{\alpha}) Mult(\vec{N}^r \vert \sum_k N_k^r, \vec{p})\\
    &\propto Dir(\vec{p} \vert \vec{\alpha}) Mult(\vec{N}^r \vert \sum_k N_k^r, \vec{p})
\end{split}
\end{equation}
where the prior on vector $\vec{N}$ is a regular Multinomial distribution of parameter $N=\sum_k N_k^r$. Note that for certain values of $r$, the vector $\vec{N}^r$ might not be made of integer values; the resulting Multinomial prior on $\vec{N}^r$ must then be expressed in terms of $\Gamma$ functions (see Eq.\ref{eq-Mult}) to be valid for $\vec{N}^r \in \mathbb{R}^{\vert \vec{N} \vert}$. 
Distributions of non-integer counts are not new in the literature \citep{Khurshid2005ConfidenceIntervNegBinDist,McCarthy2012DifferentialEA,Ghitza2013} and are essentially allowed by the generalized definition of the factorial function in terms of the gamma function. 
When $r=1$, we recover the standard Dirichlet-Multinomial prior on $\vec{p}$ for the $n^{th}$ draw; the history of draws $\vec{N}$ can be expressed as the result of $N$ \gls{independent} draws of equal probability $\vec{p}$. When $r \neq 1$, the prior on $\vec{N}$ is sampled from a Multinomial distribution in which the number of samples drawn depends on $r$ as $\sum_k N_k^r$. For instance, let $\vec{N} = (1, 2)$ and $r=2$: the resulting powered conditional Dirichlet prior would then be sampled from a Multinomial distribution $Mult(\vec{N}=(1,4) \vert N=5, \vec{p}=(p, 1-p))$. 

\subsubsection{Posterior predictive}
We now derive the posterior distribution for the $n^{th}$ draw to belong to a \gls{cluster} $c$ given all previous draws. We assume that $\vec{C_-}$ represents all previous realizations up to $n-1$, that is, the \gls{cluster} to which each previous draw has been associated. For simplicity of notation, we define the population of a \gls{cluster} $k$ at time $n-1$ as $N_k = \vert \{C_i \vert i=k\}_{i=1, 2, ..., n-1} \vert$. We are now looking at the probability distribution of its $n^{th}$ draw to belong to $c$. It is expressed as the probability of a draw from the categorical distribution given all previous observations (because there is only one new draw, it is the same as a Multinomial distribution with parameter $N=1$) combined with the powered Dirichlet prior defined Eq.\ref{eq-DirMultPrior}. Then:
\begin{equation}
    \label{eq-DirCatr}
    \begin{split}
        DirCat_r(C_n = c \vert \vec{\alpha}, \vec{C_-}) &= \int_{\vec{p}} Cat(C_n = c \vert \vec{p}) \underbrace{Dir_r(\vec{p} \vert \vec{\alpha}, \vec{N})}_{\textbf{Eq.\ref{eq-DirMultPrior}}} \\
        =& \int_{\vec{p}} \frac{1}{B(\vec{\alpha} + \vec{N}^r)} \prod_k p_k^{c_k + \alpha_k + N_k^r - 1}\\
        &= \frac{B(\vec{c} + \vec{\alpha} + \vec{N}^r)}{B(\vec{\alpha} + \vec{N}^r)}
    \end{split}
\end{equation}
where $\vec{c}$ is a vector of the same length as $\vec{\alpha}$ whose $c^{th}$ entry equals to 1, and 0 anywhere else. Alternative demonstrations of this result are possible \citep{Wilks1992,Sethuraman1994}.

\subsubsection{Powered Chinese Restaurant process}
We finally derive an expression for the Powered Chinese Restaurant process from Eq.\ref{eq-DirCatr}. We recall that $N_k = \vert \{C_{-i} \vert i=k\}_{i=1, 2, ..., n-1} \vert$. Taking back the conditional probability for the $n^{th}$ observation to belong to \gls{cluster} $c$ (Eq.\ref{eq-DirCatr}), we have:
\begin{equation}
\label{eq-derivPCRP}
\begin{split}
    p(C_n = &c \vert \vec{C_-}, \vec{\alpha}) = DirCat_r(C_n=c \vert \vec{C_-}, \vec{\alpha})\\
    =& B(\vec{c}+\vec{N^r}+\vec{\alpha}) / B(\vec{N^r}+\vec{\alpha})\\
    =& \Gamma(N_c^r + \alpha_c + 1) \frac{\prod_{k \neq c} \Gamma(N_k^r + \alpha_k)}{\Gamma(1 + \sum_k N_k^r + \alpha_k)} \frac{\Gamma(\sum_k N_k^r + \alpha_k)}{\prod_{k} \Gamma(N_k^r + \alpha_k)}\\
    =& \frac{(N_c^r + \alpha_c)}{\sum_k N_k^r + \alpha_k} \frac{\prod_{k} \Gamma(N_k^r + \alpha_k)}{\Gamma(\sum_k N_k^r + \alpha_k)} \frac{\Gamma(\sum_k N_k^r + \alpha_k)}{\prod_{k} \Gamma(N_k^r + \alpha_k)}\\
    =& \frac{N_c^r + \alpha_c}{\sum_k N_k^r + \alpha_k}
\end{split}
\end{equation}

Every \gls{cluster} with $N_c=0$ (empty \glspl{cluster}) has an identical probability of getting chosen. Besides, the result is identical if either of them gets chosen. Therefore, we can express the probability of choosing any empty \glspl{cluster} as a function of $\alpha=\sum_k \alpha_k$.

Finally, taking the limit $K \rightarrow \infty$ and defining the limit value $\lim_{K \rightarrow \infty} \sum_k^K \alpha_k = \alpha$, we find the Powered Chinese Restaurant Process ($PCRP$):
\begin{equation}
\label{eq-PowCRP}
    PCRP (C_i = c \vert \alpha, C_1, C_2, ..., C_{i-1}) = 
    \begin{cases}
    \frac{N_c^r}{\alpha + \sum_k^K N_k^r} \text{ if c = 1, 2, ..., K}\\
    \frac{\alpha}{\alpha + \sum_k^K N_k^r} \text{ if c = K+1}
    \end{cases}
\end{equation}
\myequations{\ \ \ \ PDP - Powered Dirichlet Process}

The formal derivation of the Powered Chinese Restaurant process in Eq.\ref{eq-PowCRP} and the demonstration of its link to the conditional Dirichlet prior on $\vec{p}$ are the first main contribution of this section. Besides, this demonstration uncovers the link between the prior in Eq.\ref{eq-DirMultPrior} and an exotic formulation of the Multinomial distribution, which has never been considered before.
Most importantly, it highlights that it originates from sampling every new observation from \gls{independent} weighted Multinomial distributions of different parameters $N_r = \sum_k x_k^r$. 
As stated in the introduction, special cases of the process have already been used in some occasions \citep{Qin2003,Jensen2008,Wallach2010UnifP} but never demonstrated. Furthermore, this formulation generalizes the Uniform process when $r\rightarrow 0$ \citep{Wallach2010UnifP}, the Dirichlet process when $r \rightarrow 1$ and the Pitman-Yor process when $r_k(N_k) = (log(1-\beta/N_k) + log(N_k))/log(N_k)$ and $\alpha(K) = \alpha + \beta K$ (see Eq.\ref{eq-PY}, we recall that $e^{\log x} = x$). The present expression explicitly allows for controlling the importance of the ``\gls{rich-get-richer}'' property as well as recovering state-of-the-art processes. 

\begin{figure}
    \centering
    \includegraphics[width = 0.7\textwidth]{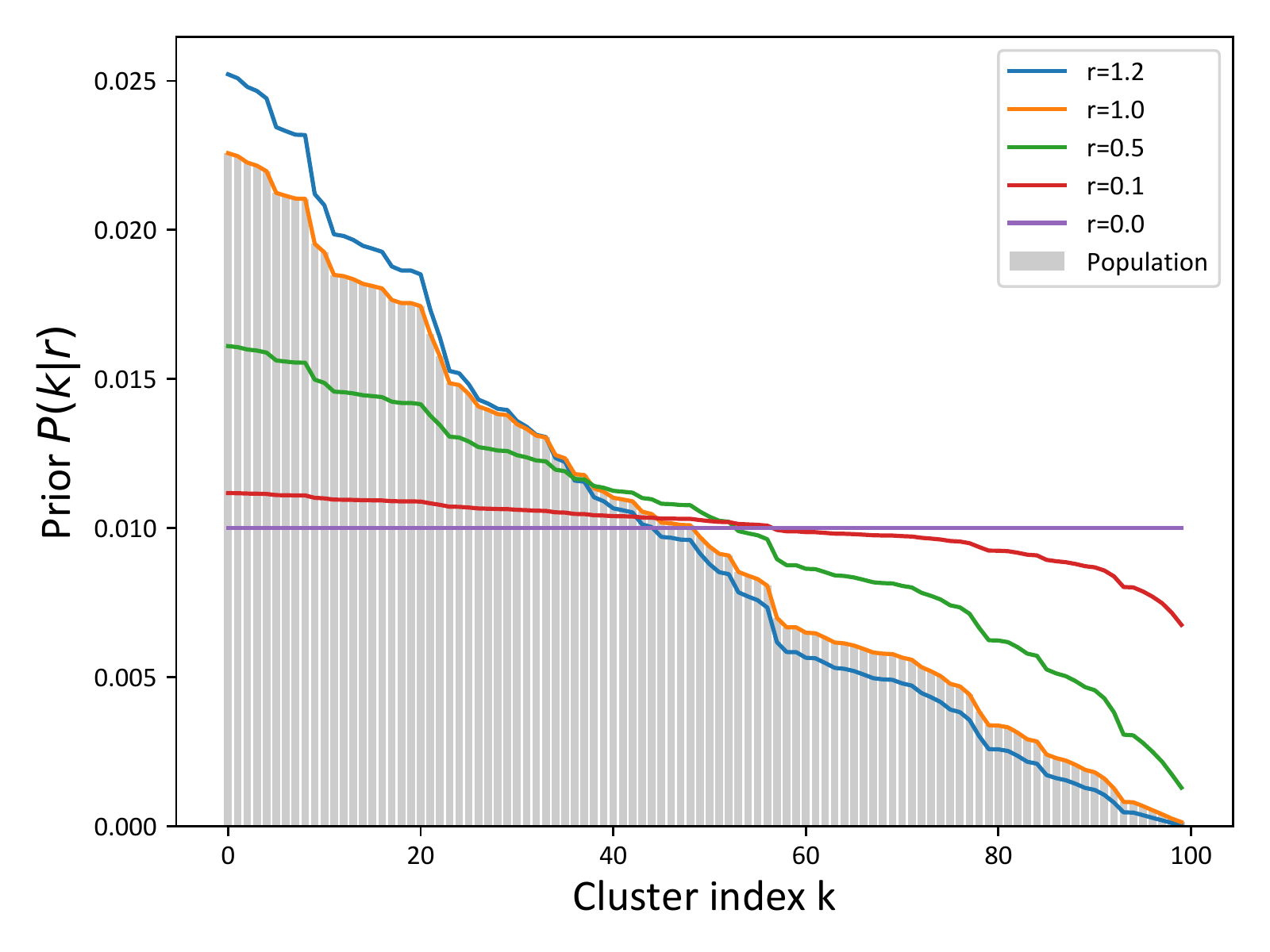}
    \caption[PDP - Effect of $r$ on the PDP]{Effect of $r$ on the Powered Chinese Restaurant process prior probability. The grey bars represent the population in each \gls{cluster} at a given time. The solid lines represent the prior probability of the next observation to belong to each of these. For $r=1$, we recover the Dirichlet Process prior.}
    \label{fig-prior}
\end{figure}
We illustrate the change in prior probability for an existing \gls{cluster} to get chosen induced by the Powered Chinese Restaurant process in Fig.\ref{fig-prior} -- we do not plot the prior probability for a new \gls{cluster} to be created. This figure plots the population of \glspl{cluster} (grey bars) and their associated prior probability of getting chosen. When $r>1$, the most populated \glspl{cluster} are associated with a more significant prior probability than in the standard CRP, whereas the less populated ones have even fewer chances to get chosen; rich-get-more-richer, the prior on population is peakier on large \glspl{cluster}. On the other hand, when $r<1$, most populated \glspl{cluster} have fewer chances to get chosen than in CRP, whereas less populated ones have an increased chance of getting chosen; rich-get-less-richer, the prior on population is flatter across \glspl{cluster} of different sizes. In the limit case $r=0$, the \glspl{cluster}' population does not play any role anymore; rich get-no-richer, the prior is flat over all \glspl{cluster}. Note that if we wanted to represent the Pitman-Yor process prior in this figure, it would correspond to the plot for $r=1$ vertically shifted of $-\beta$ (such as defined Eq.\ref{eq-PY}) leading to an increased probability of creating a new \gls{cluster} of $\beta K$ (not represented in the plot) \citep{PitmanYor1997}. Varying the parameter $\alpha$ of $\Delta \alpha$ plays a similar role as $\beta$ in this situation. It would uniformly shift the prior probability for each existing \gls{cluster} to get chosen by $\frac{\Delta \alpha}{K}$ and increase the probability of creating a new one by $\Delta \alpha$. For Both Pitman-Yor and Dirichlet processes, the linear dependence of each \gls{cluster}'s population does not change.

In Fig.~\ref{fig-prior}, we understand that the Powered Chinese Restaurant process allows for defining priors from \glspl{cluster} population that are not possible when tuning the Chinese Restaurant or Pitman-Yor processes. Introducing non-linearity in the dependence on previous observations allows giving more or less importance to the ``\gls{rich-get-richer}'' property.

\subsection{Properties of the Powered Chinese Restaurant process}
We will now investigate some key properties of the Powered Chinese Restaurant process. We recall that $N_k$ is the population of the \gls{cluster} $k$, and $N = \sum_k N_k$.

\subsubsection{Convergence}
\begin{prop}
\label{th-statio}
For $N \rightarrow \infty$, the Powered Chinese Restaurant process converges towards a stationary distribution. When $r<1$, it converges towards a uniform distribution over all the possible \glspl{cluster}, and when $r>1$, it converges towards a Dirac distribution on a single \gls{cluster}.
\end{prop}
\begin{proof}
We consider a simple situation where only 2 \glspl{cluster} are involved. The generalization to the case where $K$ \glspl{cluster} are involved is straightforward. When the \glspl{cluster}' population is large enough, we make the following Taylor approximation:
\begin{equation}
\label{eq-taylorapprox}
    \begin{split}
        (N_i+1)^r &= N_i^r(1 + \frac{1}{N_i})^r = N_i^r + r N_i^{r-1} + \mathcal{O}(N^{r-2})
    \end{split}
\end{equation}

Since the population of a \gls{cluster} $N_i$ is a non-decreasing function of $N$, we assume that first order Taylor approximation holds when $N \rightarrow \infty$. Given \glspl{cluster} population at the $N^{th}$ observation, we perform a stability analysis of the gap between probabilities $\Delta p(N) = p_1(N)-p_2(N)$. We recall that the probability for \gls{cluster} $i$ to get chosen is $p_i(N) = N_i^r / (\sum_k N_k^r)$ and that either of the \glspl{cluster} is chosen with this probability at the next step (at step $N+1$, $\Delta p(N+1) = p_1(N+1)-p_2(N)$ with probability $p_1(N)$ and $\Delta p(N+1) = p_1(N)-p_2(N+1)$ with probability $p_2(N)$). Explicitly the variation of the gap between probabilities when $N$ grows is written as:
\begin{equation}
\label{eq-varGap}
    \begin{split}
        &\frac{p_1(N) (p_1(N+1) - p_2(N)) +  p_2(N) (p_1(N) - p_2(N+1)) - \Delta p(N)}{\Delta p(N)}\\
        &\stackrel{\text{Eq.\ref{eq-taylorapprox}}}{\approx} \frac{1}{p_1(N) - p_2(N)} \times \left( p_1(N) \frac{N_1^r - N_2^r + r N_1^{r-1}}{N_1^r + N_2^r + r N_1^{r-1}}  + p_2(N) \frac{N_1^r - N_2^r - r N_2^{r-1}}{N_1^r + N_2^r + r N_2^{r-1}} \right)\\
        &\ \ =\frac{2 r N_1^r N_2^r}{(N_1^r+N_2^r+rN_1^{r-1})(N_1^r+N_2^r+rN_2^{r-1})}\left( \frac{N_1^{r-1} - N_2^{r-1}}{N_1^r-N_2^r} \right)
    \end{split}
\end{equation}

We see in Eq.\ref{eq-varGap} that the sign of the variation of the gap between probabilities depend only on the term $\frac{N_1^{r-1} - N_2^{r-1}}{N_1^r-N_2^r}$. We can therefore perform a stability analysis of the Powered Chinese Restaurant process using only this expression.

When $0<r<1$, the sign becomes negative because the following relation holds: ${N_1^{r-1} - N_2^{r-1} < 0 \Leftrightarrow N_1^{r} - N_2^{r} > 0 \ \forall N_1, N_2}$; that makes right hand side of Eq.\ref{eq-varGap} negative. Therefore, adding a new observation statistically reduces the gap between the probabilities of the two \glspl{cluster}. We could forecast this prediction from Eq.\ref{eq-taylorapprox} by seeing that adding a new observation to a large \gls{cluster} increases its probability to get chosen lesser than for a small \gls{cluster} -- rich-get-less-richer. Moreover, we see from Eq.\ref{eq-taylorapprox} that a crowded \gls{cluster} (such as $N_1^r \gg N_2^r$) see its probability evolve as $N^{r-1}$. Asymptotically, the only fixed point of Eq.\ref{eq-varGap} when $N \rightarrow \infty$ is $N_1 \rightarrow N_2$, which implies a uniform distribution. 

On the contrary, when $r>1$ we have the following relation: $N_1^{r-1} - N_2^{r-1} > 0 \Leftrightarrow N_1^{r} - N_2^{r} > 0 \ \forall N_1, N_2$; ; that makes right hand side of Eq.\ref{eq-varGap} positive. Adding a new observation statistically increases the gap between probabilities. From Eq.\ref{eq-taylorapprox}, we see that adding an observation to a large \gls{cluster} increases its probability with its population -- rich-get-more-richer. In this case, Eq.\ref{eq-varGap} has $K+1$ fixed points, with $K$ the number of \glspl{cluster}. The uniform distribution is an unstable fixed point, while $K$ Dirac distributions (each on one \gls{cluster}) are stable fixed points of the system. It means the gap converges to $1$, that is a probability of 1 for one \gls{cluster} and a probability of 0 for the others. 

When $r=1$, the right-hand side of Eq.\ref{eq-varGap} is null. It means the gap remains statistically constant $\forall N_i$, which is a classical result for the regular Dirichlet process. This convergence has already been studied on many occasions \citep{Ferguson1973,Arratia1992}. 

We note that as $r \rightarrow 0$, Eq.\ref{eq-varGap} is not defined anymore. That is because the probability for a \gls{cluster} to be chosen does not depend on its population anymore. In this case, $p_1(N) - p_2(N) \propto N_1^0 - N_2^0 = 0$: the probability for any \gls{cluster} to be chosen is equal, hence the Uniform process -- ``rich-get-no-richer''.

\end{proof}

\subsubsection{Expected number of tables}

\begin{prop}
\label{th-slowVarp}
When $N$ is large, $\sum_k N_k^r$ varies with $N$ as $N^{\frac{r^2+1}{2}}$ when $r<1$, and with $N^{r}$ when $r\geq 1$.
\end{prop}
\begin{proof}
Taking back Eq.\ref{eq-PowCRP}, we are interested in the variation of $p_i = \frac{N_i^r}{\sum_k N_k^r}$ according to $N$ when $N_i^r$ is large:
\begin{equation}
\label{eq-slowVarp}
\begin{split}
p_i(N+1)-p_i(N) \approx \begin{cases}
\frac{rN_i^{r-1} + \mathcal{O}(N^{r-2})}{\sum_k N_k^r} &\text{ if $N_i$ grows}\\
0 &\text{ otherwise}
\end{cases}
\end{split}
\end{equation}

We see in Eq.\ref{eq-slowVarp} that for $r<1$, the larger $N_i$ the slower the variation of $p_i$. It means that for large $N_i^r$, we can write $N_i \propto N p_i$, with $p_i$ a constant of $N$. Since $N_i$ is either way a non-decreasing function of $N$, we reformulate the constraint $N_i^r$ large in $N^r$ large.

For $r>1$, the probability $p_i$ varies greatly with $N$ and quickly converges to 1 for large $N$ (see Proposition \ref{th-statio}), and so $N_i \approx N$ for \gls{cluster} $i$ and $N_{j \neq i} \ll N_i \ \forall j$.

Since the sum $\sum_k N_k^r$ essentially varies according to large $N_k$, we can approximate $\sum_k N_k^r \approx N^r \sum_k p_k^r$ for large $N^r$.

Besides, we showed in Proposition \ref{th-statio} that for large $N$ the process converges towards a uniform distribution for $r<1$ and towards a Dirac distribution when $r>1$. Therefore, we can express $\sum_k^K p_k^r$ as:
\begin{equation}
    \sum_k^K p_k^r \stackrel{N \gg 1}{\approx}
    \begin{cases}
    K^{1-r} &\text{for $r<1$}\\
    1 &\text{for $r \geq 1$}
    \end{cases}
\end{equation}

Based on the demonstration of Eq.4 in \citep{Wallach2010UnifP}, we suppose that $K$ evolves with $N$ as $N^{\frac{1-r}{2}}$ when $r<1$. We verify that this assumption holds in the Experiment section.

Therefore, we can write:
\begin{equation}
    \sum_k N_k^r \approx N^r \sum_k^K p_k^r \approx \begin{cases}
    N^{r} \left( N^{\frac{1-r}{2}} \right)^{1-r} = N^{\frac{1+r^2}{2}} &\text{for $r<1$}\\
    N^{r} &\text{for $r \geq 1$}\\
    \end{cases}
\end{equation}

\end{proof}

\begin{prop}
\label{th-expK}
The expected number of tables of the Powered Chinese Restaurant process evolves with $N \gg 1$ as $H_{\frac{r^2+1}{2}}(N)$ for $r<1$ and as $H_{r}(N)$ when $r \geq 1$, where $H_m(n)$ is the generalized harmonic number.
\end{prop}
\begin{proof}
In general, the expected number of \glspl{cluster} at the $N^{th}$ step can be written as:
\begin{equation}
\label{eq-expK}
    \mathbb{E}(K \vert N, r) = \sum_1^N \frac{\alpha}{\sum_k N_k^r + \alpha} \stackrel{N^r \gg 1}{\propto} \sum_1^N \frac{1}{\sum_k N_k^r}
\end{equation}

We showed in Proposition \ref{th-slowVarp} that we can rewrite $\sum_k N_k^r \propto N^{\frac{r^2 + 1}{2}}$ when $r<1$ and $\sum_k N_k^r \propto N^r$ when $r\geq 1$.
Injecting this result in Eq.\ref{eq-expK} for $r$, we get:
\begin{equation}
    \label{eq-expK-fin}
    \mathbb{E}(K \vert N, r) \stackrel{N^r \gg 1}{\propto}
    \begin{cases} \sum_1^N\frac{1}{N^{\frac{r^2 + 1}{2}}} = H_{\frac{r^2+1}{2}}(N)\\
    \sum_1^N\frac{1}{N^r} = H_{r}(N)
    \end{cases}
\end{equation}

\end{proof}

For $r=1$, $\mathbb{E}(K \vert N, r=1) \propto H_1(N) \approx \gamma + \log(N)$ where $\gamma$ is the Euler–Mascheroni constant, which is a classical result for the regular Dirichlet process. 

When $r>1$ and $N \rightarrow \infty$, the term $H_{\frac{r^2+1}{2}}(N)$ converges towards a finite value and the sum $\sum_k p_k^r$ goes to 1 (see Proposition \ref{th-statio}). By definition $\mathbb{E}(K \vert N, r>1) \stackrel{N \rightarrow \infty}{\propto} \zeta(\frac{r^2+1}{2})$, where $\zeta$ is the Riemann Zeta function.

When $r < 1$, we can approximate the harmonic number in a continuous setting. We rewrite Eq.\ref{eq-expK-fin} as:
\begin{equation}
    \begin{split}
        \mathbb{E}(K \vert N, r) \stackrel{N^r \gg 1}{\propto} \sum_{n=1}^N\frac{1}{n^{\frac{r^2 + 1}{2}}}
        &\stackrel{N^r \gg 1}{\approx} \int_1^N n^{-\frac{r^2 + 1}{2}} dn\\
        &\ \ \ \ \ \ = \frac{2}{1 - r^2}(N^{\frac{1-r^2}{2}}-1) 
    \end{split}
\end{equation}
It can be shown that $\frac{N^{1-x}-1}{1-x} = H_x(N) + \mathcal{O}(\frac{1}{N^x})$. Therefore, the Powered Chinese Restaurant process exhibits a power-law behaviour similar to the Pitman-Yor process Eq.\ref{eq-PY} for $r = \sqrt{1 - 2\beta}$ for $0<r<1$. For values of $r>1 \Leftrightarrow \beta<0$, the equivalent Pitman-Yor process is not defined unlike the Powered Chinese Restaurant process. Note that there is \textit{a priori} no reason for $r$ to be constrained in the domain of real number. Complex analysis of the process might be an interesting lead for future works.

\begin{figure}[t!]
    \centering
    \includegraphics[width = \textwidth]{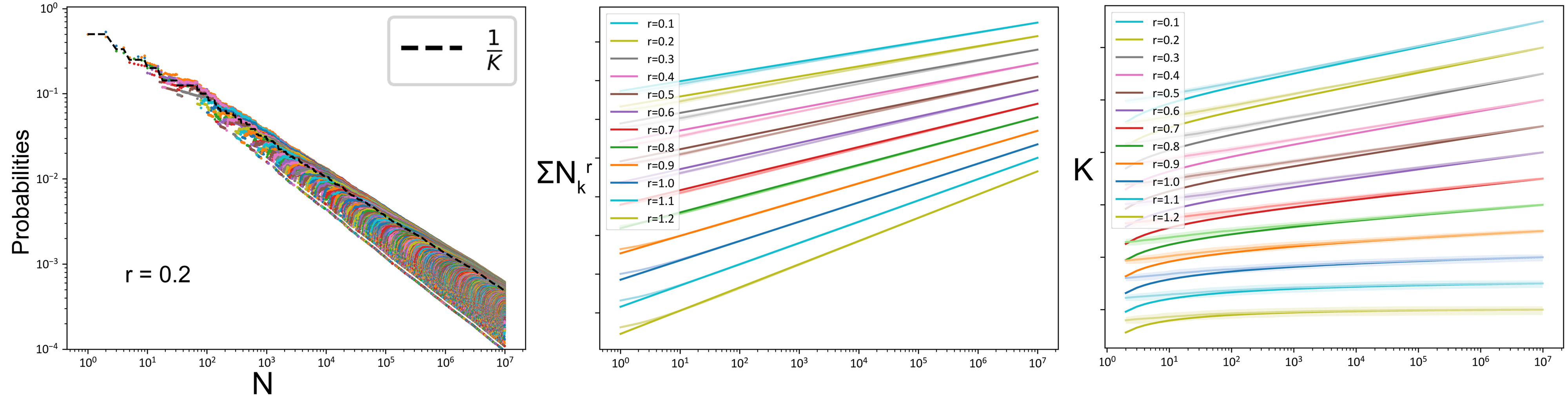}
    \caption[PDP - Numerical validation of theorems]{Numerical validation of Propositions \ref{th-statio} (left), \ref{th-slowVarp} (middle), \ref{th-expK} (right). In the first plot, $K$ is the number of non-empty \glspl{cluster}. In the second and third plots, the theoretical results are the solid lines, and the associated numerical results are the transparent lines of same colour. Except for small $N$, the difference between theory and experiments is almost indistinguishable.}
    \label{fig-validTh}
\end{figure}

\subsection{Experiments}
\label{Experiments}

\subsubsection{Numerical validation of propositions}
First of all, we present numerical confirmations of propositions stated above (Propositions \ref{th-statio}, \ref{th-slowVarp}, \ref{th-expK}) by simulating 100 \gls{independent} Powered Chinese Restaurant processes with parameter $\alpha=1$ for various values of $r$. Note that in all the theorems we verify here, $\alpha$ acts as a scaling factor without modifying the shape of the results discussed here. We present the results of numerical simulations in Fig.\ref{fig-validTh}.

In the \textbf{left} part, we plot the evolution of the probability for each \gls{cluster} to be chosen as $N$ grows for $r=0.2$ for one run. We see that the probabilities do not remain constant but instead diminish as the number of \glspl{cluster} grows. The figure suggests they all converge to a common value (a uniform probability) as shown in Proposition \ref{th-statio}. The black line shows the probability of a uniform distribution. We chose not to show the results for $r>1$; in this case, one probability goes to 1 as the other fades to 0 as $N$ grows, as expected.

In the \textbf{middle} part of the figure, we plot the expression for $\sum_k N_k^r$ derived in Proposition \ref{th-slowVarp} (solid lines) versus the value of the sum from experimental results (transparent lines), averaged over 100 runs. Note that plots are in a log-log scale and that curves have been shifted vertically for visualization purposes. As assumed in Proposition \ref{th-slowVarp}, the approximation holds for all values of $r$.

Finally, in the \textbf{right} picture, we plot the evolution of the number of \glspl{cluster} $K$ versus $N$ according to Proposition \ref{th-expK} (solid lines) and experiments (transparent lines). The error bars correspond to the standard deviation over the 100 runs. We see that the expression derived in Proposition \ref{th-expK} accounts well for the evolution of the number of \glspl{cluster}. Note that plots are in a log-log scale and that curves have been shifted vertically for visualization purposes. We must point out that there is a constant shift from experiments to the theory that does not appear on the plot (because of the rescaling). This shift comes from the approximation of large $N^r$ which is not valid at the beginning of the process. However, it does not play any role in the evolution of $K$ as $N$ grows large enough.

\subsubsection{Use case: infinite Gaussian mixture model}
We now illustrate the usefulness of a prior that alleviates the ``\gls{rich-get-richer}'' property on several synthetic datasets and on a real-world application. We choose to consider as an illustration its use as a prior in the infinite Gaussian mixture model.
We choose this application to ease visual understanding of the implications of the PCRP, but the argument holds for other models using DP priors as well (text modelling, gene expression clustering, etc.). 

\begin{figure*}
    \centering
    \includegraphics[width = \textwidth]{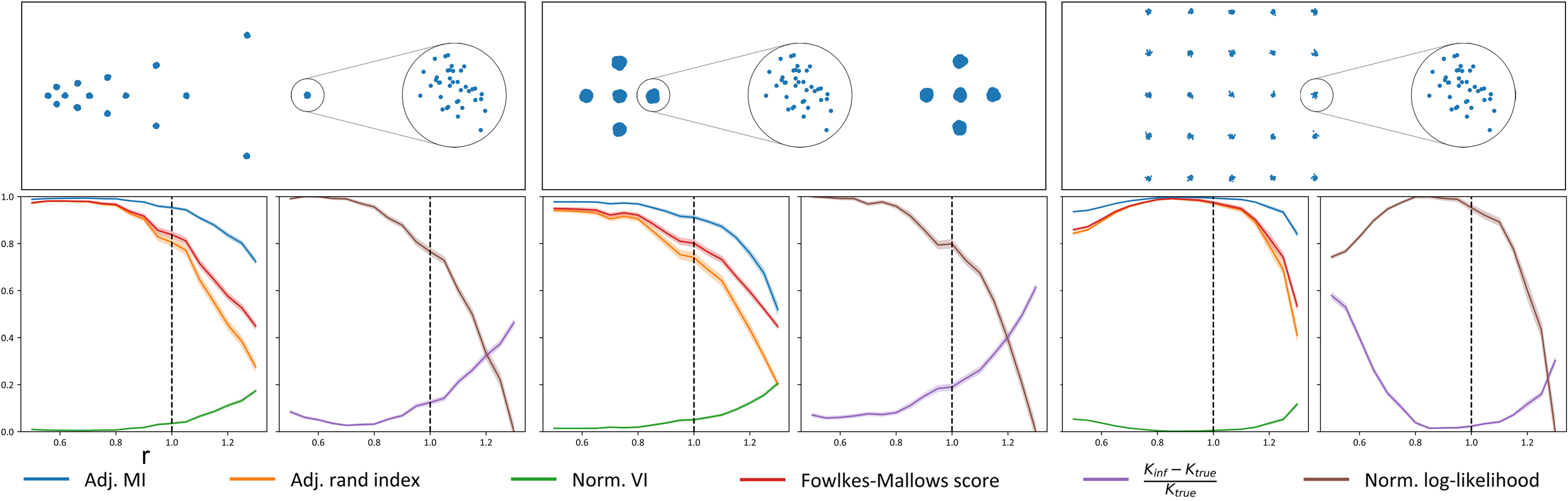}
    \caption[PDP - Experimental results on synthetic data]{Application on synthetic data. (Top) Original datasets used for the experiments (Density, Diamond, and Grid). (Bottom) Results for various values of $r$; the x and y axes are all the same. The dashed line indicates the regular DP prior as $r=1$. The error corresponds to the standard error of the mean over all runs.}
    \label{fig-Results}
\end{figure*}

We consider a classical infinite Gaussian mixture model coupled with a Powered Dirichlet process prior. We fit the data using a standard collapsed Gibbs sampling algorithm for IGMM \citep{Rasmussen1999,Wallach2010UnifP,Yin2014}, with a Normal Inverse Wishart prior on the Gaussians' parameters. The input data is shuffled at each iteration to reduce the ordering bias from the dataset. Note that we cannot completely get rid of the bias because the Powered Dirichlet Process is not exchangeable for all $r$. The problem has been addressed on numerous occasions (Uniform process \citep{Wallach2010UnifP}, distance-dependent CRP \citep{Blei2010DDCRP,Ghosh2014}, spectral CRP \citep{Socher2011}) and shown to induce negligible variations of results in the case of Gibbs sampling. We stop the sampler once the likelihood of the model reaches stability
; we repeat this procedure 100 times for each value of $r$. 
Finally, the parameter $\alpha$ is set to 1 in all experiments (see Section \ref{motivation}). 

Note that we choose not to compare to other types of clustering algorithms. In this section, we demonstrate the power of alternative forms of the Dirichlet process. The argument on a simple model (here a regular DP combined with IGMM) extends to other priors built on Dirichlet processes (Hierarchical and Nested Dirichlet processes). Besides, comparison of DP-based priors to other clustering methods (\acrshort{KNN}, DBScan, Spectral clustering, etc.) has already been done numerous times and is out of this section's scope.

\paragraph{Synthetic data}
Synthetic datasets are represented in Fig.\ref{fig-Results}-top and comprise $N$=1000 observations each. They have been generated by sampling from 2D Gaussian distributions whose relative parameters are explicit from Fig.\ref{fig-Results}. 

\input{Tables/Chapter_4/table-res-PDP}

We present the results on synthetic data in Fig.\ref{fig-Results}-bottom and in Table\ref{tabMetrics}. We consider standard metrics in clustering evaluation with a non-fixed number of \glspl{cluster}: mutual information score and rand index both adjusted for chance (\textbf{Adj.MI} and \textbf{Adj.RI}), normalized variation of information (\textbf{Norm.VI}, lower is better), Fowlkes-Mallow score, marginal likelihood (normalized for visualization) and absolute relative variation of the inferred number of \glspl{cluster} according to the number used to generate the data ($\mathbf{\frac{K_{inf}-K_{true}}{K_{true}}}$, lower is better). Note that we purposely chose stereotypical cases to illustrate the argument better. The Density dataset on the \textbf{left} of Fig.\ref{fig-Results} is informative about the change induced by $r$. Here, \glspl{cluster} are distributed at various scales in the dataset; we see that the lower the value of $r$, the better the results. Indeed, when $r$ is small, the model can distinguish \glspl{cluster} in the dense area better, whereas when $r$ is closer to 1, the \glspl{cluster} in the dense area are put together in a larger \gls{cluster}. The same happens with the Diamond dataset in the \textbf{middle} of Fig.\ref{fig-Results}, where \glspl{cluster} are distributed according to two different scales. Finally, on the Grid dataset on the \textbf{right} part of Fig.\ref{fig-Results}, we see an optimum $r$ exists to distinguish the \glspl{cluster} distributed on a grid; it makes sense since only one scale in \glspl{cluster} distribution is involved in this dataset. In Table.~\ref{tabMetrics}, we explicitly report the values of the PDP optimal $r$ and compare them to the values yielded by the same model using either a Dirichlet Process (\textbf{\acrshort{DP}}) prior or a Uniform Process (\textbf{\acrshort{UP}}) prior.

\begin{figure*}
    \centering
    \includegraphics[width = \textwidth]{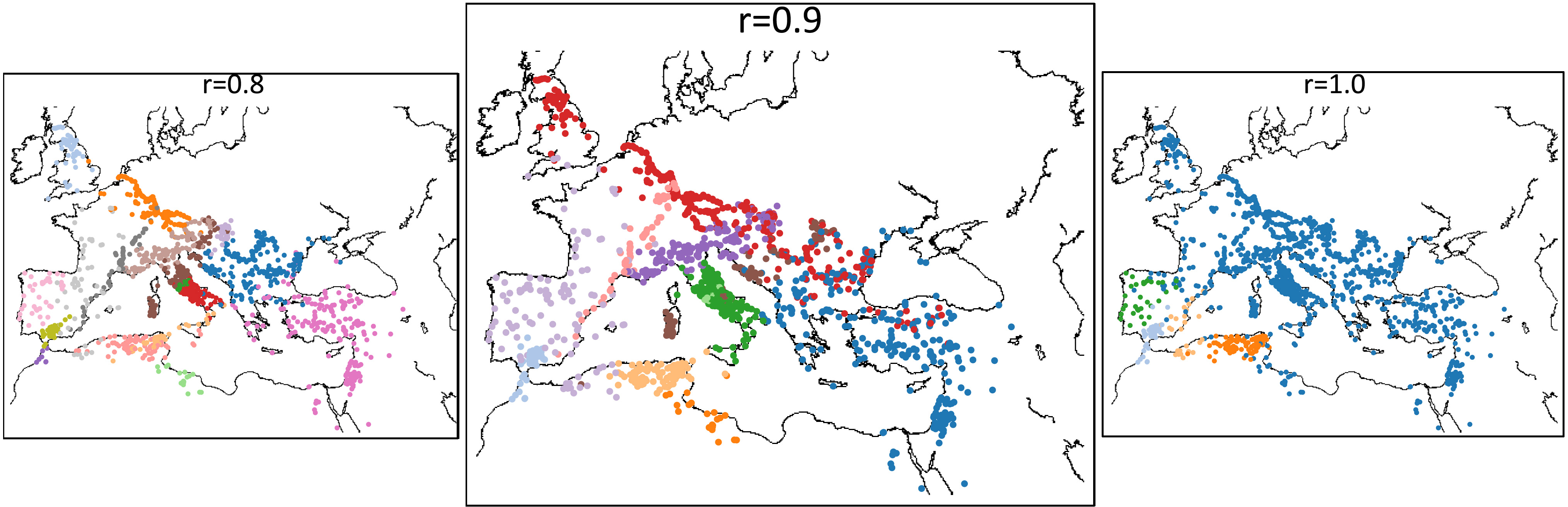}
    \caption[PDP - Application to spatial clustering on geolocated Latin inscriptions]{Application to spatial clustering on geolocated data for $r=0.8$ (left), $r=1$ (right) and $r=0.9$ (middle). We see that the model using a Powered Dirichlet process prior for $r=0.9$ and $r=0.8$ describes the data better than the same model using a Dirichlet process prior ($r=1$).}
    \label{fig-Antoni}
\end{figure*}

\paragraph{Real-world data}
In Table~\ref{tabMetrics}, we report the results for the 20Newsgroup (20-NG) real-world dataset, which is a collection of 18,000 user posts published on Usenet, organized in 20 Newsgroup (which are our target thematic \glspl{cluster}). As a model, we consider a modified version of LDA \citep{Blei2003LDA} that uses a PDP prior instead of DP in the words sampling step. Note that because the number of \glspl{cluster} must be provided to LDA, we do not compute $\mathbf{\frac{K_{inf}-K_{true}}{K_{true}}}$. 
We also run an additional experiment on three additional well-known datasets: Iris (4 attributes, 3 classes), Wines (13 attributes, 3 classes) and Cancer (30 attributes, 2 classes). We see PDP yields improved performances on every dataset.

We now illustrate the interest of using an alternate form of prior for the Infinite Gaussian Mixture model on real-world data. We consider a dataset of 4.300 roman sepulchral inscriptions comprising the substring ``Antoni'' that have been dated between 150AC and 200AC and assigned with map coordinates. The dates correspond to the reign of Antoninus Pius over the Roman empire. The dataset is available on Clauss-Slaby repository \citep{ClaussSlabyDataset}. It was common to give children or slaves the name of the emperor; the dataset gives a global idea of the main areas of the roman empire at that time \citep{Hanson2017}. The task is to discover spatial \glspl{cluster} of individuals named after the emperor. We expect to find geographical \glspl{cluster} around: Italy, Egypt, Gauls, Judea, and all along the \textit{limes} (borders of the roman empire, which concentrate lots of sepulchral inscriptions for war-related reasons) \citep{Hanson2016}. We present the results in Fig.\ref{fig-Antoni}.

We see that when $r=1$, the classical DP prior is not fit for describing this dataset, as it misses most of the \glspl{cluster}. On the other hand, when $r=0.9$, the infinite Gaussian mixture model retrieves the expected \glspl{cluster}. It also makes some \glspl{cluster} that were not expected, such as the north Italian \gls{cluster} or the long \gls{cluster} going through Spain and France that corresponds to the layout of the Roman roads (via Augusta and via Agrippa; it was common to bury the dead on roads edges). Finally, when $r=0.8$, we get even more detail: some of the main \glspl{cluster} are broken into smaller ones (Italy breaks into Rome, North Italy, and South Italy; Britain becomes an \gls{independent} \gls{cluster}, etc.). In this case, changing $r$ controls the level of details of the clustering. Note that we do not compute metrics for this experiment in the absence of ``ground-truth'' clustering; there is no such thing as the right clustering in this case. Applied to spatial data, the PDP prior allows to control the clustering granularity. We see how different results can be according to the extent the model relies on the ``\gls{rich-get-richer}'' prior and how its control is needed to make modelling relevant to a given situation.

\subsection{Conclusion}
In this section, we discussed the necessity of controlling the ``\gls{rich-get-richer}'' property that arises from the classical Chinese Restaurant Process usual formulation. We showed cases where this modelling hypothesis must be alleviated or strengthened to describe data more accurately. 
To this end, we derive the Powered Chinese Restaurant Process from a powered version of the Dirichlet-Multinomial distribution. This formulation allows reducing the expected number of \glspl{cluster}, which is not possible in the standard Pitman-Yor processes, while generalizing the standard Dirichlet process and the Uniform process. 

The main feature of this formulation is that it allows for direct control of the ``\gls{rich-get-richer}'' priors' importance. We derive elementary results on convergence and the expected number of \glspl{cluster} of the new process. 

Finally, we show that it yields better results on synthetic data when coupled to a standard Gaussian mixture model and illustrates a possible use case with real-world data. For future works, it might be interesting to investigate cases where $r$ takes non-positive values (which might lead to a ``poor-get-richer'' kind of process), and to develop a procedure to infer it automatically for specific problems (by minimizing a dispersion criterion for instance).

The regular Chinese Restaurant process has been used for decades as a prior in many real-world applications. However, alternate forms for this prior have been little explored. We are convinced that controlling the impact of the ``\gls{rich-get-richer}'' hypothesis will bring significant changes in many state-of-the-art models. 

We now propose to illustrate this claim and to further answer our task at hand. To this end, we use PDP to redefine the Dirichlet-Hawkes Process in terms of the Powered Dirichlet-Hawkes Process (\acrshort{PDHP}) and investigate the new possibilities brought by this reformulation.

\section{Powered Dirichlet-Hawkes Process -- Modelling self \gls{interacting} \glspl{cluster}}
\label{PDHP}
\textit{This work has been published, see \citep{Poux2021PDHP}}

\subsection{Introduction}
\subsubsection{PDHP as an answer to DHP's limits}
In this section, we develop the Powered Dirichlet-Hawkes process (\acrshort{PDHP}) as an answer to the limits of DHP raised in Section~\ref{DHP-sota}. It makes use of the redefinition of Dirichlet Processes detailed in Section~\ref{PDP}. 

We highlighted earlier that DHP has several limits. It fails to model data correctly when textual information is scarce (e.g., short texts such as tweets, or when topics’ vocabulary overlap) \citep{Yin2018ShortTextDHP} and when temporal information is scarce (overlapping dynamics, few observations). In addition, the usual assumption is that publication dynamics and textual \gls{content} are perfectly correlated, which cannot be true in real-world applications. For instance, two identical textual \glspl{content} will trigger different publication dynamics depending on who published them, or when they have been published.

\subsubsection{Contributions}

We overcome all these limitations by developing the Powered Dirichlet-Hawkes Process (\acrshort{PDHP}), which yields better results than DHP on every dataset considered (up to +0.3 \acrshort{NMI}). In particular, when textual or temporal information is scarce, PDHP allows to control the importance given to one or the other and thus retrieves better \glspl{cluster}. PDHP is the DHP equivalent of controlling the importance of the ``\gls{rich-get-richer}'' hypothesis in \acrshort{DP} using \acrshort{PDP}, only using temporal information as a prior. PDHP also allows to distinguish \textit{textual \glspl{cluster}} from \textit{temporal \glspl{cluster}} (documents that follow the same dynamic independently from their \gls{content}). 

Our contributions are listed below:
\begin{itemize}
    \item We highlight and explain the limitations of the DHP prior: it does not handle weakly informative temporal and textual information and it is not designed to consider different dynamics between text and time.
    \item We derive the Powered Dirichlet-Hawkes process (\acrshort{PDHP}) as a new prior in Bayesian non-parametric for the temporal clustering of a stream of textual documents, which is a generalization of the Dirichlet-Hawkes process (DHP) and of the Uniform process (\acrshort{UP}).
    \item We show how the PDHP prior performs better than DHP and UP priors through a thorough evaluation and comparison on several synthetic datasets and real-world datasets from Reddit.
    \item We show that PDHP prior allows choosing whether \glspl{cluster} are based more on textual or temporal information, or a mixture of both. We can favour their generation more or less according to documents' textual \gls{content} or their publication dynamics.
    \item We perform a detailed analysis of PDHP's results on real-world datasets from Reddit. We illustrate our approach with examples of real-world topics uncovered by our method. We systematically analyse the influence of the hyperparameter $r$ introduced in PDHP. Our method allows us to recover more or less bursty events from data streams depending on $r$. %
\end{itemize}

\subsection{Model and algorithm}
\subsubsection{Dirichlet prior and alternatives}
We briefly recall the definition of a Dirichlet prior (see Section~\ref{PDP} for an extensive discussion on DP). A Dirichlet prior for clustering implements the assumption that the more a \gls{cluster} is populated, the more chances a new observation belongs to it (``\gls{rich-get-richer}'' property).
Besides, there is still a chance that a new observation gets assigned to a newly created \gls{cluster}. It is often expressed using a metaphor, the Chinese Restaurant Process (\acrshort{CRP}), and it goes as follows: if an $i^{th}$ client arrives in a Chinese restaurant, they will sit at one of the $K$ already occupied tables with a probability proportional to the number of persons already sat at this table. They can also sit alone at a new table $K+1$ with a probability inversely proportional to the total number of clients in the restaurant. When their choice is made, the next client arrives, and the process is repeated.
Let $c$ be the \gls{cluster} chosen by the $i^{th}$ customer, $\vec{C^-}$ the table assignment of previous customers up to $i-1$, $N_c$ the population of table $c$, $C$ the number of already occupied tables and $\alpha_0 \in \mathbb{R}^+$ the concentration parameter. The process can be written formally as:
\begin{equation}
\label{eq-CRP-PDHP}
    \text{CRP} (C_i = c \vert \vec{C^-}, \alpha_0) = 
    \begin{cases}
    \frac{N_c}{\alpha_0 + N} \text{ if c = 1, 2, ..., C}\\
    \frac{\alpha_0}{\alpha_0 + N} \text{ if c = C+1}
    \end{cases}
\end{equation}

The Uniform process \citep{Wallach2010UnifP} has been proposed as an alternative to the DP prior. In this context, a new customer entering the restaurant has an identical chance to sit at either of the occupied tables, and a chance to sit at an empty table inversely proportional to the number of occupied tables. Formally:
\begin{equation}
\label{eq-UP-PDHP}
    \text{U-CRP} (C_i = c \vert \vec{C^-}, \alpha_0) = 
    \begin{cases}
    \frac{1}{\alpha_0 + C} \text{ if c = 1, 2, ..., C}\\
    \frac{\alpha_0}{\alpha_0 + C} \text{ if c = C+1}
    \end{cases}
\end{equation}

Finally, the Powered Dirichlet process that we detailed Section~\ref{PDP} generalizes the two processes above, stating that the probability for a new client to sit at a new table depends arbitrarily on the number of customers already sat at this table:
\begin{equation}
\label{eq-PCRP}
    \text{P-CRP} (C_i = c \vert r, \vec{C^-}, \alpha_0) = 
    \begin{cases}
    \frac{N_c^r}{\alpha_0 + \sum_{c'} N_{c'}^r} \text{ if c = 1, 2, ..., C}\\
    \frac{\alpha_0}{\alpha_0 + \sum_{c'} N_{c'}^r} \text{ if c = C+1}
    \end{cases}
\end{equation}
where $r \in \mathbb{R}^+$ is a hyper-parameter. Varying $r$ allows to give more or less importance to the ``\gls{rich-get-richer}'' hypothesis of DP. Note that $P-CRP (r=0, \vec{C^-}, \alpha_0) = U-CRP(\vec{C^-}, \alpha_0)$ and that $P-CRP(r=1, \vec{C^-}, \alpha_0) = CRP(\vec{C^-}, \alpha_0)$. We will use this more general form in the rest of this section and make $r$ vary to compare those priors in the experimental section.

\subsubsection{Hawkes processes}
The Hawkes process has already been introduced and discussed in Section~\ref{sota-DHP-Hawkes}. We recall that Hawkes processes are defined as self-stimulating temporal point processes. They are fully characterized by their intensity function $\lambda (t)$, which is related to the probability $P(t_{events} \in [t;t+\Delta t])$ of an event happening between $t$ and $t + \Delta t$ by $\lambda(t) = \lim_{\Delta t \rightarrow 0} \frac{P(t_{events} \in [t;t+\Delta t])}{\Delta t}$. 

Same as in \citep{Du2015DHP}, the intensity of the Hawkes process associated with \gls{cluster} $c$ is defined as:
\begin{equation}
\label{eq-PDHP-HawkesClus}
    \lambda_c(t \vert \mathcal{H}_{<t, c}) = \sum_{\mathcal{H}_{<t, c}} \vec{\alpha_c}^T \cdot \vec{\kappa}(t_{i,c})
\end{equation}
We recall that $\vec{\alpha_c}$ is a vector of weights and $\vec{\kappa}(t)$ is a vector of user-defined kernel functions with the same dimension as $\vec{\alpha}$. The kernel functions are defined at the beginning of the algorithm and are not modified afterwards. We will want to infer the weights vector $\vec{\alpha}$ to determine which entries of the kernel vector are the most relevant for a given situation.

\subsubsection{Powered Dirichlet-Hawkes process}
In \citep{Du2015DHP}, the authors create DHP by substituting the counts $N_k$ in the DP with the intensities of the Hawkes processes associated with each \gls{cluster}. Instead, we simply substitute the counts in PDP with the intensities of the Hawkes processes associated with each \gls{cluster}, forming the Powered Dirichlet-Hawkes Process (\acrshort{PDHP}):
\begin{equation}
    \label{eq-PDHP-prior}
    P(C_i = c\vert t_i, r, \lambda_0, \mathcal{H}_{<t_i,c}) = 
    \begin{cases}
    \frac{\lambda_c^r(t_i)}{\lambda_0 + \sum_{c'} \lambda_{c'}^r(t_i)} \text{ if c$\leq$C}\\
    \frac{\lambda_0}{\lambda_0 + \sum_{c'} \lambda_{c'}^r(t_i)} \text{ if c=C+1}
    \end{cases}
\end{equation}
\myequations{\ \ \ \ PDHP - Powered Dirichlet-Hawkes Process}
where $t_i$ is the arrival time of document $i$. We reformulate the Dirichlet-Hawkes process to allow nonlinear dependence ($r$) on the non-integer counts ($\vec{\lambda}$). We illustrate how $r$ influences the prior probability of \gls{cluster} selection in Fig.~\ref{fig-illustr-PDHP}.

\begin{figure}
    \centering
    \includegraphics[width=0.8\textwidth]{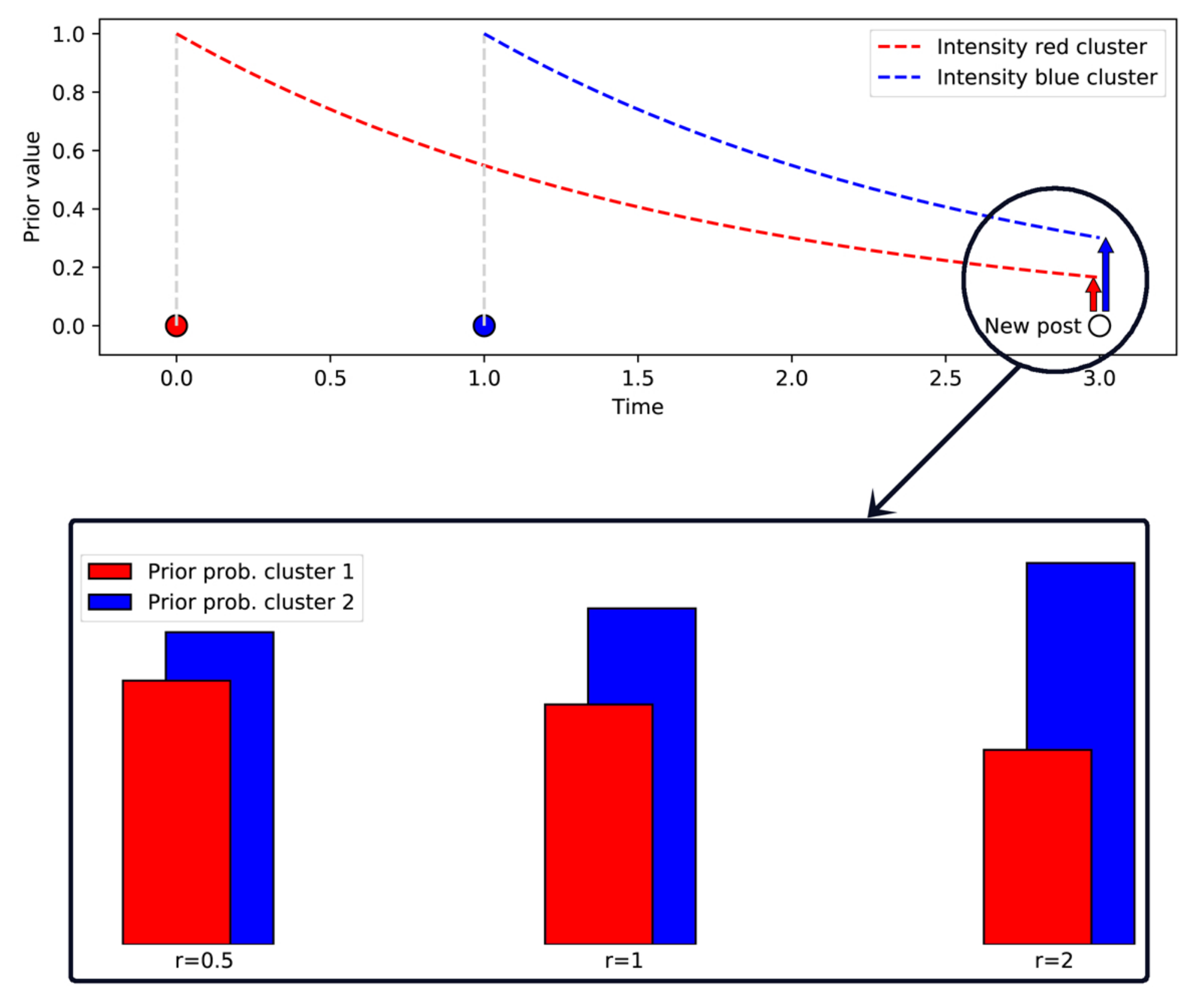}
    \caption[PDHP - Illustration of the PDHP]{\textbf{Temporal prior on a new observation \gls{cluster}} --- \textbf{(Top)} Prior probability for two \glspl{cluster} to get chosen at all times according to DHP (dotted lines) \textbf{(Bottom)} Same prior probability for each \gls{cluster} as a function of $r$ in the PDHP at a given time.}
    \label{fig-illustr-PDHP}
\end{figure}

\subsubsection{Textual modelling}
To compare to DHP on solid ground, we consider the same textual model the authors used in the original paper, the Dirichlet-Multinomial model. This model is detailed in Section~\ref{sota-DHP-text} and it considers the textual \gls{content} as a bag of words.

The likelihood of the $i^{th}$ document belonging to \gls{cluster} $c$ reads:
\begin{equation}
    \begin{split}
        \label{eq-PDHP-text}
        \mathcal{L}(C_i=c \vert N_{<i,c}, n_i, \theta_0) &= P(n_i \vert C_i=c, N_{<i,c}, \theta_0)\\ 
        &=\frac{\Gamma(N_c+\theta_0)}{\Gamma(N_c+n_i+\theta_0)} \prod_v \frac{\Gamma(N_{c,v} + n_{i,v} + \theta_{0,v})}{\Gamma(N_{c,v}+\theta_0)}
    \end{split}
\end{equation}
where $N_c$ is the total number of words in \gls{cluster} $c$ from observations previous to $i$, $n_i$ is the total number of words in document $i$, $N_{c,v}$ the count of word $v$ in \gls{cluster} $c$, $n_{i,v}$ the count of word $v$ in document $i$ and $\theta_0 = \sum_v \theta_{0,v}$.

\subsubsection{Posterior distribution}
The resulting posterior distribution of the $i^{th}$ document over \glspl{cluster} is calculated using Bayes theorem. It is proportional to the product of the textual likelihood (Eq.\ref{eq-PDHP-text}) and the temporal Powered Dirichlet-Hawkes prior (Eq.\ref{eq-PDHP-prior}):
\begin{equation}
\label{eq-likModelTot}
\begin{split}
    &P(C_i = c \vert r, n_i, t_i, N_c, \mathcal{H}_{<t, c})\\
    \propto &\underbrace{P(n_i \vert C_i=c, N_{<i,c}, \theta_0)}_{\text{Textual likelihood}} \underbrace{P(C_i = c\vert t_i, r, \lambda_0, \mathcal{H}_{<t_i,c})}_{\text{Temporal prior}} \\
    = &\frac{\Gamma(N_c+\theta_0)}{\Gamma(N_c+n_i+\theta_0)} \prod_v \frac{\Gamma(N_{c,v} + n_{i,v} + \theta_{0,v})}{\Gamma(N_{c,v}+\theta_0)}\\
    &\times \begin{cases}
    \frac{\lambda_c^r(t_i)}{\lambda_0 + \sum_{c'} \lambda_{c'}^r(t_i)} \text{ if c = 1, ..., C}\\
    \frac{\lambda_0}{\lambda_0 + \sum_{c'} \lambda_{c'}^r(t_i)} \text{ if c = C+1}
    \end{cases}
\end{split}
\end{equation}
We recall that $\lambda_c(t)$ is defined Eq.~\ref{eq-PDHP-HawkesClus}. The textual likelihood to open a new \gls{cluster} $C+1$ is computed by setting $N_{C+1,v}=0$ -- because it is empty before being opened.

\subsubsection{Algorithm for parameters inference}
\label{PDHP-SMC}
We use a similar algorithm to the one in \citep{Du2015DHP}. Briefly, the algorithm is a sequential Monte-Carlo (\acrshort{SMC}) that takes one document at a time in their order of arrival. The algorithm starts with a number $N_{part}$ of particles whose weights are $\omega_{p} = \frac{1}{N_{part}}$, each of which will keep track of a hypothesis on documents \glspl{cluster}. After a few iterations, particles that contained unlikely allocation hypotheses are discarded and replaced by more likely ones. The likeliness of a hypothesis is encoded in the weights of each particle $\omega_p$. Such genetic algorithms are favoured for optimizing DHP-based models. These models are not convex, and genetic algorithms allow to tackle the problem sequentially.

\begin{figure}
    \centering
    \includegraphics[width=\textwidth]{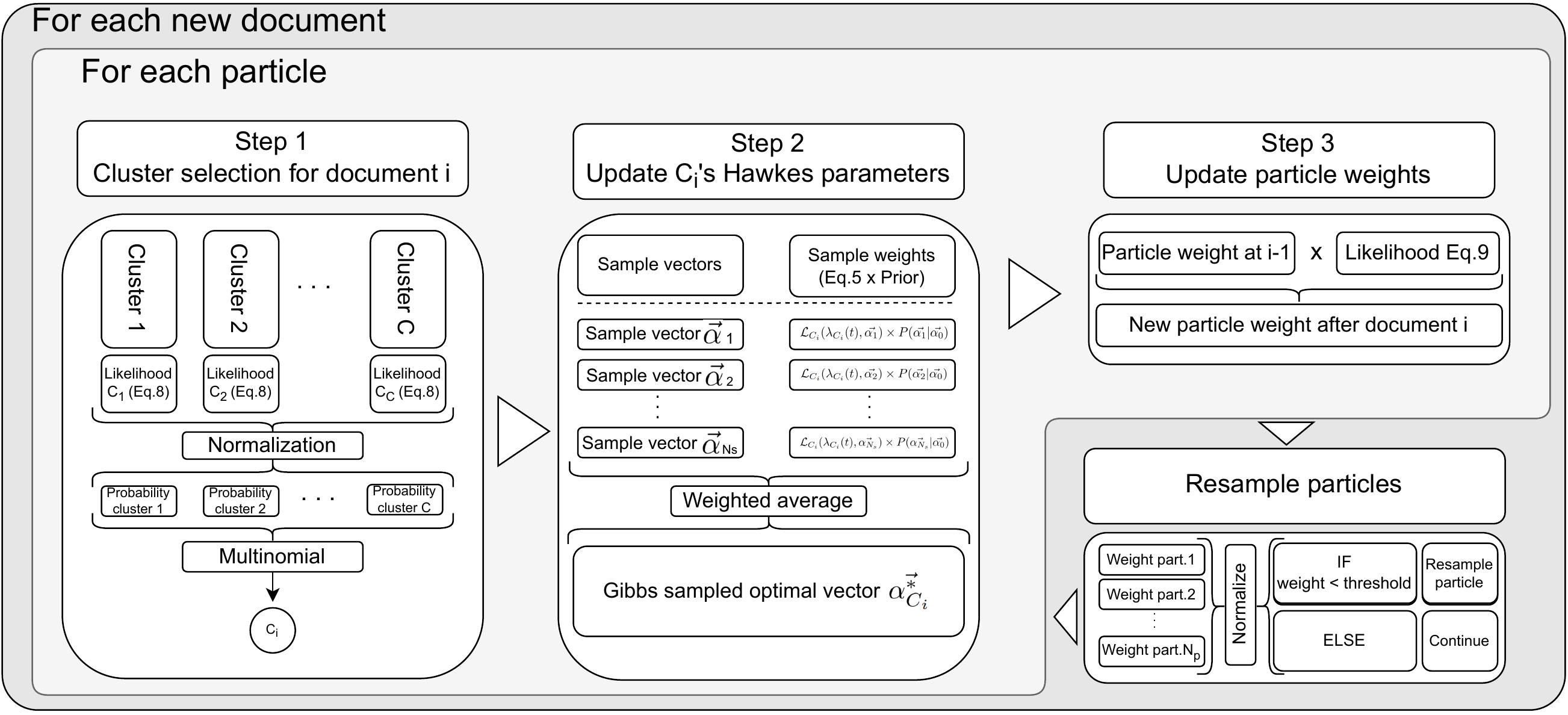}
    \caption[PDHP - Sequential Monte Carlo algorithm]{\textbf{Schematic workflow of the SMC algorithm} --- For each new observation from a stream of document, we run steps 1 (sample document's \gls{cluster}), 2 (update sampled \gls{cluster}'s internal dynamics) and 3 (update particle hypothesis' likeliness) for each particle, and then discard particles containing the less likely hypothesis on \gls{cluster} allocation.}
    \label{fig-SchemaSMC}
\end{figure}

For each particle, when a new document arrives, (1) the \gls{cluster} of the document is sampled according to a Categorical distribution over all \glspl{cluster}, whose weights are determined by Eq.~\ref{eq-likModelTot}. After the \gls{cluster} of the new document has been sampled, (2) the kernel weights $\vec{\alpha}$ from Eq.~\ref{eq-PDHP-HawkesClus} are updated using Eq.~\ref{eq-likHawkes}. For efficiency purpose, we sample $\vec{\alpha}$ from a set of $N_s$ pre-computed $\vec{\alpha}$ vectors. We finally (3) update the weights $\omega_p$ of each particle according to the posterior Eq.~\ref{eq-likModelTot} such as $\omega_p^{(n+1)} = \omega_p^{(n)} \times \text{Eq.~\ref{eq-likModelTot}}$. If the weight of a particle falls below a value $\omega_{thres}$, the particle is discarded and replaced by another existing one with sufficient weight. The full process is illustrated in Fig.~\ref{fig-SchemaSMC}. By updating incrementally the likelihood associated with each of the pre-computed $\vec{\alpha}$ sample vectors, the algorithm processes each new observation in constant time $\mathcal{O}(1)$.

The task of updating kernel coefficients (2) is the same as in any Hawkes process, and the task of updating particle weights and resampling them (3) is common to any SMC algorithm. The change induced by the PDHP compared to the DHP happens in step (1). First of all, we note that for $r=1$ the PDHP prior is identical to the DHP prior. From Section~\ref{PDP}, lowering the value of $r$ reduces the ``\gls{rich-get-richer}'' aspect of the PDP (``rich-get-less-richer''), whereas increasing it leads to a ``rich-get-more-richer'' effect. These metaphors can be translated as follows in our temporal context: for lower values of $r$, the relative difference between \gls{cluster}'s temporal intensities plays a less significant role in \gls{cluster} selection, whereas higher values of $r$ tend to exacerbate these differences and make the temporal aspect of the greatest consequence on the choice of a \gls{cluster}. In other words, tuning the value of $r$ allows giving more or less importance to the temporal aspect of the clustering. This is illustrated in Fig.~\ref{fig-ex-varr}. 

In Fig.~\ref{fig-illustr-PDHP}, we plot the situation when a new observation gets assigned a \gls{cluster}. The associated Hawkes intensities are the base to compute the prior probability for either \gls{cluster}. This quantity is then modulated using $r$ to give more or less importance to intensity differences between \glspl{cluster}. In Fig.~\ref{fig-ex-varr}, we plot the probability for various \glspl{cluster} to be chosen (which is directly proportional to the posterior distribution, see Eq.~\ref{eq-likModelTot}) according to $r$ when their textual likelihood and Hawkes process intensity are known. Note that for $r=0$, the probability for any \gls{cluster} to get chosen is linearly proportional to its textual likelihood (Dirichlet-Uniform process), whereas when $r$ increases, the probability of getting chosen gets closer to a selection based on the temporal aspect only.

\begin{figure}
    \centering
    \includegraphics[width=0.7\textwidth]{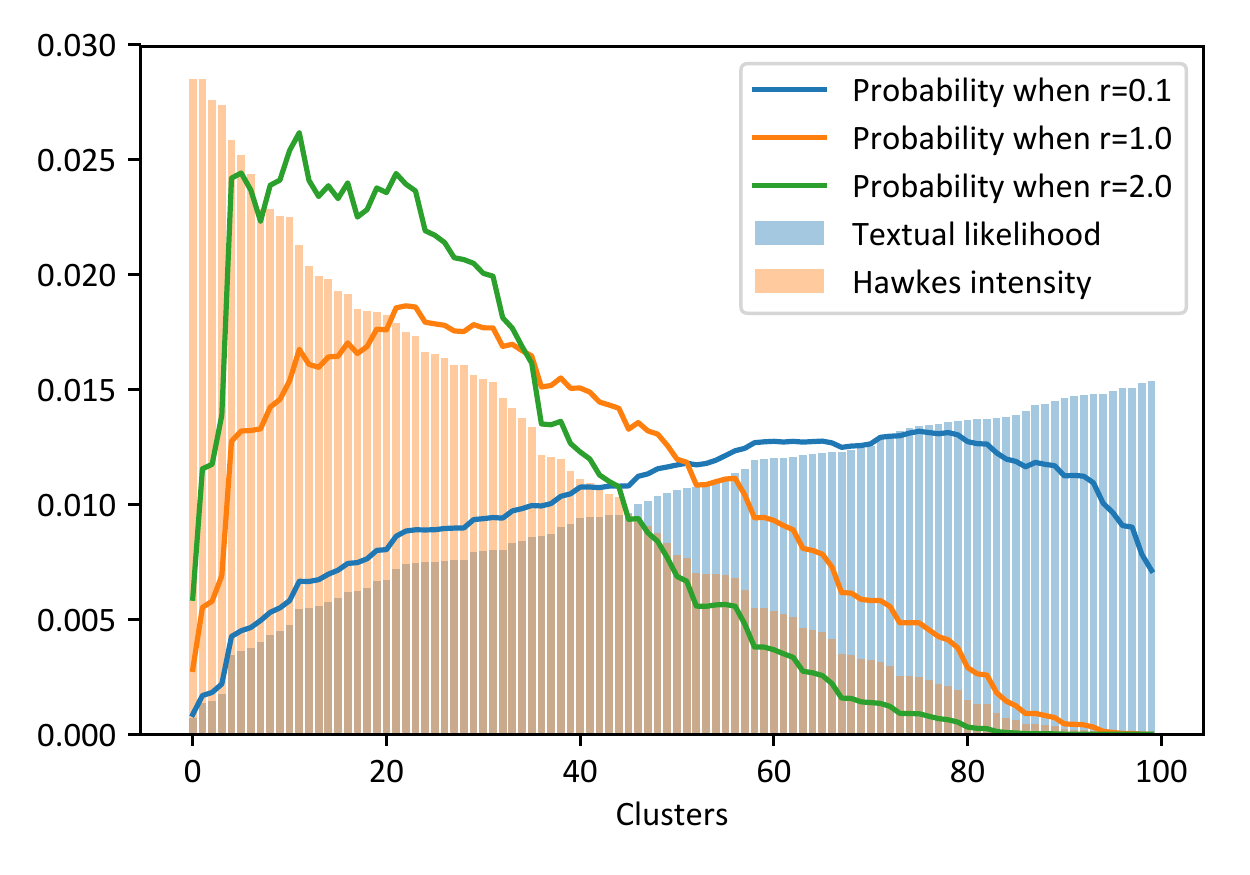}
    \caption[PDHP - Effect of $r$ on \gls{cluster} selection probabilities]{\textbf{Effect of $r$ on \gls{cluster} selection probabilities} --- The probability for each \gls{cluster} to get chosen (solid lines) for several values of $r$ and fixed individual textual likelihood (blue bars) and Hawkes intensity (orange bars).}
    \label{fig-ex-varr}
\end{figure}

This makes the main interest of the PDHP model. Tuning the parameter $r$ allows choosing whether inferred \glspl{cluster} are based on textual or temporal considerations. It generalizes several state-of-the-art works, which are special cases of the PDHP for different values of $r$. The \acrshort{DHP} \citep{Du2015DHP} is equivalent to \acrshort{PDHP} for $r=1$; the \acrshort{UP} \citep{Wallach2010UnifP} is equivalent to \acrshort{PDHP} when $r=0$.
In the following sections, we show how fine-tuning $r$ systematically yields significantly better results than setting it to $r=0$ or $r=1$ (up to a gain of 0.3 on our experiments' normalized mutual information metric). We also show how varying it allows to recover one kind of clustering or the other (textual or temporal) with high accuracy and see how it affects clustering results on several real-world datasets.

\subsection{Experiments}
\subsubsection{Synthetic data}
\paragraph{Synthetic data generation}
We simulate a case where only two \glspl{cluster} are considered. Each \gls{cluster} has its own vocabulary distribution over 1,000 words and its own kernel weights $\vec{\alpha}$, with Gaussian Hawkes kernel functions $\vec{\kappa (t)}$ of parameters $(\mu, \sigma)$=(3, 0.5), (7, 0.5) and (11, 0.5) (see Eq.~\ref{eq-PDHP-HawkesClus}). Finally, we set $\lambda_0=0.05$.
We first simulate one \gls{independent} Hawkes process per \gls{cluster} using the Tick Python library \citep{Bacry2017Tick}. The processes are stopped at time $t=1500$, which makes a rough average of 7,000 events per run. Then we associate each simulated observation with a sample of 20 words drawn from the corresponding \gls{cluster}'s word distribution. The inference has been performed using an 8 core processor (i7-7700HQ) with 8GB of RAM on a laptop, which underlines how scalable the algorithm is. As stated before, the algorithm processes each new document in constant time $\mathcal{O}(1)$, which ranged from 0.05s on synthetic data to maximum 1s on real-world data. Note that this number is directly proportional to the number of active inferred \glspl{cluster}, and thus depends strongly on the dataset.

\begin{figure}
    \centering
    \includegraphics[width=\columnwidth]{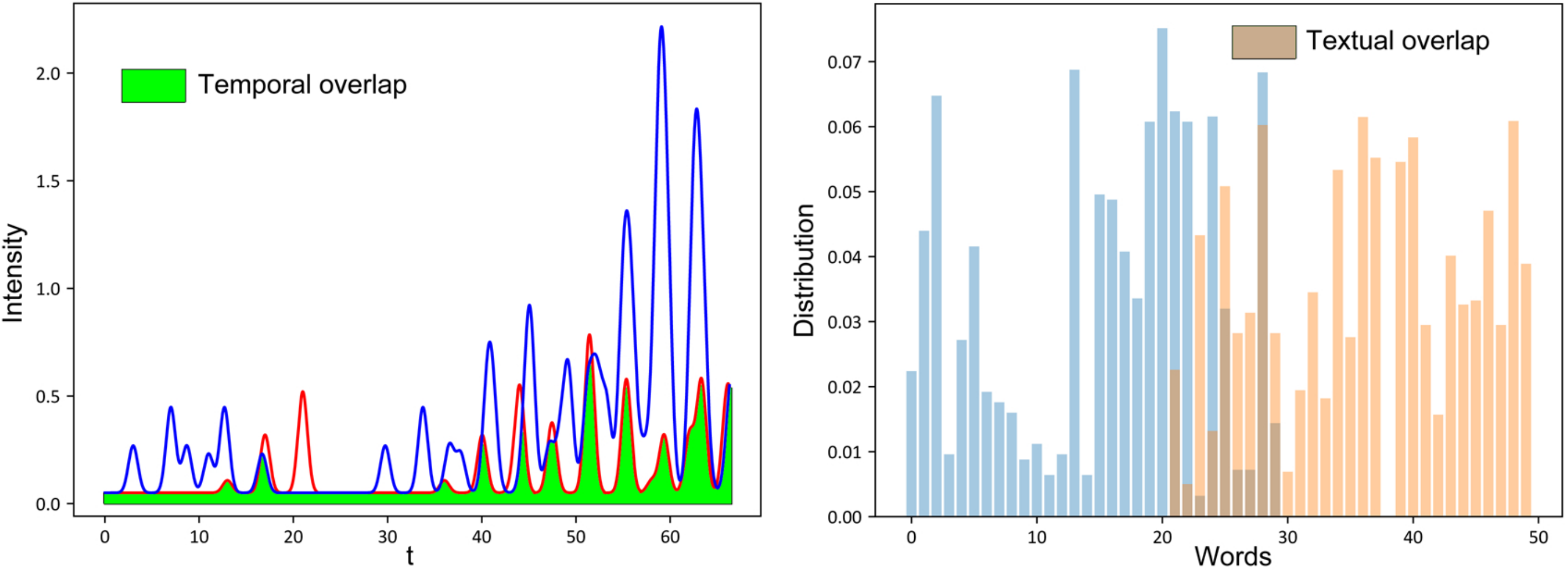}
    \caption[PDHP - Overlaps]{\textbf{Overlaps} --- (Left) Temporal overlap is defined as the ratio between the area common to two Hawkes intensities and the total area under the intensity functions. (Right) Textual overlap is defined as the proportion of vocabulary that is common to two \glspl{cluster}, weighted by the probability of words within their respective \gls{cluster}.}
    \label{fig-illustration-overlaps}
\end{figure}

We generate ten such datasets for every considered value of vocabulary overlap and Hawkes intensities overlap, which leave us with 200 datasets (5 values of textual overlap x 4 values of temporal overlap x 10 datasets). Overlap is defined as the shared area of two distributions, normalized by the total area under the distributions. For instance, if the vocabulary of one \gls{cluster} ranges from words "1" to "100" with uniform distribution, and the vocabulary of another \gls{cluster} from words "50" to "150" with uniform distribution, the overlap equals 50\%. We define the overlap of Hawkes process intensity in the same way. If the triggering Hawkes kernel of one \gls{cluster} is a Gaussian function with $(\mu, \sigma)=(3, 1)$ and one associated observation at $t=0$, and the triggering kernel of the other is also a Gaussian function but with $(\mu, \sigma)=(5, 1)$ also with an associated observation at $t=0$, the overlap equals 32\% (see Fig.~\ref{fig-illustration-overlaps}). When computing the Hawkes intensity overlap, every observation within a \gls{cluster} and its associated timestamp are considered. The definition of overlaps is illustrated in Fig.~\ref{fig-illustration-overlaps}. To enforce a given vocabulary overlap (Fig.~\ref{fig-illustration-overlaps}-right), we shift the word distributions of the \glspl{cluster} from which events' vocabulary is sampled. To enforce a given Hawkes intensities overlap (Fig.~\ref{fig-illustration-overlaps}-left), we shift the event times of every event in one of the \glspl{cluster} until we get the correct overlap ($\pm 5\%$).

Note that we consider ten different datasets instead of considering ten runs per dataset for two reasons. Firstly, the generation of Hawkes processes is highly stochastic, so a model might perform significantly better on a single dataset only by chance. Secondly, given the way the SMC algorithm works, the standard deviation between runs is small: at each iteration, $N_{part}$ clustering hypotheses are tested, which is equivalent to running $N_{part}$ times a single clustering algorithm. We heuristically set $N_{part}=8$, as we observe no significant improvement using more particles.

The other parameters we use for clustering synthetic data are: $\alpha_0=0.1$, ${\theta_0=1}$, $\vec{\kappa(t)} = [\mathcal{G}(t; 3, 0.5), \mathcal{G}(t; 7, 0.5), \mathcal{G}(t; 11, 0.5)]$ with $\mathcal{G}(t; \mu, \sigma)$ the Gaussian function, $N_{samples}=2.000$ and $\omega_{thres}=\frac{1}{2N_{part}}$.

\begin{figure}
    \centering
    \includegraphics[width=\columnwidth]{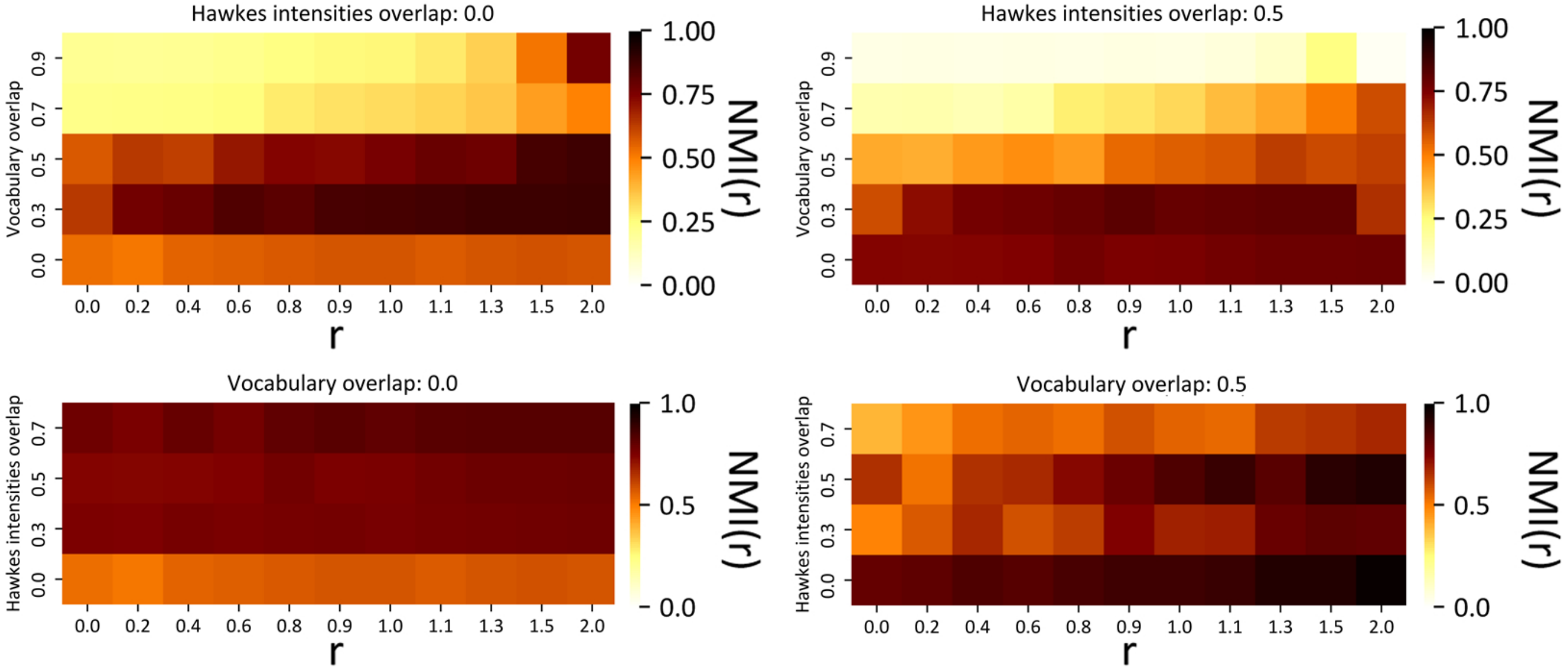}
    \caption[PDHP - PDHP yields good NMI values on synthetic data]{\textbf{PDHP yields good NMI values} --- Normalized mutual information (\acrshort{NMI}) for various values of $r$, intensities overlap and vocabulary overlap, for one dataset per combination. The results for $r=0$ are the output of the Uniform process, the results for $r=1$ are the output of the DHP \citep{Du2015DHP}, and the other values of $r$ correspond to other particular cases of PDHP. The darker the better. Overall, PDHP yields good NMI values (the maximum being 1).}
    \label{fig-res-NMI}
\end{figure}

We are interested in varying both vocabulary and intensities overlap to exhibit the limits of DHP and how PDHP overcomes them. Note that in the synthetic data experiments in \citep{Du2015DHP} (Figs.3a and 3b), the intensities overlap is almost null, which makes the task easier for the Hawkes part of the algorithm. The primary metric we use throughout the experimental section is the normalized mutual information (\acrshort{NMI}). During the experiments, we also considered the Adjusted normalized rand index and the V-measure, which are adapted to evaluate clustering results when the number of inferred \glspl{cluster} is different from the true number of \glspl{cluster}. The observed trends in results from these other metrics are identical to the ones observed for NMI. Therefore, we choose to report only the results of the latter for clarity. 

The NMI metric is standard when evaluating non-parametric clustering models. It compares two \gls{cluster} partitions (i.e., the inferred and the ground truth ones); it is bounded between 0 (each true \gls{cluster} is represented to the same extent in each of the inferred ones) and 1 (each inferred partition comprises 100\% of one true \gls{cluster}).

\paragraph{PDHP yields better results as vocabulary overlap increases}
We report our results when the intensities overlap is null, with varying $r$ and the vocabulary overlap in Fig.~\ref{fig-res-overlaps}a. Because we consider ten different datasets for each set of overlap parameters, it makes no sense to report the absolute average NMI since it can vary greatly from one dataset to the other. Instead, we plot the relative NMI difference between PDHP and DHP ($r=1$), which we expect to be less dependent on the datasets we consider. However, to give an idea of the typical performance for some parameters, we also provide raw results for one run in Fig.~\ref{fig-res-NMI}.

\begin{figure}
    \centering
    \includegraphics[width=\columnwidth]{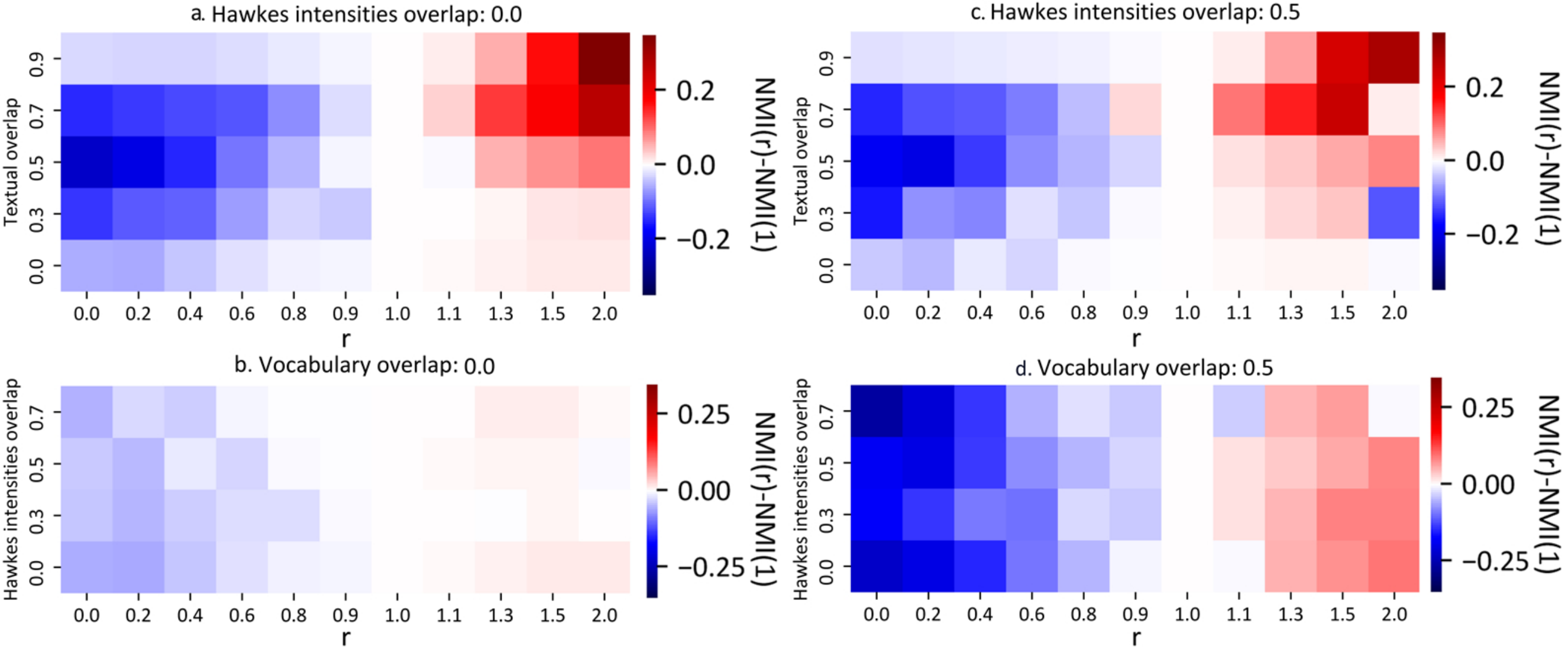}
    \caption[PDHP - Experimental comparison to DHP]{\textbf{PDHP performs better than DHP} --- Difference between the normalized mutual information (\acrshort{NMI}) of PDHP and DHP model \citep{Du2015DHP} for various values of $r$, intensities overlap and vocabulary overlap, averaged over all the datasets. Red means PDHP performed better, blue means PDHP performed less well. Because PDHP($r=1$)=DHP, the column $r=1$ show no difference. PDHP allows to increase results on NMI by as much as 0.3 over DHP.}
    \label{fig-res-overlaps}
\end{figure}

There is a clear correlation between efficiency, vocabulary overlap and $r$, with a gain on NMI up to $+30\%$ of its maximal value over DHP. As stated at the end of the "Model" section, this result was expected: the more vocabulary overlap grows, the less textual \gls{content} carries valuable information to \gls{cluster} the documents. This observation supports the concerns raised in \citep{Yin2018ShortTextDHP} about the efficiency of DHP for clustering short text documents. However, Hawkes intensities overlap being null, the arrival time of events carries highly valuable information when textual \gls{content} does not allow to distinguish \glspl{cluster} well. Therefore, PDHP provides a way to tackle the problem raised in \citep{Yin2018ShortTextDHP} without the need to sample observations.

Conversely, when vocabulary overlap is null, the textual \gls{content} provides enough information to distinguish \glspl{cluster} correctly. The temporal dimension only allows refining the results with no significant improvement for all values of $r$.

Finally, we can see how the Dirichlet-Uniform process (DUP, $r=0$) consistently yields worse performances under these settings. Once again this is expected, since in this synthetic experiment intensities overlap carries valuable information about events clustering; DUP only considers textual information and therefore misses valuable clues.

\paragraph{PDHP yields similar results for null vocabulary overlap}
We report comparable results in Fig.~\ref{fig-res-overlaps}b. Here, we consider a null vocabulary overlap for various values of $r$ and of Hawkes intensities overlap. The situation is now the opposite: the textual \gls{content} always carries valuable information about \glspl{cluster}, whereas the temporal aspect does not. We observe the same trend as in Fig.~\ref{fig-res-overlaps}a --note that the colour scale is the same. Varying the value of $r$ does not significantly change the performance of clustering, meaning the textual \gls{content} always carries enough information. This plot shows that PDHP can handle greater intensities overlap without collapsing into unrealistic clustering. Since in most real-case applications, many \glspl{cluster} with various dynamics may coexist simultaneously, it is comforting that the PDHP can also handle this case.

\paragraph{PDHP yields better results in more realistic situations}
We finally report the results for intermediate values of intensities and vocabulary overlaps in Fig.~\ref{fig-res-overlaps}c,d. In real-world applications, it seldom happens that topics’ vocabularies do not overlap at all. For instance, a quick analysis of \textit{The Gutenberg Webster's Unabridged Dictionary by Project Gutenberg} shows that 22\% of English words are associated with more than one definition. A more detailed analysis would need to consider the usage frequency of words to get correct statistics. Still, this number provides an estimate of the effective vocabulary overlap in real-world situations.

In Fig.~\ref{fig-res-overlaps}c, we present the results for a fixed intensities overlap of 0.5 versus various values of $r$ and vocabulary overlaps, and in Fig.~\ref{fig-res-overlaps}d for a fixed vocabulary overlap of 0.5 versus various values of $r$ and intensities overlaps. Once again, we see that, on average, using PDHP can increase the NMI over DHP up to +20\% of the maximum possible value.

\paragraph{PDHP finds textual or temporal \glspl{cluster} depending on r}
\begin{figure}
    \centering
    \includegraphics[width=0.8\columnwidth]{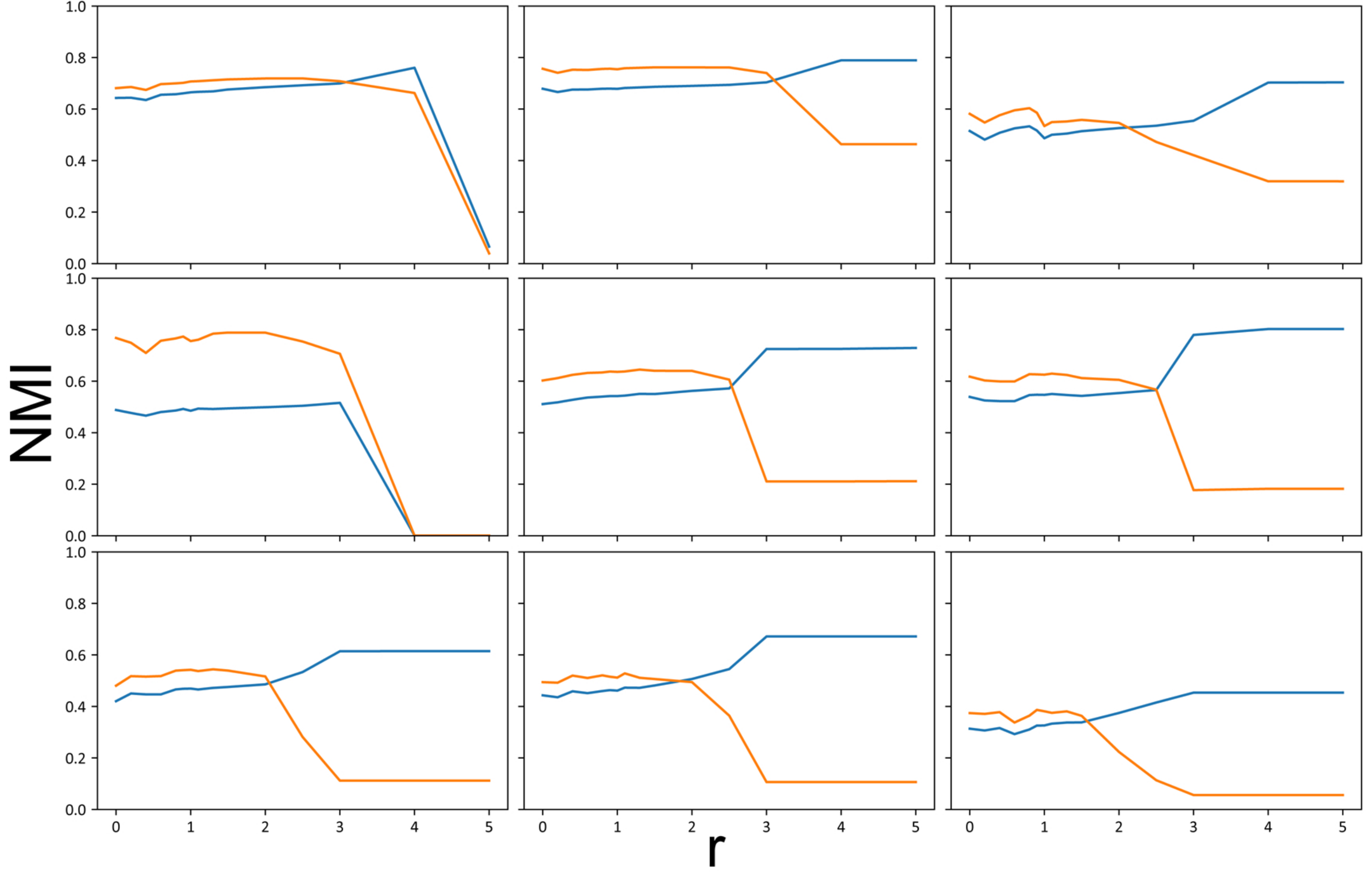}
    \caption[PDHP - Results when \gls{content} and dynamics are decorrelated]{\textbf{Textual (orange) and temporal (blue) NMI vs r when textual and temporal \glspl{cluster} are decorrelated} --- From top-left to bottom-right, there are 10\%, 20\%, 30\%, 40\%, 50\%, 60\%, 70\%, 80\% and 90\% of generated events that have been randomly re-assigned a textual \gls{cluster}. The orange curves are the textual NMI vs $r$, which evaluate how well events whose vocabulary has been sampled from the same distribution are clustered together; the blue curves are the temporal NMI vs $r$, which evaluate how well events following the same temporal dynamic are correctly clustered together. The values presented are for one dataset. We see that varying $r$ allows retrieving the right temporal ($r$ large) or textual \glspl{cluster} ($r$ small).}
    \label{fig-res-perc_rand1run}
\end{figure}
We now slightly modify our experimental setup. Instead of considering that textual \glspl{cluster} and Hawkes intensities are perfectly correlated, we consider a decorrelated case. A document whose vocabulary is drawn from \gls{cluster} $C_1$ can now follow the same temporal dynamics as \gls{cluster} $C_2$. If we imagine a dataset of news articles published online, it is clear why this might happen frequently. If popular newspapers such as New York Times or Reuters publish an article on a topic $A$ at a time t, it is likely to trigger snowball publications of related articles from less popular journals. ``Popularity'' is chosen as an indicator in this example, but it may be any other external parameter (centrality in news networks, support of publications, etc.). In this case, the article's textual \gls{content} allows to uncover a ``story of publication'', that is, how the article has been \gls{spread}, when publication spikes are, etc. However, the temporal information would help understand the dynamics of publications \gls{interaction}: which reduced set of articles triggered the publication of subsequent ones. 

In \citep{Du2015DHP}, it is assumed that every document within \glspl{cluster} follows a unique dynamic. We relax this hypothesis in our datasets as follows. For null textual and temporal overlaps, after a dataset has been generated, we resample the textual \glspl{cluster} of a fraction of randomly selected events, as well as the words associated with the event. In doing so, we decorrelate temporal and textual \glspl{cluster}. Therefore, an event is now described by two \gls{cluster} indicators: its temporal \gls{cluster} (which Hawkes intensity made the event appear where it is) and a textual \gls{cluster} (which vocabulary has been used to sample the words the event contains).

For completeness, we also show the results for various decorrelations for one run in Fig.~\ref{fig-res-perc_rand1run}. To better understand the tendency of NMIs with respect to $r$, we plot the average difference between the NMI of textual clustering and the NMI temporal clustering over all the datasets. Explicitly: $\Delta NMI = NMI_{text}-NMI_{temp}$. The results are reported in Fig.~\ref{fig-res-decorr}.

\begin{figure}
    \centering
    \includegraphics[width=\columnwidth]{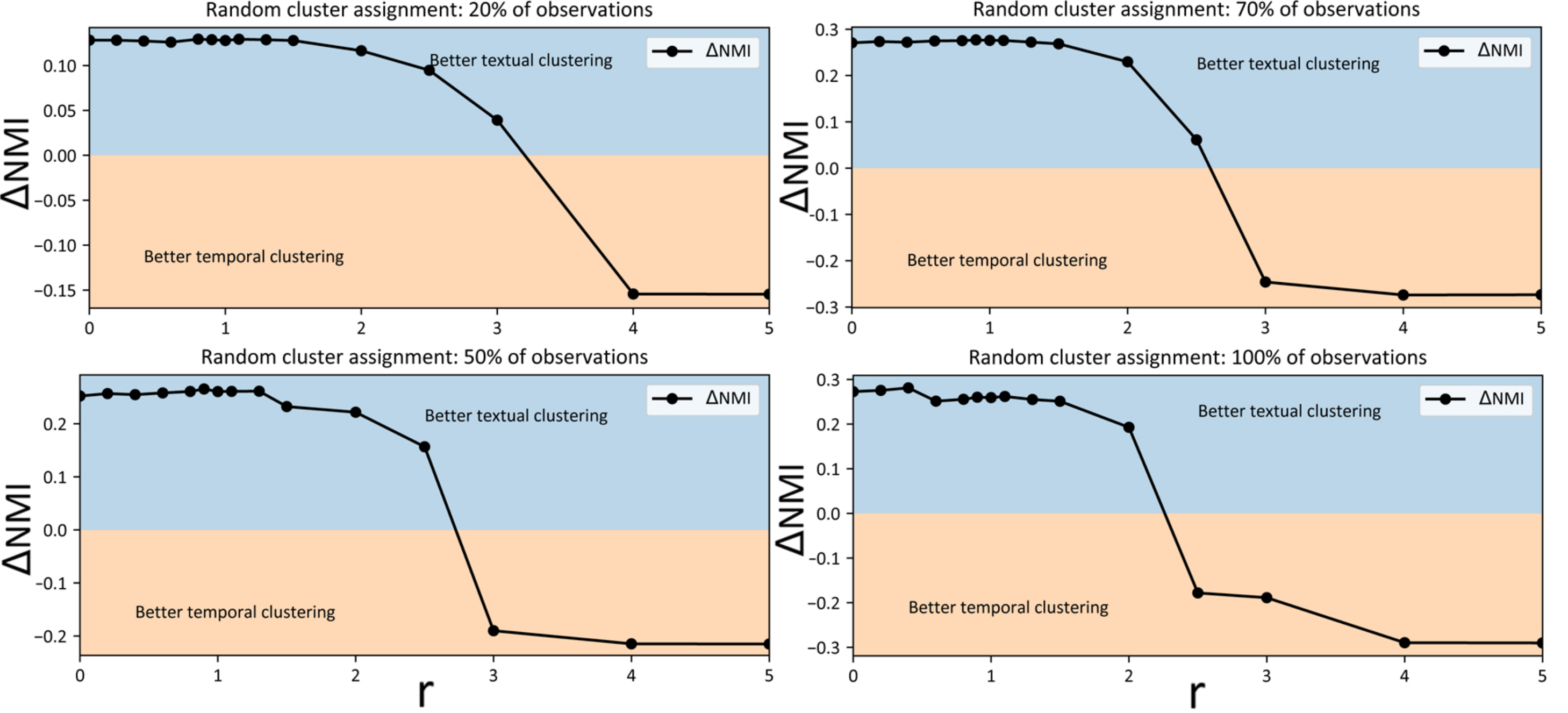}
    \caption[PDHP - Varying $r$ to choose between textual or temporal clustering]{\textbf{Varying $r$ allows to choose between textual or temporal clustering} --- The black line plots the difference between the NMI of textual and temporal clustering. For small $r$, textual clustering is far better than temporal clustering, and for large $r$, the situation is reversed. This is because $r$ determines the importance given to the temporal dimension and therefore allows choosing between retrieving temporal or textual \glspl{cluster}.}
    \label{fig-res-decorr}
\end{figure}

As supposed at the end of the ``Model'' section, varying $r$ allows retrieving one clustering or the other. Note that the value $r$ of transition from text to time clustering depends directly on the dataset considered: number of words sampled, vocabulary size, overlaps, etc.

\paragraph{PDHP efficiently infers the temporal dynamics of each \gls{cluster}}
Finally, we show that PDHP correctly infers kernels' parameters in every situation where events are correctly assigned to their temporal \gls{cluster}. The results are reported in Fig.~\ref{fig-err-kernels}. We looked at the mean absolute error (\acrshort{MAE}) and the mean Jensen-Shannon divergence (MJS) between the vector $\vec{\alpha}$ used to generate the dataset and the inferred one. We note in Fig.~\ref{fig-err-kernels} that when the textual overlap is small, the inferred kernel is close to the real one and $r$ has a minor impact on the result. This is because the inferred kernels mostly depend on the correctness of inferred \glspl{cluster}: when observations are allocated to the right \glspl{cluster}, the Hawkes process inference considers relevant information when inferring these \glspl{cluster}' dynamics. However, when observations are misallocated, the Hawkes process infers dynamics also based on times that are not supposed to contribute to this \gls{cluster}'s dynamic. When the clustering task is simple and yields good results (that is, when the textual overlap is small, see Fig.~\ref{fig-res-NMI}), the PDHP infers correct temporal dynamics ($\sim 5\%$ MAE); this shows our method correctly accounts for \glspl{cluster} dynamics given the available information.

When vocabulary overlap is large, the value of $r$ significantly influences the kernel inference performances. However, when $r$ is chosen so that \glspl{cluster} are correctly inferred, the kernel inference retrieves well the expected kernels ($\sim 5\%$ MAE). Finally, the temporal kernel inference is expected to be less precise when temporal overlap increases, which is what happens in Fig.~\ref{fig-err-kernels}-bottom-right. In this case, the model does not retrieve well the synthetic kernels even for the optimal $r$. Besides, the error bars get wider as a consequence of the \gls{cluster} allocation being more challenging. Overall, provided the right \glspl{cluster}, we conclude that our method correctly retrieves the inferred temporal kernels.

\begin{figure}
    \centering
    \includegraphics[width=\columnwidth]{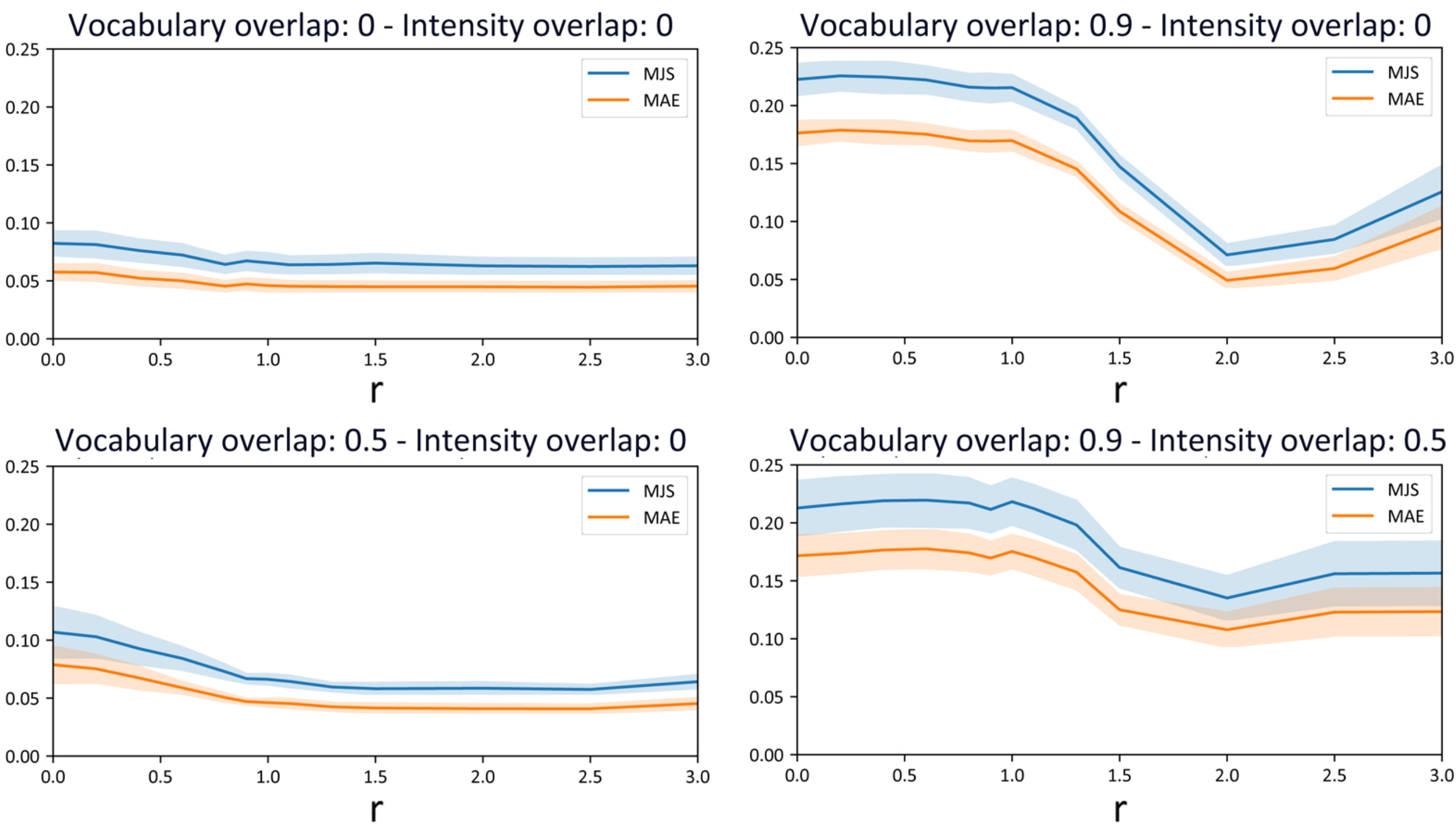}
    \caption[PDHP - Varying $r$ allows to better capture the dynamics at stake]{\textbf{Varying $r$ allows to better capture the dynamics at stake} --- We plot the mean average error and the mean Jensen-Shannon divergence of the inferred kernel function $\alpha$ with respect to the kernel used to generate the data, for various values of temporal and textual overlaps. The standard error bars are computed over 10 \gls{independent} runs. The higher the temporal overlap, the larger the error bars; the larger the textual overlap the more influence has $r$.}
    \label{fig-err-kernels}
\end{figure}

\subsubsection{Real-world application on Reddit}
We use the PDHP prior to model real streams of textual documents. We consider three Reddit datasets \citep{Baumgartner2020PushshiftRedditDataset} about different topics. The \textbf{News dataset} is made of 73.000 titles extracted from the subreddits inthenews, neutralnews, news, nottheonion, offbeat, open\_news, qualitynews, truenews and worldnews, from April 2019. We chose this month because of the wide variety of events that happened then (for instance, Sri Lanka Easter bombings, Julian Assange arrest, first direct picture of a black hole, Notre-Dame cathedral fire). We also consider 15.000 post titles of the subreddit TodayILearned (\textbf{TIL dataset}) and 13.000 post titles of the subreddit AskScience (\textbf{AskScience dataset}) on January 2019. We extracted the nouns, verbs, adjectives, and symbols from the textual data. We run the experiments using the following parameters: $\alpha_0=0.5$, $\theta_0=0.01$, $N_{samples}=2.000$, $N_{part}=8$ and $\omega_{thres}=\frac{1}{2N_{part}}$. The kernel vector $\vec{\kappa}$ is chosen as in \citep{Du2015DHP}. It is made of Gaussian functions, with means located at 0.5, 1, 4, 8, 12, 24, 48, 72, 96, 120, 144 and 168 hours. The variance of each is set to 1, 1, 3, 8, 12, 12, 24, 24, 24, 24, 24 and 24 hours. The shape of the kernel is chosen so that it can account for a dynamic that can occur at different timescales. The algorithm will then infer the weights $\vec{\alpha}$ associated with each entry of the kernel vector $\vec{\kappa}$ for each \gls{cluster}.

\begin{figure}
    \centering
    \includegraphics[width=\columnwidth]{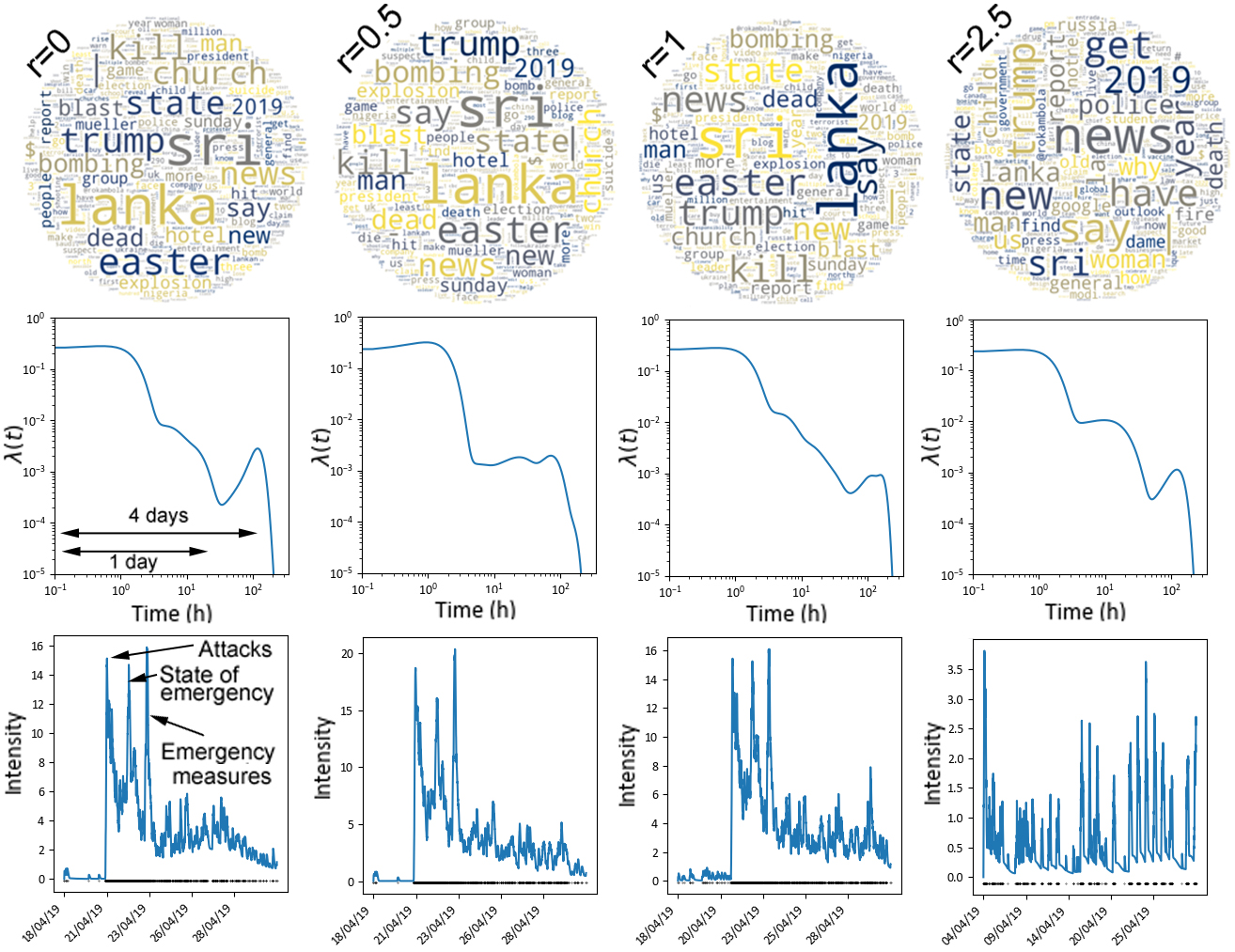}
    \caption[PDHP - Results for the \gls{cluster} about Sri Lanka 2019 bombings]{Wordclouds, triggering kernels and intensities for \glspl{cluster} the most closely related to Sri Lanka 2019 bombings for various values of $r$. The points at the bottom of the intensity plots are individual publication events. Note that triggering kernels are plotted on a log-log scale for visualisation purposes because most of the intensity is focused on small times: dynamics are bursty.}
    \label{fig-res-Srilanka}
\end{figure}
\paragraph{PDHP recovers meaningful stories}
As an illustrative example, we consider the inferred \glspl{cluster} the most closely related to the Sri Lanka Easter bombings of April 2019 in Fig.~\ref{fig-res-Srilanka}. The main bursts in the news related to this event happened on the 21$^{st}$, 22$^{nd}$ and 23$^{rd}$ of January, and respectively correspond to the bombings themselves, the declaration of the state of emergency, and finally their application on the 23$^{rd}$. We plot the temporal kernels associated with this event on a log-log scale because most of the intensity is focused on small times: dynamics of information \gls{spread} are bursty \citep{Karsai2018BurstyHumanDynamics}. We see that inferred dynamics change with $r$ as well as the \gls{cluster}'s vocabulary, which is expected since \glspl{cluster} do not contain the same documents. For $r=0$, the Uniform process infers \glspl{cluster} based on textual information only; the triggering kernel is inferred afterwards. For $r=2.5$ on the contrary, \glspl{cluster} are formed based on the triggering kernel, and textual information follows; we see from the right-plot that this \gls{cluster} captures publications exhibiting a daily intensity cycle; this is visible both in the intensity plot (the bump around $2.10^2$h which is not present on other temporal kernels) and in the real-time axis where publications seem to be packed around specific points in time roughly corresponding to a daily cycle. Given the intensity spikes on 21$^{st}$, 22$^{nd}$, and 23$^{rd}$, it is not surprising that articles about the Sri Lanka bombings are also part of this \gls{cluster}. Note that the more $r$ increases, the more intense the triggering kernel is around 24h. We see from Fig.~\ref{fig-res-Srilanka} that DHP is a specific case of our modelling, and that tweaking the $r$ parameter allows us to retrieve completely different results.

\paragraph{PDHP favours temporal or textual clustering depending on r}
We report the values of log-likelihoods for every dataset and various values of $r$ in Fig.~\ref{fig-likReddit}. The textual likelihood is defined in Eq.\ref{eq-likModelLg}, and the likelihood of a Hawkes process is defined in Eq.\ref{eq-likHawkes}. Note that $r$ does not appear in either Eq.\ref{eq-likModelLg} or Eq.\ref{eq-likHawkes}; the plot in Fig.~\ref{fig-likReddit} thus only reflects the relevancy of the proposed textual modelling or temporal modelling independently from PDHP. Those likelihoods evaluate how well the textual or temporal aspect of the dataset is modelled with no consideration of the model being used. As expected from the synthetic experiments, varying $r$ makes the model more sensitive to either textual or temporal data --note the similarity to Fig.~\ref{fig-res-perc_rand1run}. A low $r$ favours the textual information clustering and is thus better at modelling documents' textual \gls{content}, whereas a high $r$ favours temporal information which makes PDHP better at capturing the publication dynamics.

\begin{figure}
    \centering
    \includegraphics[width=\columnwidth]{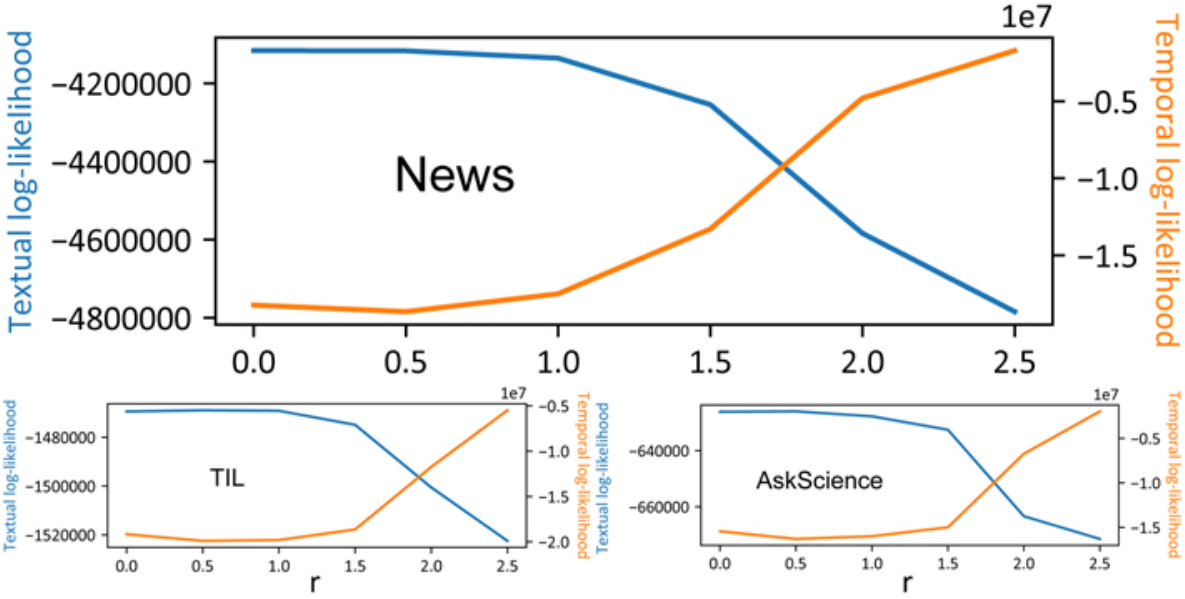}
    \caption[PDHP - Favouring text or time based \glspl{cluster} on real world datasets]{\textbf{$r$ allows to favour text-based of time-based clustering on real world datasets} --- Textual likelihood and Hawkes process likelihood for various values of $r$. The lower $r$ the higher textual likelihood is, and the higher $r$ the higher Hawkes process likelihood is.}
    \label{fig-likReddit}
\end{figure}

\paragraph{PDHP infers sharper textual \glspl{cluster} for low r}
We evaluate how meaningful textual \glspl{cluster} are using an entropy measure. We assume that a \gls{cluster} is meaningful when it contains a reduced set of words; a \gls{cluster} talking about one topic only is more likely to have a smaller vocabulary than a \gls{cluster} about two or more topics.
A way to measure this is to see how \gls{spread} the vocabulary of a \gls{cluster} is using Shannon entropy. Let $N_{c,v}$ be the count of word $v$ in \gls{cluster} $c$. The normalized Shannon entropy of a \gls{cluster} $c$ is defined as:
\begin{equation}
    \label{eq-Entropy}
    S(\vec{N_c}) = \frac{1}{-\log (V)}\sum_v^V \log (\frac{N_{c,v}}{\sum_v' N_{c,v'}})\frac{N_{c,v}}{\sum_v' N_{c,v'}}
\end{equation}
An entropy of 0 means the vocabulary of the \gls{cluster} is concentrated on a single word among the V possible words in the vocabulary; an entropy of 1 means that every of the $V$ words is present to the same extent with probability $\frac{1}{V}$. In Fig.~\ref{fig-entropy-txt}, we plot the mean entropy for various values of $r$ for all the datasets, along with the standard error over the \glspl{cluster}. The results show that on average vocabulary is more concentrated within \glspl{cluster} for low values of $r$. The inflection point of the curves corresponds to what has been previously observed with likelihoods in Fig.~\ref{fig-likReddit} and Fig.~\ref{fig-res-decorr}. On the contrary, higher values of $r$ lead to \glspl{cluster} that comprise a less concentrated vocabulary. This is expected because as $r$ increases, the textual information is no longer the most relevant data for \gls{cluster} formation.

\begin{figure}
    \centering
    \includegraphics[width=0.6\columnwidth]{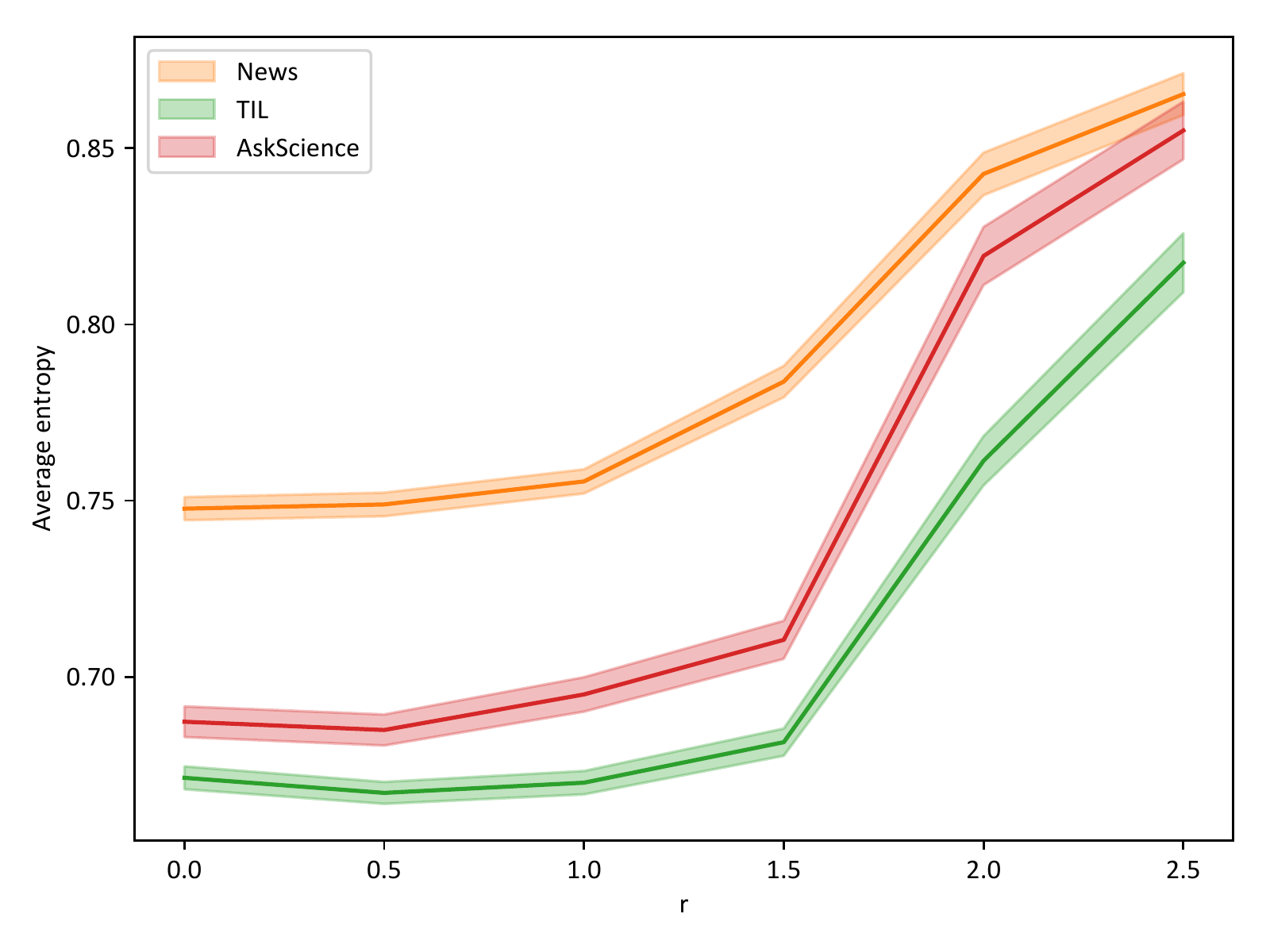}
    \caption[PDHP - Textual \glspl{cluster} are more informative for low values of $r$]{\textbf{Textual \glspl{cluster} are more informative for low values of $r$} --- Weighted average entropy of words distribution for every dataset. Weights corresponds to the number of words within \glspl{cluster}. The error bar represents the standard error over all the \glspl{cluster}.}
    \label{fig-entropy-txt}
\end{figure}

\begin{figure}
    \centering
    \noindent\makebox[\textwidth]{
    \includegraphics[width=1.2\columnwidth]{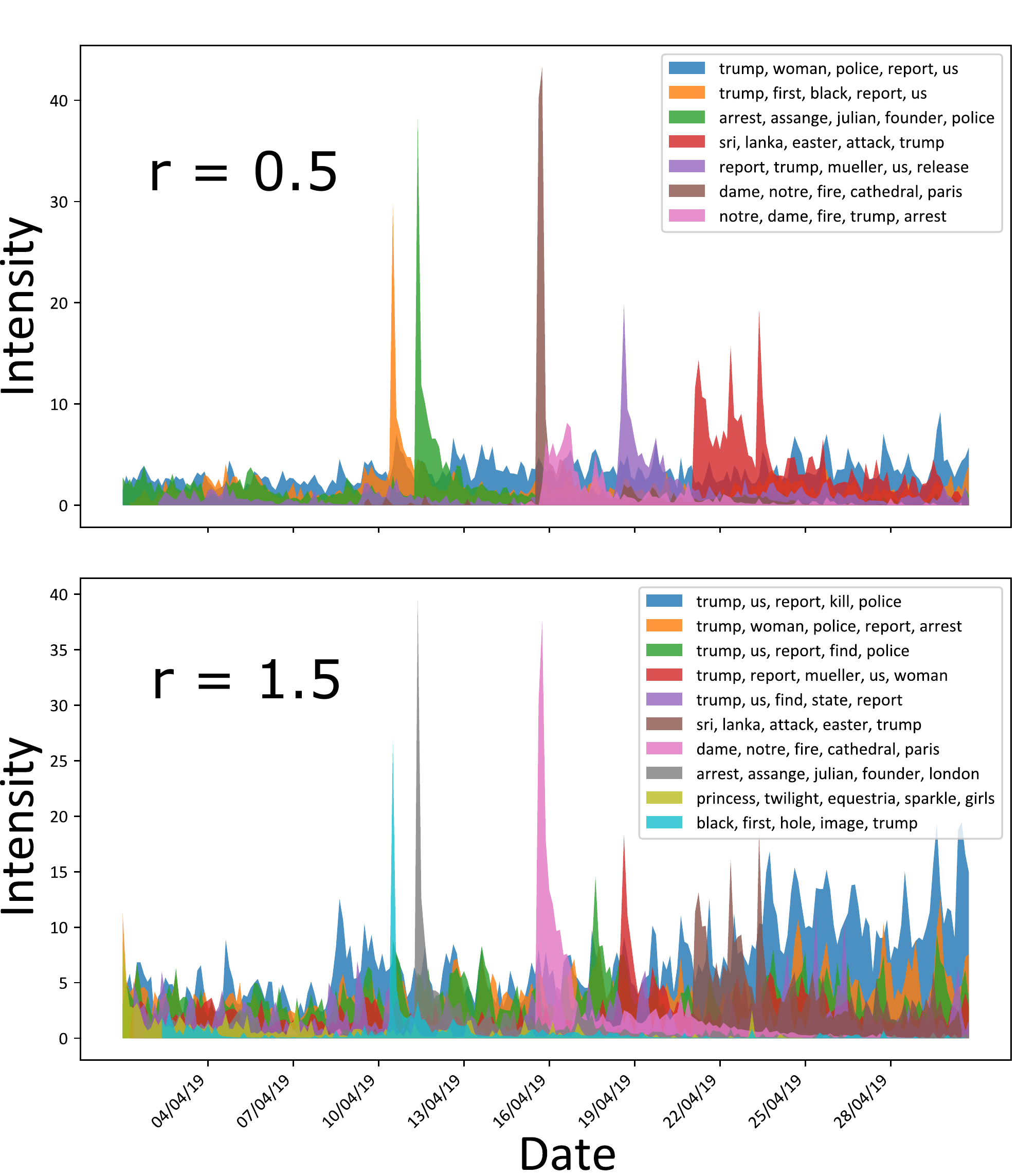}
    }
    \caption[PDHP - PDHP allows for modelling bursty events]{\textbf{PDHP allows for modelling bursty events} --- PDHP's intensity for the News dataset for two values of $r$. A lower $r$ finds globally relevant \glspl{cluster}, whereas a higher $r$ allows to recover shorter bursty events.}
    \label{fig-bursty}
\end{figure}

\paragraph{PDHP controls the burstiness}
In Fig.~\ref{fig-bursty}, we plot the intensity function associated with the News dataset on the real-time axis for several \glspl{cluster} for $r=0.5$ and $r=1.5$. Note that not each of the $\sim$300 inferred \glspl{cluster} are represented, but instead we consider only the ones whose intensity went above 10 at least once, the rest being considered as noise. First of all, both values of $r$ allow to recover the major events of April 2019 (in order of appearance): the first direct picture of a black hole (10/04), the arrest of Julian Assange (11/04), the fire of Notre-Dame de Paris Cathedral (15/04), the release of the Mueller report on Donald Trump (18/04) and the Sri Lanka Easter bombings (21/04). The top 5 words of every \gls{cluster} are reported in the legend. 

When $r$ increases, PDHP retrieves new \glspl{cluster} associated with shorter bursty events. For instance, the \gls{cluster} associated with the release of a new episode of Equestria Girls that went unnoticed for $r=0.5$ has been detected with $r=1.5$. This happens because the episode has not been discussed over a long period, and associated articles have a vocabulary significantly overlapping with other \glspl{cluster}' one. A model relying mostly on textual information might not detect specific words (twilight, Equestria, sparkle, etc.). If detected and a new \gls{cluster} is created, it might then fail to associate subsequent events to this new \gls{cluster} if temporal information plays a lesser role. 
On the other side, a model favouring temporal information is much more likely to associate subsequent events to a new \gls{cluster} despite textual information fitting well older and more populated \glspl{cluster}. 

This can be seen in Fig.~\ref{fig-res-Srilanka}, where the intensity of a kernel peaks at short times.
This results in encouraging the burstiness. When $r$ is large, a given event is likely to be associated with subsequent ones even when the associated vocabularies are only vaguely similar. On the other side, when $r$ is small, older events with closer vocabularies have more chances of getting associated with it despite their intensity not peaking at the new event time.

\paragraph{Recovering publication cycles}
The limit case of encouraging events burstiness is the deterministic allocation of documents to a \gls{cluster} based only on their relative positions on the time axis. This is achieved when $r$ is large. In this case, textual information does not matter and only regularities in the time distributions are detected. We illustrate such a case on the News dataset in Fig.~\ref{fig-deterministicnews}.

In Fig.~\ref{fig-deterministicnews}, we plot on the left the intensity associated with the events for each \gls{cluster} on the real-time axis for $r=2$. We see that the two most populated \glspl{cluster} follow precise dynamics. We added on the right side of the plot the temporal kernel corresponding to each of these \glspl{cluster}. On the right plot, we retrieve the cause of the daily and weekly cycles observed for the largest \gls{cluster} on the left plot. The second most populated \gls{cluster} follows similar dynamics, except that it seems to be shifted by half a day on the real-time axis; the peaks are in phase opposition with the largest \gls{cluster}. It is worth noting that the Notre-Dame fire \gls{cluster} is still detected; this is due to its vocabulary being different enough from the existing \gls{cluster}'s ones to trigger its own \gls{cluster}, and the associated number of documents being consequent in a short time window. Interestingly, immediately after this \gls{cluster} emerged, the dynamics on the real-time axis also follow a decaying circadian circle over three days.

\begin{figure}
    \centering
    \includegraphics[width=\columnwidth]{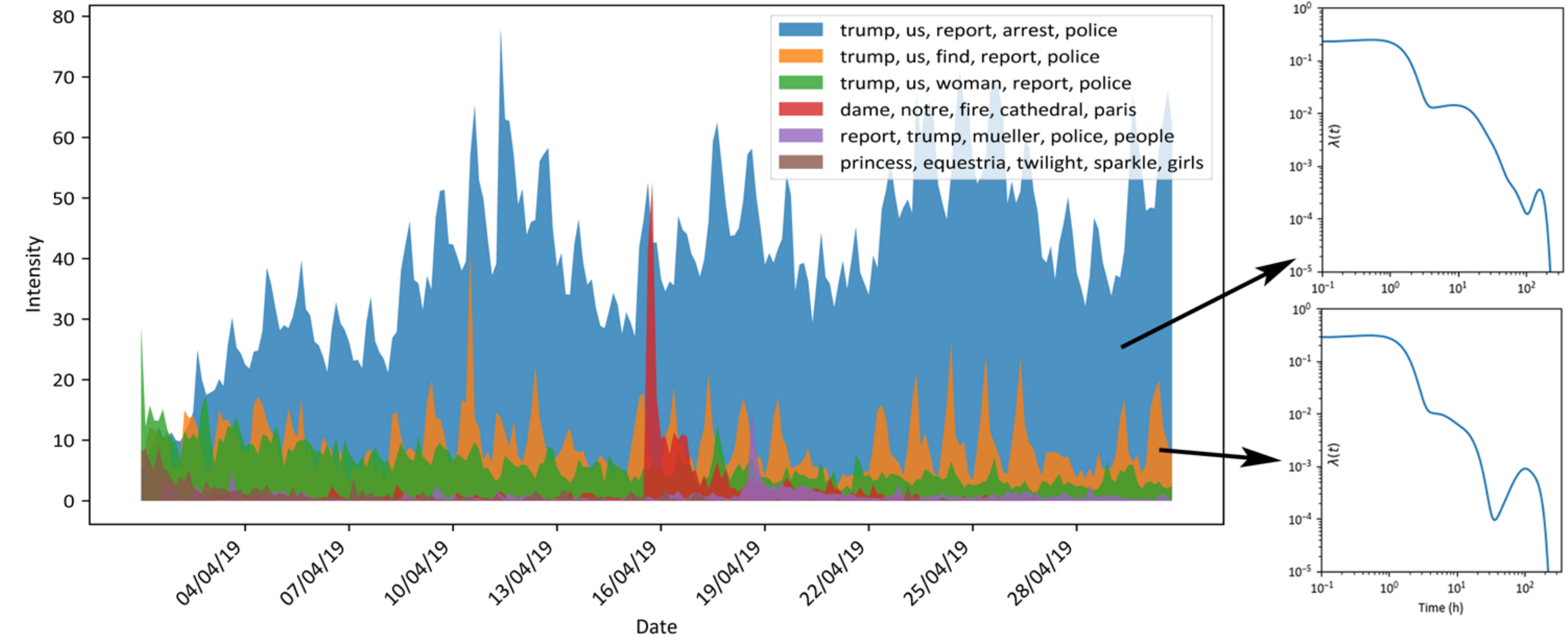}
    \caption[PDHP - Limit case of encouraging bursty events clustering]{\textbf{Limit case of encouraging bursty events clustering} --- We plot PDHP's intensity for the News dataset over the observation period for a large value of $r$. In this case, textual information plays a marginal role, and \glspl{cluster} are inferred based on the events publication dates only.}
    \label{fig-deterministicnews}
\end{figure}

\paragraph{Heuristics}
\textbf{Choosing $r$} --- We saw that in all the previous experiments, the optimal $r$ got determined from a grid-search-like strategy. We did not come up with a way to automatically infer the optimal $r$ without trying several values. 

However, we provide some heuristics regarding the tuning of $r$. 
\begin{itemize}
    \item As $r$ increases, we usually get a smaller number of inferred \glspl{cluster}. This is because considering the temporal dimension adds consistency to the language model; the temporal intensity prior for a new observation is likely to be non-null, which increases the probability of \textit{not} opening a new \gls{cluster} with respect to a model that does not consider time.
    \item As $r$ goes to infinity, we only infer one large \gls{cluster} that comprises all the observations. This is because even the slightest difference in the prior intensities leads to a deterministic \gls{cluster} choice.
\end{itemize}

Our leads to automatically determine $r$ involve computing an ad-hoc objective metric to optimize jointly with the likelihood. Given there is no gold standard for clustering in real-world processes, the choice of $r$, and therefore the choice of such metric, is left to the user. As we showed in Fig.~\ref{fig-res-decorr} and later in Fig.~\ref{fig-likReddit}, the choice of $r$ simply tunes the information on which \glspl{cluster} are based. The clustering objective is to be defined for each situation,  which is made possible by manual tuning of $r$. Such an objective could consist in minimizing the \glspl{cluster}' word distribution entropy, the standard deviation of the effective triggering kernel, or the average distance between events within a \gls{cluster}. A possible procedure for such optimization would involve a multi-arms bandit to deal with this trade-off.
 
\textbf{Number of \glspl{cluster}} --- In previous experiments, we compared our clustering results to the ground truth using the \acrshort{NMI} score. We chose this metric because it allows us to compare a different number of \glspl{cluster} together. Indeed, it is seldom the case that PDHP infers exactly the right number of \glspl{cluster}.

Typically, in our synthetic experiments with 2 ground-truth \glspl{cluster}, the number of \glspl{cluster} could differ significantly at the beginning of the algorithm (up to 10 \glspl{cluster} at once for small values of $r$). However, as the number of observations increases, smaller \glspl{cluster} are discarded as the algorithm converges toward the 2 correct \glspl{cluster}.

In real-world data, the number of \glspl{cluster} can grow very high --even more for small values of $r$. However, the number of observations each of these \glspl{cluster} comprise seems to follow a power-law distribution. Many of the \glspl{cluster} contain very few observations (5 documents or less); they are leftovers from the process as it converges towards more robust statistics. This is why in Fig.~\ref{fig-res-Srilanka} and Fig.~\ref{fig-bursty}, we restrict ourselves to the study of the largest \glspl{cluster} only.

\subsection{Conclusion}
In this section, we built the Powered Dirichlet-Hawkes process as a generalization of the Dirichlet-Hawkes process and Uniform process. We demonstrated how it improves performance on various datasets. When textual information conveys little information, or when temporal information conveys little information, and when both do, our model can correctly retrieve the original \glspl{cluster} used in the generation process with high accuracy.

A central consideration in document clustering is that there are no ``right'' \glspl{cluster}. For instance, we illustrate how textual \gls{content} and temporal dynamics can be decorrelated in real-life applications. PDHP is flexible enough to allow to choose the weight they wish to give to temporal or textual information depending on the situation; when textual and temporal \glspl{cluster} are decorrelated, the model allows to choose which of those to infer.

Many future extensions are possible for PDHP. For instance, it would be interesting to develop its hierarchical version (PHDHP) as it has already been done with HDHP for DHP \citep{Valera2017HDHP}. 
Given several recent works have been based on the regular Dirichlet-Hawkes process, it would be insightful to study how their results vary when using the Powered Hawkes-Dirichlet process instead. A study of the influence of the language model used along with PDHP would also be interesting since the text model we used here was simple on purpose (our focus being on the PDHP prior and not on the model it gets associated with). Typically, we would expect online Bayesian models that account for mutations of textual \glspl{cluster} over time (e.g., shifts in vocabulary, mutating words, etc.) to bring a significant improvement in modelling real-world systems \citep{Blei2006DynamicTopicModel,AlSumait2008OnlineLDA,Bassiou2014OnlinePLSA,Yin2014}

Regarding \gls{interaction} modelling, we are now close to our objective. With PDHP, we can:
\begin{itemize}    
    \item Consider \glspl{entity}' \gls{content}. An \gls{entity} is no more described as a unique identifier, but instead by its semantic \gls{content}. Two \glspl{entity} that convey the same information are now considered as such and clustered together as a more global \gls{entity} (i.e., a topic here).
    \item Model sparse \glspl{interaction}. \Glspl{entity} are now clustered together into temporal \glspl{cluster}. It makes it feasible to spot \gls{interaction} terms between sets of \glspl{entity}. The lifespan of \glspl{entity} is no more a problem since \glspl{cluster} can comprise \glspl{entity} spanning over extended periods, which also increases the data available for each \gls{cluster}.
    \item Model dynamic \glspl{interaction}. Each \gls{cluster} is associated with its own intensity function, which determines its effect on ulterior observations. Eventually, \glspl{entity}’ influence fades away as time goes by.
\end{itemize}
In addition, we can tune the importance given to the temporal and textual dimensions when modelling \glspl{interaction}. By varying $r$, a \gls{cluster} can self-stimulate essentially according to its \glspl{entity} publication times, or essentially according to their \gls{content}.

A major shortcoming of the proposed method is that \glspl{interaction} can only take place within a \gls{cluster}; they are self-\glspl{interaction}. As stated in the introduction, this can be interpreted as the diagonal of the \gls{interaction} matrix in Fig.~\ref{figGpesInter} in its temporal version. Ultimately, we are also interested in studying how different \glspl{cluster} of \glspl{entity} influence each other. Therefore, we propose to extend both \citep{Du2015DHP} and the PDHP (Section~\ref{PDHP}) to the multivariate case in the next section.

\section{Multivariate Powered Dirichlet-Hawkes Process -- Final model}
\label{MPDHP}

\subsection{Introduction}

\subsubsection{Multivariate extension of PDHP}
In this section, we extend the (univariate) Powered Dirichlet-Hawkes process introduced in Section~\ref{PDHP} to the multivariate case as the Multivariate Powered Dirichlet-Hawkes Process (\acrshort{MPDHP}). The various publications in a document's stream will now be able to influence each other. We detail and overcome several technical challenges that arise from considering \gls{interacting} topics, and we conduct a systematic evaluation of MPDHP on a range of synthetic datasets. As a first step, evaluation is conducted on synthetic datasets only, so that we can discuss the performances and limitations of MPDHP in a completely controlled environment. By the end of this section, we want to determine whether it is possible to use MPDHP in a real-world setting

In previous works \citep{Blei2006DynamicTopicModel,GomezRodriguez2013SurvivalAnalysis,Du2015DHP,Valera2017HDHP}, the understanding of large flows of data boils down to summarizing them into a composition of \gls{independent} groups --\glspl{cluster}. However, as discussed throughout this manuscript, the processes at stake are more complex than that, given these \glspl{cluster} are not \gls{independent} of each other; they interact.

Such \gls{interaction} is illustrated in Fig.~\ref{fig:illustration-MPDHP}. This figure has to be considered in regard to Fig.~\ref{fig-illustr-PDHP}, where each group of documents was restricted to self-\glspl{interaction} only; a given topic is assumed to only trigger observations from the same topic. In Fig.~\ref{fig:illustration-MPDHP}, it would mean that the red \gls{cluster} can only trigger observations from the red \gls{cluster}, as in Fig.~\ref{fig-illustr-PDHP}. Here instead, we consider that the red \gls{cluster} also influences the probability of triggering an observation from the blue \gls{cluster}, and conversely.

\subsubsection{Workflow}
In this section, we extend the models discussed in Section~\ref{DHP-sota} and Section~\ref{PDHP} to account for \gls{cluster} \gls{interaction} mechanisms. 
\textbf{Firstly}, we detail the challenges that arise when developing the Multivariate Powered Dirichlet-Hawkes Process (\acrshort{MPDHP}). We show that alleviating them makes it possible to achieve a linear time complexity $\mathcal{O}(N)$ (as in the original \citep{Du2015DHP} and Section~\ref{PDHP}) along with getting good clustering results. In doing so, we also relax the near-critical Hawkes process hypothesis made in \citep{Du2015DHP,Valera2017HDHP}, and correct a major flow regarding the kernel choice in previous works. 
\textbf{Secondly}, we conduct a systematic evaluation of the MPDHP in a variety of synthetic situations. Our goal is to clearly identify the limits of MPDHP regarding textual and temporal overlaps, computation time, the amount of available data, the number of co-existing \glspl{cluster}, etc. This section is intended as a technical report. By the end of it, a potential user should know in which situations MPDHP can be useful, and in which ones alternative modelling choices are required.
.

\begin{figure}
    \centering
    \includegraphics[width=\columnwidth]{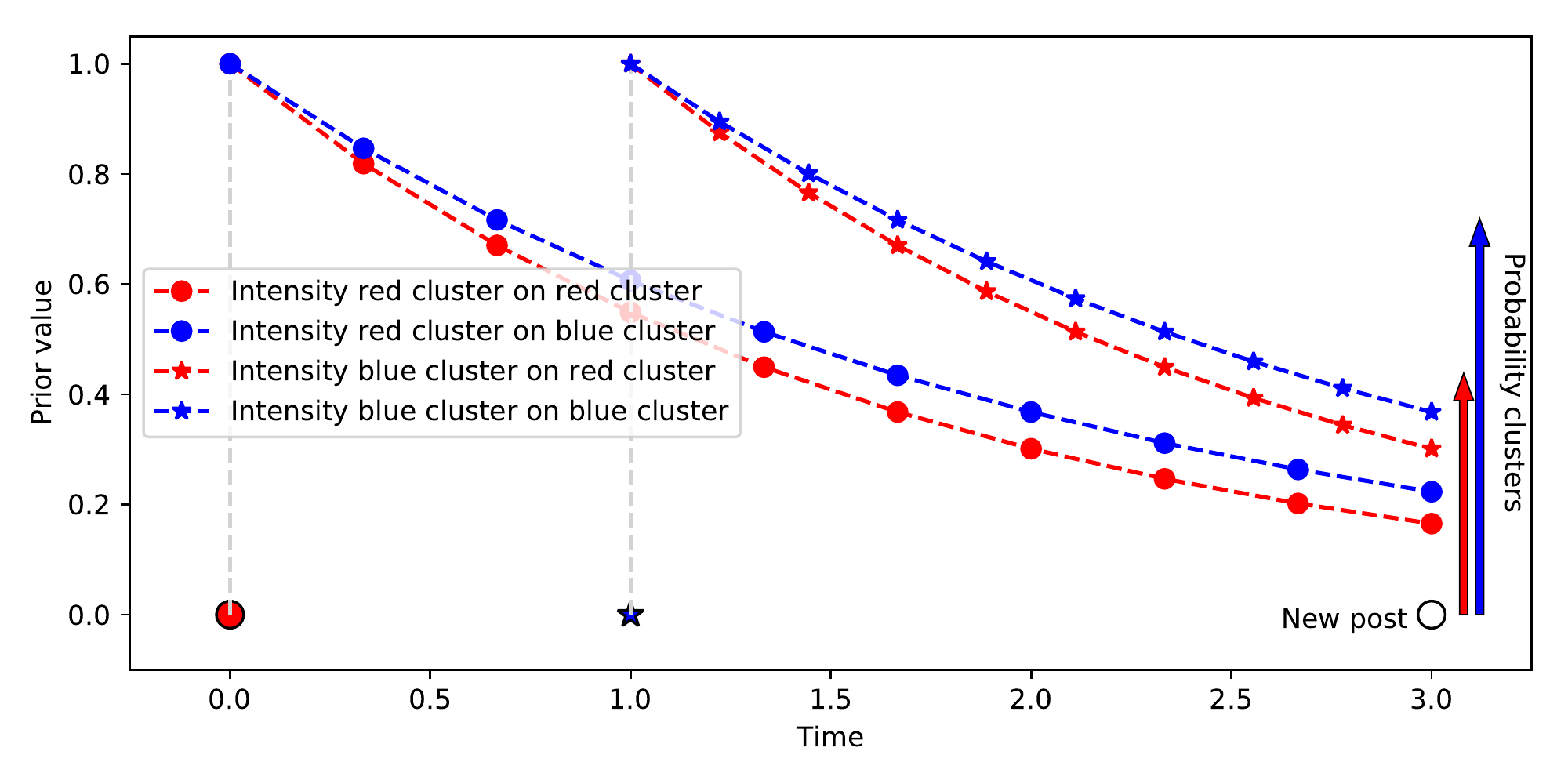}
    \caption[MPDHP - Illustration of the MPDHP]{\textbf{Illustration of the Multivariate Powered Dirichlet-Hawkes Process prior} --- A new event appears at time $t=3$ from a \gls{cluster} which is yet to be determined. The \textit{a priori} probability that this event belongs to a given \gls{cluster} $c_{red}$ depends on the sum of the red dotted intensity functions at time $t=3$. Similarly, the \textit{a priori} probability that this event belongs to a \gls{cluster} $c_{blue}$ depends on the sum of the blue dotted lines at time $t=3$. In previous existing models, this prior probability depends on each \gls{cluster} self-stimulation only.}
    \label{fig:illustration-MPDHP}
\end{figure}

\subsection{The Multivariate Powered Dirichlet-Hawkes process}
\subsubsection{Multivariate Hawkes process}
As discussed earlier, a Hawkes process is a temporal point process where the appearance of new events is conditional on the realization of previous events. It is fully characterized by an intensity function, noted $\lambda (t)$ that depends on the history events generated \textit{by this process} up to $t$, $\mathcal{H}_{<t}$. We recall that the term $\mathcal{H}_{<t}$ is implicit anytime the intensity function $\lambda$ is mentioned. A multivariate Hawkes process is an extension of the Hawkes process, where the intensity function $\lambda (t)$ depends on the history events generated \textit{by other Hawkes processes}. It means that in the definition of $\lambda (t)$, we cannot only consider the events that happened in the same \gls{cluster}, as in Eq.~\ref{eq-DHP-HawkesClus} and Eq.~\ref{eq-PDHP-HawkesClus}.

Similarly to Section~\ref{PDHP}, we associate one single Hawkes process to each \gls{cluster}. However, in (P)DHP, each of them is associated with a \textbf{univariate} Hawkes process, which depends only on the history of events comprised in this \gls{cluster}. In our case, instead, we associate each \gls{cluster} to a \textbf{multivariate} Hawkes process that depends on all the observations previous to the time being. 
Let $t_i^c$ be the time of realization of the $i^{th}$ event belonging to \gls{cluster} $c$. We write the intensity function for \gls{cluster} $c$ at all times as:
\begin{equation}
    \label{eq-multiHawkes}
    \lambda_c (t) = \sum_{t_i^{c'} < t} \vec{\alpha}^T_{c,c'} \cdot \vec{\kappa}(t-t_i^{c'})
\end{equation}
In Eq.~\ref{eq-multiHawkes}, $\vec{\alpha}_{c,c'}$ is a vector of $L$ parameters to infer, and $\vec{\kappa}(t-t_i^{c'})$ is a vector of $L$ temporal kernel functions depending only on the time difference between two events. As we will see later, $\alpha_{c,c',l}$ is readily interpreted as the influence of $c'$ on $c$ according to the $l^{th}$ entry of the temporal kernel. Once again, we consider a Gaussian \acrshort{RBF} kernel, which allows us to model a range of different intensity functions:
\begin{equation}
\label{eq-kernel}
    \kappa_l(\Delta t) = \frac{1}{\sqrt{2 \pi \sigma_l^2}}e^{-\frac{(\Delta t-\mu_l)^2}{2 \sigma_l^2}}\ \ \ \forall l\in L
\end{equation}

The log-likelihood of a multivariate Hawkes process for all observations up to a time $T$ is identical to the univariate case:
\begin{equation}
    \label{eq-lik-multiHawkes}
    \log \mathcal{L}(\alpha \vert \mathcal{H}) = \sum_c \int_{0}^T \lambda_c (t) dt + \sum_{t_i^{c}} \lambda_c (t_i^{c})
\end{equation}

\subsubsection{Multivariate Powered Dirichlet-Hawkes Process}
The Multivariate Powered Dirichlet-Hawkes Process (\acrshort{MPDHP}) arises from the merging of the Powered Dirichlet Process (Section~\ref{PDP}) and the Multivariate Hawkes Process (MHP), described in the previous paragraph. As in \citep{Du2015DHP,Valera2017HDHP} and Section~\ref{PDHP}, the counts in PDP are substituted with the intensity functions of a temporal point-process, here MHP. The \textit{a priori} probability that a new event is associated with a given \gls{cluster} no longer depends on the population of this \gls{cluster}, but on its temporal intensity at the time the new observation appears. This is illustrated in Fig.~\ref{fig:illustration-MPDHP}, where two events from two different \glspl{cluster} $c_{red}$ and $c_{blue}$ have already happened at times $t_0=0$ and $t_1=1$. A new event appears at time $t=3$. The \textit{a priori} probability that this event belongs to the \gls{cluster} $c_{red}$ depends on the sum of the intensity functions of observations at $t_0$ and $t_1$ on \gls{cluster} $c_{red}$ at time $t=3$ --sum of the red dotted lines. Similarly, the \textit{a priori} probability that this event belongs to the \gls{cluster} $c_{blue}$ depends on the sum of the blue dotted lines at time $t=3$.

Let $t_i$ be the time at which the $i^{th}$ event appears. The resulting expression reads:
\begin{equation}
    \label{eq-MPDHP-prior}
    P(C_i = c\vert t_i, r, \lambda_0, \mathcal{H}) = 
    \begin{cases}
    \frac{\lambda_c^r(t_i)}{\lambda_0 + \sum_{c'} \lambda_{c'}^r(t_i)} \text{ if c$\leq$K}\\
    \frac{\lambda_0}{\lambda_0 + \sum_{c'} \lambda_{c'}^r(t_i)} \text{ if c=K+1}
    \end{cases}
\end{equation}
\myequations{\ \ \ \ MPDHP - Multivariate Powered Dirichlet-Hawkes Process}
In Eq.~\ref{eq-MPDHP-prior}, $\lambda_c$ in defined as in Eq.~\ref{eq-multiHawkes}, and the parameter $\lambda_0$ is the equivalent of the concentration parameter described in Eq.~\ref{eq-PowCRP}. Taking back the illustration in Fig.~\ref{fig:illustration-MPDHP}, this parameter corresponds to a time-independent intensity function. It has a chance to get chosen typically when the other intensity functions are below it (meaning they do not manage to explain the dynamic aspect of a new event). In this case, a new topic is opened, and gets associated with its own intensity function.

\subsubsection{Language model}
Similarly to what has been done in the previous sections and in \citep{Du2015DHP}, the MPDHP must be associated with a Bayesian model given it is a prior on sequential data. Since we study applications on textual data, we choose to side the MPDHP prior with the same Dirichlet-Multinomial language model as in previous sections. We recall the likelihood of the $i^{th}$ document belonging to \gls{cluster} $c$ reads (see Section.~\ref{sota-DHP-text}):
\begin{equation}
    \begin{split}
        \label{eq-MPDHP-text}
        \mathcal{L}(C_i=c \vert N_{<i,c}, n_i, \theta_0) &= P(n_i \vert C_i=c, N_{<i,c}, \theta_0)\\ 
        &=\frac{\Gamma(N_c+\theta_0)}{\Gamma(N_c+n_i+\theta_0)} \prod_v \frac{\Gamma(N_{c,v} + n_{i,v} + \theta_{0,v})}{\Gamma(N_{c,v}+\theta_0)}
    \end{split}
\end{equation}
where $N_c$ is the total number of words in \gls{cluster} $c$ from observations previous to $i$, $n_i$ is the total number of words in document $i$, $N_{c,v}$ the count of word $v$ in \gls{cluster} $c$, $n_{i,v}$ the count of word $v$ in document $i$ and $\theta_0 = \sum_v \theta_{0,v}$.

\subsection{Implementation}
\subsubsection{Base algorithm}
The Sequential Monte Carlo (\acrshort{SMC}) algorithm used for the optimization has already been described in Section~\ref{PDHP-SMC} and in Fig.~\ref{fig-SchemaSMC}. We briefly review it in this section, before discussing the challenges that arise when using it in the multivariate case.

\paragraph{SMC algorithm}
The goal of the SMC algorithm is to jointly infer textual documents' \glspl{cluster} and the dynamics associated with them. It runs as follows. First, the algorithm computes each \gls{cluster}'s posterior probability for a new observation by multiplying the temporal prior on \gls{cluster} allocation (see Eq.~\ref{eq-MPDHP-prior}, illustrated Fig.~\ref{fig:illustration-MPDHP}) with the textual likelihood (see Eq.~\ref{eq-MPDHP-text}). It results in an array of $K+1$ probabilities, where $K$ is the number of non-empty \glspl{cluster}. A \gls{cluster} label is then sampled from this probabilities vector. 
If the empty $(K+1)^{th}$ \gls{cluster} is chosen, the new observation is added to this \gls{cluster}, and its dynamics are randomly initialized (i.e., a $(K+1)^{th}$ row and a $(K+1)^{th}$ column are added to the parameters matrix $\alpha$). If a non-empty \gls{cluster} is chosen, its dynamics are updated by maximizing the new likelihood Eq.~\ref{eq-lik-multiHawkes}. The process then goes on to the next observation.

This routine is repeated $N_{part}$ times in parallel. Each parallel run is referred to as a \textit{particle}. Each particle keeps track of a series of \gls{cluster} allocation hypotheses. After an observation has been processed, we compute the particles likelihood given their respective \gls{cluster} allocations hypotheses. Particles that have a likelihood relative to the other particles' one below a given threshold $\omega_{thres}$ are discarded and replaced by a more plausible existing particle.

\paragraph{Sampling the temporal kernel}
The parameters $\alpha$ are inferred using a sampling procedure. A number $N_{sample}$ of pre-computed vectors are drawn from a Dirichlet distribution with probability $P(\alpha \vert \alpha_0)$, with $\alpha_0$ a concentration parameter. As the SMC algorithm runs, within each existing \gls{cluster}, each of these candidate vectors is associated with a likelihood computed from Eq.~\ref{eq-lik-multiHawkes}, noted $P(\mathcal{H} \vert \alpha)$, where $\mathcal{H}$ represents the data. The sampling procedure returns the average of each of the $N_{sample}$ pre-computed $\alpha$, weighted by the posterior distribution associated with them ${P(\alpha \vert \mathcal{H}) \propto P(\mathcal{H} \vert \alpha)P(\alpha \vert \alpha_0)}$. The so-returned matrix is guaranteed to be a good statistical approximation of the optimal matrix, provided the number of sample matrices $N_{sample}$ is large enough.

\paragraph{Limits}
This algorithm described here works well for the univariate case but fails for the multivariate case. In particular, updating the multivariate intensity function of each \gls{cluster} requires knowing the number of already existing \glspl{cluster}, which vary over time. Therefore, we cannot pre-compute the sample matrices in advance --they must be updated as the algorithm runs to account for the right number of non-empty \glspl{cluster}. Moreover, the number of parameters to estimate evolves linearly with the number of active \glspl{cluster} $K$, instead of remaining constant as in DHP and variants \citep{Du2015DHP,Valera2017HDHP}. Because the number of parameters is not constant anymore, their candidate values cannot be sampled from a Dirichlet distribution anymore. In the following, we review these challenges and propose our solutions to overcome them. Eventually, we manage to preserve a constant time complexity for each observation.

\subsubsection{Optimization challenges}

\paragraph{Updating the triggering kernels}
In the univariate case \citep{Du2015DHP,Valera2017HDHP} and Section~\ref{PDHP}, the coefficients $\alpha_{c} \in \mathbb{R}^L$ are sampled from a collection of existing sample vectors computed at the beginning of the algorithm (where $L$ is the size of the kernel vector). However, we must now infer a matrix instead. We recall that matrix $\alpha_c$ represents the weights given to the temporal kernel vector of every \gls{cluster} influence on $c$ --see Eq.~\ref{eq-multiHawkes}. The likelihood Eq.~\ref{eq-lik-multiHawkes} can be updated incrementally for each sample matrix. A given \gls{cluster} $c$ has a likelihood value associated with each of those $N_{sample}$ sample matrices, which represents how fit one sample matrix is to explain one \gls{cluster}'s dynamics. The final value of the parameters matrix is then computed by sampling, simply averaging the sample matrices weighted by their likelihood for a given \gls{cluster} times the prior probability of these vectors being drawn in the first place.

However, such sampling was possible in the univariate case, where each sample matrix was in fact a vector of fixed length. In our case, because Hawkes processes are multivariate, each entry $\alpha_{c} \in \mathbb{R}^{K \times L}$ is now a matrix. Moreover, the number of existing \glspl{cluster} $K$ increases over time and can grow large. Each time a new \gls{cluster} is added to the computation, a row is appended to the $\alpha_c$ matrix --it accounts for the influence of this new \gls{cluster} regarding $c$.

Our solution is that some events can be discarded from the computation so that some old \glspl{cluster} can also be discarded. \Glspl{cluster} whose last observation has exceeded a certain age has a near-zero chance to get sampled once again. It means these \glspl{cluster}' contribution to the likelihood Eq.~\ref{eq-lik-multiHawkes} is fixed. Therefore, they do not have a role in the computation of the parameters matrix $\alpha_c$. The row corresponding to each of these \glspl{cluster} in the parameters matrix can then be discarded in every sample matrix. Put differently, the last sampled value for their influence on $c$ will remain unchanged for the remaining of the algorithm. The dimension of $\alpha_c$ only depends on the number of \textit{active} \glspl{cluster}, whose intensity function has not faded to zero. For a given dataset, the number of active \glspl{cluster} typically fluctuates around a constant value, making one iteration running in constant time $\mathcal{O}(1)$.

\paragraph{A beta prior on parameters}
Another problem inherent to the proposed multivariate modelling is the prior assumption on sample vectors. In \citep{Du2015DHP,Valera2017HDHP}, each sample vector is sampled from a Dirichlet distribution. This choice is to infer Hawkes processes that are nearly unstable: the spectral radius of their temporal kernel function $\lambda_{c}(t)$ is close to 1. However in our case, such an assumption is not possible because the size of each sample matrix can vary as the number of active \glspl{cluster} evolve. Drawing one Dirichlet vector of size $L$ for each entry $\alpha_{c,c'}$ would force the spectral radius of $\alpha_c$ to equal $K$, which transcribes a highly unstable Hawkes process. Our solution is to consider the parameters as completely \gls{independent} of each other. Each entry of the matrix $\alpha$ is drawn from an \gls{independent} $\beta$ distribution of parameter $\beta_0$. In this way, we make no assumption on the spectral radius of the Hawkes process, and sample rows/columns corresponding to new \glspl{cluster} can be generated one after the other.

\paragraph{On the temporal concentration parameter $\lambda_0$}
\label{implem-discussionlambda0}
While not specifically related to the implementation of the multivariate case, we discuss in this paragraph an important consideration when designing DHP-based models. In most recently published works on the topic \citep{Du2015DHP,Valera2017HDHP} and in the experiments conducted in Section~\ref{PDHP}, inference on real-world processes is done using an \acrshort{RBF} temporal kernel. It means that time is paved with Gaussian functions centered at various points in time; the parameter $\alpha$ in DHP-based models accounts for the weights given to each of these Gaussian functions.

In these works, the kernel is chosen so that it accounts for different time scales by centering Gaussian functions on unevenly spaced points in time. The standard deviation of each of these entries varies to account for larger time ranges. However, all values of a Gaussian function are small when their standard deviation is large, for normalization reasons --the maximum value of a Gaussian function whose standard deviation is $\sigma$ is $\frac{1}{\sqrt{2 \pi \sigma^2}}$. 

In the SMC algorithm, this RBF kernel is evaluated at a single point in time and confronted with the temporal concentration parameter $\lambda_0$ (see Eq.\ref{eq-MPDHP-prior}) to determine whether to open a new \gls{cluster}. In \citep{Du2015DHP}, such values are compared to $\lambda_0$ constant in time. It means that, mechanically, these methods cannot detect observations triggered by such Gaussian functions as their value is systematically lower than $\lambda_0$ --typically at long time ranges in \citep{Du2015DHP,Valera2017HDHP}, which can be seen from these articles' kernel plots that fade as time goes. We explicitly illustrate how this leads to a bias in the modelling in Fig.~\ref{fig:illustr-pb-lamb0}.

\begin{figure}
    \centering
    \includegraphics[width=\textwidth]{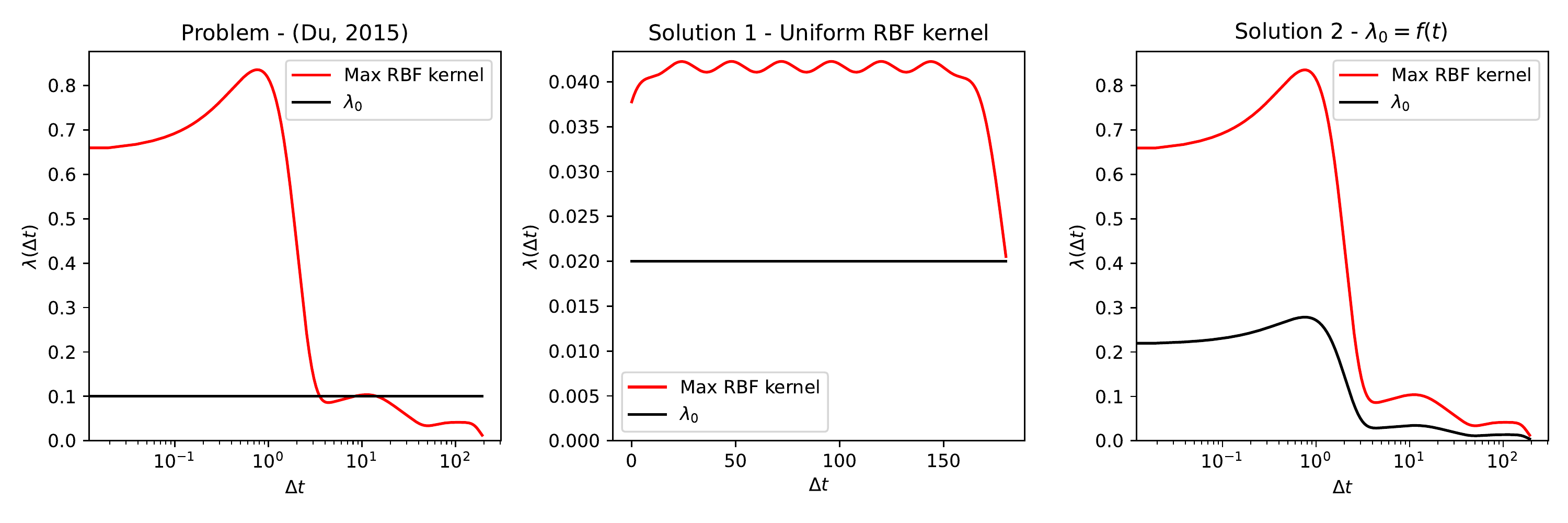}
    \caption[MPDHP - Choosing the concentration parameter]{\textbf{Choosing the right temporal concentration parameter $\lambda_0$} --- The choice of the temporal concentration parameter $\lambda_0$ can lead to bias. \textbf{(Left)} The problem with its choice in \citep{Du2015DHP} is that events happening at large time ranges are likely to go undetected, as the Hawkes intensity at these ranges cannot be larger than $\lambda_0$. \textbf{(Middle)} A first solution consists in paving the space with evenly spaced Gaussian functions that all share the same standard deviation. \textbf{(Right)} A second solution is to make $\lambda_0$ a function of time so that its ratio with the temporal kernel remains constant.}
    \label{fig:illustr-pb-lamb0}
\end{figure}

Consider for instance the \acrshort{RBF} kernel used in \citep{Du2015DHP,Valera2017HDHP}, plotted in Fig.~\ref{fig:illustr-pb-lamb0}, with the Gaussian means equal to 0.5, 1, 8, 12, 24, 48, 72, 96, 120,
144 and 168 hours, and the corresponding deviations equal to 1, 1, 8, 12, 12, 24, 24, 24, 24, 24, and 24 hours. The authors used $\lambda_0=0.01$. For the last entry of their RBF kernel, the maximum value of the Gaussian function $\mathcal{G}(\mu=168;\sigma=24)$ is about $3.10^{-4}$, which is much smaller than $\lambda_0=1.10^{-2}$. It means that even for a \gls{cluster} whose intensity function only acts at long-ranges, the chances of spotting events triggered by such \glspl{cluster} are about 3\%. This makes the models presented in \citep{Du2015DHP,Valera2017HDHP} unfit to spot long-range \glspl{interaction}.

There are two ways to overcome this problem (Fig.~\ref{fig:illustr-pb-lamb0}-Left) so that $\lambda_0$ can be consistently confronted to the \glspl{cluster}' temporal kernels (see Eq.~\ref{eq-multiHawkes}):
\begin{itemize}
    \item Fig.~\ref{fig:illustr-pb-lamb0}-Middle -- To consider an RBF kernel whose Gaussian functions all share the same deviation, while keeping $\lambda_0$ constant. We choose this solution in the follow-up experimental section.
    \item Fig.~\ref{fig:illustr-pb-lamb0}-Right -- To consider a $\lambda_0$ that can vary in time according to the maximum value of the RBF kernel at different time points --which depends on their standard deviation.
\end{itemize}

\subsection{Experiments}
\subsubsection{Setup}
We now design a series of experiments to explore the possible use domain for the Multivariate Powered Dirichlet-Hawkes Process.
We list the parameters we consider in our experiments. When a parameter does not explicitly vary, it takes a default value given between parentheses. These parameters are: the textual overlap (0) and the temporal overlap (0) discussed further in the text,
the temporal concentration parameter $\lambda_0$ (0.01), the strength of temporal dependence $r$ (1), the number of synthetically generated \glspl{cluster} $K$ (2), the number of words associated with each document $n_{words}$ (20), the number of particles $N_{part}$ (10) and the number of sample matrices used for sampling $\alpha$, noted $N_{sample}$ (2,000). For the detail of these parameters, please refer to Eq.~\ref{eq-MPDHP-prior}. 

Note that the overlap $o(f_1,...,f_N)$ between $N$ functions is defined as the sum over each function $f_i$ of its intersecting area with the largest of the $N-1$ other functions, divided by the sum of each function's total area. This value is bounded between 0 (perfectly separated functions) and 1 (identical functions). Mathematically:
\begin{equation}
    \label{eq:overlap}
    o(f_1,...,f_N) = \frac{1}{\sum_i \int_{\mathbb{R}} f_i(x) dx} \sum_i \int_{\mathbb{R}} min(f_i(x), max(\{f_j(x)\}_{j \neq i})) dx
\end{equation}

For each combination of parameters considered, we generate 10 different datasets. In all datasets, we consider a fixed size vocabulary $V=1000$ for each \gls{cluster}. All datasets are made of 5,000 observations. Observations for each \gls{cluster} $c$ are generated using an \acrshort{RBF} temporal kernel $\vec{\kappa} (t)$ weighted by a parameter matrix $\alpha_c$. Explicitly, we set $\vec{\kappa} (t) = \left[ \mathcal{G}(3;0.5) ; \mathcal{G}(7;0.5) ; \mathcal{G}(11;0.5) \right]$ where $\mathcal{G}(\mu;\sigma)$ is a Gaussian function of mean $\mu$ and deviation $\sigma$ --following the discussion raised in Section~\ref{implem-discussionlambda0}.
We note $L=3$ the number of entries of $\vec{\kappa}$. The inferred entries of $\alpha$ determine the amplitude (or weight) of each Gaussian kernel function.

The generation process is as follows. First, we draw a random matrix $\alpha \in \mathbb{R}^{K \times L}$ and normalize it so that its spectral radius equals 1 --near unstable Hawkes process. We repeat this process until we obtain the wanted temporal overlap.\footnote{Note that the overlap as defined here is different from the one used in Section.~\ref{PDHP}. In the latter, we considered the overlap of the intensity function plot on the real-time axis. Here instead we consider the overlap between the kernel intensity functions. This is because the temporal overlap as defined in Section~\ref{PDHP} is always close to 1 in the multivariate case, because different \glspl{cluster}' events strongly interact with each other.}
Then, we simulate the multivariate Hawkes process using the triggering kernels $\vec{\alpha} \cdot \vec{\kappa}(t)$, where $\vec{\kappa} (t)$ is the RBF kernel as defined earlier. Given the Hawkes process is multivariate, each event is associated with its class it has been generated from among $K$ possible classes. For each so generated event, we draw $n_{words}$ words from a vocabulary of size $V$. The vocabularies are drawn from a multinomial distribution and shifted over this distribution so that they overlap to a given extent (see Eq.\ref{eq:overlap}).

\subsubsection{Baselines}
We compare our approach to 3 closely related baselines.
\begin{itemize}
    \item \textbf{Dirichlet-Hawkes process (\acrshort{DHP})} \citep{Du2015DHP} -- In this model, \glspl{cluster} can only self-replicate. It means that the intensity function of a \gls{cluster} $c$ Eq.~\ref{eq-multiHawkes} only considers past events that happened in the same \gls{cluster} $c$.
    \item \textbf{Dirichlet process (\acrshort{DP})} -- This prior is standard in clustering problems. It corresponds to a special case of Eq.~\ref{eq-PowCRP} where $r=1$. It assumes that the prior probability for an observation to belong to a \gls{cluster} depends linearly on the population of this \gls{cluster}.
    \item \textbf{Uniform process (\acrshort{UP})} \citep{Wallach2010UnifP} -- This prior corresponds to a special case of Eq.~\ref{eq-PowCRP} where $r=0$. It assumes that the prior probability for an observation to belong to a \gls{cluster} does not depend on any information about this \gls{cluster} (neither population nor dynamics).
\end{itemize}

As in previous experiments, we evaluate our results in terms of normalized mutual information (\acrshort{NMI}) score. We recall that this metric is standard when evaluating non-parametric clustering models. It compares two \gls{cluster} partitions (i.e., the inferred and the ground truth ones); it is bounded between 0 (each true \gls{cluster} is represented to the same extent in each of the inferred ones) and 1 (each inferred partition comprises 100\% of one true \gls{cluster}).

\subsubsection{Results}
\begin{figure}[h]
    \centering
    \includegraphics[width=\textwidth]{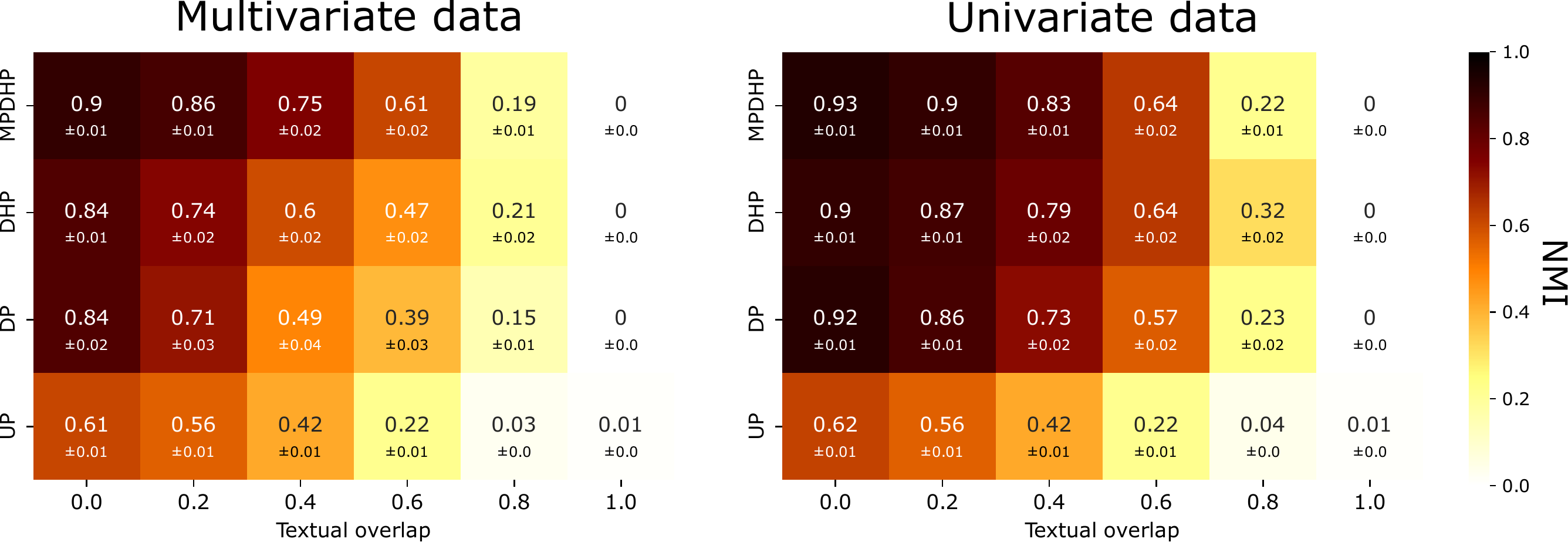}
    \caption[MPDHP - Numerical results on synthetic datasets]{\textbf{Experimental results on synthetic data} --- MPDHP consistently outperforms other baselines by considering both textual information and temporal information.}
    \label{fig:mainres}
\end{figure}
\paragraph{MPDHP outperforms state-of-the-art} %
In Fig.~\ref{fig:mainres}, we plot our results for datasets that has been generated using a Multivariate Hawkes process (\glspl{cluster} have an influence on each other) and a Univariate Hawkes process (\glspl{cluster} can only influence themselves). We compare MPDHP to the proposed baselines for various values of textual overlap. We draw the following conclusions:
\begin{itemize}
    \item \textbf{MPDHP} systematically outperforms the proposed baselines on multivariate data --when \glspl{cluster} interact with each other. Considering that \glspl{cluster} interact with each other improves our description of the datasets. 
    \item \textbf{MPDHP} performs at least as good as PDHP on univariate data --when \glspl{cluster} can only self-stimulate. The complexity induced by MPDHP does not make it unfit for simpler tasks.
    \item \textbf{PDP} performs well for small textual overlaps, but rapidly fails when the textual overlap increases. This is expected since only the textual information is considered by the PDP. It highlights the importance of also considering the temporal information.
    \item \textbf{PDHP} performs better than MPDHP when the textual overlap is large (textual overlap of 0.8). This is due to the increased complexity of MPDHP over PDHP. In challenging situations such as this one, a simpler model takes fewer efforts to overcome initialization mistakes, as there are fewer parameters to put back on the right track.
    \item \textbf{All models} fail when the textual overlap is complete; \glspl{cluster} cannot be inferred from temporal information only.
\end{itemize}

\paragraph{Uninformative textual \gls{content} and entangled dynamics} %
In Fig.~\ref{fig:XP1}, we plot the results of MPDHP for different values of textual and temporal overlap. Textual overlap is defined as in Eq.~\ref{eq:overlap}. According to Eq.~\ref{eq-multiHawkes}, the influence kernel of \gls{cluster} $c'$ on \gls{cluster} $c$ can be written $\vec{\alpha}_{c,c'} \cdot \vec{\kappa} (t)$. For each \gls{cluster} $c$, we generate values of $\alpha_{c,c'}$ so that the overlap between all the functions in the set $\{ \vec{\alpha}_{c,c'} \cdot \vec{\kappa} (t) \}_{c'}$ equals a given value. The idea is to evaluate whether MPDHP is robust when \glspl{cluster} have similar dynamics.

Overall, we see that when the textual overlap is small, MPDHP yields good results independently from the temporal overlap. It means that in this case, the textual \gls{content} is enough to differentiate \glspl{cluster} despite their dynamics being similar. However, as textual \gls{content} gets less informative (textual overlap $\geq$ 0.6), results are better when the temporal overlap is low. In these cases, textual information is not enough and MPDHP relies more on temporal data. Overall, MPDHP handles challenging cases provided either textual or temporal information is informative enough -- for instance temporal overlap of 0 and textual overlap of 0.7, or temporal overlap of 1 and textual overlap of 0.4. It fails when both are uninformative -- for instance, temporal and textual overlaps of 1.
\begin{figure}[h]
    \centering
    \includegraphics[width=0.55\textwidth]{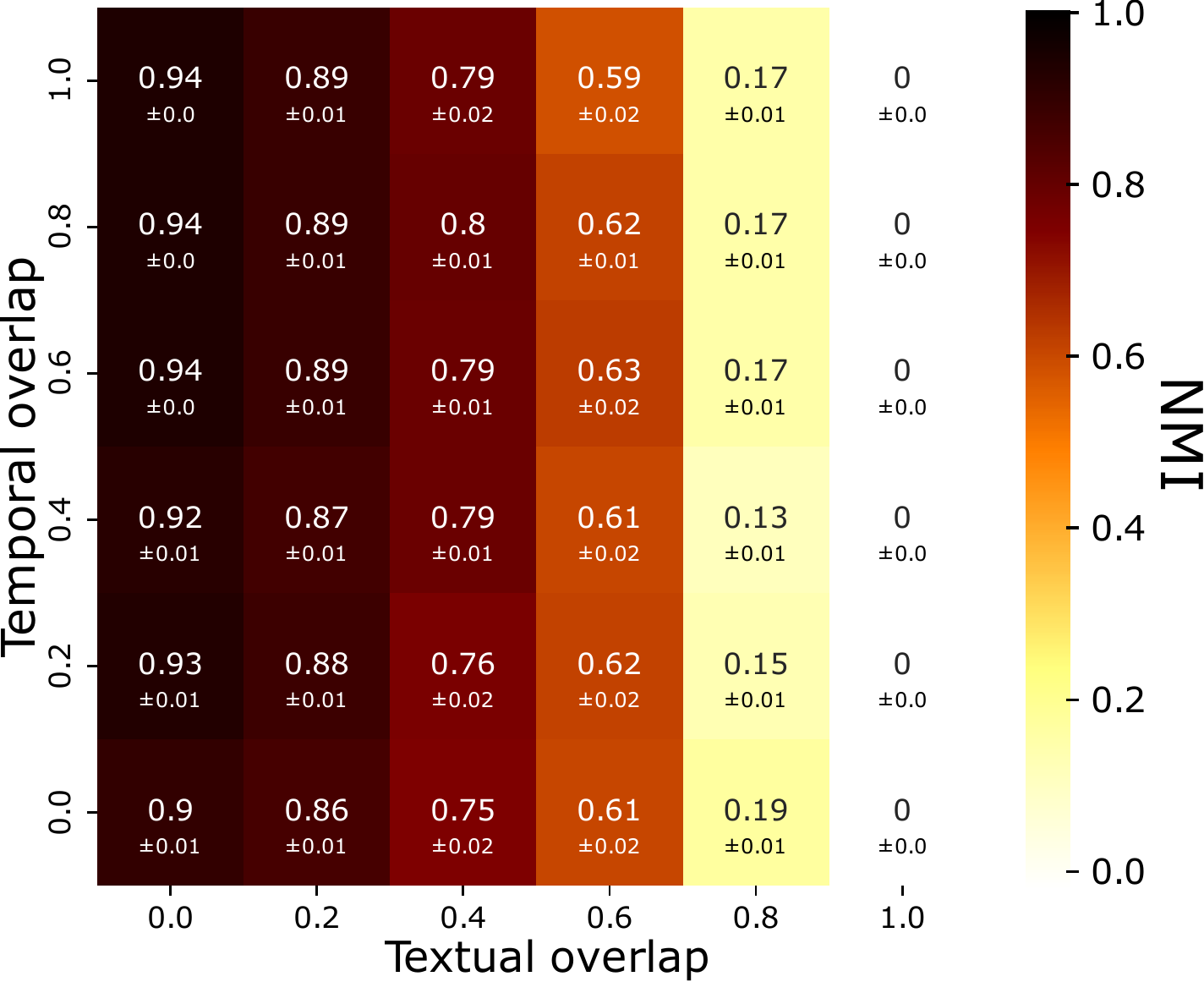}
    \caption[MPDHP - MPDHP handles scarce textual or temporal information]{\textbf{MPDHP handles scarce textual or temporal information} --- MPDHP handles challenging cases provided either textual of temporal information is informative enough (temporal overlap of 0 and textual overlap of 0.7; temporal overlap of 1 and textual overlap of 0.4) and fails when both are uninformative (overlaps of 1).}
    \label{fig:XP1}
\end{figure}

\paragraph{Highly \gls{interacting} processes} %
Next, we assess whether MPDHP works when a large number of \glspl{cluster} coexist simultaneously. The rate at which new \glspl{cluster} get opened is mainly controlled by the $\lambda_0$ hyperparameter (see Eq.~\ref{eq-MPDHP-prior}), which we vary to see whether MPDHP is robust against it. In Fig.~\ref{fig:XP2}, we plot the performances of MPDHP according to these two parameters. We can draw two conclusions:
\begin{itemize}
    \item MPDHP can handle a large number of coexisting \glspl{cluster} and still correctly identify to which one each document belongs.
    \item MPDHP is robust versus variations of $\lambda_0$. In this case, results are similar for $\lambda_0$ varying over 5 orders of magnitude. It means MPDHP does not have to be fine-tuned according to the number of expected \glspl{cluster} in cases where this number is not known in advance.
\end{itemize}
\begin{figure}[h]
    \centering
    \includegraphics[width=0.8\textwidth]{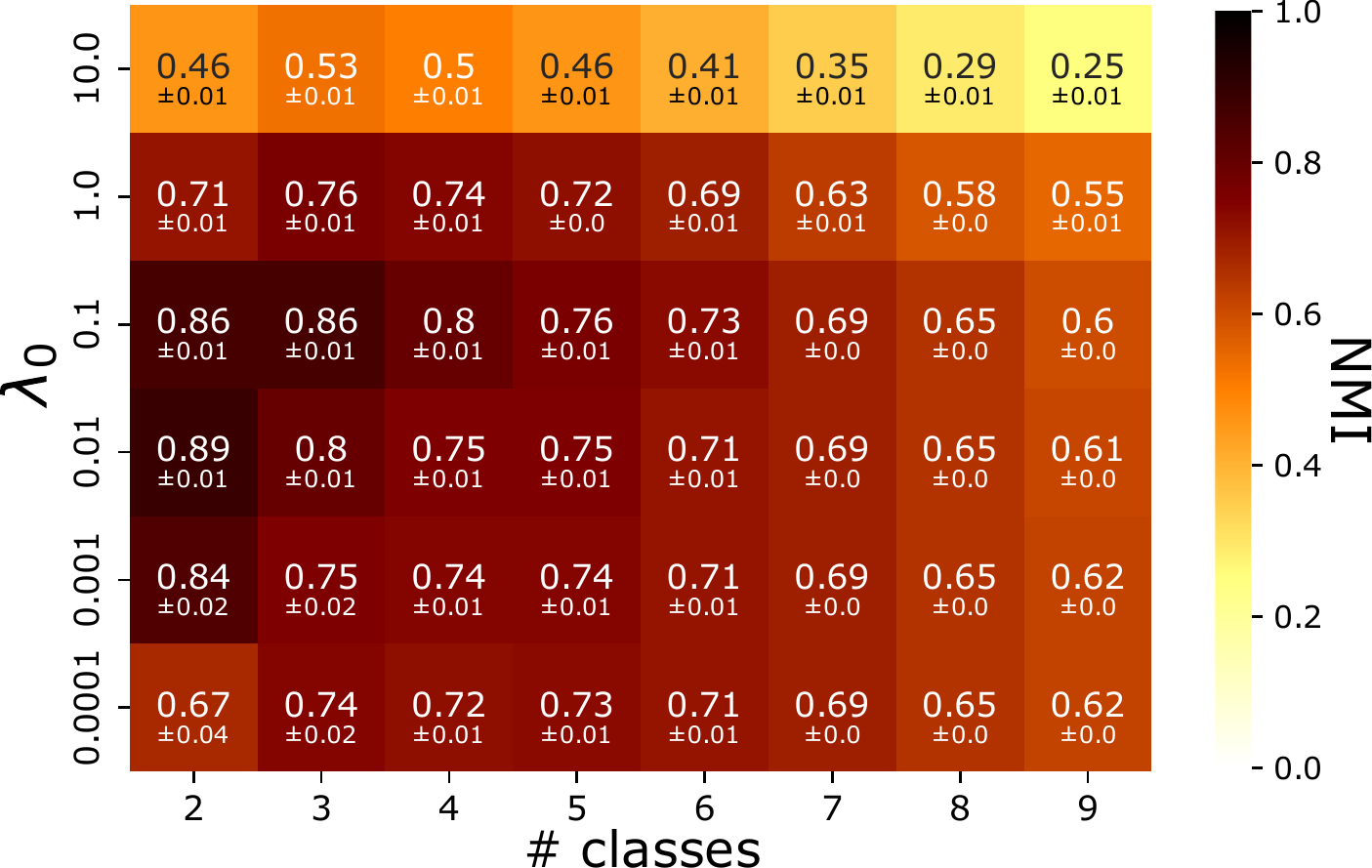}
    \caption[MPDHP - MPDHP handles many coexisting \glspl{cluster}]{\textbf{MPDHP can handle a large number of coexisting \glspl{cluster}} --- MPDHP yields good results when a large number of \glspl{cluster} coexist simultaneously. It is also robust against variations of $\lambda_0$ over 5 orders of magnitude.}
    \label{fig:XP2}
\end{figure}

\paragraph{Handling scarce textual information} %
In this paragraph, we determine how much data should be provided to MPDHP to get satisfying results. In Fig.~\ref{fig:XP3}, we plot the performances of MPDHP with respect to the number of words generated by each observation and to the \glspl{cluster}' vocabulary overlap. MPDHP needs a fairly small number of words to yield good results over 5,000 observations. For reference, the overlap between topics can be estimated at around 0.25 (\citep{Duran2019OverlapRW}, in Spanish). Similarly, we can estimate an average of $\sim$10-20 named \glspl{entity} per Twitter post (240 characters). These results support the application of MPDHP to model real-world situations.
\begin{figure}[h]
    \centering
    \includegraphics[width=0.99\textwidth]{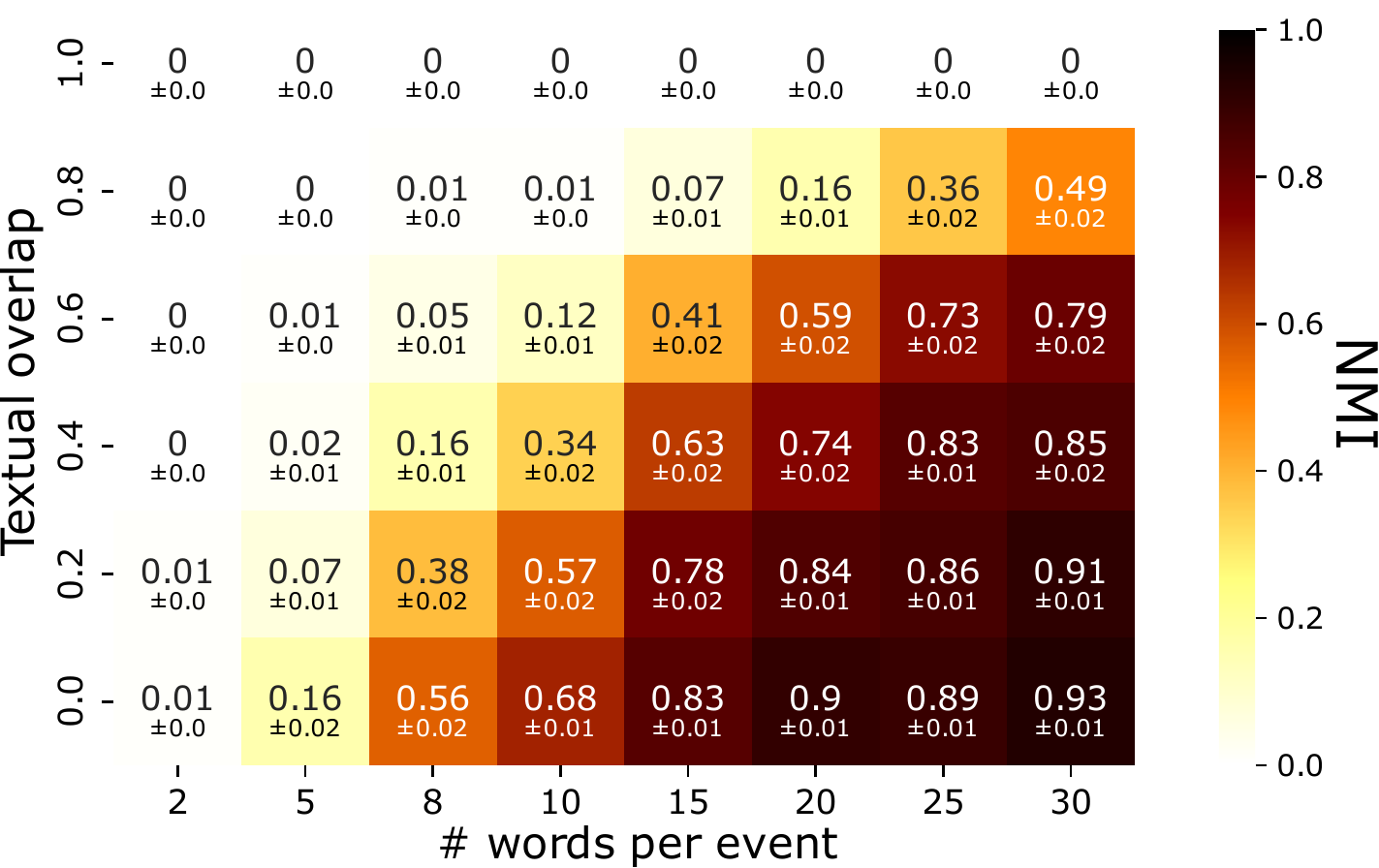}
    \caption[MPDHP - How much data to use MPDHP]{\textbf{How much data to use MPDHP} --- MPDHP performances according to the number of words associated with each event and the overlap between \glspl{cluster}' vocabulary. Overall, MPDHP needs little data. This plot provides a map of how much data must be provided to MPDHP to make it work. For reference, topic overlaps in the real world can be estimated at around 0.25 and the number of named \glspl{entity} per Twitter post around 20.}
    \label{fig:XP3}
\end{figure}

\paragraph{Computational needs} %
Finally, we investigate how much computational resources should be allocated to MPDHP's sequential Monte-Carlo (\acrshort{SMC}) inference algorithm to obtain good results. In Fig.~\ref{fig:XP6}, we plot the model's performance against the two main optimization parameters --the number of sample matrices and the number of particles. We recall that the sample matrices are used to infer each value of the kernel weights matrices for \gls{cluster} $c$, noted $\alpha_c$; the more sample matrices, the better the estimation. The number of particles represents the number of different \gls{cluster} allocations hypotheses explored by the SMC algorithm at each step; the more particles, the more hypotheses are tested simultaneously. Overall, we see that MPDHP works well with few resources. In our experiments, results do not seem to improve significantly when using more than 20 particles, and when using more than 1000 sample vectors.
\begin{figure}[h]
    \centering
    \includegraphics[width=0.99\textwidth]{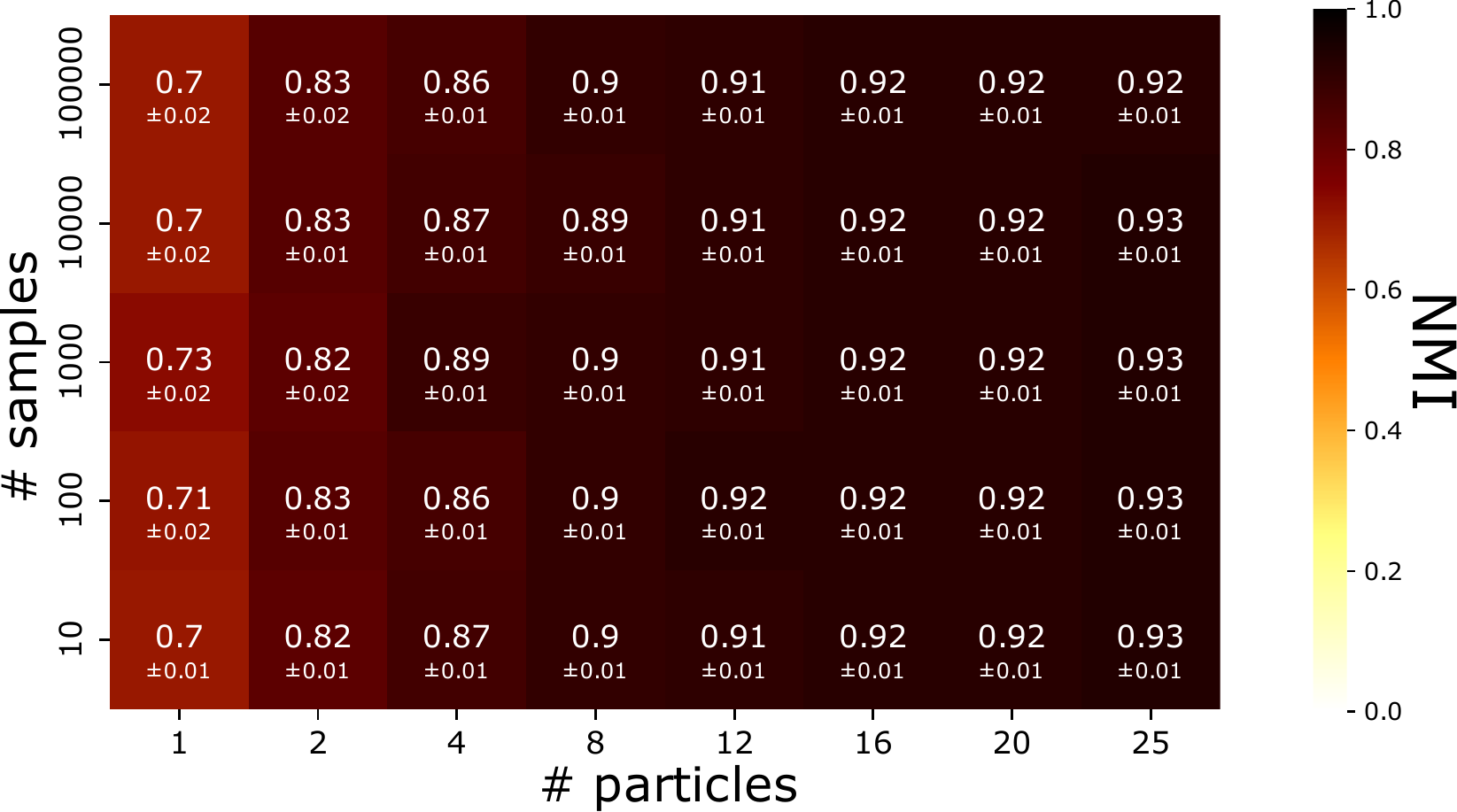}
    \caption[MPDHP - How complex should the algorithm be]{\textbf{How complex should the algorithm be} --- Performance of MPDHP using different versions of the Sequential Monte-Carlo algorithm. Here, we plot the model's performance with respect to the number of sample matrices used to estimate the kernel's weights $\alpha_c$ (see Eq.~\ref{eq-multiHawkes}) and the number of particles $N_{part}$ used for the inference. Overall, MPDHP functions well with few computational resources.}
    \label{fig:XP6}
\end{figure}

\subsection{Conclusion}
In this section, we extended existing priors (the Dirichlet-Hawkes process and the Powered Dirichlet-Hawkes process) so that they can consider multivariate Hawkes processes, resulting in the resulting Multivariate Powered Dirichlet-Hawkes process (\acrshort{MPDHP}). This new process is used as a Bayesian prior coupled to a textual model to infer \glspl{cluster} temporal \gls{interaction} network from textual data flow. Along with its derivation came several optimization challenges, that we overcome to preserve a computational time that scales linearly with the size of the dataset $\mathcal{O}(N)$.

Through systematic experiments, we tested our approach against state-of-the-art models and explored its limits by varying the parameters used for synthetic data generation. We show that MPDHP outperforms existing baselines when \glspl{cluster} are allowed to interact with each other, and performs at least as well as the PDHP baseline when \glspl{cluster} are only allowed to self-interact (which PDHP is designed to model). We show that MPDHP can handle cases where textual \gls{content} is uninformative better than other baselines. Besides, it handles cases where temporal dynamics are similar across \glspl{cluster}. We also show that MPDHP is robust against tuning of the temporal concentration parameter $\lambda_0$, which allows it to handle highly intricate processes where a lot of \glspl{cluster} coexist simultaneously. We evaluated the robustness of MPDHP against the number of words observed for each event and the overlap between various \glspl{cluster} and showed that it performs well with few textual data when vocabularies overlap is not total. Finally, we discussed the computational needs of PDHP, and show that it works correctly with minimal computational resources.

The present section was intended as a report on what can and what cannot be achieved using MPDHP. Our various results suggest that this prior can be applied in a robust way to a broad range of problems for a minimal computational cost. In particular, the results from these extensive experiments support the possibility of applying MPDHP to real-world situations.

Regarding the task at hand, it appears that MPDHP provides a robust way to model \glspl{interaction} in information \gls{spread}. The PDHP allowed to model sparse and dynamic self-\glspl{interaction} between semantic \glspl{entity}; these \glspl{interaction} can now take place between different \gls{entity} \glspl{cluster}.

\section{Case study on a real-world dataset -- Reddit news}
\label{MPDHP-Reddit}

\subsection{Introduction} 
Throughout the previous sections, we developed a plethora of models to investigate \glspl{interaction} in information \gls{spread}. A first approach underlined the necessity of considering clustering and a second one the necessity of considering the temporal dimension. To answer these conclusions, we extended a promising class of models, the Dirichlet-Hawkes process, so that it becomes possible to spot temporal \glspl{interaction} between \glspl{cluster} of textual documents in large-scale settings. The resulting Multivariate Powered Dirichlet-Hawkes Process (\acrshort{MPDHP}) answers all these constraints and yields interpretable results on \gls{interacting} processes.

As a closure of our work on \glspl{interaction} modelling in information \gls{spread}, we propose to conduct a large-scale experiment on a real-world dataset using MPDHP. In this section, we first describe and justify the choice of a large-scale real-world dataset from Reddit. In a second time, we conduct an in-depth analysis of the results yielded by MPDHP. We finally conclude on the role of \glspl{interaction} in the \gls{spread} of news headlines on Reddit.

\subsection{Dataset}
\subsubsection{Origin and raw data}
\paragraph{Why Reddit?}
Reddit is a social network platform that gathers 48M of monthly users, as of today. Reddit is composed of a galaxy of forums (or subreddits) that are often specific to a topic or an organisation. Within each forum, users can open a new post made of a title and a \gls{content} (typically textual or visual). Users can then comment on the post and get answered to.
The \glspl{content} published on the platform are therefore user-generated and highly dynamic, with an estimated 1M of new publications per day. Each post can be upvoted (score +1) or downvoted (score -1); the popularity of the post is defined as the difference between upvotes and downvotes. 

For these reasons, this corpus fits our demonstration: information emerges from a large user-generated data flow, its \glspl{content} are formatted, textual, timestamped and user-generated, and the number of daily publications makes it likely that some of them interact with each other.

\paragraph{Data}
We collected our dataset from the Pushshift Reddit repository \citep{Baumgartner2020PushshiftRedditDataset}. As its authors describe it, \textit{``Pushshift is a social media data collection, analysis, and archiving platform that since 2015 has collected Reddit data and made it available to researchers. Pushshift’s Reddit dataset is updated in real-time and includes historical data back to Reddit’s inception. [...] The Pushshift Reddit dataset makes it possible for social media researchers to reduce time spent in the data collection, cleaning, and storage phases of their projects.''}

In practice, the dataset comprises the entirety of the \gls{content} posted on Reddit up to June 2021. In particular, for each Reddit post, we can retrieve the subreddit it came from, the title of the publication, its publication date and its score.

\subsubsection{Preprocessing}
\paragraph{Selecting the news subreddits}
For the need of our study, we restrict ourselves to consider only popular English news subreddits. Namely, we select only posts from the following subreddits: inthenews, neutralnews, news, nottheonion, offbeat, open news, qualitynews, truenews, worldnews. This first routine leaves us with 867,328 headlines, which makes a total of 1,111,955 words drawn from a vocabulary of size 36,284. 

\paragraph{Cleaning the textual data}
As it is common in natural language processing, we must clean the raw text extracted from the Reddit posts so it becomes usable. To do so, we conduct the following routine:
\begin{enumerate}
    \item Remove the web addresses
    \item Put the text in lowercase
    \item Remove punctuation signs
    \item Remove extra white spaces
    \item Remove all English stopwords (imported from nltk)
    \item Remove all words whose length is lesser than 4
    \item Remove all words that appear less than 3 times in the original dataset
\end{enumerate}

\paragraph{Removing uninformative documents}
Next, we remove publications that carry lesser textual or temporal information. 

Firstly, we choose not to consider the publications that have a popularity lesser than 20 -- meaning that they received less than 20 positive votes more than negative votes. We make this choice so that we consider publications that are visible enough to have any influence on the data generation process.

Secondly, we remove the publications that comprise less than 3 words. The semantic information carried is expected to be poor and is not considered in our analysis.

\paragraph{Final dataset}
After curating the dataset in the way described above, we are left with 102,045 news headlines (one-eighth of the original data), which makes a total of 875,334 tokens (named entities, verbs, numbers, etc.) drawn from a vocabulary of size 13,241 (one-third of the original vocabulary). The characteristics of this dataset are shown in Fig.~\ref{fig:statsDSRedditafter}. 

\begin{figure}
    \centering
    \includegraphics[width=\textwidth]{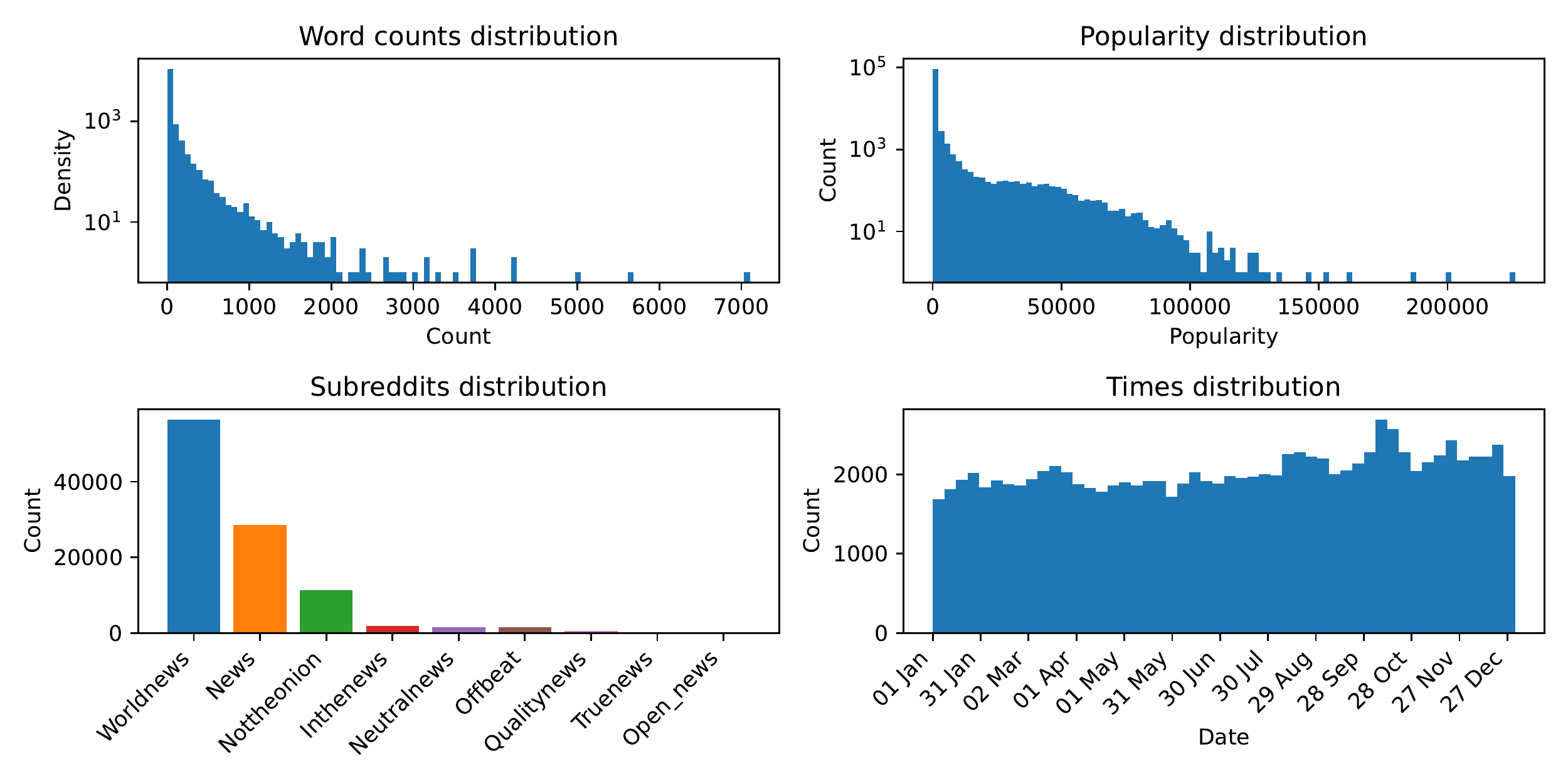}
    \caption[MPDHP x News - Characteristics of the News dataset]{\textbf{Characteristics of the News dataset} --- For $\sim$100,000 headlines and $\sim$13,000 different words (Top-Left) Distribution of the words count (Top-Right) Distribution of headlines popularity (Bottom-Left) How many headlines per subreddit (Bottom-Right) How publications \gls{spread} over time}
    \label{fig:statsDSRedditafter}
\end{figure}

\subsection{Experimental setup}
As announced earlier, we will apply the MPDHP of Section~\ref{MPDHP} to the News dataset detailed above. First, we must determine which hyper-parameters to use.

\paragraph{Temporal kernel}
We run our experiments using three different \acrshort{RBF} kernels, which should account for publication dynamics at three different timescales: minute, hour, and day. 
The ``\textbf{Minute}'' RBF kernel is made of Gaussian functions centered at the following times: $\left[ 0, 10, 20, 30, 40, 05, 60, 70, 80 \right]$ minutes; each entry shares a same standard deviation $\sigma$ of 5 minutes, and $\lambda_0=0.01$.
The ``\textbf{Hour}'' RBF kernel has Gaussians centered around $\left[ 0, 2, 4, 6, 8 \right]$ hours, with a standard deviation $\sigma$ of 1 hour, and $\lambda_0=0.001$.
The ``\textbf{Day}'' RBF kernel is centered around $\left[ 0, 1, 2, 3, 4, 5, 6 \right]$ days, with a standard deviation $\sigma$ of 0.5 days, and $\lambda_0=0.0001$.

For each of these kernels, we set the concentration parameter $\lambda_0$ so that it reaches roughly the value of one Gaussian function evaluated at $2\sigma$. It means that an event which is $2\sigma$ away from the center of the Gaussian kernel of a single observation has 50\% chances of getting associated with this Gaussian kernel entry, and 50\% chances of opening a new \gls{cluster}. It also translates that the chances to open a new \gls{cluster} given an observation right in the center of a single RBF kernel entry are roughly $\frac{1}{8}$.

This allows us to spot \glspl{interaction} that might occur at different timescales. However, one of these timescales is likely to explain the data much better than the others. If that is the case, it would mean that the dynamics at stake are not scale-free --that \glspl{interaction} do not have noteworthy influence at every temporal scale.

\paragraph{Language model}
We use the same Dirichlet-Multinomial language model as in \citep{Du2015DHP} and in Section~\ref{PDHP} and Section~\ref{MPDHP}. We expect the vocabulary to show enough variations across the range of topics discussed in the News dataset. 
We vary the concentration parameter of this model between $\theta_0=0.001$ and $\theta_0=0.01$.
The choice of these values is supported by other works using a similar model \citep{Du2015DHP} and by our own observations in Section~\ref{PDHP}. A larger value of $\theta_0$ makes the inferred \glspl{cluster} cover a broader range of document types, whereas a small value makes the inferred \glspl{cluster} more specific to a topic.

\paragraph{SMC algorithm}
Finally, we detail and justify the parameters used for running the \acrshort{SMC} algorithm. 
The procedure to update the optimal parameters matrix $\alpha$ relies on the average of a set of sample vectors weighted by these vectors' likelihood.
The number of active \glspl{cluster} can grow large, and the number of parameters scales as the square of this number. We therefore need to have enough sample matrices to guarantee a good approximation of the optimal $\alpha$ --the value recommended for 2 \glspl{cluster} in Fig.~\ref{fig:XP6} is likely to be too small for the task.

We set $N_{samples}=100000$. From our observations, the number of coexisting \glspl{cluster} can go as large as 80 coexisting ones (roughly 40,000 parameters to estimate), that is 2.5 sample values per parameter. In practice however, the number of coexisting \glspl{cluster} remains fairly low, around 10 coexisting \glspl{cluster} (roughly 1,000 parameters), allowing sampling each parameter from 100 candidate values.

Each sample parameter is drawn from an identical Beta distribution of concentration parameter $\alpha_0=2$. We set this value so that extreme values (0 and 1) for these parameters are rare, and so that the sampled parameters matrices can show great variations from one another.

Finally, we consider 8 particles for the SMC algorithm, similarly to what is done in \citep{Du2015DHP} and in Section~\ref{PDHP}, and recommended in Fig.~\ref{fig:XP6}.

\subsection{Results}
\subsubsection{Overview of the experiments}
In our experiments, we used three temporal kernels that account for dynamics at different time scales, and tested diverse values for the language model concentration parameter $\theta_0$ and for the exponent $r$. In Table~\ref{tab-diffruns}, we represent the main characteristics for each of these runs in terms of number of inferred \glspl{cluster} $K$, the average \gls{cluster} population $<N>$ (where $< \cdot >$ denotes the average), the average normalized entropy of the vocabulary of the top 20 \glspl{cluster} $S_{text}^{(20)}$, the average normalized entropy of the subreddits partition of the top 20 \glspl{cluster} $S_{sub}^{(20)}$.

We recall that the normalized entropy is defined as:
\begin{equation}
    S(\vec{x}) = -\frac{1}{\ln \vert \vec{x} \vert}\sum_{i}^{\vert \vec{x} \vert} x_i \ln x_i
\end{equation}
where $\vec{x}$ is a vector that sums to $1$ and $\vert \vec{x} \vert$ its cardinal (length). Each entry $x_i$ represents the probability of $i$. When considering counts, $\vec{x}$ can be set equal to the frequency of each observation. The entropy is normalized between 0 (minimal \gls{spread}, $\vec{x}_i = \delta_{ij} \, \forall i$) and 1 (maximal \gls{spread}, $\vec{x}_i = \frac{1}{\vert \vec{x} \vert} \, \forall i$). In our case, a low entropy $S_{text}^{(20)}$ (resp. $S_{sub}^{(20)}$) means that clusters contain documents that are concentrated around a reduced set of words (resp. of subreddits); conversely, a large entropy means that clusters do not account for documents concentrated around a specific vocabulary (resp. set of subreddits).

\input{Tables/Chapter_4/table-MPDHP-fit}

We can make several observations from Table~\ref{tab-diffruns}:
\begin{itemize}
    \item We recover the fact that textual \glspl{cluster} have a lower entropy for small values of $r$; this is because their creation is based more on textual coherence than on temporal coherence.
    \item The subreddit entropy \textit{seems} to grow with $r$, but no clear trend is visible due to large error bars. A possible interpretation is that favouring the temporal information for \gls{cluster} creation results in larger \glspl{cluster}. They would be too large to account for subreddit-specific dynamics. However, the entropy remains fairly low. The entropy of the distribution Fig.~\ref{fig:statsDSRedditafter}-top-left is equal to 0.51.
    \item The number of inferred \glspl{cluster} decreases with $r$, and their average population increases.
    \item The number of \glspl{cluster} grows large for the ``Minute'' kernel. This is because the short time range considered does not allow for \glspl{cluster} to last in time. A \gls{cluster} that does not replicate within 1h30 is forgotten.
\end{itemize}

\paragraph{Choosing a timescale}
From Table~\ref{tab-diffruns}, we see there can be a large variety of outputs to analyse, depending on the modelling choices. If we are interested in the micro-\glspl{interaction} that happen at tiny time scales, the ``Minute'' RBF kernel should be considered. However, the short time range it allows for \glspl{interaction} makes the number of \glspl{cluster} large, and the average \gls{cluster} population small. On the other hand, the ``Day'' RBF kernel spans over larger periods, which prevents discarding \glspl{cluster} too soon. For instance, it will not discard \glspl{cluster} that follow a circadian publication dynamic, unlike the ``Minute'' kernel that cannot account for such time ranges.

\paragraph{Choosing $\theta_0$}
In a Dirichlet-Multinomial textual model, such as written in Eq.~\ref{eq-likModelLg}, the hyperparameter $\theta_0$ controls the concentration of topics' vocabulary. A small value of $\theta_0$ makes it so that a new document should have an almost identical word distribution as a given \gls{cluster} to enter it; there are more chances that small discrepancies lead to opening a new \gls{cluster}. Conversely, larger values of $\theta_0$ allow documents to belong to \glspl{cluster} even if their word distribution does not fit exactly the \gls{cluster}'s \gls{content}. We tested a small value $\theta_0=0.001$ and a large value $\theta_0=0.01$, which are standard in usual text modelling \citep{Blei2003LDA,Blei2006DynamicTopicModel,Du2015DHP}. The choice of this parameter controls the level of specificity wanted for the textual \glspl{cluster}.

\paragraph{Choosing $r$}
We saw the experiments of Section~\ref{PDHP} that the choice of $r$ allows us to nudge the clustering towards textual-based clustering of temporal-based clustering. Smaller values of $r$ favour the textual \gls{content} as the main information in the creation of the \glspl{cluster}, whereas larger values of $r$ make their composition rely more on the inferred temporal dynamics. 

\begin{figure}
    \makebox[\textwidth][c]{
    \centering
    \includegraphics[width=1.15\textwidth]{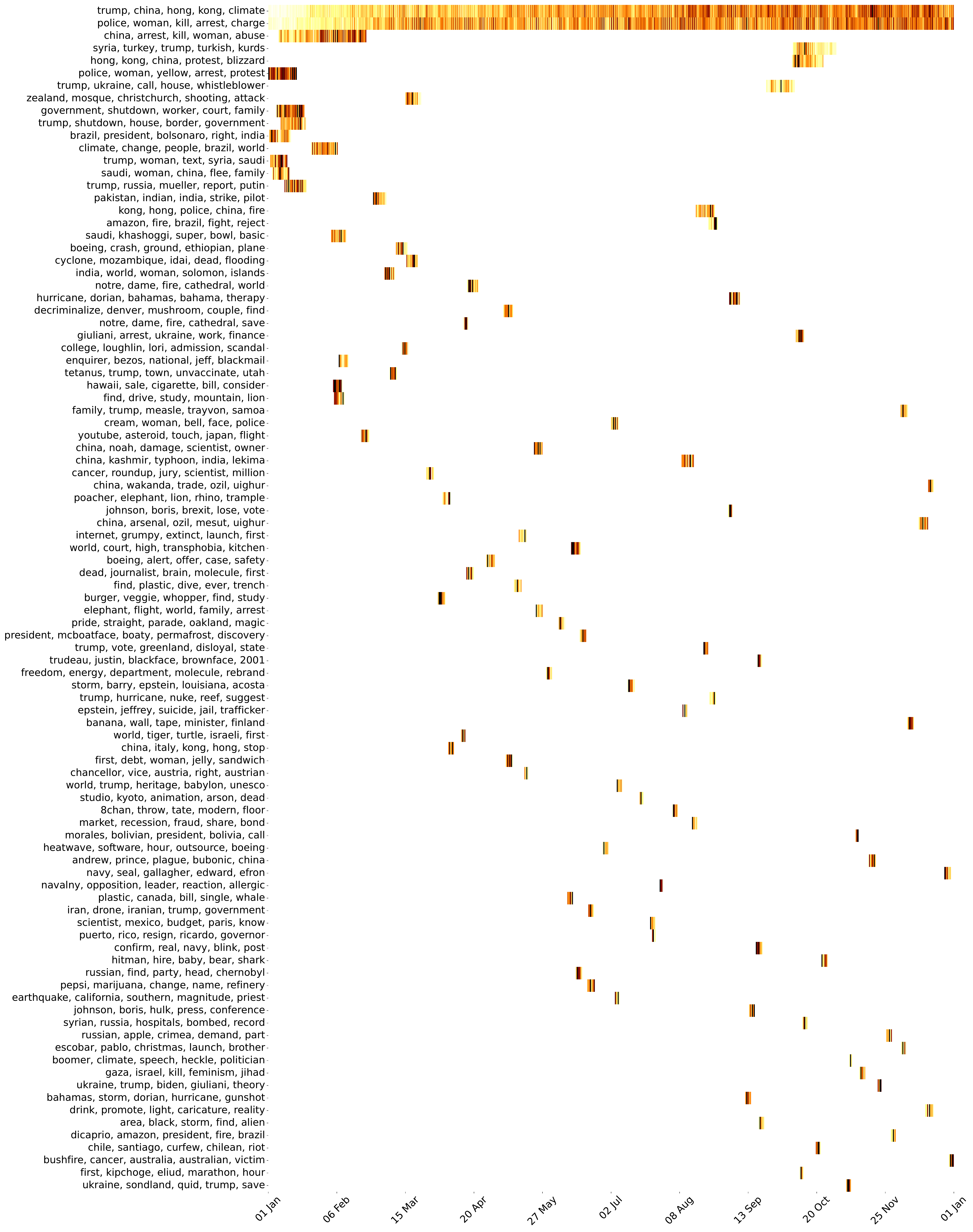}
    }    
    \caption[MPDHP x News - Timeline news 2019]{\textbf{Timeline of the top inferred \glspl{cluster} from news headlines from 01/2019 to 12/2019} --- Each line is normalized with respect to its maximum value. Each bin accounts for half a day. The darker, the more observations in the \gls{cluster} at a given time.}
    \label{fig-timeline-MPDHPNews}
\end{figure}
In our experiments, we see that smaller values of $r$ tend to increase the number of inferred \glspl{cluster}. This can be explained in the following way: the temporal concentration parameter $\lambda_0$ has been fixed so that an observation $2\sigma$ away from the center of one RBF kernel entry (that is $\sim$95\% chances not to be triggered by it) has $\sim$50\% chances to open a new \gls{cluster}. Mechanically, it is likely that $\lambda_0$ is lesser than the value of the RBF kernel most of the time. However, for smaller values of $r$, the gap between $\lambda_0$ and the temporal intensity $\lambda (t)$ fades to 0 (see Fig.~\ref{fig-illustr-PDHP}), making it more likely to open new \glspl{cluster} from the temporal perspective. In the limit $r \rightarrow 0$, we recover a Uniform Prior \citep{Wallach2010UnifP}, making the opening of new \glspl{cluster} entirely governed by the parameter $\theta_0$ --\glspl{cluster} are created and filled based on textual \gls{content} only. On the contrary, when $r$ is large, the gap between $\lambda_0$ and $\lambda (t)$ is greater, making it less likely to open new \glspl{cluster} from a temporal perspective. In the limit $r \rightarrow \infty$, the opening and filling of \glspl{cluster} is deterministically governed by the temporal information, as the smallest gap in the intensities leads to a Dirac distribution on the \gls{cluster} with the largest temporal intensity.

\subsubsection{Visualizing topics over time}
\label{MPDHPxNews-visutopics}
In Fig.~\ref{fig-timeline-MPDHPNews}, we plot the timeline of the inferred \glspl{cluster} on a real-time axis for one of our experiments (kernel ``Hour'', $\theta_0=0.01$, $r=1$) for illustration purpose. Each bin represents a half-day period.
We can make several interesting observations from this figure.

\textbf{Firstly}, two \glspl{cluster} seem to be always present. Their intensity does not follow any visible bursty dynamics. When we look at their composition, we notice that the first \gls{cluster} is made of 75\% of articles from the subreddit r/worldnews, which is +20\% from what one would expect from chance (55\% of the corpus is from r/worldnews, see Fig.~\ref{fig:statsDSRedditafter}). Similarly, the second \gls{cluster} comprises 46\% of r/news articles, which is also roughly +20\% from expected at random (28\% of the corpus is from r/news, see Fig.~\ref{fig:statsDSRedditafter}). These two \glspl{cluster} therefore significantly account for publications from either of these subreddits, independently from the textual \gls{content}. Our first intuition is that there are strong \glspl{interaction} between these subreddits. Both are general news forums with a large audience; an article that gets posted in one of them is highly likely to be copy-pasted on the other.

\textbf{Secondly}, topics that are not part of these two \glspl{cluster} appear and fade quickly in time. This is in line with the expected behaviour of news on the internet, which typically bursts for a few hours/days before being replaced by the latest news. We see however from this plot that only a small fraction of the inferred \glspl{cluster} coexists simultaneously in the dataset. The chance of spotting an \gls{interaction} is therefore weak, as noted in Chapter~\ref{Chapter-SBMs}.

\textbf{Thirdly}, it seems that \glspl{cluster} expand on larger periods at the beginning of the algorithm. This is due to the cold start of MPDHP. It needs some time before statistically distinguishing \glspl{cluster} and explores several directions at once. The early \glspl{cluster} are artefacts of such a cold-start effect.

\subsubsection{Quantifying \glspl{interaction}}
\paragraph{Effective \gls{interaction}}
We introduce the parameters we are going to use in follow-up analyses. The output of MPDHP consists of a list of \glspl{cluster} comprising timestamped bags of words --news headlines. Between each pair of \glspl{cluster}, MPDHP inferred a temporal influence function $\lambda (t)$, that represents the probability for one \gls{cluster} to trigger publications from another. Therefore, our model yields an adjacency matrix $A \in \mathbb{R}^{K \times K \times L}$, where $K$ is the number of \glspl{cluster} and $L$ the size of the RBF kernel $\vec{\kappa} (t)$. One entry $a_{i,j,l}$ represents the strength of the influence of $j$ in $i$ due to the $l^{th}$ entry of $\vec{\kappa}(t)$.

However, we must consider the effective number of \glspl{interaction} to get relevant metrics. A given triggering function $\lambda (t)$ could be inferred from the observation of very few observations only; we must weigh these \glspl{interaction}. To do so, we simply consider a normalized weight matrix $W \in \mathbb{R}^{K \times K \times L}$, whose entries $w_{i,j,l}$ are the average of the intensity of $i$ above $\lambda_0$ due to $j$ from the kernel entry $l$ for all observations. Explicitly:
\begin{equation}
    \label{eq-effInter}
    w_{i,j,l} = \frac{1}{\vert \mathcal{H}_i \vert}\sum_{t_i \in \mathcal{H}_i} \sum_{t_j<t_i} max(a_{i,j,l} \kappa_l (t_i-t_j)-\lambda_0, 0)
\end{equation}
where $t_x \in \mathcal{H}_x$ is the publication time of an observation from \gls{cluster} $x$ and $\lambda_0$ the temporal concentration parameter. Note that we retract $\lambda_0$ from the intensity term, because it is considered as a background probability for a publication to happen --similarly to the \gls{virality} in Chapter~\ref{Chapter-SBMs} and to the background noise Chapter~\ref{Chapter-InterRate}. Note that $W$ can also be interpreted as the instantaneous increase in probability due to \glspl{interaction}.

Note that in all the following computations, we do not consider \glspl{cluster} that comprise less than 10 documents. Such \glspl{cluster} are considered leftovers from the algorithm.

\paragraph{\Glspl{interaction} strength}
In Table~\ref{tab-interStrength}, we investigate the effective impact of \glspl{interaction} in the dataset. We consider the following metrics:
\begin{itemize}
    \item $<A>$: the average value of the whole adjacency matrix. It tells us to which extent \glspl{piece of information} are connected to each other according to MPDHP.
    \item $<W>$: the average value of the effective \glspl{interaction}. It tells us the extent to which the \glspl{interaction} (encoded in $A$) effectively happen in the dataset.
    \item $<A>_W$: the average of the inferred \gls{interaction} matrix $A$ weighted by the effective \glspl{interaction} $W$. In this case, $W$ can be interpreted as our confidence in the corresponding entries of $A$ given the data with which they were inferred. This value quantifies the overall role of \glspl{interaction} in the dataset.
    \item $\frac{<W^{intra}>}{<W^{extra}>}$: ratio of the intra-\gls{cluster} effective \glspl{interaction} with the extra-\gls{cluster} effective \glspl{interaction}. This tells us how much different \glspl{cluster} influence each other versus how much they influence themselves.
\end{itemize}

When computing the means, we discard the entries of $A$ and $W$ equal to 0. This is because all \glspl{cluster} do not exist simultaneously, and thus should not be considered. An \gls{interaction} strictly equal to 0 means that \glspl{cluster} simply did not exist at the same time.

\input{Tables/Chapter_4/table-MPDHP-interStrength}

The main conclusion of the results Table~\ref{tab-interStrength} is that most \glspl{interaction} are weak. The average value of $A$ tells us that the average value of the inferred parameters is around $0.05$, which is few given the value is bounded between 0 and 1. The metric ${<W>}$ tells us that on all events, the \gls{interaction} between \glspl{cluster} rose the probability of publication by 0.1\%-1\% on average. We can also note that the values of ${<W>}$ are of the same order of magnitude as $\lambda_0$ (0.01 for the ``Minute'' kernel, 0.001 for ``Hour'', and 0.0001 for ``Day''). We can interpret this as the probability for a new document belonging to a \gls{cluster} or being from a new \gls{cluster} is roughly the same from a temporal perspective. The metric $<A>_W$ tells us that when weighting the average of $A$ with the effective \gls{interaction}, the values of $A$ are slightly higher than $0.05$; we can now be confident in this value, given it has been inferred on a statistically significant number of observations. However, it still tells us that only some \glspl{interaction} are significant, which correlates with the findings of Chapter~\ref{Chapter-SBMs}. Finally, the last metric $\frac{<W^{intra}>}{<W^{extra}>}$ finds that most effective \glspl{interaction} take place more often within the same \gls{cluster}, meaning that \glspl{cluster} tend to self-replicate. Only for the ``Hour'' kernel and $\theta_0=0.01$ this value is lesser than one. It is because in this case, MPDHP infers two large \glspl{cluster} that exist for the entire year (see the 2 first rows of Fig.~\ref{fig-timeline-MPDHPNews}) and side topic-specific \glspl{cluster}. These \glspl{cluster} strongly influence each other, and all the topic-specific \glspl{cluster} can interact with them. In fact, the same effect explains the decrease of $\frac{<W^{intra}>}{<W^{extra}>}$ as $r$ grows: fewer \glspl{cluster} are inferred, and the probability of having large \glspl{cluster} that last for the whole period increases.

Another major observation from Table~\ref{tab-interStrength} is that standard deviations of effective \glspl{interaction} are large. It means that despite most \glspl{interaction} being weak, some of them play a significant role in the dataset. In Fig.~\ref{fig-MPDHP-figinterstrength}, we plot the distribution of effective \glspl{interaction} for one specific run (``Hour'' kernel, $\theta_0=0.01$, $r=1$). Note that we recover the same trend in all other experiments. The results of this figure are similar to the ones in Chapter~\ref{Chapter-SBMs} (Fig.~\ref{fig-IMMSBM-histRelInter}): most \glspl{interaction} are weak, and only a few of them are significant.

\begin{figure}
    \centering
    \includegraphics[width=0.7\textwidth]{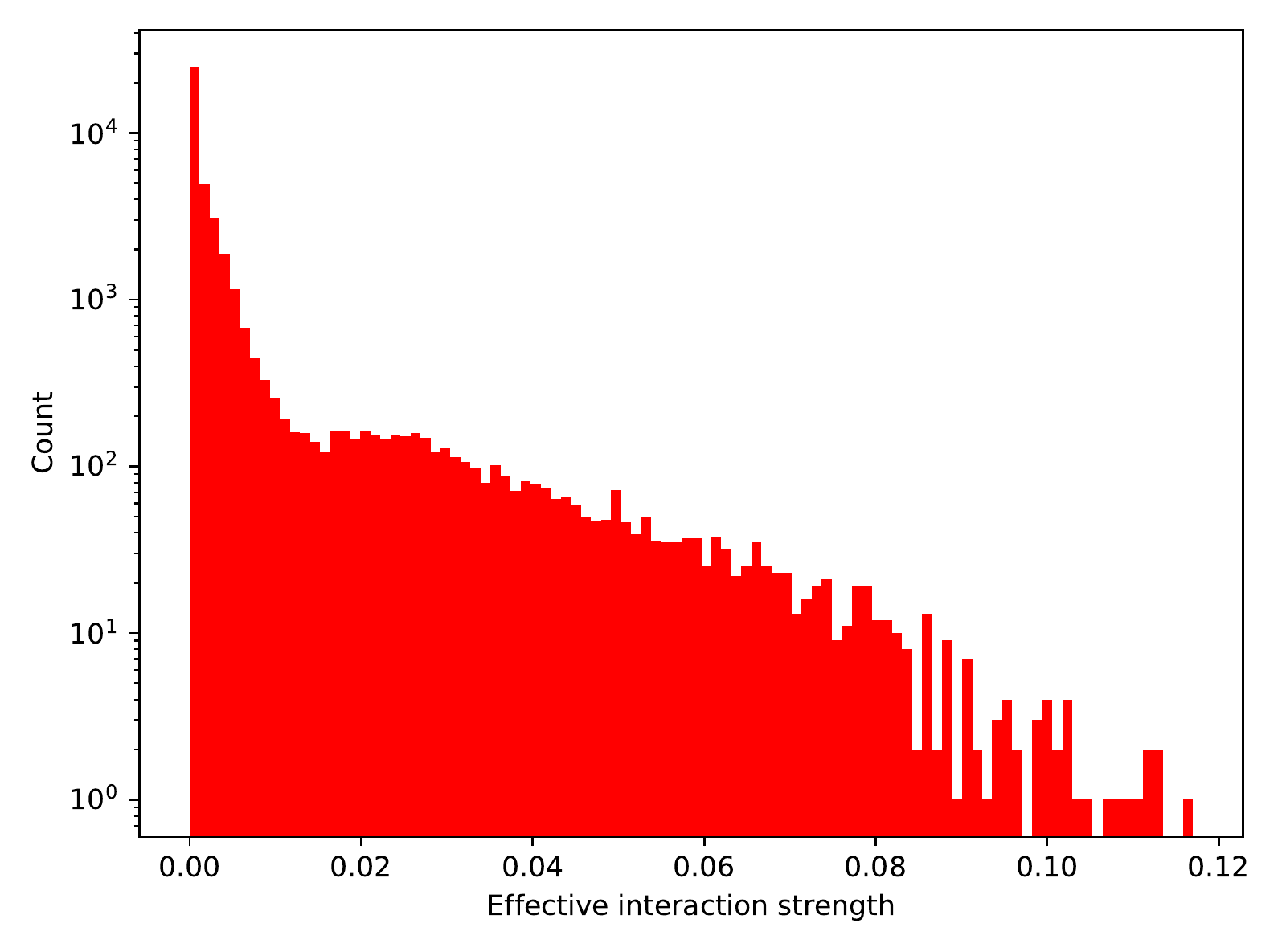}
    \caption[MPDHP x News - Distribution of \gls{interaction} strength]{\textbf{Distribution of \gls{interaction} strength} -- Most \glspl{interaction} are weak.}
    \label{fig-MPDHP-figinterstrength}
\end{figure}

\paragraph{\Glspl{interaction} range}
Finally, in Table~\ref{tab-interRange}, we investigate the range of effective \glspl{interaction} in every experiment. The \gls{interaction} range studies the persistence of \glspl{interaction}. To do so, we compute the effective \gls{interaction} for each kernel entry individually (and not the aggregated \gls{interaction} as in Table~\ref{tab-interStrength}) and average it over all existing \glspl{cluster}.

\input{Tables/Chapter_4/table-MPDHP-InterRange}

Importantly, the raw values of effective \gls{interaction} corresponding to the first entry of the triggering kernel $\kappa_1$ are consistently smaller than subsequent values. This is induced by our kernel choice, because $\kappa_1$ is always centered around $t=0$, which makes half of the associated Gaussian function account for (impossible) backwards influence. Therefore, where other kernels can contribute on both sides of their means, $\kappa_1$ cannot. In Table~\ref{tab-interRange}, we extrapolate their value as twice the computed one.

We see in Table~\ref{tab-interStrength} that influence tends to decrease over time for all the kernels considered, after reaching a first peak. Overall, the \gls{interaction} between documents seems to play a marginal role still. We did not plot the standard deviation for visualization purposes, but they are as large as in Table~\ref{tab-interStrength}. Therefore, \textit{most} \glspl{interaction} do not play a significant role in the publication of subsequent documents over time, but it greatly helps identify the right \gls{cluster} for some of them. Overall, the increase in probability for a new document to belong to a \gls{cluster} due to \glspl{interaction} is within 0.1\%-1\% (we recall that $\lambda (t) = \lim_{\Delta t \rightarrow \infty} \frac{P(\text{event in }\Delta t)}{\Delta t}$).

\subsubsection{Visualizing topical \glspl{interaction}}
In this section, we plot the temporal \gls{interaction} network between the inferred \glspl{cluster} both on the global and the monthly scale. The \glspl{cluster}' composition is given explicitly. Each edge between a pair of \glspl{cluster} represents two metrics: the inferred \gls{interaction} strength $A$, and the effective \gls{interaction} $W$. The inferred \gls{interaction} strength $A$ is plotted using a colour code (darker is stronger). The effective \gls{interaction} is represented using transparency (the less transparent the stronger the effective \gls{interaction}). Therefore, a barely visible dark edge means that MPDHP inferred a strong \gls{interaction}, but that few of them were effectively observed.

\paragraph{Experiment considered for subsequent analyses}
In the following paragraphs, we consider the ``Hour'' kernel, with $\theta_0 = 0.01$ and $r = 1$. We justify this choice for easing the interpretation of our results. This way, we restrict our analysis to \glspl{cluster} that do not contain fragments of a whole. For instance, we prefer to have only one \gls{cluster} about the Notre-Dame cathedral fire and related news instead of three \glspl{cluster} containing fragments of the news, such as the initial fire, the political reactions, funds raising, etc. Therefore, we choose to consider $\theta_0 = 0.01$, which avoids \glspl{cluster} to be overly specific. We choose the ``Hour'' kernel, which spans over long enough periods so that news fragments about a similar topic are considered as possibly related. Besides, from direct observation, it seldom happens for news to stick around for more than a few days, which the ``Hours'' kernel is fit to capture.
The choice of $r$ is based on an arbitrary trade-off between textual and temporal information. We do not want to consider extreme values ($r=0$ or $r > 2$) so that we exploit both pieces of information. Besides, we saw in Table~\ref{tab-interStrength} and Table~\ref{tab-interRange} that only slight variations are observed across the range of $r \in \{ 0.5, 1, 1.5 \}$ considered. Finally, we are interested in seeing how the large \glspl{cluster} spanning over the whole period relate to topic-specific smaller \glspl{cluster} seen in Fig.~\ref{fig-timeline-MPDHPNews}.

\begin{figure}
    \centering
    \includegraphics[width=\textwidth]{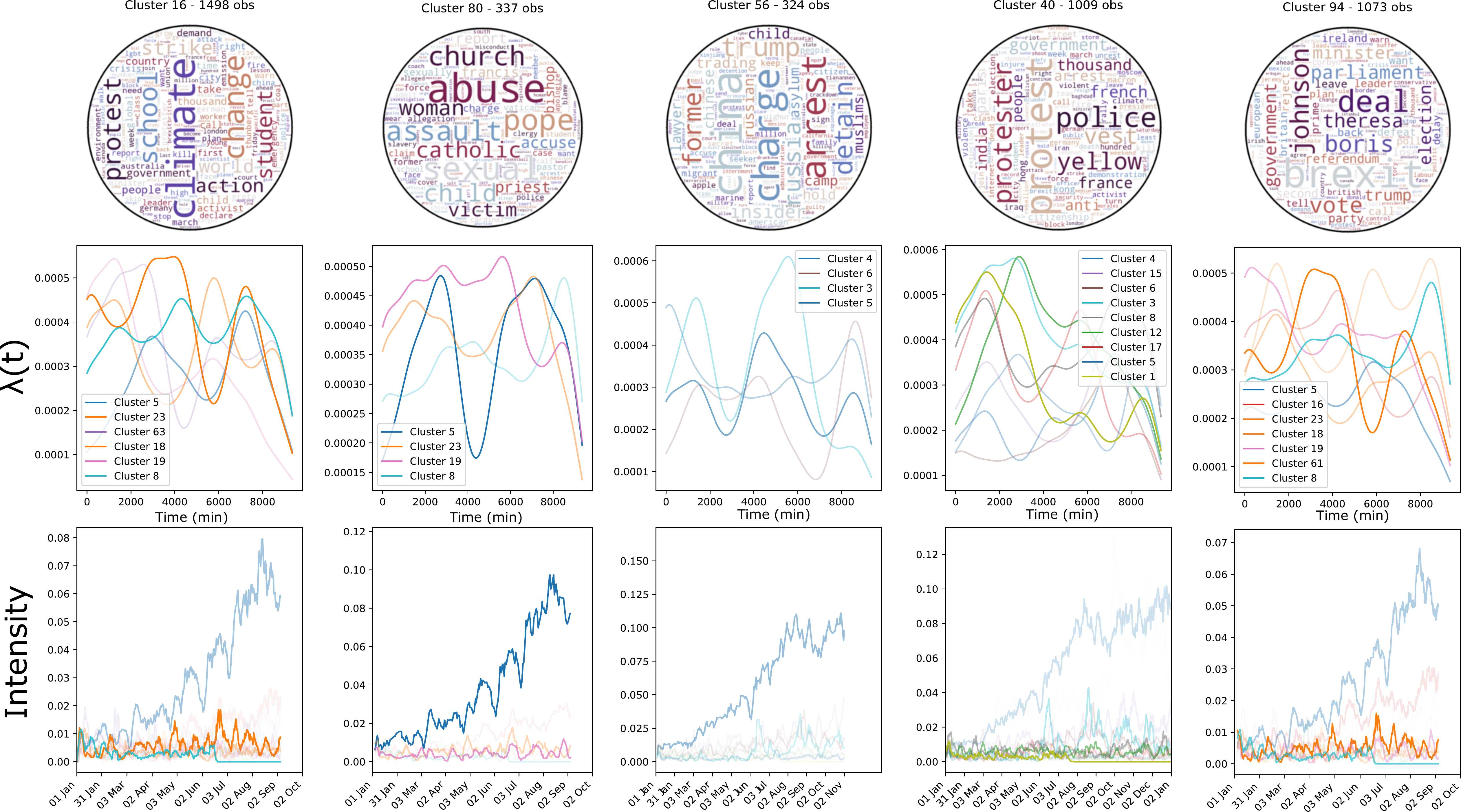}
    \caption[MPDHP x News - Typical output]{\textbf{A typical MPDHP output} -- A set of manually selected \glspl{cluster} along with the vocabulary of their documents (top), their inferred dynamics (middle) and the \glspl{cluster} that influenced them on the real-time axis (bottom).}
    \label{fig-illustrMPDHPClus}
\end{figure}
\paragraph{Representing individual \glspl{cluster}}
In Fig.~\ref{fig-illustrMPDHPClus}, we represent the raw data that we will use in ulterior visualizations for some selected \glspl{cluster}. In the top part, we represent their textual \gls{content} as a wordcloud. In this case, we chose to pick \glspl{cluster} about climate change inaction protests, catholic church child abuse scandals, China's Uighur detention camps, ``Gilets Jaunes'' protests in France, and Brexit. For each \gls{cluster}, we represent the top temporal influence that other \glspl{cluster} may exert on them. Transparency accounts for the effective \gls{interaction} $W$ discussed in the previous sections. The bottom plot represents the influence exerted on these \glspl{cluster} by all other \glspl{cluster} at all times --which does not mean this influence led to a publication.

This way of representing the data fits well in the univariate case, as in \citep{Du2015DHP}, but does not allow to capture the complexity of the inferred mechanisms at stake. In the following sections, we propose alternative visualizations in the form of temporal networks.

\paragraph{Globally}
In Fig.~\ref{fig-globInterGraph-MPDHPNews}, we plot the \glspl{cluster} \gls{interaction} network over the whole period we considered (12 months). This figure uses the same data as Fig.~\ref{fig-timeline-MPDHPNews}. Both the adjacency matrix and the transparency matrix are normalized by their highest value. We see that the strongest \gls{interaction} happens between the blue and the orange \glspl{cluster}. Besides, this \gls{interaction} seems to be stronger at smaller times. The \gls{interaction} between other \glspl{cluster} is essentially directed toward these two \glspl{cluster}. Shortly after a news \gls{cluster} appears, it is likely to find an echo on either r/news or r/worldnews. After existing for a while, subsequent publications are merged into one of these \glspl{cluster} depending on where they got published. This explains why most \glspl{cluster} last so few over time in Fig.~\ref{fig-timeline-MPDHPNews}.

We see that there are very few \glspl{interaction} between the \glspl{cluster} which account for actual news --in opposition to the \glspl{cluster} that account for a subreddit, see Section~\ref{MPDHPxNews-visutopics}.

\begin{figure}
    \centering
    \includegraphics[width=\textwidth]{Figures/Chapter_4/MPDHP-Reddit/GlobalInterGraph.pdf}
    \caption[MPDHP x News - Global \glspl{cluster} temporal \gls{interaction} graph 2019]{\textbf{Global \glspl{cluster} temporal \gls{interaction} from 01/2019 to 12/2019} --- Each node represents one \gls{cluster}, and each edge represents the \gls{interaction} strength at various times -- the darker and the less transparent the stronger the effective \gls{interaction}. The \glspl{cluster} composition is given below the network.}
    \label{fig-globInterGraph-MPDHPNews}
\end{figure}

\begin{figure}
    \centering
    \includegraphics[height=\textheight]{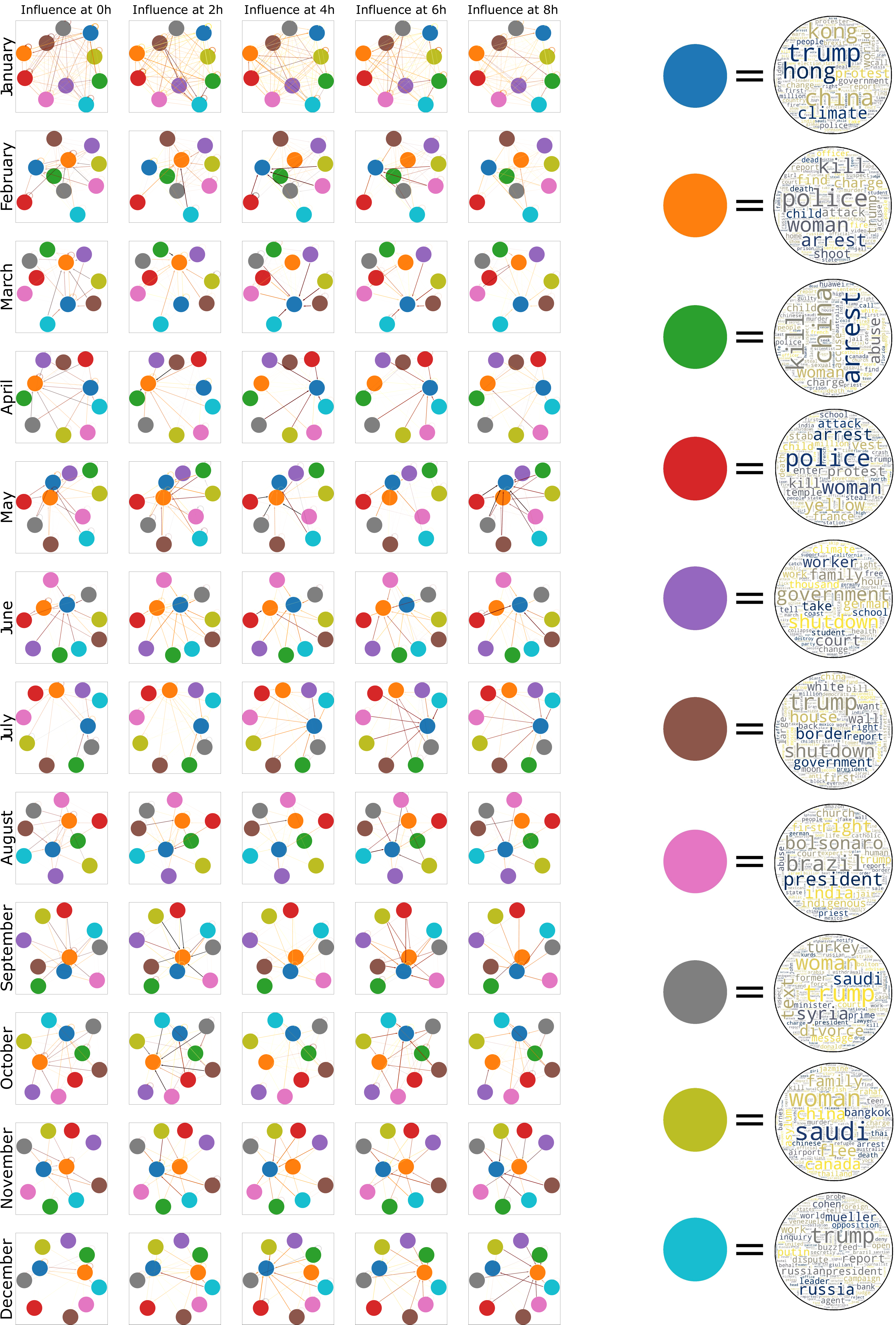}
    \caption[MPDHP x News - Monthly \glspl{cluster} temporal \gls{interaction} graph 2019]{\textbf{Monthly \glspl{cluster} temporal \gls{interaction} from 01/2019 to 12/2019} --- Each line stands for one month, each node represents one \gls{cluster}, and each edge represents the \gls{interaction} strength at various times -- the darker and the less transparent the stronger the effective \gls{interaction}. The \gls{cluster} composition is given below the network.}
    \label{fig-monthInterGraph-MPDHPNews}
\end{figure}
\paragraph{Monthly}
In Fig.~\ref{fig-monthInterGraph-MPDHPNews}, we plot the same \gls{interaction} between \glspl{cluster} as in Fig.~\ref{fig-globInterGraph-MPDHPNews}, but for each month. This figure uses the same data as Fig.~\ref{fig-timeline-MPDHPNews}, except it is now broken down into temporal slices. From this figure, we recover that at the early times of the algorithm, MPDHP has not converged yet. It makes the \gls{interaction} plot messy, with several \glspl{interaction} that are likely to be inaccurate. As time goes by, \glspl{interaction} are better defined --in particular between the blue and orange \glspl{cluster}. Another interesting fact is that most \glspl{interaction} happen in a 1h time frame: the publication of a document on r/news is highly likely to trigger a similar publication on r/worldnews within the hour. Here again, we see that most side \glspl{cluster} do not interact significantly with each other.

\subsection{Conclusion}
\paragraph{A real-world application}
In this section, we conducted extensive experiments on a single real-world large-scale dataset from Reddit. We described how we gather a year worth of news data and pre-processed it. We then argued how the dataset was fit to apply to \acrshort{MPDHP}. We extensively discussed the various hyper-parameters that had to be tuned to run our experiments (in particular $\kappa (t)$, $r$ and $\theta_0$). In the end, we conducted 6 different experiments, one for each combination of parameters. The choice of other hyper-parameters has also been justified. 

\paragraph{\Glspl{interaction} do not appear to play a significant role in this dataset}
We proposed several ways to assess the role of \glspl{interaction} in the dataset. In particular, we introduced the notion of \textit{effective \gls{interaction}} as a way to evaluate how confident we can be in MPDHP's output. On this basis, we analysed the importance of \glspl{interaction} in general, as well as from a temporal perspective. We recovered our conclusions from Chapter~\ref{Chapter-SBMs}: \glspl{interaction} are sparse. We also recovered our conclusions from Chapter~\ref{Chapter-InterRate}: \gls{interaction} strength decays over time. These observations emphasize the adequateness of our approach to model \glspl{interaction}. Building on these considerations and looking at the global effective \gls{interaction} average, we conclude that \glspl{interaction} play a minor role in such a dataset. Overall, they only increase the instantaneous probability for a new observation to appear by 1\%. Even the most extreme values represented in Fig.~\ref{fig-MPDHP-figinterstrength} seem to only increase this probability by 12\% top. 

\paragraph{Perspectives}
However, despite intending our study as exhaustive, there is room for improvement in \gls{interaction} modelling using \acrshort{MPDHP}. In particular, there are two biases that we could not explore yet. Firstly, the parameter $\lambda_0$ has been set according to a heuristic (so that a new \gls{cluster} is opened with fifty percent chances when we are 95\% sure that it does not match the existing one). Its direct inference would robustify the approach, despite seeming complicated. Indeed, $\lambda_0$ does not account for individual events realizations, but for Hawkes processes starts, which is not a trivial difference. Secondly, another possible improvement would be to allow \glspl{cluster} to passively replicate --without the need for an \gls{interaction}. We expect that this would boil down to adding a time-independent kernel entry to $\vec{\kappa}$, in a similar fashion as in Chapter~\ref{Chapter-InterRate}. However, non-technical questions may arise from such modification: when to consider a \gls{cluster} as extinct given a non-fading kernel? How should this kernel relate to the temporal concentration parameter $\lambda_0$?

We believe such improvements would make MPDHP more robust and interpretable, and find applications beyond \gls{interaction} modelling.

%% file: Tables/Chapter_4/table-res-PDP.tex
\begin{table}[t]
    \caption[PDP - Numerical results on synthetic datasets]{Numerical results of the Powered Dirichlet process, Uniform process and Dirichlet process priors coupled to a standard Infinite Gaussian Mixture Model, for the 3 synthetic datasets plotted Fig.~\ref{fig-Results} and 2 real-world datasets (100 runs each). We see that using PDP as prior makes the model outperform the baselines consistently on every metric. The standard error is given in shorthand form -- $0.123(12) \Leftrightarrow 0.123 \pm 0.012$. \label{tabMetrics}}
\centering
\setlength{\lgCase}{2.2cm}
\begin{tabular}{ |p{0.08\lgCase}|p{1.29\lgCase}|p{0.93\lgCase}|p{0.93\lgCase}|p{0.93\lgCase}|p{0.93\lgCase}|p{0\lgCase}}
\cline{3-6}
  \multicolumn{1}{c}{\rotatebox[origin=c]{90}{}} & \multicolumn{1}{c|}{\rotatebox[origin=c]{90}{}} & \centering Adj.MI & \centering Adj.RI & \centering Norm.VI & \centering $\frac{K_{inf}-K_{true}}{K_{true}}$ & \\
 \cline{1-6}
    \centering\multirow{4}{*}{\centering\rotatebox[origin=c]{90}{\ \ \ \ \ \footnotesize \textbf{Density}}} & \centering PDP (r=0.60) & \centering \textbf{0.992(1)} & \centering \textbf{0.980(2)} & \centering \textbf{0.006(1)} & \centering \textbf{0.045(5)} & \\
     & \centering DP & \centering 0.951(4) & \centering 0.797(17) & \centering 0.037(3) & \centering 0.128(10) & \\
     & \centering UP & \centering 0.939(2) & \centering 0.854(4) & \centering 0.050(1) & \centering 0.548(1) & \\

 \cline{1-6}

     \centering\multirow{4}{*}{\centering\rotatebox[origin=c]{90}{\ \ \ \ \ \footnotesize \textbf{Diamond}}} & \centering PDP (r=0.50) & \centering \textbf{0.982(2)} & \centering \textbf{0.956(5)} & \centering \textbf{0.011(1)} & \centering \textbf{0.063(7)} & \\
     & \centering DP & \centering 0.909(7) & \centering 0.731(19) & \centering 0.053(4) & \centering 0.202(12) & \\
     & \centering UP & \centering 0.927(2) & \centering 0.844(6) & \centering 0.051(2) & \centering 0.544(2) & \\

 \cline{1-6}
 
     \centering\multirow{4}{*}{\centering\rotatebox[origin=c]{90}{\ \ \ \ \ \footnotesize \textbf{Grid}}} & \centering PDP (r=0.85) & \centering \textbf{0.997(1)} & \centering \textbf{0.990(2)} & \centering \textbf{0.003(1)} & \centering \textbf{0.014(2)} & \\
     & \centering DP & \centering 0.995(1) & \centering 0.977(4) & \centering \textbf{0.004(1)} & \centering \textbf{0.018(3)} & \\
     
     & \centering UP & \centering 0.811(1) & \centering 0.517(3) & \centering 0.154(1) & \centering 2.120(1) & \\

 \cline{1-6}
 
      \centering\multirow{4}{*}{\centering\rotatebox[origin=c]{90}{\ \ \ \ \ \footnotesize \textbf{Iris}}} & \centering PDP (r=0.90) & \centering \textbf{0.868(4)} & \centering \textbf{0.866(7)} & \centering \textbf{0.057(2)} & \centering \textbf{0.000(0)} & \\
      
     & \centering DP & \centering 0.843(6) & \centering 0.820(12) & \centering 0.065(2) & \centering 0.030(10) & \\
     
     & \centering UP & \centering 0.544(2) & \centering 0.295(3) & \centering 0.303(2) & \centering 2.777(32) & \\

 \cline{1-6}
 
       \centering\multirow{4}{*}{\centering\rotatebox[origin=c]{90}{\ \ \ \ \ \footnotesize \textbf{Wines}}} & \centering PDP (r=0.10) & \centering \textbf{0.712(15)} & \centering \textbf{0.637(20)} & \centering \textbf{0.102(5)} & \centering \textbf{0.157(17)} & \\
      
     & \centering DP & \centering 0.589(19) & \centering 0.461(16) & \centering 0.128(4) & \centering 0.327(13) & \\
     
     & \centering UP & \centering \textbf{0.713(17)} & \centering \textbf{0.657(21)} & \centering \textbf{0.103(5)} & \centering \textbf{0.147(17)} & \\

 \cline{1-6}
 
       \centering\multirow{4}{*}{\centering\rotatebox[origin=c]{90}{\ \ \ \ \ \footnotesize \textbf{Cancer}}} & \centering PDP (r=0.10) & \centering \textbf{0.254(17)} & \centering \textbf{0.278(21)} & \centering \textbf{0.118(1)} & \centering \textbf{0.000(0)} & \\
      
     & \centering DP & \centering 0.085(16) & \centering 0.094(19) & \centering 0.108(2) & \centering \textbf{0.000(0)} & \\
     
     & \centering UP & \centering \textbf{0.271(17)} & \centering \textbf{0.300(21)} & \centering \textbf{0.118(1)} & \centering \textbf{0.000(0)} & \\

 \cline{1-6}
 
     \centering\multirow{4}{*}{\centering\rotatebox[origin=c]{90}{\ \ \ \ \ \footnotesize \textbf{20-NG}}} & \centering PDP (r=0.80) & \centering \textbf{0.421(4)} & \centering \textbf{0.119(3)} & \centering \textbf{0.477(3)} & \centering - & \\
      
     & \centering DP & \centering 0.404(4) & \centering 0.105(4) & \centering 0.491(3) & \centering - & \\
     
     & \centering UP & \centering      
     
0.000(4) & \centering 0.000(0) & \centering 0.830(3) & \centering - & \\

 \cline{1-6}

\end{tabular}

\end{table}

%% file: Tables/Chapter_4/table-MPDHP-fit.tex
\begin{table}[h]
    \centering
    \begin{tabular}{|c|c|l|r|r|l|l|}
        \cline{1-7}
        $\vec{\kappa} (t)$ & $\theta_0$ & $r$ & $K$ & $<N>$ & $S_{text}^{(20)}$ & $S_{sub}^{(20)}$ \\
        \cline{1-7}
        {\multirow{8}{*}{\rotatebox[origin=c]{90}{\textbf{Minute}}}} & {\multirow{4}{*}{\rotatebox[origin=c]{90}{\textbf{0.01}}}} & 0.0 & 16150 & 6 & 0.744(68) & 0.400(36) \\
        & & 0.5 & 8498 & 12 & 0.796(55) & 0.441(24) \\
        & & 1.0 & 5069 & 20 & 0.790(49) & 0.476(53) \\
        & & 1.5 & 2730 & 37 & 0.808(43) & 0.490(52) \\
        \cdashline{2-7}
        & {\multirow{4}{*}{\rotatebox[origin=c]{90}{\textbf{0.001}}}} & 0.0 & 46304 & 2 & 0.475(48) & 0.239(56) \\
        & & 0.5 & 37277 & 3 & 0.485(40) & 0.256(60) \\
        & & 1.0 & 26858 & 4 & 0.501(37) & 0.266(72) \\
        & & 1.5 & 19275 & 5 & 0.506(39) & 0.280(58) \\
        \cline{1-7}

        {\multirow{8}{*}{\rotatebox[origin=c]{90}{\textbf{Hour}}}} & {\multirow{4}{*}{\rotatebox[origin=c]{90}{\textbf{0.01}}}} & 0.0 & 3792 & 27 & 0.798(52) & 0.474(106) \\
        & & 0.5 & 1735 & 59 & 0.791(45) & 0.469(77) \\
        & & 1.0 & 825 & 124 & 0.803(47) & 0.484(75) \\
        & & 1.5 & 426 & 240 & 0.795(37) & 0.489(69) \\
        \cdashline{2-7}
        & {\multirow{4}{*}{\rotatebox[origin=c]{90}{\textbf{0.001}}}} & 0.0 & 18012 & 6 & 0.760(86) & 0.397(63) \\
        & & 0.5 & 11923 & 9 & 0.784(72) & 0.425(38) \\
        & & 1.0 & 4837 & 21 & 0.821(50) & 0.497(30) \\
        & & 1.5 & 2368 & 43 & 0.814(42) & 0.481(91) \\
        \cline{1-7}

        {\multirow{8}{*}{\rotatebox[origin=c]{90}{\textbf{Day}}}} & {\multirow{4}{*}{\rotatebox[origin=c]{90}{\textbf{0.01}}}} & 0.0 & 609 & 168 & 0.713(34) & 0.413(103) \\
        & & 0.5 & 326 & 313 & 0.728(33) & 0.429(98) \\
        & & 1.0 & 172 & 593 & 0.743(31) & 0.461(94) \\
        & & 1.5 & 96 & 1063 & 0.755(36) & 0.464(86) \\
        \cdashline{2-7}
        & {\multirow{4}{*}{\rotatebox[origin=c]{90}{\textbf{0.001}}}} & 0.0 & 4349 & 23 & 0.705(49) & 0.396(103) \\
        & & 0.5 & 2654 & 38 & 0.721(59) & 0.404(104) \\
        & & 1.0 & 1399 & 73 & 0.734(57) & 0.431(106) \\
        & & 1.5 & 764 & 134 & 0.741(54) & 0.442(98) \\
        \cline{1-7}
         
    \end{tabular}
    
    \caption[MPDHP x News - Results for each experiment]{\textbf{Results for each experiment} --- We tried various combinations of parameters $\vec{\kappa} (t)$, $\theta_0$ and $r$ and observe how they result in a variety of outputs. We characterize these outputs in terms of clusters (number, size, textual entropy, subreddits entropy). The standard deviation on the last digits is given in standard notation -- $0.123(12) \Leftrightarrow 0.123 \pm 0.012$.}
    \label{tab-diffruns}
\end{table}

%% file: Tables/Chapter_4/table-MPDHP-interStrength.tex
\begin{table}[h]
    \centering
    \begin{tabular}{|c|c|l|c|c|c|l|}
        \cline{1-7}
        $\vec{\kappa} (t)$ & $\theta_0$ & $r$ & $<A>$ ($10^{-3}$) & $<W>$ ($10^{-5}$) & $<A>_W$ ($10^{-3}$) & $\frac{<W^{intra}>}{<W^{extra}>}$ \\
        \cline{1-7}
        {\multirow{6}{*}{\rotatebox[origin=c]{90}{\textbf{Minute}}}} & {\multirow{3}{*}{\rotatebox[origin=c]{90}{\textbf{0.01}}}} & 0.5 & 49(21) & 342(889) & 66(17) & 1.8(62) \\
        & & 1.0 & 48(20) & 478(1124) & 60(17) & 1.4(43) \\
        & & 1.5 & 48(20) & 746(1901) & 60(17) & 1.0(33) \\
        \cdashline{2-7}
        & {\multirow{3}{*}{\rotatebox[origin=c]{90}{\textbf{0.001}}}} & 0.5 & 50(22) & 316(882) & 66(17) & 3.1(138) \\
        & & 1.0 & 50(21) & 279(752) & 67(16) & 2.6(105) \\
        & & 1.5 & 50(22) & 268(665) & 67(16) & 2.3(84) \\
        \cline{1-7}

        {\multirow{6}{*}{\rotatebox[origin=c]{90}{\textbf{Hour}}}} & {\multirow{3}{*}{\rotatebox[origin=c]{90}{\textbf{0.01}}}} & 0.5 & 49(18) & 389(843) & 56(17) & 0.5(13) \\
        & & 1.0 & 49(18) & 478(1187) & 56(17) & 0.6(15) \\
        & & 1.5 & 48(17) & 471(789) & 52(15) & 0.7(13) \\
        \cdashline{2-7}
        & {\multirow{3}{*}{\rotatebox[origin=c]{90}{\textbf{0.001}}}} & 0.5 & 50(21) & 110(398) & 61(17) & 1.7(67) \\
        & & 1.0 & 50(18) & 133(506) & 57(17) & 1.4(60) \\
        & & 1.5 & 49(17) & 183(554) & 55(17) & 1.1(37) \\
        \cline{1-7}

        {\multirow{6}{*}{\rotatebox[origin=c]{90}{\textbf{Day}}}} & {\multirow{3}{*}{\rotatebox[origin=c]{90}{\textbf{0.01}}}} & 0.5 & 49(18) & 41(97) & 55(17) & 1.2(34) \\
        & & 1.0 & 49(19) & 63(131) & 54(17) & 1.2(31) \\
        & & 1.5 & 49(19) & 91(187) & 53(18) & 1.2(31) \\
        \cdashline{2-7}
        & {\multirow{3}{*}{\rotatebox[origin=c]{90}{\textbf{0.001}}}} & 0.5 & 50(20) & 18(90) & 60(19) & 1.1(59) \\
        & & 1.0 & 50(19) & 23(101) & 58(17) & 1.0(50) \\
        & & 1.5 & 50(19) & 37(111) & 56(18) & 1.0(36) \\
        \cline{1-7}
         
    \end{tabular}
    
    \caption[MPDHP x News - Interaction strength]{\textbf{Interaction strength} --- Overall, interaction between clusters is weak. The large standard deviations suggest that there is a large variety of interacting behaviours. Interactions tend to happen within a cluster (self-interactions).
    }
    \label{tab-interStrength}
\end{table}

%% file: Tables/Chapter_4/table-MPDHP-InterRange.tex
\begin{table}[h]
    \centering
    \begin{tabular}{|c|c|l|c|c|c|c|c|c|c|c|c|c|}
        \cline{1-12}
        $\vec{\kappa} (t)$ & $\theta_0$ & $r$ & $\kappa_1$ & $\kappa_2$ & $\kappa_3$ & $\kappa_4$ & $\kappa_5$ & $\kappa_6$ & $\kappa_7$ & $\kappa_8$ & $\kappa_9$ \\
        \cline{1-12}
        & & & 0m & 10m & 20m & 30m & 40m & 50m & 60m & 70m & 80m \\
        \cline{1-12}
        {\multirow{6}{*}{\rotatebox[origin=c]{90}{\textbf{Minute (m)}}}} & {\multirow{3}{*}{\rotatebox[origin=c]{90}{\textbf{0.01}}}} & 0.5 & 266 & 421 & 407 & 451 & 428 & 403 & 395 & 345 & 121 \\
        & & 1 & 396 & 591 & 580 & 607 & 532 & 575 & 521 & 507 & 224 \\
        & & 1.5 & 616 & 937 & 893 & 914 & 955 & 840 & 808 & 810 & 304 \\
        \cdashline{2-12}
        & {\multirow{3}{*}{\rotatebox[origin=c]{90}{\textbf{0.001}}}} & 0.5 & 436 & 509 & 457 & 424 & 371 & 340 & 313 & 178 & 52 \\
        & & 1 & 284 & 435 & 396 & 388 & 343 & 327 & 272 & 187 & 45 \\
        & & 1.5 & 208 & 388 & 366 & 353 & 326 & 333 & 290 & 215 & 61 \\
        \cline{1-12}

        & & & 0h & 2h & 4h & 6h & 8h & - & - & - & - \\
        \cline{1-12}
        {\multirow{6}{*}{\rotatebox[origin=c]{90}{\textbf{Hour (h)}}}} & {\multirow{3}{*}{\rotatebox[origin=c]{90}{\textbf{0.01}}}} & 0.5 & 494 & 430 & 502 & 456 & 324 & & & & \\
        & & 1 & 658 & 538 & 549 & 542 & 451 & & & & \\
        & & 1.5 & 558 & 615 & 532 & 526 & 411 & & & & \\
        \cdashline{2-8}
        & {\multirow{3}{*}{\rotatebox[origin=c]{90}{\textbf{0.001}}}} & 0.5 & 124 & 149 & 119 & 137 & 92 & & & & \\
        & & 1 & 154 & 164 & 172 & 149 & 111 & & & & \\
        & & 1.5 & 208 & 244 & 223 & 197 & 156 & & & & \\
        \cline{1-12}

        & & & 0d & 1d & 2d & 3d & 4d & 5d & 6d & - & - \\
        \cline{1-12}
        {\multirow{6}{*}{\rotatebox[origin=c]{90}{\textbf{Day (d)}}}} & {\multirow{3}{*}{\rotatebox[origin=c]{90}{\textbf{0.01}}}} & 0.5 & 44 & 45 & 47 & 46 & 47 & 47 & 37 & & \\
        & & 1 & 70 & 71 & 72 & 68 & 68 & 70 & 61 & & \\
        & & 1.5 & 102 & 100 & 101 & 105 & 98 & 105 & 82 & & \\
        \cdashline{2-10}
        & {\multirow{3}{*}{\rotatebox[origin=c]{90}{\textbf{0.001}}}} & 0.5 & 18 & 20 & 21 & 21 & 21 & 21 & 17 & & \\
        & & 1 & 24 & 26 & 24 & 27 & 28 & 26 & 22 & & \\
        & & 1.5 & 21 & 41 & 42 & 41 & 41 & 41 & 35 & & \\
        \cline{1-12}

    \end{tabular}
    
    \caption[MPDHP x News - Interaction range]{\textbf{Interaction range} --- All the values for effective interaction are given in ten-thousandth ($10^{-5}$). Influence tends to decrease over time for all the kernels considered.}
    \label{tab-interRange}
\end{table}

%% file: Chapters/5_Conclusion.tex
\chapter{Conclusion}
\label{Chapter-Conclusion}

\section{Contributions}
\subsection{Overview}
Throughout this manuscript, we developed seven original models that allowed us to tackle various aspects of \glspl{interaction} between \glspl{piece of information} in their \gls{spread}: \bl{\acrshort{IMMSBM}}, \bl{\acrshort{SIMSBM}}, \bl{SDSBM}, \ora{InterRate}, \gr{\acrshort{PDP}}, \rd{\acrshort{PDHP}}, \rd{\acrshort{MPDHP}}. Along the road, this research led us to study seemingly distant fields of machine learning. A first venture in \bl{stochastic block modelling} provided us with insights on the frequency at which \glspl{interaction} happen. An exploration of \ora{convex network inference methods} allowed us to conclude on the time range of \glspl{interaction}. Driven by the need for more global models, a dive into fundamental bricks of machine learning led us to explore and modify canonical \gr{Dirichlet processes}. Finally, answering our problematic required a \rd{junction between Dirichlet process and point processes}. We recall the summary of our contributions in Table~\ref{table-contrib-ccl} --this table is identical to Table~\ref{table-codesandDS} and is simply recalled here.

\begin{table}
    \caption[Conclusion - Contributions]{\textbf{Contributions presented in this manuscript} --- Our contributions are listed below in the order of their appearance in the text.}
    \label{table-contrib-ccl}
	\centering
    \noindent
    
    \makebox[\textwidth]{\resizebox{\textwidth}{!}{
	\begin{tabular}{|c|c|c|c|c|c|c|c|}
	\cline{1-8}
	& \hyperref[SIMSBM]{SIMSBM} & \hyperref[IMMSBM]{IMMSBM}& \hyperref[SDSBM]{SDSBM} & \hyperref[Chapter-InterRate]{InterRate} & \hyperref[PDP]{PDP} & \hyperref[PDHP]{PDHP} & \hyperref[MPDHP]{MPDHP}\\
	& Chap.~\ref{Chapter-SBMs} & Chap.~\ref{Chapter-SBMs} & Chap.~\ref{Chapter-SBMs} & Chap.~\ref{Chapter-InterRate} & Chap.~\ref{Chapter-DHPs} & Chap.~\ref{Chapter-DHPs} & Chap.~\ref{Chapter-DHPs} \\
	\cline{1-8}
	Self-\glspl{interaction}    & x & x & x & x & x & x & x \\
	\cline{1-8}
	Pair-\glspl{interaction}    & x & x & x & x &   & x & x \\
	\cline{1-8}
	N-\glspl{interaction}       & x &  & x &   &   & x & x \\
	\cline{1-8}
	Clustering           & x & x & x &   & x & x & x \\
	\cline{1-8}
	Discrete time    &   &   & x & x &   & x & x \\
	\cline{1-8}
	Continuous time  &   &   &   & x &   & x & x \\
	\cline{1-8}
	Online inference     &   &   &   &   & x & x & x \\
	\cline{1-8}
	\end{tabular}
	}}
	
\end{table}

\subsection{Answers to our problematic}
\subsubsection{Q1: How frequent are \glspl{interaction}? \bl{$\bullet$}}
To answer our first question \textbf{Q1}, we considered Stochastic Block Modelling to be a good approach. 

Indeed, a straightforward way to represent \glspl{interaction} is to embed them as a network. \Glspl{entity} are nodes of this network, and \glspl{interaction} between these \glspl{entity} are link between these nodes. These links can account for distinct types of relations that we can associate with different labels. 
Given the number of \textit{possibly} \gls{interacting} \glspl{entity} can be high in some real-world systems, we need to group them to see whether groups of \glspl{entity} interact in a similar way instead of dealing with each \gls{entity} individually. In practice, typically only a few of these possible \glspl{interaction} actually happen, but we did not know this beforehand. Besides, we considered that \glspl{entity} could interact as pairs, triplets or more, and that they could interact with different elements of their environment (time, user's metadata, etc.). After a careful review of the literature, we noticed that several state-of-the-art models could be expressed as special cases of a more global framework: the Serialized \Gls{interacting} Mixed \gls{membership} Stochastic Block Model (\acrshort{SIMSBM}, Section~\ref{SIMSBM-model}). This generalization allows modelling an arbitrarily large context (number of input pieces of information) as well as an arbitrarily high order of \glspl{interaction}.

Along with this generalization, we introduced a simple procedure to incorporate time modelling into any iteration of SIMSBM. This extension takes the form of a prior on the model's parameters. This approach relies on the single assumption that dynamics are not abrupt. 

We considered a special case of SIMSBM, SIMSBM(2) or \Gls{interacting} Mixed \Gls{membership} SBM (\acrshort{IMMSBM}), and used it to model \glspl{interaction} between \glspl{entity} in several real-world datasets. We investigated the role of \glspl{interaction} between different spreading \glspl{entity} (hashtags, words, memes, etc.) and quantified their importance in several corpora --see Section~\ref{IMMSBM}.

In particular, we focused on the study of a Twitter dataset and investigated the role of \gls{interaction} between Twitter URLs on their spreading probability.

The conclusion of this study answered the first question \textbf{Q1} raised in the introduction: \textbf{\glspl{interaction} are sparse}. On the Twitter dataset, significant \glspl{interaction} took place only between a limited fraction of \gls{cluster} pairs, and between an even smaller fraction of \gls{entity} pairs.

\subsubsection{Q2: How persistent are \glspl{interaction}? \ora{$\bullet$}}
Despite the generality of the proposed SIMSBM, this model was not fit for modelling the time range over which \glspl{interaction} may take place. However, taking the \gls{interaction} range into account is fundamental following a simple consideration: that the influence of \glspl{entity} on someone cannot last indefinitely as they get gradually forgotten by the users.

We introduced a convex multi-kernel model designed for this task: InterRate --see Chapter~\ref{Chapter-InterRate}. This model was able to infer the evolution of pair \glspl{interaction}' intensity over time, for any pair of \glspl{entity}, called the \gls{interaction profile}. It allowed us to study how users exhibit different behaviours according to \textit{combinations} of exposures they have been exposed to in continuous time. We showed that the joint effect of two exposures on a user is more than the disjoint sum of their individual effect, which means there is an \gls{interaction}.

The answer of this study to the second question (\textbf{Q2}) raised in the introduction is that \textbf{\glspl{interaction} are brief}. A study of the processes at stake in the Twitter dataset revealed that \gls{interaction} between two \glspl{entity} is significant only when those are close to each other in time. Typically, the intensity of an \gls{interaction} decreases exponentially with the time separating the \gls{interacting} \glspl{entity}.

\subsubsection{Q3: Can we efficiently model \glspl{interaction}? \gr{$\bullet$} \rd{$\bullet$}}
Our third question raised the problem of efficiently model \glspl{interaction}. Designing a dedicated study to get insights into one aspect or the other of the \gls{interacting} processes is one thing. Designing a scalable and useful method that answers the possible challenges raised by the task at hand is quite another.

The conclusions drawn from Chapter~\ref{Chapter-SBMs} and Chapter~\ref{Chapter-InterRate} are that \glspl{interaction} are sparse (we need \glspl{cluster} to model them) and that \glspl{interaction} are brief (we need to consider time). Besides, we discussed the fact that \gls{interacting} \glspl{entity} may have a short lifespan, and that their semantic \gls{content} must be taken into account. In Chapter~\ref{Chapter-DHPs}, we introduced the steps paving our way to such a model. 

Our approach was based on in-depth modifications of an existing Bayesian prior: the Dirichlet-Hawkes Process. This model was promising regarding our problem: it allowed us to create \glspl{cluster} (\bl{Q1}) that can have a lasting influence in time on ulterior observations (\ora{Q2}). Moreover, it could perform this task in linear time with the size of the dataset. However, it also suffered from several limitations. In particular, it was not fit to retrieve meaningful \glspl{cluster} when textual information conveyed little information or when temporal dynamics were hard to unveil. Furthermore, it assumed that the textual \gls{content} of a document is perfectly correlated to its temporal dynamics, which cannot be rigorously true in real-world processes.

As an indirect way to overcome these limits, we explored the Powered Dirichlet Process as an alternative to the Dirichlet Process. This new prior alleviates the ``\gls{rich-get-richer}'' property of vanilla Dirichlet processes.

We then incorporated this Powered Dirichlet Process into the standard Dirichlet-Hawkes process to create the Powered Dirichlet-Hawkes Process. This formulation improved the results when temporal information or textual \gls{content} were weakly informative and alleviated the hypothesis that textual \gls{content} and temporal dynamics were always perfectly correlated. 
Thus, our approach eventually allowed us to correctly model self-\glspl{interaction} within a \gls{cluster}.

As a final improvement on the proposed approach, we extended the Powered Dirichlet-Hawkes Process to the multivariate case. This way, we can efficiently model \glspl{interaction} between all pairs of \glspl{cluster}, and by extension for all pairs of \glspl{entity} they comprise. The model can also run in constant time, and typically takes 300ms per document on large-scale real-world datasets. 

We therefore provided a way to efficiently model \glspl{interaction}, positively answering our third introductory question (\textbf{Q3}). Efficiency can be understood both as technical efficiency and relevance efficiency. This former because our approach scales well on large real-world datasets and takes a minimal time to process a consequent amount of data. The latter because the final model can answer all the crucial challenges raised by \gls{interaction} modelling, as it...
\begin{itemize}    
    \item considers \glspl{entity}' \gls{content}. An \gls{entity} is no more described as a unique identifier, but instead by its semantic \gls{content}. Two \glspl{entity} that convey the same information are now considered as such and clustered together as a more global \gls{entity} --a topic.
    \item models sparse \glspl{interaction}. \Glspl{entity} are now clustered together into temporal \glspl{cluster}. It makes it feasible to spot \gls{interaction} terms between sets of \glspl{entity}. The lifespan of \glspl{entity} is no more a problem since \glspl{cluster} can comprise \glspl{entity} spanning over extended periods, which also increases the data available for each \gls{cluster}.
    \item models dynamic \glspl{interaction}. Each \gls{cluster} is associated with its own intensity function, which determines its effect on ulterior observations. Eventually, \glspl{entity}’ influence fades away as time goes by.
    \item models multivariate \glspl{interaction}. Every \gls{cluster} interacts as a pair with other \glspl{cluster}, enabling the study of \glspl{interaction} between the inferred topics.
\end{itemize}

\subsubsection{Q4: Do \glspl{interaction} play a significant role in spreading processes? \rd{$\bullet$}}
We conducted extensive experiments on a real-world large-scale dataset made of one year of news headlines gathered from Reddit. This dataset seems fit for \gls{interaction} modelling: \glspl{content} are user-generated, dynamic, and mutating. 
We proposed several ways to assess the role of \glspl{interaction} in this specific dataset. The notion of effective \gls{interaction} has been chosen as a privileged way to evaluate how much \gls{interaction} has a role in the assumed data generation process. Firstly, we recovered our conclusions from Chapter~\ref{Chapter-SBMs}: \glspl{interaction} are sparse. Secondly, we recovered our conclusions from Chapter~\ref{Chapter-InterRate}: \gls{interaction} strength decays over time. By looking at the average role of \glspl{interaction} in this dataset, we saw that \glspl{interaction} play a minor role in such a dataset. On average, \glspl{interaction} increase the instantaneous probability for a new observation to appear by 1\% over its assumed random appearance probability. Without looking at the aggregated statistics, extreme values typically increase this probability by $\sim$10\%, which does not hint toward strong \gls{interaction} effects.

\subsection{General uses for our models}
Each of the models introduced in this manuscript has been presented as a novel way to tackle the \gls{interaction} modelling problem. However, efforts have systematically been made to illustrate other, more general use cases. As a final note, we argue and detail some use cases of our approaches outside of \gls{interaction} modelling.

\subsubsection{Powered Dirichlet Processes}
The Dirichlet process (\acrshort{DP}) is canonical in Bayesian nonparametric modelling. Enumerating the models based on DP is not feasible; among the most popular ones, we find the Latent Dirichlet Allocation model, the Infinite Gaussian Mixture Model, and the non-parametric MMSBMs and the Dirichlet-Hawkes process, all of which have seen numerous extensions depending on the use context. The popularity of Dirichlet processes is explained by the possibility to automatically infer the number of latent groups along with their composition, and the possibility to process data sequentially, in the order of arrival.

A current way to extend models based on DP is to make them hierarchical. Using the hierarchical DP, one can consider a mixture of partitions, by drawing the base distribution of a DP from another DP. Other extensions such as the Pitman-Yor process and the nested DP have been discussed in Section~\ref{PDP}. By proposing the Powered DP (\acrshort{PDP}), we paved the way for a whole new class of extensions. We illustrated the impact of redefining DP by adding a nonlinear dependence term. Possible applications for this advance would already comprise all existing models based on DP and models based on its extension, as the Powered DP is not exclusive. In the case of hierarchical DP for instance, the base distribution from which is drawn a Powered DP would also be drawn from a Powered DP, making the Powered hierarchical DP. Similar extensions and combinations are possible for the nested DP, the hierarchical nested DP, etc.

\subsubsection{Stochastic Block Models}
When working on stochastic block models, our main goal was to design a framework (\acrshort{SIMSBM}, Section~\ref{SIMSBM-model}) that allows us to model almost any type of categorical data. It has been designed so that it can take as many pieces of information as needed as an input context, as many symmetric \glspl{interaction} between these \glspl{entity} as needed, and consider the temporal dimension to model dynamic \gls{cluster} \glspl{membership} (SDSBM extension, Section~\ref{SDSBM}). The flexibility provided by this work fits \gls{interaction} modelling purposes, but also many other applications. The IMMSBM (Section~\ref{IMMSBM-model}) has been accepted as a contribution to recommender systems, which is not explicitly about \glspl{interaction} modelling \citep{Poux2021IMMSBM}. 

Already discussed use cases in recommender systems include product recommendations on online retail websites, movies recommendation on streaming platforms, and songs on music streaming platforms. We also showed these models could be used to predict players' next move on real-world datasets and replicated studies of \citep{Poux2021MMSBMMrBanks} that identified elementary players' behaviours. We argued for a possible application in automated medical diagnosis, where the combination of a few symptoms leads to high-quality diagnoses. We showed it could help predict the next tweet retweeted by Twitter users.

Further uses are encouraged in the field of humanities, where data often lies unexploited due to the lack of explainable, readily usable models. We illustrated a possible use case using a Latin epigraphy database. Projects using SIMSBM to infer the gender and age of antic remains given the tools and objects found in their tombs are currently ongoing. We strongly argued for such use in the SDSBM section. We paid particular attention to humanities recurrent challenges on data availability; our approach is shown to work even with scarce data.

The point is, that the field of applications for (dynamical) SIMSBM ranges way beyond simply \gls{interaction} modelling. Dedicated studies using these new models may provide interesting results in social sciences, medicine and online recommendation.

\subsubsection{Dirichlet-Point processes}
In Chapter~\ref{Chapter-DHPs} and in this Conclusion, we extend the idea of Dirichlet-Hawkes processes to a broader class of models. The challenges inherent to Dirichlet-Hawkes processes have been answered by developing the Multivariate Powered Dirichlet-Hawkes process. This model can specifically be used to create a time-lined summary of event streams. At a time when internet \gls{content} appears at an unprecedented pace, dedicated tools for big data summarizing will become increasingly necessary in many applications. Understanding ongoing trends on social platforms would help identify topics of interest, or simply provide these platform’s users with an overview of what is going on. Several digital newspapers manually compile trends of interest that appeared over a day, a week; such tasks could be helped if not automated by such tools.

Another possible application specific to the \acrshort{MPDHP} is the understanding of publication mechanisms. Significant research is being done in identifying and countering fake-news \gls{diffusion} on social media. The tool we developed allows us to get insights into the way topics relate to each other. Dedicated studies on the interplay between fake-news topics and disclaimers could help develop countering strategies. Going one step further, we showed Dirichlet-Point processes could be used to make unsupervised modelling of topic-dependent spreading subnetworks. In the same line as before, automatically identifying fake news spreading subnetworks could help to surgically burst opinion bubbles by encouraging new links to other communities. As a more direct approach, it could help target specific nodes with disclaimers.

\section{Perspectives}
\subsection{Towards more general block-modelling approaches}
In Chapter~\ref{Chapter-SBMs}, we developed a global framework that allows us to model \glspl{interaction} of any given order using any size of context. We then proposed an extension to our method that allows us to consider time. We believe that this modelling flexibility can serve as a base to develop improved, even more flexible models to tackle a range of problems. We consider two possible extensions below.

\subsubsection{Considering time as a continuous variable}
In Section~\ref{SDSBM}, we proposed a way to model time as an additional constraint to the stochastic block modelling approach. There, time is discrete and one model is inferred for each time slice, conditional on models from other time slices. However, we discussed in subsequent sections how slicing time in discrete intervals can induce biases in the modelling. A possible lead to alleviate this problem would be to merge our work on \acrshort{SIMSBM} (Section~\ref{SIMSBM}) with the Dirichlet-Point processes discussed in Chapter~\ref{Chapter-DHPs}. In particular, SIMSBM uses a Dirichlet prior as \textit{a priori} on its \gls{membership} vectors. It happens that some works explored this direction to derive a non-parametric version of \acrshort{MMSBM}, by expressing the Dirichlet prior as a Chinese Restaurant process \citep{Fan2015DynInfMMSBM} --as what we did in Chapter~\ref{Chapter-DHPs}. Once the MMSBM is expressed as a sequential Dirichlet process, it might be possible to include the advances in Dirichlet-Point processes as an explicit way to model time as a continuous variable. In general, this approach could make SIMSBM-based approaches fit to consider continuous data in general, provided it obeys an underlying point process.

\subsubsection{Considering nodes' metadata}
Up to now, we modelled \glspl{interaction} at the level of the \gls{interacting} \glspl{entity} themselves. Using our Stochastic Block Modelling framework, it is now possible to account for the \glspl{entity}’ \gls{content} and \gls{interaction} time. However, this is not the most elegant way to consider the context in which a \gls{piece of information} interacts with others. Recent works based on similar models proposed to model nodes' metadata as an additional layer, whose links can be activated or deactivated \citep{Oscar2022NodeMetadata}. In this case, metadata does not have to be of a given type, not it is mandatory for it to be useful. The inclusion of these advances into the proposed framework would make a significant step towards an \acrshort{SBM} that could be applied to any problem at hand with minimal model designing effort.

\subsection{Improving the Multivariate Powered DHP}
\subsubsection{Accounting for exogenous data generation}
A consideration that we factored out from our analysis is the role of exogenous data generation. It has been underlined on some occasions that documents can get published according to dynamics external to the dataset. Its apparition could have been conditioned by other media sources (TV, radio, ...), social links not accounted for on Twitter or simply a demonstration of free will, and therefore should not be included in our dataset generative assumption. In \citep{Myers2012ExternalInfluence}, the authors model the rate of arrival of documents from such exogenous influence using temporal point processes. They conclude that the role of external influence is not trivial. Supporting this claim, \citep{He2015HawkesTopic} introduced a term that accounts for exogenous events in Hawkes-based modelling. The fact that both works make extensive use of temporal point processes in their modelling suggests that the inclusion of their findings into \acrshort{MPDHP} is doable and would certainly yield insightful results. 

\subsubsection{Going further than Dirichlet-Hawkes processes}
In Chapter~\ref{Chapter-DHPs}, we studied in-depth the Dirichlet-Hawkes process and various models that can be built on them. However, the full picture is broader than simply the association of Dirichlet processes and Hawkes processes. Instead, the method described throughout Chapter~\ref{Chapter-DHPs} can be applied to merge any Dirichlet process (hierarchical, nested, or powered) or variants (Indian Buffet Process, Pitman-Yor process, etc.) with any point process (Hawkes, Cox, Poisson, Determinantal, Geometric, etc.). We believe that the resulting Dirichlet-Point processes are powerful tools that can adapt to many modelling problems. This field of such combinations has been little explored up to now and may offer interesting insights for further studies.

\subsection{Considering the network structure}
In this manuscript, we considered both the \gls{content} and the dynamics of spreading \glspl{entity} to unveil \glspl{interaction}. However, our models are not fit to consider the structure of the network \glspl{entity} \gls{spread} on. An interesting lead to explore is how \glspl{interaction} between pieces of information differ at the user level, depending on their position in the spreading network.

Both as a final perspective and as an illustration of a broader use for Dirichlet-Point processes, we propose to develop a model that considers the network structure based on Dirichlet-Point processes. We sketch a model that can jointly infer dynamic (Chapter~\ref{Chapter-InterRate}) \glspl{cluster} (Chapter~\ref{Chapter-SBMs}) of textual documents (Chapter~\ref{Chapter-DHPs}) spreading online \textit{and} the subnetworks they \gls{spread} along. We did not find any previous attempt to \textit{jointly} infer these parameters, by using an iteration of the Dirichlet-Point processes discussed in this conclusion.

\subsubsection{Possible lead: Dirichlet-Survival process}
We can bridge the gap between network inference and dynamic clustering models by defining the \textbf{Dirichlet-Survival process}.

\paragraph{Network inference as a survival process}
In \citep{GomezRodriguez2013SurvivalAnalysis}, the authors demonstrate that most of the then-existing underlying network inference models can be derived as a special case of a global framework. Without entering the details here, it is shown that using a specific formulation of Survival processes allows retrieving network inference models such as NetRate \citep{GomezRodriguez2011NetRate}, KernelCascade \citep{Du2012KernelCascade}, MoNet \citep{Wang2012Monet}, InfoPath \citep{GomezRodriguez2013InfoPath}, etc.

Each of these models can be characterized by a single intensity function, similar to the Hawkes intensity discussed in Chapter~\ref{Chapter-DHPs}, named the hazard rate. For the interested reader, we detail how to get to this function for the NetRate model \citep{GomezRodriguez2011NetRate} in Appendix~\ref{Appendix5}.

We write this intensity function $\lambda (t_i^c \vert t_j^c, \alpha_{j,i})$. It represents the instantaneous probability for a node $i$ to be \gls{infected} at time $t_i^c$ by an infection $c$ because of the node $j$ that got \gls{infected} by $c$ at time $t_j^c$. The strength of the link between $i$ and $j$ is encoded in $\alpha_{j,i}$. As this intensity function results from a survival point-process similar to the Hawkes process, we will see that we can substitute it to the counts of a Dirichlet process to create a Dirichlet-Survival process.

\paragraph{Dirichlet-Survival prior}
As in \citep{Du2015DHP,Valera2017HDHP,Tan2018IBHP} and what we did throughout Chapter~\ref{Chapter-DHPs}, we will create a new Dirichlet-Point process by merging the Dirichlet Process in its Chinese Restaurant iterative metaphor, to a Point process such as the one characterizing NetRate. We write the intensity of such process $\lambda (t_i^c \vert t_j^c, \alpha_{j,i})$.

Doing so essentially breaks down the temporal dynamics at the level of the network's nodes' in-going edges. Each edge is now associated with its own point process. This makes a yet unexplored bridge between Dirichlet processes and Survival analysis. 

Instead of considering a single network on which data spreads, we assume there is any number of such subnetworks, each associated with a \gls{cluster}. Instead of associating one point process to one \gls{cluster} as in Chapter~\ref{Chapter-DHPs}, we associate a collection of point processes to one \gls{cluster} --one per edge in the network. We write $\lambda (t_i^c \vert t_j^c, \alpha_{j,i}^{(k)})$ the intensity associated with the edge between $i$ and $j$ given their infection times by $c$ are separated of a time $\Delta t = t_j-t_i$ \gls{cluster} $k$.

An infection event from \gls{cascade} $c$ is now assumed to have a given probability of being trigger by any of the $k$ existing \glspl{cluster} (or subnetworks) on which information \gls{spread}. By substituting the counts in the Dirichlet process with the total intensity on node $i$ due to all its neighbours $\mathcal{H}_{i,c}$ that got \gls{infected} earlier by $c$ in subnetwork $k$, noted $\Lambda_i^{c,(k)} = \sum_{j \in \mathcal{H}_{i,c}} \lambda (t_i^c \vert t_j^c, \alpha_{j,i}^{(k)})$. Let $\lambda_0$ be the probability that the observation did not get triggered by any existing subnetwork --the concentration parameter.
\begin{equation}
    \label{eq-DSP}
        P(s_n = k \vert \mathcal{H}_{i,c}) = \begin{cases} 
        \frac{\Lambda_i^{c,(k)}}{\lambda_0^{(K+1)} + \Lambda_i^{c,(k)}} \text{ if k = 1, ..., K}\\
        \frac{\lambda_0^{(K+1)}}{\lambda_0^{(K+1)} + \Lambda_i^{c,(k)}} \text{ if k = K+1}
        \end{cases}
\end{equation}
\myequations{\ \ \ \ Conclusion - Dirichlet-Survival Process}

\input{Tables/Chapter_5/table-res-houston}
\paragraph{Some preliminary experimental results}
Similarly to what we did in Chapter~\ref{Chapter-DHPs}, we coupled the Dirichlet-Survival prior to a Dirichlet-Multinomial language model. 
In this section, we sketch some experimental results as a way to support the relevance of further studies using Dirichlet-Point processes, especially for modelling \glspl{interaction} in information \gls{spread}.

We run the Dirichlet-Survival process on synthetic data first. We generate an Erd\"os-Renye network (\textbf{ER}) \citep{Erdos1960ERgraph}, on which we simulate information \gls{spread} using a Dirichlet-Survival process. Using 500 nodes, we create 5 subnetworks of size of approximately 250 nodes and 700 edges, one per different topic.

We present the results obtained using Dirichlet-Survival processes in Table~\ref{Houston-tabMetrics}. We compare to several models: TC \citep{Du2013TopicCascade}, DHP \citep{Du2015DHP} and NetRate \citep{GomezRodriguez2011NetRate}. The NetRate model infers edges without taking \glspl{cluster} or textual \gls{content} into account. The DHP considers the dynamics and the textual \gls{content}, but not the network's structure. The TC model first infers textual \glspl{cluster} and then infers one subnetwork for each of them. Overall, we see that accounting for all three data types (\gls{content}, time, and structure) increases the model’s performance.

We finally illustrate a possible use case on real-world data in Fig.~\ref{fig-outputrw}. We used data from the Memetracker dataset \citep{Leskovec2009Memetracker}. %

\begin{figure*}
    \centering
    \includegraphics[width=\textwidth]{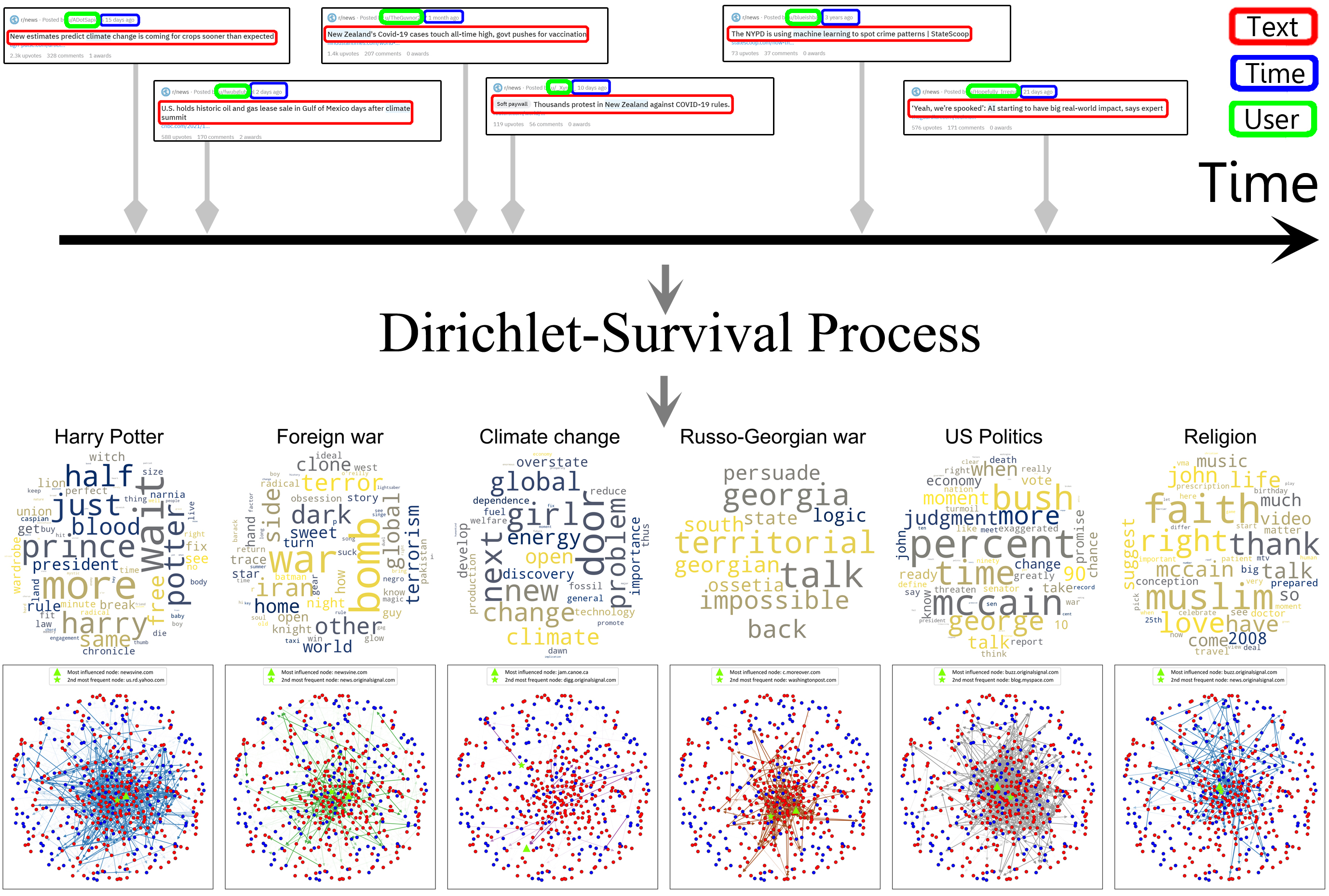}
    \caption[DSP - Dirichlet-survival Process on real-world data]{Dirichlet-survival process applied to real data. (Top row) The most frequent words within the inferred \gls{cluster}. (Bottom row) The inferred subnetworks associated with each \gls{cluster}. Mass media nodes are represented in red, and blogs nodes are represented in blue. The most influenced node is the one having the highest total in-going transmission rate. The nodes’ positions are identical for every network to ease comparison.}
    \label{fig-outputrw}
\end{figure*}

\subsubsection{Perspectives on \gls{interaction} modelling}
We briefly showed a possible use case for Dirichlet-Point processes in information \gls{spread} modelling by defining the Dirichlet-Survival process.
In this case, \glspl{interaction} are not modelled since \glspl{cascade} \gls{spread} independently from each other. 

However, it paves the way for defining even more complex Dirichlet-Point processes that consider time, \gls{content} and structure of information \gls{spread} in \gls{interaction} modelling. In our case, substituting the non-\gls{interacting} \glspl{cascade} model with a more elaborated underlying network inference model might allow us to uncover \gls{interaction} mechanisms. In particular, recent years have seen some works tackling the network inference problem by assuming a Hawkes process on each edge of a network. Considering a multivariate Hawkes process instead, similarly to what we did Section~\ref{MPDHP} would also be an interesting lead.

It should be noted, however, that the Dirichlet-Survival process introduced in this chapter may not be the best approach to model such phenomena. In particular, comparing its performances to \citep{He2015HawkesTopic,Barbieri2017SurvivalFactorization} is needed to get solid results on this specific application. However, the mere fact that we can almost instantly elaborate new models to tackle new problems argues in favour of the use of Dirichlet-Point processes in a broader context.

\section{Final words}
The work presented in this manuscript is the \gls{outcome} of three years of questioning, exploring, and discovering various aspects of \glspl{interaction} at stake in information \gls{spread}. As it seems to be the norm in research, this initial question served as a guiding thread. A thread that sewed this manuscript through large and disconnected areas of the machine learning canvas: stochastic block models, dynamic networks inference, Dirichlet processes, temporal point processes. Here, our developments over those pieces are used to answer our problematic. However, we essentially focused on improving the backbone of these areas, whose specific application to \gls{interaction} modelling yields interesting insights. From a broader perspective, our work is also intended as a contribution to machine learning in general; Formulating alternative Dirichlet Processes, granting flexibility and dynamism to block models, merging Dirichlet and Point processes, have implications that range beyond solely \glspl{interaction} modelling. Our efforts resulted in explainable, scalable and flexible methods that can tackle a wide range of problems. It is our sincere hope that the advances presented in this manuscript will be of greater help for researchers of all horizons, either to improve over them or to use them as tools for answering concrete real-world questions.

%% file: Tables/Chapter_5/table-res-houston.tex
\begin{table}[h]
    \caption[DSP - Numerical results of Dirichlet-Survival process on synthetic data]{Numerical results of Dirichlet-Survival process, TopicCascade, Dirichlet-Hawkes process and NetRate models.
    The AUC, F1 score and MAE are computed considering every top cluster's edges at once so there is no error to report. 
    \label{Houston-tabMetrics}}
    \centering
    \setlength{\lgCase}{2.6cm}
    \noindent\makebox[\textwidth]{
    \begin{tabular}{|p{0.08\lgCase}|p{0.5\lgCase}|p{0.8\lgCase}|p{0.8\lgCase}|p{0.8\lgCase}|p{0.8\lgCase}|p{0\lgCase}}
    \cline{3-6}
      \multicolumn{0}{c}{\rotatebox[origin=c]{90}{}} &  & \centering Dir-Surv & \centering TC & \centering DHP & \centering NetRate & \\

     \cline{1-6}
        \centering\multirow{5}{*}{\centering\rotatebox[origin=c]{90}{ER}} & \centering NMI & \centering \textbf{0.787} & \centering 0.711 & \centering 0.638 & \centering - & \\
        
         & \centering ARI & \centering \textbf{0.631} & \centering 0.488 & \centering 0.411 & \centering - & \\
         
     \cdashline{2-6}
        & \centering AUC & \centering \textbf{0.849} & \centering 0.800 & \centering - & \centering 0.659 & \\
        
         & \centering F1 & \centering \textbf{0.263} & \centering 0.176 & \centering - & \centering 0.005 & \\
         
         & \centering MAE & \centering \textbf{0.229} & \centering 0.278 & \centering - & \centering 0.481 & \\

     \cline{1-6}
    \end{tabular}
    }
\end{table}

%% file: Appendices/Appendix-2.tex
\chapter{Appendix -- Stochastic Block Models}
\label{Appendix2}

\section{SIMSBM - Additional experimental results}
\label{SI-SIMSBM-Additional-Results}
\input{Tables/Chapter_2/table-res-SIMSBM-SM}

\section{IMMSBM - Datasets}
\label{SI-IMMSBM-datasets}
\subsection{Medical records}
The Pubmed dataset collect has been inspired by \citep{Zhou2014DiseaseSymptomNet}. Every article on PubMed is manually annotated by experts with a list of keywords describing the main topics of the publication. We downloaded a list of 322 symptoms and 4,442 diseases provided by \citep{Zhou2014DiseaseSymptomNet}. Then, we used the PubMed API to query each one of the symptom/disease keywords aforementioned. For each result, we got a list of every publication in which the keyword is among the main topics. Then, we build the dataset by considering every publication in which there is at least one symptom and one disease. Finally, we create the triplets (symptom1, symptom2, disease) by looking at all the pairs of symptoms in an article and linking each one of them to all the diseases observed in the same article. In the end, we are left with a total of 52,833,690 observed triplets, distributed over 15,809,271 PubMed publications.

\subsection{Spotify}
The Spotify dataset has been collected using the Spotify API. We randomly sampled 2,000 playlists using the keywords "english" and "rock", which corresponds to a total of 135,100 songs. Then, for each playlist, we used a running window of 4 songs to build the dataset. The artist of the song immediately after the running window is the output we aim at predicting, x, and the artists of the 4 songs within the running window are the interacting inputs. Once again, we consider all the possible pairs of artists in the running window and associate them to the output artist x. Note that we only considered the artists that appear more than 50 times in the whole dataset, for the sake of statistical relevance. The resulting dataset is consists of 1,236,965 triplets for 2,028 artists.

\subsection{Twitter}
We gathered the Twitter dataset used in \citep{Hodas2014DataSetTwitter}. It consists in a collection of all the tweets containing URLs posted during the month of October 2010. A first operation consisted in cleaning the dataset of URLs that are considered as aggressive advertising. To do so, we considered only the URLs whose retweets built a chain of length at least 50; this choice comes from the idea that commercial spams are not likely to be retweeted by actual users and therefore do not create chains. Secondly, we considered only the users who have not tweeted a given URL more than 5 times, this behavior being an activity typical of spamming bots. Doing so, we are left with tweets that are mostly coming from the non-commercial activity of human users. Then, we follow a dataset building process similar to \citep{Myers2012CoC}. For each user, we slice her feed + tweets temporal sequence in intervals separated by the tweets of the user. Every time a user tweets something, the interval ends. An interval therefore consists of the tweet of the user and all the tweets she has been exposed to right before tweeting. Following the suggestion of \citep{Myers2012CoC}, we only consider the 3 last tweets the user has been exposed to before retweeting one of them. Each one of these intervals form an entry of our message+answer dataset (3 last entries in the feed + next tweets). In the end, we are left with 284,837 intervals containing a total of 2,110 different tweeted URLs. Our dataset then consists of 1,181,543 triplets.

\subsection{Reddit}
Finally, we downloaded the May 2019 Reddit dataset from the data repository pushhift.io, which stores regular saves of the comments posted on the website. We chose to consider only the comments made in the subreddit r/news. Reddit's comments system work as a directed tree network, where each answer to a given comment initiates a new branch. We considered pairs of messages such as one (the answer) has for direct parent the other one (the comment). We then extracted all the named entities in both of them using the Spacy Python library. For each pair of named entities in the comment, we associated every name entity in the answer. We consider here only the named entities that appear at least 200 times in the subreddit, for the same reason as for the Spotify dataset. The final dataset results in 35,364,725 triplets for a total of 1,656 named entities.

\section{IMMSBM - Upper limit to predictions}
\label{SI-IMMSBM-upperLim}
We derive an analytical expression for the upper-limit to our model for a given dataset. Explicitly, we analytically maximize the likelihood according to each entry of the dataset.

We enforce the constraint that the sum over the output space of probabilities given any observations made has to sum to 1. To do so, we use a Lagrange multiplier $\lambda_{obs}$ for every different observation (in the case of our model: for every different triplet). The log-likelihood then takes the following form:
\begin{equation}
    \label{eqLimHauteP}
    \begin{split}
        &\ell = \sum_{(obs, x)} \ln P_{obs}(x) - \sum_{obs}\lambda_{obs}(\sum_x P_{(obs,x)}-1)\\
        \Leftrightarrow\ &\frac{\partial \ell}{\partial P_{obs_i}(x_i)} = \sum_{\partial (obs_i,x_j)}\frac{1}{P_{obs_i}(x_j)} - \lambda_{obs_i} = 0 \\
        \Leftrightarrow\ &P_{obs_i}(x_j) = \frac{1}{\lambda_{obs_i}}\sum_{\partial (obs_i,x_j)} 1
    \end{split}
\end{equation}
Where $obs_i$ correspond to any given couple of inputs (i,j) in the model presented in the main paper. We use the following notation: ${\partial (obs_i,x_j) = \{ obs \vert (obs_i, x_j) \in R^{\circ} \}}$, with $R^{\circ}$ the dataset entries. Therefore, we can define $\sum_{\partial (obs_i,x_j)} 1 \equiv N_{obs_i,x_j}$ the number of times $obs_i$ appears jointly with $x_j$ in the dataset. We are now looking for the $\lambda_{obs_i}$ maximizing the likelihood:
\begin{equation}
    \label{eqLimHauteLamb}
    \begin{split}
        &\frac{\partial \ell}{\partial \lambda_{obs_i}} = \sum_x P_{obs_i}(x)-1  = 0 \\
        \Leftrightarrow\ &\sum_x \frac{N_{obs_i,x}}{\lambda_{obs_i}} = 1 \\
        \Leftrightarrow\ &\lambda_{obs_i} = \sum_x N_{obs_i,x}\\
    \end{split}
\end{equation}
Finally, plugging Eqs.\ref{eqLimHauteP} and \ref{eqLimHauteLamb} together, we obtain:
\begin{equation}
    P_{obs_i}(x_j) = \frac{N_{obs_i, x_j}}{\sum_x N_{obs_i, x}}
\end{equation}
This equation gives the probability maximizing the likelihood for any entry of the dataset. In the main model, it translates to $P_{(i,j)}(x) = N_{(i,j), x}/\sum_x' N_{(i,j), x'}$ with $N_{(i,j),x}$ the number of times output x has been witnessed after a pair of inputs (i,j). Note that this is simply the frequency of an output given a pair of input entities.

Keep in mind that the result we just derived gives perfect predictions only for a particular dataset, and therefore and has no global predictive value. This tool is useful when it comes to assess the performance of other models' results, but is unusable in prediction tasks.

\comment{

\subsection{Clusters composition}
We provide in Tables \ref{TabPubMed},\ref{TabRed} and \ref{TabSpot} the exhaustive list of entities belonging to each cluster with more than 50\% membership (80\% for Spotify) for the datasets considered. 

Each one of the clusters has been manually given a name according to the lists present in the tables. The names presented here are the same as the ones shown in Fig.3 of the main article.
\begin{table*}
\caption{Composition of clusters for the PubMed dataset.}
\label{TabPubMed}
\centering
\setlength{\lgCase}{12cm}
\begin{tabular}{|p{0.2\lgCase}|p{\lgCase}|p{0\lgCase}}
 \cline{1-2}
 \centering Cluster & \centering Entities &  \\
 \cline{1-2}
1 - Vision & anisocoria, blindness, color vision defects, diplopia, eye hemorrhage, eye manifestations, eye pain, miosis, photophobia, pseudophakia, pupil disorders, scotoma, sensation disorders, tonic pupil, usher syndromes, vision, vision disorders, waterhouse-friderichsen syndrome & \\ 
 \cline{1-2} 
 2 - Blood & ecchymosis, glossalgia, oral hemorrhage, purpura & \\ 
 \cline{1-2} 
 3 - Seizures & seizures & \\ 
 \cline{1-2} 
 4 - Neurons & neurologic manifestations & \\ 
 \cline{1-2} 
 5 - Male genitals & dysmenorrhea, dysuria, encopresis, hirsutism, pelvic pain, prostatism, urinary bladder, virilism & \\ 
 \cline{1-2} 
 6 - Eating dis. & hyperphagia, obesity, overweight & \\ 
 \cline{1-2} 
 7 - Coma & coma, consciousness disorders, hypothermia, tetany, unconsciousness & \\ 
 \cline{1-2} 
 8 - Pain & pain & \\ 
 \cline{1-2} 
 9 - Facial pain & facial pain, toothache, trismus & \\ 
 \cline{1-2} 
 10 - Back pain & back pain, hot flashes, low back pain & \\ 
 \cline{1-2} 
 11 - Muscle & muscle cramp, muscle hypertonia, muscle rigidity, muscle spasticity, muscle weakness, myotonia & \\ 
 \cline{1-2} 
 12 - Mental & cafe-au-lait spots, deaf-blind disorders, hydrops fetalis, mental retardation, morning sickness & \\ 
 \cline{1-2} 
 13 - Paralysis & paralysis & \\ 
 \cline{1-2} 
 14 - Undefined & neuralgia, piriformis muscle syndrome & \\ 
 \cline{1-2} 
 15 - Speech disorder & articulation disorders, communication disorders, language disorders, speech disorders & \\ 
 \cline{1-2} 
 16 - Kidney & albuminuria, flank pain, hemoglobinuria, oliguria, polyuria, proteinuria, renal colic & \\ 
 \cline{1-2} 
 17 - Voice & aphonia, dysphonia, hoarseness, orthostatic intolerance, respiratory aspiration, vocal cord paralysis, voice disorders & \\ 
 \cline{1-2} 
 18 - Neuropsychatry & agraphia, alien hand syndrome, aphasia, gerstmann syndrome, headache, hemiplegia, paresis, systolic murmurs & \\ 
 \cline{1-2} 
 19 - Sceptisemia & hypergammaglobulinemia, pruritus, purpura fulminans, tinea pedis & \\ 
 \cline{1-2} 
 20 - Vertigo & dizziness, earache, hearing disorders, hyperacusis, motion sickness, presbycusis, tinnitus, vertigo & \\ 
 \cline{1-2} 
 21 - Mental & agnosia, alexia, anomia, auditory perceptual disorders, bulimia, delirium, dyslexia, hallucinations, learning disorders, memory disorders, mutism, neurobehavioral manifestations, perceptual disorders, phantom limb, pseudobulbar palsy, stuttering & \\ 
 \cline{1-2} 
 22 - Articulations & arthralgia, metatarsalgia, shoulder pain & \\ 
 \cline{1-2} 
 23 - Muscle weight & birth weight, body weight, labor pain, sarcopenia & \\ 
 \cline{1-2} 
 24 - Heart & acute coronary syndrome, amaurosis fugax, angina, angina pectoris, chest pain, heart murmurs, syncope & \\ 
 \cline{1-2} 
 25 - Liver & hyperemesis gravidarum, jaundice, necrolytic migratory erythema & \\ 
 \cline{1-2} 
 26 - Breath & cough, dyspnea, hemoptysis, respiratory paralysis & \\ 
 \cline{1-2} 
 27 - Sleep & sleep deprivation, sleep disorders, snoring & \\ 
 \cline{1-2} 
 28 - Abdominal pain & abdomen, abdominal pain, diarrhea, dyspepsia, eructation, flatulence, gastroparesis, vomiting & \\ 
 \cline{1-2} 
 29 - Movement & amblyopia, ataxia, athetosis, catalepsy, cerebellar ataxia, chorea, dyskinesias, dystonia, hyperkinesis, hypokinesia, myoclonus, supranuclear palsy, tics, torticollis, tremor & \\ 
 \cline{1-2} 
 30 - Brain & ageusia, amnesia, cerebrospinal fluid otorrhea, cerebrospinal fluid rhinorrhea, decerebrate state, persistent vegetative state & \\ 
 \cline{1-2} 
 
 \cline{1-2}
\end{tabular}
\end{table*}

\begin{table*}
\caption{Composition of clusters for the Reddit dataset.}
\label{TabRed}
\centering
\setlength{\lgCase}{12cm}
\begin{tabular}{|p{0.2\lgCase}|p{\lgCase}|p{0\lgCase}}
 \cline{1-2}
 \centering Cluster & \centering Entities &  \\
 \cline{1-2}
1 - Komodo dragon & dragon, eye, king, komodo, link, mate, ring & \\ 
 \cline{1-2} 
 2 - Diseases & antivaxxer, autism, child, disease, gene, herpes, immunity, measles, outbreak, parent, vaccination, vaccine & \\ 
 \cline{1-2} 
 3 - Weapons & ar15, felon, firearm, gun, handgun, knife, nra, nz, ownership, pistol, rifle, shooter, shooting, shotgun, weapon, zealand & \\ 
 \cline{1-2} 
 4 - Driving & accident, car, driver, driving, drunk, dui, license, road, spider, traffic, vehicle & \\ 
 \cline{1-2} 
 5 - Media & click, cnn, coverage, forum, fox, karma, marathon, media, meme, motive, news, outlet, page, post, reddit, sub, subreddit, td, upvote & \\ 
 \cline{1-2} 
 6 - J. Smollet hoax & hoax, jussie, maga, smollett, supporter & \\ 
 \cline{1-2} 
 7 - Punishment & capital, cruel, death, execution, inmate, murderer, offender, penalty, prison, prisoner, punishment, rehabilitation, revenge, row, sentence, torture & \\ 
 \cline{1-2} 
 8 - Catholicism & aid, bible, catholic, church, conversion, gay, homosexuality, marriage, mormon, pedophile, priest, satan, scout, temple, therapy, troop & \\ 
 \cline{1-2} 
 9 - Planes & air, airline, airport, battery, boeing, bottle, delivery, flight, passenger, pilot, plane, seat, tip, tsa & \\ 
 \cline{1-2} 
 10 - Statistics & average, cdc, homicide, increase, math, number, percent, rate, stat, statistic, study, suicide & \\ 
 \cline{1-2} 
 11 - Managment & business, ceo, charity, corporation, economy, employee, executive, fund, incentive, income, investment, market, pay, payer, profit, raise, revenue, salary, shareholder, spending, stock, tax, wage, wealth & \\ 
 \cline{1-2} 
 12 - Undefined & fallacy, slavery & \\ 
 \cline{1-2} 
 13 - Animals & animal, bag, beef, burger, cat, chicken, cow, dog, dumpster, factory, farm, farmer, food, fry, meat, milk, pet, pig, pit, plastic, poacher, puppy, restaurant, taste, vegan & \\ 
 \cline{1-2} 
 14 - Notre-Dame fire & art, building, cathedral, construction, dame, design, fire, france, glas, lead, notre, pari, roof, spire, stone, structure, tower & \\ 
 \cline{1-2} 
 15 - US politics & bernie, bush, campaign, candidate, clinton, congres, dem, democrat, election, gop, governor, hillary, obama, party, patriot, president, republican, senator, vote, voter & \\ 
 \cline{1-2} 
 16 - Police & cop, department, officer, polouse, taser, union & \\ 
 \cline{1-2} 
 17 - Feminism & birth, consent, discrimination, equality, feminism, feminist, gap, gender, incel, male, partner, peni, pill, sexism, tran, woman & \\ 
 \cline{1-2} 
 18 - School & bathroom, bully, clas, class, college, district, grade, school, student, teacher, university & \\ 
 \cline{1-2} 
 19 - US laws & amendment, constitution, freedom, liberty, ruling, speech, supreme, violation & \\ 
 \cline{1-2} 
 20 - International & arabia, china, east, eu, iran, iraq, oil, regime, saudi, socialism, tank, venezuela, war & \\ 
 \cline{1-2} 
 21 - Medicine & addict, addiction, cure, doctor, drug, fentanyl, healthcare, hospital, insulin, insurance, med, medication, medicine, patient, pharma, professional, substance, surgery, tooth, treatment, va & \\ 
 \cline{1-2} 
 22 - GAFA & amazon, apple, artist, computer, device, facebook, password, phone, software, tech, technology, thief & \\ 
 \cline{1-2} 
 23 - Smokers & capacity, cigarette, corner, disney, pack, park, smell, smoke, smoker, smoking, weed & \\ 
 \cline{1-2} 
 24 - Immigration & asylum, border, citizenship, entry, felony, illegal, immigrant, immigration, mexico, migrant, militia, patrol, port & \\ 
 \cline{1-2} 
 25 - Trial & accusation, attorney, bail, charge, client, condom, da, judge, jury, lawyer, plea, precedent, prosecution, prosecutor, rape, sweden, trial & \\ 
 \cline{1-2} 
 26 - Extremisms & alt, antifa, antisemitism, atheist, christian, christianity, extremist, fascism, fascist, hitler, ideology, islam, islamic, israel, jew, kkk, mosque, muslim, nationalism, nationalist, nazi, religion, supremacist, supremacy, symbol, terrorism, terrorist & \\ 
 \cline{1-2} 
 27 - Wikileak & alright, assange, attitude, dirt, dnc, email, info, intelligence, journalist, leak, proof, propaganda, public, putin, russia, russian, support, west, wikileak & \\ 
 \cline{1-2} 
 28 - US Coasts & apartment, area, baltimore, chicago, city, coast, francisco, homeles, housing, la, neighborhood, nyc, poop, san, sf, shelter, shithole, south, town, weather, york & \\ 
 \cline{1-2} 
 29 - Hobbies & band, channel, character, club, cult, episode, festival, game, golf, hitman, movie, music, night, pizza, player, podcast, robert, season, series, tiger, watch & \\ 
 \cline{1-2} 
 30 - Students & age, beer, card, credit, debt, drink, drinking, loan, minimum & \\ 
 \cline{1-2}

 \cline{1-2}
\end{tabular}
\end{table*}

\begin{table*}
\caption{Composition of clusters for the Spotify dataset. Note that here we define the membership threshold at 80\%: since T is smaller, more entities are likely to belong to any cluster more than 50\%, hence leading to hardly human-readable results.}
\label{TabSpot}
\centering
\setlength{\lgCase}{12cm}
\begin{tabular}{|p{0.25\lgCase}|p{\lgCase}|p{0\lgCase}}
 \cline{1-2}
 \centering Cluster & \centering Entities &  \\
 \cline{1-2}
1 - Latino & alaska y dinarama, alejandro sanz, andrés calamaro, arena hash, asian kung-fu generation, aterciopelados, bacilos, café tacvba, caifanes, d.a.n., donots, duncan dhu, ekhymosis, el gran silencio, el tri, elefante, fito paez, fobia, frankie valli and the four seasons, gustavo cerati, hawkwind, heroes del silencio, hombres g, jaguares, jarabe de palo, jethro tull, jr jr, jumbo, la ley, la mosca tse-tse, la unión, los bunkers, los claxons, los enanitos verdes, los fabulosos cadillacs, los prisioneros, los tres, magneto, maldita vecindad y los hijos del 5to. patio, maná, mar de copas, marty robbins, mikel erentxun, molotov, mägo de oz, nacha pop, nek, paellas, pedro suárez-vértiz, rata blanca, soda stereo, the horrors, van der graaf generator, victimas del doctor cerebro, vilma palma e vampiros & \\ 
 \cline{1-2} 
 2 - Blues rock & cartel de santa, david gilmour, hank williams, johnny hallyday, keith richards, kiss, melbourne symphony orchestra, myles kennedy & \\ 
 \cline{1-2} 
 3 - Christian rock \& Alt. rock & 6cyclemind, angee rozul, audio adrenaline, bamboo, beyond creation, charles bradley, cheese, chicosci, chris quilala, cueshé, eraserheads, francism, franco, hale, hilera, hillsong united, imago, kamikazee, kjwan, mayra andrade, mosaic msc, newsboys, orange and lemons, parokya ni edgar, petra, powerspoonz, rico blanco, rivermaya, sandwich, shawn mcdonald, slapshock, static-x, stray kids, the afters, the dawn, the speaks, the stone foxes, the twang, tobymac, typecast, urbandub & \\ 
 \cline{1-2} 
 4 - Metal & a skylit drive, fort minor, heilung, limp bizkit, p.o.d., papa roach, powerwolf & \\ 
 \cline{1-2} 
 5 - J-pop \& Alt. rock & alison mosshart, bigbang, bring me the horizon, crossfaith, dani filth, egoist, exo, exo-k, grimes, hyde, lisa, one ok rock, placebo, primal scream, reol, self deception, the oral cigarettes, vamps, ves tal vez, vixx, while she sleeps & \\ 
 \cline{1-2} 
 6 - Electro \& House & afrojack, diplo, dragonland, gloria trevi, james bay, labrinth, lost frequencies, lsd, nayer, scouting for girls, the word alive, walk off the earth & \\ 
 \cline{1-2} 
 7 - New wave & alphaville, baltimora, bananarama, barbra streisand, belinda carlisle, billy ocean, bronski beat, cool kids of death, culture club, cyndi lauper, desireless, dreamtale, erasure, eric 'et' thorngren, fine young cannibals, gentle giant, googolplex, haschak sisters, huey lewis and the news, imagination, information society, johnny hates jazz, kool and the gang, laura branigan, level 42, men at work, michael sembello, opeth, orchestral manoeuvres in the dark, pet shop boys, roxy music, soft cell, spandau ballet, starship, tears for fears, the bangles, the buggles, the human league, the j. geils band, the romantics, thompson twins, wang chung & \\ 
 \cline{1-2} 
 8 - Rock'n'roll \& J.Lennon & bill haley and his comets, chuck berry, frankie valli, jay and the americans, john denver, john lennon, little richard, the beach boys, the flux fiddlers, the plastic ono band, van morrison & \\ 
 \cline{1-2} 
 9 - Undefined & jeff buckley, karnivool, nirvana & \\ 
 \cline{1-2} 
 10 - Heavy metal & anthrax, arsenal, avenged sevenfold, beatsteaks & \\ 
 \cline{1-2} 
 11 - Folk rock & coldplay, dark moor, dermot kennedy, james blunt, power quest, taburete, tom walker & \\ 
 \cline{1-2} 
 12 - Recent pop & alan walker, alessia cara, alesso, alexandra porat, aloe blacc, anne-marie, arizona zervas, au/ra, ava max, axwell /\ ingrosso, becky hill, before you exit, benny blanco, billie eilish, blackbear, blackpink, bloodpop®, bonn, bryson tiller, camila cabello, cardi b, chance the rapper, clean bandit, daddy yankee, dakota, daya, demi lovato, digital farm animals, disciples, dj khaled, dj snake, elina, farruko, fifth harmony, francesco yates, grey, hailee steinfeld, halsey, illenium, jack and jack, jax jones, jessie reyez, jonas blue, jonas brothers, jp cooper, julia michaels, justin bieber, k-391, kane brown, khalid, kygo, lauv, lennon stella, liam payne, little mix, luis fonsi, major lazer, maren morris, marshmello, martin garrix, mø, niall horan, nick jonas, noah cyrus, normani, sabrina carpenter, sandro cavazza, selena gomez, shawn mendes, steeleye span, stefflon don, tainy, tiësto, tomine harket, troye sivan, why don't we, william singe, zara larsson, zedd & \\ 
 \cline{1-2} 
 13 - Alt. rock & blutengel, brian eno, greta van fleet, ledger, mew, muse, starsailor & \\ 
 \cline{1-2} 
 14 - Pop/Glam-rock & andy black, asobi seksu, el haragán y compañía, funeral for a friend, panic! at the disco, suspekt, twenty one pilots & \\ 
 \cline{1-2} 
 15 - Garage rock & interpol, kings of leon, liquido, mono, rainbow99, the hives & \\ 
 \cline{1-2} 
\end{tabular}
\end{table*}
}

\section{SDSBM - Explicit derivation of the E-step}
\label{SI-SDSBM-derivationEstep}
\subsection{Short derivation}
This demonstration can be found in \citep{Antonia2016AccurateAndScalableRS,Tarres2019TMBM,Poux2021IMMSBM}. We recall the log-likelihood as defined in the main paper:
\begin{align}
    \label{SDSBM-eq-L}
    \log P(\theta, p \vert R^{\circ}) &\propto \log P(R^{\circ} \vert \theta, p) \prod_t\prod_i P(\theta_i^{(t)})\prod_k P(p_k^{(t)}) \notag \\ 
    = &\sum_{(i,o,t) \in R^{\circ}} \log \sum_{k \in K} \theta_{i,k}^{(t)} p_k^{(t)}(o) \\
    &+ \sum_t \sum_i \log P(\theta_i^{(t)}) \sum_k \log P(p_k^{(t)}) \notag \\
    \geq &\sum_{(i,o,t) \in R^{\circ}}\sum_{k \in K} \omega_{i,o}^{(t)}(k) \log \frac{\theta_{i,k}^{(t)} p_k^{(t)}(o)}{\omega_{i,o}^{(t)}(k)} \notag \\
    &+ \sum_t \sum_i \log P(\theta_i^{(t)}) \sum_k \log P(p_k^{(t)}) \notag
\end{align}
In Eq.\ref{SDSBM-eq-L}, we introduced a proposal distribution $\omega_{i,o}^{(t)}(k)$, that represents the probability of one cluster allocation $k$ given the observation $(i,o,t)$. The last line followed from Jensen's inequality assuming that $\sum_k \omega_{i,o}^{(t)}(k) = 1$. We notice that Jensen's inequality holds as an equality when:
\begin{equation}
    \omega_{i,o}^{(t)}(k) = \frac{\theta_{i,k}^{(t)} p_k^{(t)}(o)}{\sum_{k'} \theta_{i,k'}^{(t)} p_{k'}^{(t)}(o)}
\end{equation}
which provides us with the expectation formula. The prior terms $P(\theta_i^{(t)})$ and $P(p_k^{(t)})$ have no effect on the result as they cancel in the inequality \ref{SDSBM-eq-L}.

\subsection{Full derivation}

The derivation presented in this section follows a well-known general derivation of the EM algorithm, which can be found in C.M. Bishop's \textit{Pattern Recognition and Machine Learning}-p.450 for instance.

We recall that one entry of the dataset $R^{\circ}$ takes the form of a tuple $(i, o, t)$, where $i$ is the input item and $o$ an associated label at time $t$. $k \in K$ denotes the latent variable accounting for cluster allocation among $K$ possible values. 
The total log-likelihood is the sum of each observation's log-likelihood. Without loss of generality, we focus on a single observation $(i,o,t)$. The expression of the log-posterior distribution for one observation reads:
\begin{align}
\label{SDSBM-eq-L2}
    &\log P(\theta^{(t)}, p^{(t)} \vert (i,o,t)) \\
    &\propto \log P(R^{\circ} \vert \theta^{(t)}, p^{(t)}) P(\theta_i^{(t)})\prod_k P(p_k^{(t)}) \notag \\ 
    &= \log P^{(t)}(i,o \vert \theta^{(t)}, p^{(t)}) + \log P(\theta_i^{(t)}) + \sum_k \log P(p_k^{(t)}) \notag
\end{align}
For an observation $(i,o,t) \in R^{\circ}$, we assume a distribution $Q_{i,o}^{(t)}(k)$ on the latent variables associated to it; this distribution is yet to be defined. Because $k$ takes values among $K$ possible ones, we have $\sum_{k \in K} Q_{i,o}^{(t)}(k) = 1$. Given this normalization condition, we can decompose each summed term in Eq.\ref{SDSBM-eq-L2} for any distribution $Q_{i,o}^{(t)}(k)$ as:

\begin{align}
    \label{SDSBM-eq-decomp}
    &\log P^{(t)}(i, o \vert \theta^{(t)}, p^{(t)}) \notag \\
    &= \underbrace{\log P^{(t)}(i, o, k \vert \theta^{(t)}, p^{(t)}) - \log P^{(t)}(k \vert i, o, \theta^{(t)}, p^{(t)})}_{\text{Does not depend on $k$}} \notag \\
    &=\sum_{k \in K} Q_{i,o}^{(t)}(k) \log P^{(t)}(i, o, k \vert \theta^{(t)}, p^{(t)}) \notag \\
    &\,\,\,\,- \sum_{k \in K} Q_{i,o}^{(t)}(k) \log P^{(t)}(k \vert i, o, \theta^{(t)}, p^{(t)}) \notag \\
    &=\sum_{k \in K} Q_{i,o}^{(t)}(k) \log \frac{P^{(t)}(i, o, k \vert \theta^{(t)}, p^{(t)})}{Q_{i,o}^{(t)}(k)}  \notag \\
    &\,\,\,\,- \sum_{k \in K} Q_{i,o}^{(t)}(k) \log \frac{P^{(t)}(k \vert i, o, \theta, p^{(t)})}{Q_{i,o}^{(t)}(k)}
\end{align}

We note that the term in the last line of Eq.\ref{SDSBM-eq-decomp}, $- \sum_{k \in K} Q_{i,o}^{(t)}(k) \log \frac{P^{(t)}(k \vert i, o, \theta^{(t)}, p^{(t)})}{Q_{i,o}^{(t)}(k)}$, is the Kullback-Leibler (KL) divergence between $P^{(t)}$ and $Q_{i,o}^{(t)}$, noted $KL(P^{(t)} \vert \vert Q_{i,o}^{(t)})$. The KL divergence obeys $KL(P^{(t)} \vert \vert Q_{i,o}^{(t)}) \geq 0$, and is null iif $P^{(t)}$ equals $Q_{i,o}^{(t)}$. Therefore, the term in the before-last line of Eq.\ref{SDSBM-eq-decomp}, $\sum_{k \in K} Q_{i,o}^{(t)}(k) \log \frac{P^{(t)}(i, o, k \vert \theta^{(t)}, p^{(t)})}{Q_{i,o}^{(t)}(k)}$, is interpreted as a lower bound on the log-likelihood $\log P^{(t)}(i, o \vert \theta^{(t)}, p^{(t)})$. 

The aim of the E-step is to find the expression of $Q_{i,o}^{(t)}(k)$ that maximizes the lower bound of the log-likelihood with respect to the latent variables $k$. Given that the log-likelihood does not depend on $Q_{i,o}^{(t)}(k)$ and $KL(P^{(t)} \vert \vert Q_{i,o}^{(t)}) \geq 0$, the lower-bound is maximized when $KL(P^{(t)} \vert \vert Q_{i,o}^{(t)}) = 0$, which occurs when $Q_{i,o}^{(t)}(k) = P^{(t)}(k \vert i, o, \theta^{(t)}, p^{(t)})$. In this case, the lower-bound on the log-likelihood equals the likelihood itself and thus reaches a global maximum with respect to the latent variables $k$ for fixed parameters $\theta^{(t)}$ and $p^{(t)}$.

Given the definition of our simple model, the derivation of $P(k \vert i, o, \theta^{(t)}, p^{(t)})$ is straightforward. The probability of one combination of clusters $k$ among $K$ possible combinations given an input features vector and an output $o$ is proportional to $p_{k}^{(t)}(o) \theta_{i,k}$. Therefore:

\begin{align}
    P^{(t)}(k \vert i, o, \theta^{(t)}, p^{(t)}) = \frac{p_{k}^{(t)}(o) \theta_{i,k}^{(t)}}{\sum_{k' \in K} p_{k'}^{(t)}(o)\theta_{i,k'}^{(t)}}
\end{align}

which is the expression of $\omega_{i, o}^{(t)}(k)$ in the main article. 

\section{SDSBM - Explicit derivation of the M-step for p}
\label{SI-SDSBM-derivationMstep}

\begin{align}
    &\frac{\partial \left(\log P(\theta, p \vert R^{\circ}) - \sum_{k',t'} \psi_{k'}^{(t')}(\sum_{o}p_{k'}^{(t')}(o)-1)\right)}{\partial p_{k}^{(t)}(o)} = 0 \notag \\
    &\Leftrightarrow \sum_{(i,t) \in \partial o} \frac{\omega_{i,o}^{(t)}(k)}{p_{k}^{(t)}(o)} + \frac{\beta\langle p_{k}^{(t)}(o)\rangle}{p_{k}^{(t)}(o)} - \psi_{k}^{(t)} = 0 \notag \\
    &\Leftrightarrow \sum_{(i,t) \in \partial o} \omega_{i,o}^{(t)}(k) + \beta\langle p_{k}^{(t)}(o)\rangle = \psi_{k}^{(t)}p_{k}^{(t)}(o) \notag \\
    &\Leftrightarrow \sum_{(i,t) \in \partial o} \sum_o \omega_{i,o}^{(t)}(k) + \beta\underbrace{\sum_o \langle p_{k}^{(t)}(o)\rangle}_{=1 \text{}} = \psi_{k}^{(t)} \notag \\
    &\Leftrightarrow \frac{\sum_{(i,t) \in \partial o} \omega_{i,o}^{(t)}(k) + \beta\langle p_{k}^{(t)}(o)\rangle}{\sum_{(i,o,t) \in R^{\circ}} \omega_{i,o}^{(t)}(k) + \beta} = p_{k}^{(t)}(o)
\end{align}

\section{SDSBM - Using the prior in related works}
\label{SI-SDSBM-inclusion-SotA}
Throughout this section, we highlight the changes brought by our method to the EM equations derived in the mentioned papers. In summary, we see that out method allows to make these works dynamic with minimal changes of the original models.

\subsection{Bi-MMSBM}
In \citep{Antonia2016AccurateAndScalableRS}, the authors apply a MMSBM to a labeled bipartite network. The nodes on each side of the bipartite network are associated to their own membership matrix; membership of nodes $i \in I$ over $K$ clusters is encoded into $\theta \in \mathbb{R}^{I \times K}$, and membership of nodes $j \in J$ over $L$ clusters is encoded into $\eta \in \mathbb{R}^{J \times L}$. The block-interaction matrix for the label $o \in O$ is noted $p(o) \in \mathbb{R}^{K \times L}$.

Assuming a temporal version, items $i$ and $j$ to be linked by a label $o$ at time $t$ reads:
\begin{equation}
    P^{(t)}(o \vert i,j) = \sum_{k \in K}\sum_{l \in L} \theta_{i,k}^{(t)}\eta_{j,l}^{(t)}p_{k,l}^{(t)}(o)
\end{equation}

Given the same set of observation s $R^{\circ}$ as in the main article, the posterior distribution follows:
\begin{align}
\label{eq-postBiMMSBM}
    P(\theta, \eta, p \vert R^{\circ}) &= P(R^{\circ} \vert \theta, \eta, p) \\
    &\times \prod_t \left(\prod_i P(\theta_i^{(t)}) \prod_j P(\eta_j^{(t)}) \prod_{k,l} P(p_{k,l}^{(t)}) \right) \notag
\end{align}
such that:
\begin{align}
    &P(R^{\circ} \vert \theta, \eta, p) = \prod_{(i,j,t,o)\in R^{\circ}} \sum_{k \in K}\sum_{l \in L} \theta_{i,k}^{(t)}\eta_{j,l}^{(t)}p_{k,l}^{(t)}(o) \\
    &P(\theta_i^{(t)}) \propto \prod_k {\theta_{i,k}^{(t)}}^{\beta \langle \theta_{i,k}^{(t)} \rangle} \\
    &P(\eta_j^{(t)}) \propto \prod_l {\eta_{j,l}^{(t)}}^{\beta \langle \eta_{j,l}^{(t)} \rangle} \\
    &P(p_{k,l}^{(t)}) \propto \prod_o {p_{k,l}^{(t)}}(o)^{\beta \langle p_{k,l}^{(t)}(o) \rangle}
\end{align}

where $\langle x^{(t)} \rangle = \frac{\sum_{t' \neq t} \kappa(t,t') x^{(t')}}{\sum_{t' \neq t} \kappa(t,t')}$. The expectation step is not influenced by the priors choice and is the same as in \citep{Antonia2016AccurateAndScalableRS} for each temporal slice. The new maximization steps are:

\begin{align}
    &\theta_{i,k}^{(t)} = \frac{\sum_l \sum_{(o,j) \in \partial(i,t)} \omega_{i,j,o}^{(t)}(k,l) \rd{+ \beta\langle\theta_{i,k}^{(t)}\rangle}}{N_{i,t}\rd{+\beta}} \notag\\
    &\eta_{j,l}^{(t)} = \frac{\sum_k \sum_{(o,i) \in \partial(j,t)} \omega_{i,j,o}^{(t)}(k,l) \rd{+ \beta\langle\eta_{j,l}^{(t)}\rangle}}{N_{j,t}\rd{+\beta}} \notag\\
    &p_{k,l}^{(t)}(o) = \frac{\sum_{(i,j,t) \in \partial o} \omega_{i,j,o}^{(t)}(k,l) \rd{+ \beta\langle p_{k,l}^{(t)}(o)\rangle}}{\sum_{(i,j,o,t) \in R^{\circ}} \omega_{i,j,o}^{(t)}(k,l)\rd{+\beta}} \notag
\end{align}

Here again, $\beta$ is set fixed for demonstration purposes, but can be tuned at will by the user. This allows to choose the extent to which dynamics shall be smoothed, or ignored.

\subsection{T-MBM}
The T-MBM is essentially the same model as \citep{Antonia2016AccurateAndScalableRS} but with one type of entry that can appear twice in one observation. Both entries share the same membership matrix $\theta$. The probability of a label of type $o$ given entries $h$; $i$ and $j$ at time $t$ is now written:
\begin{equation}
    P(o \vert h,i,j,t) = \sum_{k \in K}\sum_{l \in L}\sum_{m \in M} \theta_{h,k}^{(t)}\theta_{i,l}^{(t)}\eta_{j,m}^{(t)}p_{k,l,m}^{(t)}(o)
\end{equation}

The posterior distribution follows the same expression as in Eq.\ref{eq-postBiMMSBM}. The expectation step is left unchanged by the choice of the priors, and the new maximization equations are given below:
\begin{align}
    &\theta_{h,k}^{(t)} = \frac{\sum_{l,m} \sum_{(o,i,j) \in \partial(h,t)} \omega_{h,i,j,o}^{(t)}(k,l,m) \rd{+ \beta\langle\theta_{h,k}^{(t)}\rangle}}{N_{h,t}\rd{+\beta}} \notag\\
    &\eta_{i,l}^{(t)} = \frac{\sum_{k,m} \sum_{(o,h,j) \in \partial(i,t)} \omega_{h,i,j,o}^{(t)}(k,l,m) \rd{+ \beta\langle\eta_{i,l}^{(t)}\rangle}}{N_{i,t}\rd{+\beta}} \notag\\
    &p_{k,l,m}^{(t)}(o) = \frac{\sum_{(h,i,j,t) \in \partial o} \omega_{h,i,j,o}^{(t)}(k,l,m) \rd{+ \beta\langle p_{k,l,m}^{(t)}(o)\rangle}}{\sum_{(h,i,j,o,t) \in R^{\circ}} \omega_{h,i,j,o}^{(t)}(k,l,m)\rd{+\beta}} \notag
\end{align}

\section{SDSBM - Inferring two dynamic matrices of parameters}
\label{SI-SDSBM-bothdynmatrices}
In the main article, $p$ is provided to the model and only $\theta$ has to be inferred. Doing so, we can confront inferred membership vectors to the ground truth while avoiding label-switching issues \citep{Matias2017DynSBM,Lee2019ReviewSBMs}. When $p$ is also to be inferred, finding a correspondence between the inferred clusters and the ground-truth is not a trivial task, and cannot be performed in unbiased ways. However, the good results yielded by the model, presented in Fig.\ref{fig-varPinferred}, when also inferring $p$ hints that the membership vectors are correctly inferred.

\begin{figure}
    \centering
    \includegraphics[width=\columnwidth]{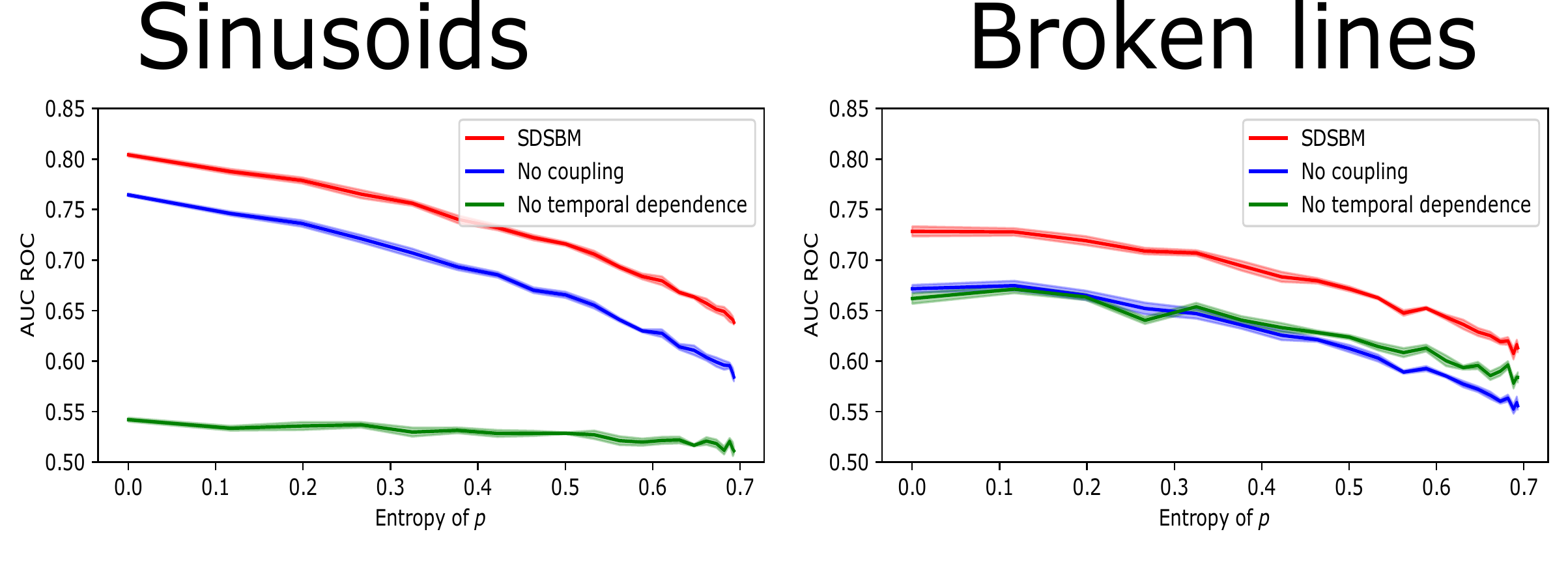}
    \caption[SDSBM - Experimental results when inferring both $\theta$ and $p$ jointly]{Experimental results when inferring both $\theta$ and $p$ jointly. The AUC-ROC is as good as when $p$ is provided to the model.}
    \label{fig-varPinferred}
\end{figure}

\section{SDSBM - Clusters composition for the Epigraphy experiment}
\label{SI-SDSBM-clusComp}
\begin{itemize}
    \item Cluster 0
    \begin{itemize}
    	\item Roma (98.0\%)
    \end{itemize}
    \item Cluster 1
    \begin{itemize}
    	\item Latium et Campania (32.0\%)
    	\item Venetia et Histria (14.0\%)
    	\item Samnium (11.0\%)
    	\item Umbria (10.0\%)
    	\item Apulia et Calabria (8.0\%)
    \end{itemize}
    \item Cluster 2
    \begin{itemize}
    	\item Pannonia superior (15.0\%)
    	\item Dalmatia (11.0\%)
    	\item Noricum (10.0\%)
    	\item Hispania citerior (7.0\%)
    	\item Gallia Narbonensis (6.0\%)
    \end{itemize}
    \item Cluster 3
    \begin{itemize}
    	\item Dacia (24.0\%)
    	\item Pannonia inferior (17.0\%)
    	\item Moesia inferior (15.0\%)
    	\item Syria (6.0\%)
    	\item Numidia (6.0\%)
    	\item Pannonia superior (5.0\%)
    \end{itemize}
    \item Cluster 4
    \begin{itemize}
    	\item Germania superior (24.0\%)
    	\item Mauretania Caesariensis (11.0\%)
    	\item Asia (11.0\%)
    	\item Etruria (11.0\%)
    	\item Galatia (9.0\%)
    \end{itemize}
\end{itemize}

\clearpage

%% file: Tables/Chapter_2/table-res-SIMSBM-SM.tex
\begin{table}[h]
	\centering
	\caption[SIMSBM - Experimental results on real-world datasets]{Replication results on two datasets used in \citep{Antonia2016AccurateAndScalableRS} and \citep{Poux2021MMSBMMrBanks}, referenced in the main text. The standard error on the last digits over all 100 runs is indicated in standard notation -- $0.123(12) \Leftrightarrow 0.123 \pm 0.012$. Overall, we retrieve the same results as those presented in \citep{Antonia2016AccurateAndScalableRS} and \citep{Poux2021MMSBMMrBanks}. The models presented in this chapter are \underline{underlined}.}
	\noindent\makebox[\textwidth]{\resizebox{\textwidth}{!}{
	\begin{tabular}{|l|l|l|S|S|S|S|S|S|S}

		\cline{1-9}
		& & & \text{F1} & \text{P@1} & \text{AUCROC} & \text{AUCPR} & \text{RAP} & \text{NCE} \\ 

		\cline{1-9}
		\multirow{5}{*}{\rotatebox[origin=c]{90}{\footnotesize \text{\textbf{Imdb}}}}

		& \multirow{5}{*}{\rotatebox[origin=c]{90}{\footnotesize \text{\textbf{User, Movie}}}}
		& \underline{SIMSBM(1,1)} & \maxf{ \num{ 0.3995 +- 0.0002 } } & \maxf{ \num{ 0.3558 +- 0.0003 } } & \maxf{ \num{ 0.7665 +- 0.0001 } } & \maxf{ \num{ 0.3406 +- 0.0003 } } & \maxf{ \num{ 0.5805 +- 0.0002 } } & \maxf{ \num{ 0.1593 +- 0.0001 } } \\ 
		& & TF &  \num{ 0.2570 } &  \num{ 0.2348 } &  \num{ 0.5031 } &  \num{ 0.1541 } &  \num{ 0.4627 } &  \num{ 0.2573 } \\ 
		& & KNN &  \num{ 0.2668 } &  \num{ 0.2002 } &  \num{ 0.5558 } &  \num{ 0.1735 } &  \num{ 0.3308 } &  \num{ 0.4834 } \\ 
		& & NB &  \num{ 0.2585 } &  \num{ 0.2382 } &  \num{ 0.5377 } &  \num{ 0.1660 } &  \num{ 0.4664 } &  \num{ 0.2536 } \\ 
		& & BL &  \num{ 0.2570 } &  \num{ 0.2349 } &  \num{ 0.5000 } &  \num{ 0.1525 } &  \num{ 0.4647 } &  \num{ 0.2557 } \\ 

		\cline{1-9}
		
		\multirow{5}{*}{\rotatebox[origin=c]{90}{\footnotesize \text{\textbf{MrBanks}}}}

		& \multirow{5}{*}{\rotatebox[origin=c]{90}{\footnotesize \text{\textbf{Ply, Full sit}}}}
		& \underline{SIMSBM(1,1)} & \maxf{ \num{ 0.7126 +- 0.0002 } } & \maxf{ \num{ 0.6688 +- 0.0004 } } & \maxf{ \num{ 0.7126 +- 0.0003 } } & \maxf{ \num{ 0.7180 +- 0.0004 } } & \maxf{ \num{ 0.8344 +- 0.0002 } } & \maxf{ \num{ 0.1656 +- 0.0002 } } \\ 
		& & TF &  \num{ 0.6795 } &  \num{ 0.6037 } &  \num{ 0.5176 } &  \num{ 0.5363 } &  \num{ 0.8019 } &  \num{ 0.1981 } \\ 
		& & KNN &  \num{ 0.6940 } &  \num{ 0.6433 } &  \num{ 0.6668 } &  \num{ 0.6430 } &  \num{ 0.8217 } &  \num{ 0.1783 } \\ 
		& & NB &  \num{ 0.6795 } &  \num{ 0.6037 } &  \num{ 0.5907 } &  \num{ 0.5822 } &  \num{ 0.8019 } &  \num{ 0.1981 } \\ 
		& & BL &  \num{ 0.6795 } &  \num{ 0.6037 } &  \num{ 0.5000 } &  \num{ 0.5215 } &  \num{ 0.8019 } &  \num{ 0.1981 } \\ 

		\cline{1-9}
	\end{tabular}
	}}
\end{table}

%% file: Appendices/Appendix-3.tex
\chapter{Appendix -- Temporal diffusion networks}

\label{Appendix3}

\section{Implementation of Clash of the Contagions}
\label{InterRate-implemCoC}
In this appendix, we provide technical details on the way the Clash of Contagions baseline is implemented. Following the directions given in the reference article \citep{Myers2012CoC}, we implemented a Stochastic Gradient Descent (SGD) method for parameters inference. Given the small number of entities considered in the experiments, each iteration of the SGD is computed using the full dataset instead of slicing it into mini-batches. 

\subsection{Setup}
For each corpus, we run the SGD algorithm 100 times, from which we save the parameters maximizing the likelihood the most. At the beginning of each run, parameters M and $\Delta$ are randomly initialized. The stopping condition makes the algorithm ends when the relative variation of the likelihood according to the last iteration is been lesser than $10^{-6}$ for more than 30 times in a row; those numbers have been chosen empirically to maximize the performances of the algorithm. The hyper-parameters have been set to: T=5 (number of clusters) and K=20 (number of considered time steps).

\subsection{Update rule}
In each iteration, we update the parameters in the direction of the gradient descent (noted G). However, a major problem when dealing with SGD is to choose the line step length $\eta$ (the amplitude of the variation of the parameters in the direction of the gradient G). After each iteration, we compare several update rules, and we select the one maximizing the likelihood. Those rules are as follows:
\begin{itemize}
\item AdaGrad: $\eta^{AG} \times G$
\item AdaDelta: $\eta^{AD} \times G$
\item Line search in the direction of the gradient: $\eta^{LS} \times G$
\item Line search in the direction of AdaDelta: $\eta^{LS} \times \eta^{AD} \times G$
\end{itemize}
The line search snippets consider 50 values of $\eta^{LS}$ logarithmically distributed in the interval $[ 10^{-6} ; 10^3 ]$. 

\subsection{Constraints on the parameters}
The membership vectors entries $M_{i,t}$ (membership of i to cluster t) must be positive and sum to 1 over all the clusters ($\sum_{t} M_{i,t} = 1$). In order to enforce this constraint, we consider the following variable change: $M_{i,t} \rightarrow \frac{\phi_{i,t}^2}{\sum_{t'} \phi_{i,t'}^2}$. This transformation guarantees the membership vector properties with no need for penalty methods in the implementation.

Besides, as stated in \citep{Myers2012CoC}, it can happen that a probability is larger then 1 or lesser than 0. In the absence of complementary details in the main paper, we implemented our own method to force the probabilities to stay within reasonable bounds. Here it is impossible to make a simple variable change to enforce this constraint, since the probability results of a non-linear combination of the model's parameters. We added to the likelihood an exponential penalty term. Let P denote a quantity we want to constrain between 0 and 1. The penalty term equals $e^{-\lambda P} + e^{\lambda(P-1)}$. $\lambda$ here is a parameter that tunes the intensity of the penalty, and is empirically set to $\lambda=75$. This penalty function has the form of a well with very steep walls in x=0 and x=1. In this way, it seldom happens that probabilities are larger than 1 or lesser than 0, as said in the main article. When such cases happen, we simply set it back to the closest bound for methods comparisons.

\clearpage

%% file: Appendices/Appendix-5.tex
\chapter{Appendix -- Dirichlet-Survival Process}

\label{Appendix5}

\section{Deriving NetRate}
The input of NetRate is a collection of observed cascades $C=\{ \Vec{c} \}_{\Vec{c}=\Vec{c_1},\Vec{c_2},...}$. Each cascade is a collection of events with timestamps $\Vec{c} = \{ (u_i^c, t_i^c) \}_i$, where $u_i^c$ is the node on which the $i^{th}$ event occurred and $t_i^c$ the time at which it happened in cascade $c$. Using survival analysis, \citep{GomezRodriguez2011NetRate} denotes the likelihood of an infection of node $u_i^c$ at time $t_i^c$ in cascade $c$ by any other node $u_j^c$ previously infected at time $t_j^c$ in the same cascade as $f(t_i^c \vert t_j^c, \alpha_{u_j^c,u_i^c})$, where $\alpha_{u_j^c,u_i^c}$ is an entry of the objective network's adjacency matrix. In survival analysis' framework, $f(t_i^c \vert t_j^c, \alpha_{u_j^c,u_i^c})$ is linked to the instantaneous infection rate (or hazard rate) of $u_i^c$ at time $t_i^c$ by $u_j^c$ previously infected at time $t_j^c$, noted $\lambda (t_i^c \vert t_j^c, \alpha_{u_j^c,u_i^c})$, and the probability of non-infection of $u_i^c$ up to time $t_i^c$ by $u_j^c$ previously infected at time $t_j^c$, noted $S(t_i^c \vert t_j^c, \alpha_{u_j^c,u_i^c})$, by the following relation:
\begin{equation}
    \label{eq-survival}
    f(t_i^c \vert t_j^c, \alpha_{u_j^c,u_i^c}) = \lambda (t_i^c \vert t_j^c, \alpha_{u_j^c,u_i^c}) S(t_i^c \vert t_j^c, \alpha_{u_j^c,u_i^c})
\end{equation}

Within a cascade, the likelihood that a node get infected by only one neighbour $u_j^c$ previously infected at time $t_j^c$ can be written as the likelihood of infection at time $t_i^c$ by $u_j^c$ previously infected at time $t_j^c$ times its probability of survival to every previously infected node:
\begin{equation}
    \label{eq-lik-link}
    f(t_i^c \vert t_j^c, \alpha_{u_j^c,u_i^c}) \prod_{k \neq j, t_k^c<t_i^c} S(t_i^c \vert t_k^c, \alpha_{u_k^c,u_i^c})
\end{equation}

The likelihood of an infection by any neighbour then becomes the sum of those candidate disjoint events:
\begin{equation}
    \label{eq-infection}
    \begin{split}
          &\sum_{j, t_j^c<t_i^c} f(t_i^c \vert t_j^c, \alpha_{u_j^c,u_i^c}) \prod_{k \neq j, t_k^c<t_i^c} S(t_i^c \vert t_k^c, \alpha_{u_k^c,u_i^c})\\
          \stackrel{\text{Eq.\ref{eq-survival}}}{=} &\sum_{j, t_j^c<t_i^c} \lambda (t_i^c \vert t_j^c, \alpha_{u_j^c,u_i^c}) \prod_{t_k^c<t_i^c} S(t_i^c \vert t_k^c, \alpha_{u_k^c,u_i^c})
    \end{split}
\end{equation}

The likelihood of a cascade then becomes the product of the likelihood of every event it contains, and the total likelihood $\mathcal{L}(C \vert A)$ of every cascade the product over every cascade. Let $A$ be the network's adjacency matrix whose entries $\Vec{\alpha_{i,j}}$ are directed edges from $i$ to $j$. Then:
\begin{equation}
    \label{eq-cascades}
    \mathcal{L}(C \vert A) = \prod_{\Vec{c} \in C} \prod_{t_i^c \in \Vec{c}} \sum_{j, t_j^c<t_i^c} \lambda (t_i^c \vert t_j^c, \alpha_{u_j^c,u_i^c}) \prod_{t_k^c<t_i^c} S(t_i^c \vert t_k^c, \alpha_{u_k^c,u_i^c})
\end{equation}

\cleardoublepage